\documentclass[12pt,twoside]{thuthesis}

\usepackage[dvips]{graphicx} 
\usepackage{amsmath,amssymb,amsthm,mathrsfs,amsfonts,dsfont}
\usepackage{enumerate}
\usepackage{epsfig}
\usepackage{subfigure}
\usepackage{xcolor}
\usepackage{braket}
\usepackage{multicol}
\usepackage{multirow}
\usepackage{framed,color}
\usepackage{bm}

\newtheorem{lemma}{Lemma}

\newcommand{\Pro}[1]{\mathrm{P(#1)}}

\newcommand{\tr}{\mathrm{Tr}}

\newcommand{\rhoin}{\rho^{\mathrm{(in)}}}
\newcommand{\rhoout}{\rho^{\mathrm{(out)}}}
\newcommand{\rhoctc}{\rho_{\mathrm{CTC}}}

\def\ie{i.~e.~}

\def\eg{e.~g.~}

\newcommand{\st}[1]{\mathbf{#1}}

\newcommand{\map}[1]{\mathcal{#1}}
\newcommand{\spc}[1]{\mathcal{#1}}
\newcommand{\trg}{\operatorname{Tr}}
\newtheorem{axiom}{Axiom}
\begin{document}

 \pagenumbering{Roman}


%
%
%


\title{Interplay between Quantumness, Randomness, and Selftesting}

\author{Xiao Yuan}
\department{Institute for Interdisciplinary Information Sciences}

\major{Physics}

\degree{Doctor of Philosophy}

\degreemonth{December}
\degreeyear{2016}

\thesisdate{December 20, 2016}


\supervisor{Xiongfeng Ma}{Assistant Professor}
 
\maketitle




\cleardoublepage
\begin{abstractpage}

{Quantum information processing shows advantages in many tasks, including quantum communication and computation, comparing to its classical counterpart. The essence of quantum processing lies on the fundamental difference between classical and quantum states. For a physical system, the coherent superposition on a computational basis is different from the statistical mixture of states in the same basis. Such coherent superposition endows the possibility of generating true random numbers, realizing parallel computing, and other classically impossible tasks such as quantum Bernoulli factory. Considering a system that consists of multiple parts, the coherent superposition that exists nonlocally on different systems is called entanglement. By properly manipulating entanglement, it is possible to realize computation and simulation tasks that are intractable with classical means.

Investigating quantumness, coherent superposition, and entanglement can shed light on the original of quantum advantages and lead to the design of new quantum protocols. This thesis mainly focuses on the interplay between quantumness and two information tasks, randomness generation and selftesting quantum information processing. We discuss how quantumness can be used to generate randomness and show that randomness can in turn be used to quantify quantumness. In addition, we introduce the Bernoulli factory problem and present the quantum advantage with only coherence in both theory and experiment. Furthermore, we show a method to witness entanglement that is independent of the realization of the measurement. We also investigate randomness requirements in selftesting tasks and propose a random number generation scheme that is independent of the randomness source.

By investigating the interplay between quantumness and the two information tasks, we have investigated the essence of quantumness and its fundamental role in quantum information processing. Apart from the theoretical significance, the results can be experimentally tested and applied in practice.}


\end{abstractpage}

\cleardoublepage
\begin{center}
{\LARGE{\textbf{Acknowledgements}}}
\end{center}

The research presented in this Doctor of Philosophy thesis is carried out under the the supervision of Professor Xiongfeng Ma at the Institute for Interdisciplinary Information Sciences at Tsinghua University, China. Being a wonderful mentor in research and an intimate friend in life, it is Xiongfeng who, for the first time, makes me feel the pleasure of doing research.  I acknowledge him for the inspiring instruction, discussion, and encouragement and for sharing his extensive knowledge. In the mean time, I would like to thank Professor Giulio Chiribella for his guidance during the first year of my Ph.D. study. With his altruistic help, I learned the basics of quantum information and completed my first few researches.  In addition, I acknowledge Professor Mile Gu for sharing his insightful thoughts in quantum correlation and relativistic quantum information. Under his guidance, I broadened my knowledge of quantum science and completed an interesting work of quantum information in the presence of time-traveling. 

During my Ph.D. study, I am lucky to have many chances to visit several wonderful academic institutes. My special thanks go to the hosts for their kindly help and valuable discussions. In chronological order of the visiting places, the hosts include Professor Yu-Ao Chen and Jian-Wei Pan at the University of Science and Technology of China, Professor Hoi Fung Chau at the University of Hong Kong, Professor Peter Zoller at Austria University of Innsbruck, Professor Anton Zeilinger at the University of Vienna, Professor Renato Renner at ETH, Professor Nicolas Gisin at the Universit\'e de Gen\`eve, Professor Romain All\'eaume at Telecom ParisTech, Professor Qiang Zhang at the University of Science and Technology of China, Dr. Graeme Smith, Dr. John Smolin and Dr. Charles H. Bennett at IBM's Thomas J. Watson Research Center, Dr. Qingyu Cai at the Chinese Academy of Sciences Wuhan Physics and Mathematics Institute, and Professor Yeong-Cherng Liang at the National Cheng Kung University, Tainan. Especially, I would like to thank Dr. Charles H. Bennett for enthusiastically showing his `antique' of the first quantum key distribution experiment instrument and introducing the quantum side channel. 

My works have been guided by many brilliant minds, including Assad, Syed M; Chen, Luo-Kan; Chen, Tengyun; Cao, Zhu; Fan, Jinyun; Girolami, Davide; Haw, Jing Yan; Huang, Miao; Jiang, Xiao; Lam, Ping Koy; Li, Li; Liu, Ke; Li, Wei; Li, Zheng-Da; Liu, Nai-Le; Liu, Yang; Lu, Chao-Yang; Lu, He; Lutkenhaus Norbert; Ma, Yuwei; Mei, Quanxin; Pan, Jian-Wei; Peng, Cheng-Zhi; Qi, Bing; Ralph, Timothy C; Thompson, Jayne; Vijay, R; Vedral, Vlatko; Weedbrook, Christian; Wang, Weiting; Yan, Zhaopeng; Yao, Xing-Can; Xu, Yuan; Xu, Ping; Zhang, Fang; Zhang, Yan-Bao; Zhang, Zhen; Zhou, Hongyi; Zhou, Shan. I acknowledge my collaborators for sharing their knowledge and offering their help. 

Last but not least, I want to sincerely thank my parents for their selfless love and thoughtful care in every details of my life. I am also very grateful to the students and professors in our institute. I would like to thank my friends for bring pleasures in my life.

\tableofcontents
\newpage
\listoffigures
\newpage
\listoftables

\pagenumbering{arabic}
\pagestyle{plain}
\part{Introduction and Preliminaries}
\chapter{Introduction}
Manipulation of quantum information empowers many tasks such as communication \cite{bb84,tele,Ekert1991}, computation \cite{Shor97, Grover96}, and simulation \cite{Feynman1982,lloyd1996universal}. In a communication task, the quantum key distribution protocol by Charles Bennett and Gilles Brassard in 1984 (BB84)\cite{bb84} makes it possible to extend secret keys between two remote parties by transmitting quantum signals. Such a task has been proven impossible by classical methods. In quantum computing, the Shor algorithm named after its inventor Peter Shor \cite{Shor97} factorizes integers \footnote{given an integer $N$, find its prime factors} at exponentially faster speeds than any existing classical methods. 
In quantum simulation \cite{Georgescu14}, one can efficiently simulate quantum systems that requires exponential resources with a classical computer.

To investigate the origin of the quantum information processing power, we need to find the major difference between the manipulation of quantum and classical information.
Focusing on states of physical systems, the major distinction derives from the quantum features or \emph{quantumness}. In different scenarios, the quantumness manifests differently. For instance, considering the whole system, coherent superpositions on a computational basis inherently differ from classical (stochastic) mixtures of the basis states. Named quantum \emph{coherence}, such coherent superpositions underlies the quantumness of a single quantum system \cite{Baumgratz14}. In bipartite systems, the quantumness can also be defined as the nonlocal correlation between the two systems. Considering local operation and classical communication (LOCC) as free or classical operations, quantum \emph{entanglement} is the major quantumness in bipartite states \cite{Bennett96, Vedral98, Horodecki09}.

One important research direction is the quantumness of states, which aims to analyze quantumness in a systematic way. Quantum coherence and entanglement can be quantified by resource frameworks. In general, a quantumness resource framework relies on identifying classical states and classical operations. The corresponding quantumness of an operational task emerges when a quantum behavior cannot be explained by classical means. A state is called \emph{classical} when it exhibits no quantum behavior. Denote the set of classical states by $\mathcal{C}$, then a state that does not belong to $\mathcal{C}$ is called \emph{quantum}. Based on classical states, classical operations are physically realizable operations that cannot generate quantum states from any classical state.
With classical states and operations, a resource framework of quantumness is completed by defining measures, which is a real-valued functions of states. Generally, a quantumness measure should satisfy the monotonicity requirement: that is, classical operations cannot increase the quantumness of a system.

From a mathematical viewpoint, finding legitimate quantumness measures is important for completing the quantumness resource framework. Moreover, these quantities should have meaningful interpretation in detailed operational tasks. For instance, the \emph{entanglement of formation} measures the amount of maximally entangled states on average that are required to prepare the target state in the asymptotical scenario; the \emph{distillable entanglement} measures the amount of maximally entangled states on average that can be obtained via LOCC operations on the target states in the asymptotical scenario.

This thesis focuses on the operational interpretation of quantumness measures. We  consider two operational tasks: randomness processing and selftesting quantum information processing and investigate the key roles of quantumness in both tasks. We also probe the interplay between randomness and selftesting quantum information tasks.

\subsubsection{Quantumness and randomness}
In classical theory, all physical processes are deterministic due to basic Newton's laws. In contrast, Born's rule \cite{born1926quantentheorie} endows the quantum world with true randomness. Such is the counter-intuitiveness of the result that Einstein was quoted as saying `God is not playing at dice'. Nevertheless, the intrinsically random nature of  measurement outcomes is now considered a key characteristic that distinguishes quantum mechanics from classical theory \cite{bell}.

In measurement theory, decoherence, breaking coherence or superposition, in a specific (classical) computational basis results in random outcomes \cite{Zurek09}. Intuitively, from the resource perspective, the randomness can be generated by consuming coherence of a quantum state. In order to quantitatively establish this connection, one needs to find a proper way to assess the randomness of measurement, which normally contains quantum and classical processes. The superficially random outcomes in classical processes are generally not truly random, although they might appear so if information is ignored. Thus, such classical part of randomness should be precluded when quantifying a quantum feature --- coherence. A quantum process, on the other hand, can generate genuine unpredictable randomness, which we call intrinsic (quantum) randomness. Observing such intrinsic random outcomes of measurements would indicate non-classical (quantum) features of objects.

As an example, we can consider the famous Schr\"odinger's cat gedanken experiment as shown in Fig.~\ref{Fig:cat}. In a classical world, a cat might be either alive or dead before observation, which can be described by the density matrix $\rho_{\mathrm{cat}}^\mathrm{C} = (\ket{\mathrm{alive}}\bra{\mathrm{alive}} + \ket{\mathrm{dead}}\bra{\mathrm{dead}})/2$ for the case of being alive and dead equally likely, see Fig.~\ref{Fig:cat}~(a). The observation result of whether the cat is alive or dead looks random, which is due to the lack of knowledge of the cat system. After considering some hidden variables or an ancillary system $E$ that purifies $\rho_{\mathrm{cat}}^\mathrm{C}$, $\ket{\Psi} = \left(\ket{\mathrm{alive}}\ket{0}_E + \ket{\mathrm{dead}}\ket{1}_E\right)/\sqrt{2}$, we can simply observe the system $E$ to infer whether the cat is alive or dead. In quantum mechanics, the cat can be in a coherent superposition of the states of alive and dead, $\rho_{\mathrm{cat}}^\mathrm{Q} = \ket{\psi}\bra{\psi}$, where $\ket{\psi} = \left(\ket{\mathrm{alive}} + \ket{\mathrm{dead}}\right)/\sqrt{2}$, see Fig.~\ref{Fig:cat}~(b). The observation outcome would be intrinsically random according to Born's rule. That is, without directly accessing the system of the cat and breaking the coherence, we can never predict whether the cat is alive or dead better than blindly guessing.
Therefore, the existence of intrinsic randomness can be regarded as a witness for quantum coherence.

\begin{figure}[hbt]
\centering \resizebox{8cm}{!}{\includegraphics{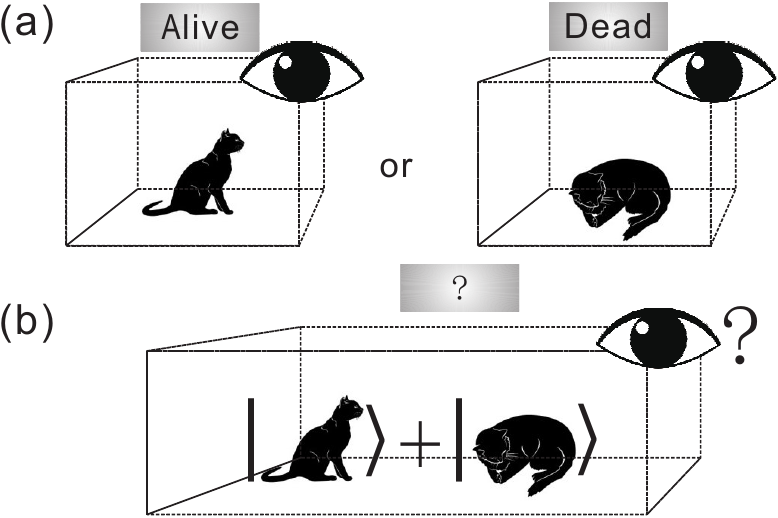}}
\caption{An illustration of Schr\"odinger's cat gedanken experiment.} \label{Fig:cat}
\end{figure}

This thesis investigates the relation between intrinsic randomness and quantum coherence \cite{Yuan15Coherence,yuan2016interplay}. We show that the amount of intrinsic randomness arising from state measurements in a basis indicates the amount of coherence in the same basis. Therefore, we can naturally regard coherence as the resource for generating randomness. In addition, we examine a simple yet interesting randomness processing task, called the Bernoulli factory \cite{Asmussen92, RSA:RSA20333} and investigate how coherence can be used to beat classical method. In the discussed randomness tasks, we reveal that coherence is an important resource for both randomness generation and processing.

\subsubsection{Quantumness and selftesting}
The concept of selftesting is unique in quantum information processing. A selftesting or device independent protocol can maintain its property even with untrusted devices that do not assume the physical implementations \cite{mayers1998quantum, acin06, Brunner14}. Considering the process in Fig.~\ref{Fig:Bell12}(a), a general picture for a single party process involves classical or quantum random input $x$ and output $a$. Without assuming the physical implementation of the box that transforms inputs into outputs, what we observe in practice is the probability distribution of $p(a|x)$. Then, a selftesting protocol is to ensure its property only based on the probability distribution of $p(a|x)$. In the bipartite scenario, Fig.~\ref{Fig:Bell12}(b), we can further impose other requirements on the two parties. For instance, in the Bell test, we generally assume no-signaling between the two parties. Given the inputs $x, y$, outputs $a, b$, and the probability distribution $p(a,b|x,y)$, the Bell test has proven remarkable power in many tasks.

\begin{figure}[hbt]
\centering \resizebox{8cm}{!}{\includegraphics{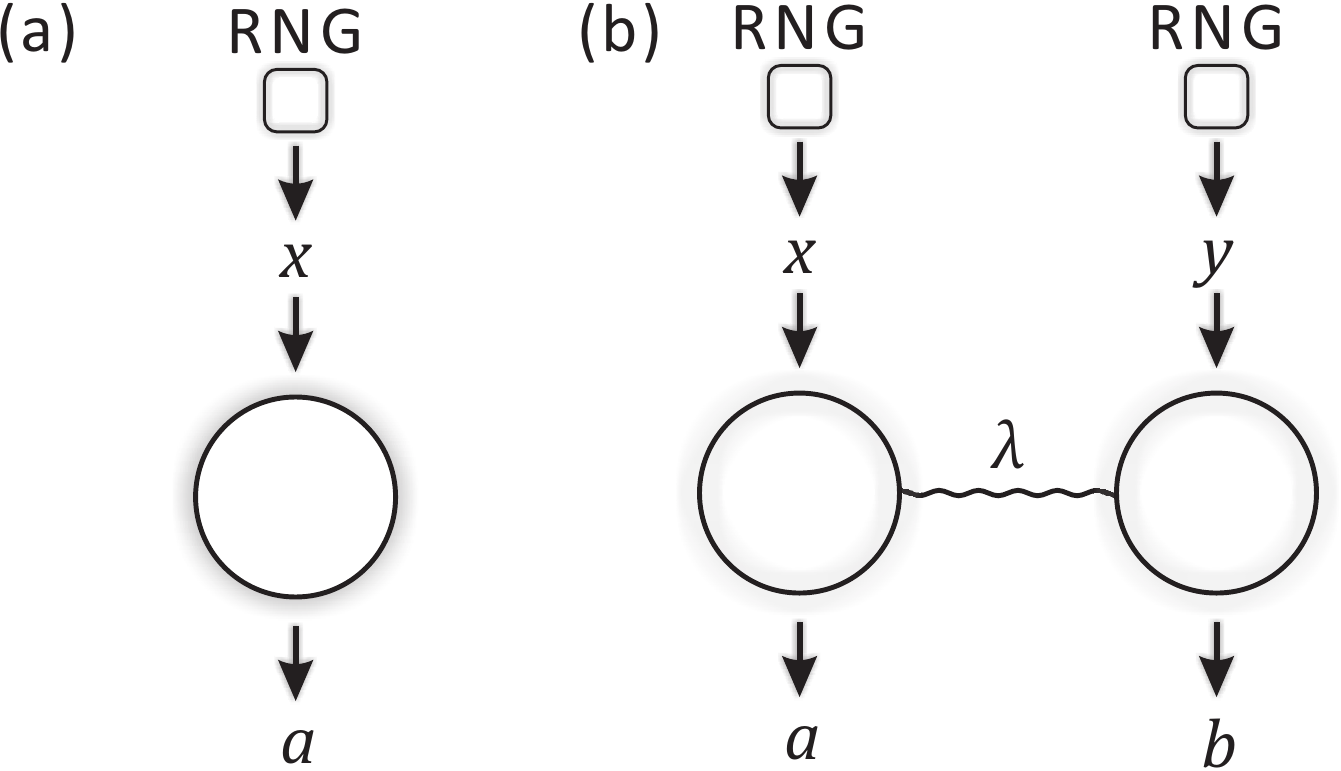}}
\caption{Device independent processing of one (a) and two (b) parties.} \label{Fig:Bell12}
\end{figure}

Here, we take randomness generation as an example. True random numbers can be generated when measuring a coherent state on its basis. Such randomness, however, assumes physical implementation (i.e., the state and its measurement). We instead consider whether random numbers can be generated without the physical implementations. The single device in Fig.~\ref{Fig:Bell12}(a) cannot be used to device independently generate randomness because one can always choose a predetermined sequence that satisfies the probability distribution $p(a|x)$. As the sequence is predetermined, no randomness is generated, thus, it appears that randomness cannot be generated in a selftesting way.

Now, we consider the two-devices case in Fig.~\ref{Fig:Bell12}(b). Following a similar argument, both parties can produce predetermined probability distributions $p(a|x)$ and $p(b|y)$ to simulate the observed probability distribution $p(a,b|x,y)$. Although the two parties cannot communicate, they can still share a predetermined strategy $\lambda$ that is independent of the inputs. The probability distribution of a given $\lambda$ is $p(a|x,\lambda)p(b|y,\lambda)$. On average, the probability distribution the two parties can simulate with predetermined strategy is
\begin{equation}\label{Eq:LHVMIntroduction}
  p(a,b|x,y)_c = \sum_\lambda p(\lambda)p(a|x,\lambda)p(b|y,\lambda),
\end{equation}
where $p(\lambda)$ is the probability distribution of $\lambda$. Such a strategy is called local hidden variable models \cite{EPR35}. We then question whether the probability distribution Eq.~\eqref{Eq:LHVMIntroduction} covers all possible probability distributions $p(a,b|x,y)$. If the answer is yes, then the observed probability distribution $p(a,b|x,y)$ cannot certify randomness; if no, the process generates true random numbers. Remarkably, there exist probability distributions $p(a,b|x,y)$ that emerge from measuring entangled states cannot be simulated as  Eq.~\eqref{Eq:LHVMIntroduction}. In practice, when observing a probability distribution that cannot be simulated, the randomness in the output $x,y$ is assured without assuming any physical realizations.

Although a selftesting protocol does not assume physical realizations, the quantumness remains the crucial ingredient for demonstrating quantum advantages. If there is no quantumness, the process becomes classical and cannot be verified device independently. Conversely, quantumness can be witnessed by selftesting protocols. For instance, entanglement is revealed by violations of Bell inequalities. This thesis investigates the relation between quantumness and selftesting \cite{Branciard13,Yuan14,PhysRevA.93.042317,Zhao16}. Specifically, we investigate the witnessing of general multipartite entanglement via a selftesting protocols with quantum states as inputs.

\subsubsection{Randomness and selftesting}
A selftesting protocol requires genuine random input. If the input is also predetermined, the device can always simulate any probability distribution by a predetermined strategy.
The randomness (freewill) loophole refers to the underlying assumption in Bell tests that different measurement settings can be chosen randomly (freely). Generally, a Bell test requires the input of each party to be fully random in order to avoid information leakage between different parties. If there is a local hidden variable that shares information about the random inputs, where in the worst scenario, the inputs are all predetermined such that each party knows exactly the input of the other party, it is possible to violate Bell inequalities just with local hidden variable model (LHVM) strategies \cite{EPR35}. Since one can always argue that there might exist a powerful creator who determines everything including all the Bell test experiments, this loophole is widely believed to be impossible to close perfectly. In this case, as we cannot prove or disprove the existence of true randomness, the assumption of freewill is indispensable in general Bell tests.

Yet, it is still meaningful to discuss the randomness requirement\footnote{The imperfect input randomness requirement is sometimes called measurement dependence in literature.} of Bell tests in a practical scenario. In this thesis, we suppose that the randomness generation devices are partially controlled by an adversary Eve, who thus possesses certain knowledge of the inputs of Alice and Bob \cite{Yuan2015CHSH,Yuan15CH}. Then Eve can make use of the information about the inputs to fake violations of Bell's inequalities \cite{Hall10} and hence lead to the device independent tasks insecure. Therefore, it is interesting to see how much of randomness needed for a Bell test in order to ensure the correctness of the conclusion. This is especially meaningful when considering a loophole free Bell test \cite{Larsson14, Kofler14} and its applications to practical tasks in the presence of an eavesdropper.

In summary, this thesis investigates three main topics--- quantumness, randomness, and selftesting, and the interplays among these features. 
Following a brief introduction to quantum information theory, we discuss these three features and their interplays from different perspectives. We also investigate two studies on quantum  information in general relativity and axioms and discuss the basic principles of quantum information theory and their extension to a general physical theory.

\chapter{Basics of quantum mechanics}
This chapter briefly covers the basics of quantum mechanics. We focus on the formalism of quantum information formalism, which involves the density matrix and the positive observable valued measures (POVMs). Due to the length limit, we only present the results that are used in this thesis. For a more detailed introduction of quantum information, please see Ref.~\cite{nielsen, wilde2013quantum, watrous2011theory}.

\section{Quantum mechanics formalism---pure states and projective measurements}
In this part, we review the Dirac notation of quantum mechanics and its equivalent form of vectors.
\subsubsection{Pure states}
Unlike classical mechanics, quantum states can be expressed as a superposition of different bases. Following the Dirac bra ket representation, a pure quantum state can be denoted as
\begin{equation}\label{}
  \ket{\psi} = \sum_{i=1}^d\alpha_i\ket{i},
\end{equation}
where the set of $\mathcal{I} = \{\ket{i}\}$ represents a state basis, such as the polarization of the photon or the energy levels of an atom. Suppose the dimension of the $\mathcal{I}$ basis is $d$, then we can regard the state space as a $d$-dimensional Hilbert space $\mathcal{H}_d$ and quantum states as vectors.
Suppose $\mathcal{I} = \{\ket{i}\}$ forms an orthogonal basis, then quantum states can be equivalently denoted as
\begin{equation}\label{}
  {\psi} = \left(\begin{array}{c}
           \alpha_1 \\
           \vdots \\
           \alpha_d \\
         \end{array}\right).
\end{equation}

\subsubsection{Measurements}
In the Dirac representation, a projective measurement can be denoted in bra form as,
\begin{equation}\label{}
  \bra{M} = \sum_{j=1}^d\beta_j\bra{j}.
\end{equation}
The measurement probability is given by
\begin{equation}\label{}
  p = |\bra{M}\psi\rangle|^2 = \left|\sum_{j=1}^d\beta_j\bra{j}\sum_{i=1}^d\alpha_i\ket{i}\right|^2 = \left|\sum_{i, j=1}^d \alpha_i\beta_j\bra{j}i\rangle\right|^2.
\end{equation}
Suppose $\mathcal{I} = \{\ket{i}\}$ forms an orthogonal basis, then the probability is given by
\begin{equation}\label{}
  p = \left|\sum_{i}^d \alpha_i\beta_i\right|^2.
\end{equation}

Similar to the vector representation, a projective measurement can be denoted as a dual vector as
\begin{equation}\label{}
  {M} = (\beta_1, \dots, \beta_d),
\end{equation}
and the probability of measuring $\ket{\psi}$ is given by the square of the inner product
\begin{equation}\label{}
  p =  \left|(\beta_1, \dots, \beta_d)\left(\begin{array}{c}
           \alpha_1 \\
           \vdots \\
           \alpha_d \\
         \end{array}\right)\right|^2 = \left|\sum_{i}^d \alpha_i\beta_i\right|^2.
\end{equation}

\subsubsection{Observables}
In quantum mechanics, an observable $O$ is a Hermitian operator that satisfies $\bra{\psi}O\ket{\phi} = \overline{\bra{\phi}O\ket{\psi}}$ for two arbitrary states $\ket{\psi}$ and $\ket{\phi}$. Here $\overline{\bra{\phi}O\ket{\psi}}$ is the complex conjugate of $\bra{\phi}O\ket{\psi}$. Therefore, the average of $O$ for state $\ket{\psi}$ is given by the real value
\begin{equation}\label{}
  \bar{O} = \bra{\psi}O\ket{\psi}.
\end{equation}
When $\ket{\psi}$ is denoted as a vector and $\bra{\psi}$ as a dual vector, $O$ can be denoted as a Hermitian matrix $O_{ij} = \bra{i}O\ket{j}$. In this case, the average value can be given by matrix multiplication $\bar{O} = \sum_{i,j}\psi_{i}^*O_{ij}\psi_{j}$.

Any Hermitian operator  $O$ has a \emph{spectral decomposition},
\begin{equation}\label{}
  O = \sum_{i}\lambda_i\ket{o_i}\bra{o_i},
\end{equation}
where $\{\ket{o_i}\}$ forms an orthogonal basis, so the average value of $O$ can also be regarded as a projective measurement on the $\{\ket{o_i}\}$ basis. The probability of the $i$th outcome is $p_i = |\bra{o_i}\ket{\psi}|^2$ and the average is
\begin{equation}\label{}
  O = \sum_i \lambda_ip_i.
\end{equation}

\subsubsection{Evolution}
In quantum mechanics, the evolution of a quantum state is determined by the Schr$\mathrm{\ddot{o}}$dinger equation. Given the Hamiltonian ${H}$ of the system, we have
\begin{equation}\label{}
  i\hbar\frac{\partial \ket{\psi(t)}}{\partial t} = {H}\ket{\psi(t)}.
\end{equation}
Considering a closed system where energy is conserved, the Hamiltonian ${H}$ is independent of time and the state can be determined by
\begin{equation}\label{}
  \ket{\psi(t)} = U(t, t_0) \ket{\psi(t_0)},
\end{equation}
where $U(t, t_0) = \int_{t_0}^{t}e^{-i {H}t/\hbar}\mathrm{d}t$ gives the evolution of the state and $\ket{\psi(t_0)}$ is the state at time $t_0$.

In quantum mechanics, the Hamiltonian ${H}$ is Hermitian. In the vector representation, ${H}$ corresponds to a Hermitian matrix that satisfies ${H}^{\dag} = {H}$. Here $\dag$ denotes the hermitian conjugate of ${H}$. Because
\begin{equation}\label{}
  {H} = \sum_{i,j}\bra{i}{H}\ket{j}\ket{i}\bra{j},
\end{equation}
the matrix representation of ${H}$ is given by $H_{i,j}=\bra{i}{H}\ket{j}$.
In this case, the evolution operator $U(t, t_0)$ can be regarded as a unitary operator that satisfies
\begin{equation}\label{}
  U(t, t_0)*U(t_0, t) = I,
\end{equation}
where $U(t_0, t) = \int_{t}^{t_0}e^{-i {H}t/\hbar}\mathrm{d}t$.

To summary, when considering a $d$-dimensional system, pure quantum states, projective measurements, observables and state evolution can be represented as vectors, dual vectors, Hermitian operators, and unitary operators, respectively. Because the vector representation is equivalent to the Dirac bra ket representation, we use them interchangeably.

\section{Composite systems and subsystems}
\subsubsection{Composite system}
We now consider two systems $A$ and $B$ that are defined in Hilbert space $\mathcal{H}_{d_A}$ and $\mathcal{H}_{d_B}$, respectively. Similarly, a pure quantum state on system $AB$ can be represented as
\begin{equation}\label{}
  \ket{\psi}_{AB} = \sum_{i_A,i_B}\alpha_{i_Ai_B}\ket{i_A}_A\ket{i_B}_B,
\end{equation}
where $\mathcal{I}_A = \{\ket{i_A}_A\}$ and $\mathcal{I}_B = \{\ket{i_B}_B\}$ forms orthogonal bases for systems $A$ and $B$, respectively. Equivalently, in vector form, we have
\begin{equation}\label{}
  \ket{\psi}_{AB} =\left(\begin{array}{c}
           \alpha_{11} \\
           \alpha_{12} \\
           \vdots \\
           \alpha_{1d_B} \\
           \vdots \\
           \alpha_{d_A1} \\
           \alpha_{d_A2} \\
           \vdots \\
           \alpha_{d_Ad_B} \\
         \end{array}\right).
\end{equation}

Projective measurements, observables and evolution can be similarly defined. Next, we move to the density matrix description of states for subsystems. Before, we first redefine the pure state formalism with states and measurement denoted as matrices.

\subsubsection{Subsystems}
In the Dirac notation, quantum states and measurements are given by $\ket{\psi}$ and $\bra{M}$ and the measurement probability is given by $|\bra{M}\psi\rangle|^2$, respectively. Equivalently, we can denote quantum state as $\rho = \ket{\psi}\bra{\psi}$ and a measurement as $P = \ket{M}\bra{M}$. Then the measurement probability is
\begin{equation}\label{}
  p = \tr[\rho P] = |\bra{M}\psi\rangle|^2,
\end{equation}
where $\tr$ is the trace operation. Thus, quantum states and measurements can also be denoted as matrices.

Consider now a projective measurement $P_A = \ket{M}_A\bra{M}_A$ on system $A$ of $\ket{\psi}_{AB}$ and the probability of the measurement. Because systems $A$ and $B$ are correlated, we cannot directly calculate the probability of solely measuring system $A$ from the state. In this case, we consider that a projective measurement on the $\mathcal{I}_B$ basis is also applied on system $B$. The probability of projecting onto $\ket{M}_A\bra{M}_A$ and $\ket{i_B}_B\bra{i_B}_B$ is calculated by
\begin{equation}\label{}
  p(M_A, i_B) = \tr[(\ket{\psi}_{AB}\bra{\psi}_{AB})(\ket{M}_A\bra{M}_A\otimes\ket{i_B}_B\bra{i_B}_B)],
\end{equation}
and the probability of projecting system $A$ onto $P_A$ is
\begin{equation}\label{}
\begin{aligned}
  p(M_A) &= \sum_{i_B}p(M_A, i_B) \\
  &= \sum_{i_B}\tr[(\ket{\psi}_{AB}\bra{\psi}_{AB})(\ket{M}_A\bra{M}_A\otimes\ket{i_B}_B\bra{i_B}_B)]\\
  & = \tr\left[\left(\ket{\psi}_{AB}\bra{\psi}_{AB}\right)\left(\ket{M}_A\bra{M}_A\otimes\sum_{i_B}\ket{i_B}_B\bra{i_B}_B\right)\right]\\
  & = \tr[\rho_A P_A],
\end{aligned}
\end{equation}
where
\begin{equation}\label{}
  \rho_A = \tr_{B}[\left(\ket{\psi}_{AB}\bra{\psi}_{AB}\right)],
\end{equation}
with $\tr_{B}$ being the trace over system $B$ only.

Therefore, when measuring a subsystem, one can equivalently describe the state with a \emph{density matrix} by tracing out the other systems. In general, it is easy to verify that a density matrix $\rho$ should have the following properties:
\begin{itemize}
  \item $\rho$ is a Hermitian operator.
  \item $\rho$ is nonnegative. That is, its spectral values are nonnegative.
  \item $\tr[\rho] = 1$.
\end{itemize}

\subsubsection{Qubit systems}
A two-level system is referred to as a quantum bit or qubit system. Its density matrix is a $2 \times 2$ Hermitian matrix. The Pauli matrices
\begin{equation}\label{}
\sigma_x = \left(
           \begin{array}{cc}
             0 & 1 \\
             1 & 0 \\
           \end{array}
         \right),
\sigma_y = \left(
           \begin{array}{cc}
             0 & -i \\
             i & 0 \\
           \end{array}
         \right),
\sigma_z = \left(
           \begin{array}{cc}
             1 & 0 \\
             0 & -1 \\
           \end{array}
         \right),
\end{equation}
together with the identity matrix $I$ form a basis for $2\times2$ Hermitian matrices. That is, given $e_0 = I$, $e_1 = \sigma_x$,  $e_2 = \sigma_y$, $e_3 = \sigma_z$, and the inner product to be the inner product of matrices, we can verify that
\begin{equation}\label{}
  \langle{e_i,e_j}\rangle = \tr[e_i^\dag e_j] = 2\delta_{i,j}, \forall i,j\in\{1,2,3,4\}.
\end{equation}

For any single qubit state $\rho$, its \emph{Bloch sphere} representation is
\begin{equation}\label{}
  \rho = \frac{I + n_x\sigma_x + n_y\sigma_y + n_z\sigma_z}{2}.
\end{equation}
Here, $n_x, n_y, n_z$ are real values and $n_x^2 + n_y^2 + n_z^2 \le 1$. It is easy to verify that
\begin{equation}\label{}
  n_i = \tr[\rho \sigma_i], \forall i = x, y ,z
\end{equation}
Two qubit states $\rho_{AB}$ can also be decomposed in the Pauli matrices basis,
\begin{equation}\label{}
  \rho_{AB} = \sum_{i,j}\lambda_{i,j}e_i\otimes e_j.
\end{equation}
Here, the coefficients are determined by the average value
\begin{equation}\label{}
  \lambda_{i,j} = \tr[\rho_{AB}(e_i\otimes e_j)]/4.
\end{equation}

\subsubsection{Purification and Schmidt decomposition}
For any state $\rho$ with a spectral decomposition of $\rho = \sum_ip_i\ket{\psi_i}\bra{\psi_i}$, one can always find its purification. That is, we can find a pure bipartite state $\ket{\psi}_{AB}$ such that 
\begin{equation}\label{}
  \ket{\psi}_{AB} = \sum_{i}\sqrt{p_i}\ket{\psi_i}\ket{i},
\end{equation}
such that $\rho = \tr_B[\ket{\psi}_{AB}\bra{\psi}_{AB}]$.

For any pure state $\ket{\psi}_{AB}$ of a bipartite system,  orthonormal bases $\{\ket{i}_A\}$ and $\{\ket{i¡®}_B\}$ exist such that:
\begin{equation} \label{Schmidt decomposition}
\begin{aligned}
\ket{\psi}_{AB}=\sum_i\sqrt{p_i}\ket{i}_A\otimes\ket{i'}_B.
\end{aligned}
\end{equation}
The subsystems $A$ and $B$ have the same eigenvalues, $p_i$s and the number of $p_i$s is called the Schmidt number of $\ket{\psi}_{AB}$. The pure state is an entangled state when the Schmidt number is greater than one. It is easy to see that the Bell states are entangled states.

\subsubsection{Positive observable valued measures}
When performing local measurements on a joint state, the state can be denoted as a density matric to simplify the calculation. We now consider performing a joint measurement and want to know the measurement probability of a local system. Suppose a joint measurement $P_{AB} = \ket{M}_{AB}\bra{M}_{AB}$ is performed on $\rho_{A}\otimes\rho_{B}$, then the probability distribution is 
\begin{equation}\label{}
  p = \tr[P_{AB}\rho_{A}\otimes\rho_{B}] = \tr[M_{A}\rho_{A}].
\end{equation}
Here, we first trace out system $B$ to get
\begin{equation}\label{}
  M_A = \tr_B[\ket{M}_{AB}\bra{M}_{AB}(I\otimes\rho_B)].
\end{equation}
In this case, when focusing only on system $A$, the effective measurement performed on system $A$ is $M_A$. Therefore, general measurements are defined by positive observable valued measures (PVOMs). That is, $\{M_i\ge0, \forall i\}$ and $\sum_i M_i = I$.

\subsubsection{Entropy of quantum states}
For any quantum state $\rho$, the definition of function $f$ acting on $\rho$ is given by acting on its spectral decomposition:
\begin{equation}\label{}
  f(\rho) = \sum_if(\lambda_i)\ket{i}\bra{i}
\end{equation}
where $\rho = \sum_i\lambda_i\ket{i}\bra{i}$ and $\{\ket{i}\}$ forms an orthogonal basis.

The entropy of quantum state $\rho$ is defined as
\begin{equation}\label{}
  S(\rho) \equiv -\tr[\rho\log\rho],
\end{equation}
or equivalently
\begin{equation}\label{}
  S(\rho) = -\sum_i\lambda_i\log\lambda_i.
\end{equation}
Here, $0\log0\equiv0$.

The relative entropy of quantum state is 
\begin{equation}\label{}
  S(\rho||\sigma) \equiv \tr[\rho\log\rho] - \tr[\rho\log\sigma].
\end{equation}
We also know that $S(\rho||\sigma)\ge0$, where the equality holds if and only if $\rho = \sigma$.

For a bipartite quantum state $\rho_{AB}$, the conditional entropy and mutual information is defined by
\begin{equation}\label{}
\begin{aligned}
S(A|B) &\equiv S(AB) - S(B)\\
S(B|A) &\equiv S(AB) - S(A)\\
S(A:B) &\equiv S(A) +S(B) - S(AB).
\end{aligned}
\end{equation}
When $\rho_{AB} = \sum_i p_i\ket{i}\bra{i}\otimes \rho_i$, we also have
\begin{equation}\label{}
  S(\rho_{AB}) = H(p_i) + \sum_i p_iS(\rho_i),
\end{equation}
where $H(p_i) = -\sum_ip_i\log p_i$.

\chapter{Quantumness, selftesting, and randomness}
This chapter introduces the basic concepts of quantumness, selftesting, and randomness. In Sec.~\ref{Sec:quantumness}, I introduce the quantumness of states, including the coherence of a single quantum system and the entanglement of multipartite systems. Sec.~\ref{Sec:selftesting} introduces the Bell nonlocality test, which is the basic tool for selftesting protocols. Finally, Sec.~\ref{Sec:randomness} reviews the development of quantum random number generation.

\section{Quantumness}\label{Sec:quantumness}
This section introduces the basics of quantumness of states. For a single quantum system, we focus on its coherence on the computational basis. We refer to Ref.~\cite{2016arXiv160902439S, marvian2016quantify} for recent reviews on this subject. For multipartite systems, we mainly focus on entanglement correlations. Nice reviews on this subject are available in Ref.~\cite{Plenio2007Measures,Horodecki09,guhne2009}.
\subsection{Quantum coherence}
As a key feature of quantum mechanics,  coherence measures the superposition power on the computational basis and is often considered as a basic ingredient for quantum technologies \cite{giovannetti2011advances, lambert2013quantum}. Considerable effort has been undertaken to theoretically formulate the quantum coherence \cite{Glauber63, Sudarshan, luo2005quantum, Aberg06, monras2013witnessing,Baumgratz14,Aberg14,Girolami14}. Recently, a comprehensive framework of coherence quantification has been established \cite{Baumgratz14}, by which coherence is considered to be a resource that can be characterized, quantified and manipulated in a manner similar to that of another important feature--- quantum entanglement \cite{Bennett96, Vedral98, Plenio2007Measures, Horodecki09}.
Here, we focus on the resource framework of coherence.

In a general $d$-dimensional Hilbert space and a computational basis ${I}= \{\ket{i}\}_{i=1,2,\dots,d}$, coherence measures its superposition power on the basis. Note that, any state that can be represented by a diagonal state of ${I}$, that is,
\begin{equation}\label{Eq:sigma}
  \delta = \sum_{i=1}^d p_i\ket{i}\bra{i},
\end{equation}
has no superposition. Thus, such state is called incoherent (classical) state and the set of such state is denoted by $\mathcal{I}$. Conversely, a maximally coherent state is given by the maximal superposition state
\begin{equation}\label{Eq:Psid}
\ket{\Psi_d} = \frac{1}{\sqrt{d}}\sum_{i = 1}^d \ket{i},
\end{equation}
up to arbitrary relative phases between the components $\ket{i}$.

When considering coherence as a resource, incoherent states are thus ``useless'' or ``free'' states. If we make an analogy with the theory of thermodynamics, incoherent state would be similar to thermal states from which no energy can be extracted. In thermodynamics, thermal operations are generally considered as free operations. Applying a thermal operations on thermal state results in a thermal state. In the same spirit, one can define a ``free'' or incoherent operation to be the operation that transforms incoherent state only to incoherent state. That is, incoherent operations are defined by incoherent completely positive trace preserving (ICPTP) maps $\Phi_{\mathrm{ICPTP}}(\rho) = \sum_n K_n\rho K_n^\dag$, where the Kraus operators $\{K_n\}$ satisfy $\sum_n K_n K_n^\dag = I$ and  $K_n \mathcal{I} K_n^\dag \subset \mathcal{I}$. In the case, where post-selections are enabled, the output state corresponding to the $n$th Kraus operation is given by $  \rho_n = {K_n\rho K_n^\dag}/{p_n}$, where $p_n = \mathrm{Tr}\left[ K_n\rho K_n^\dag\right]$ is the probability of obtaining the outcome $n$.

Given the definition of incoherent states and incoherent operations, we can measure the amount of coherence.
Generally, a measure of coherence is a map $C$ from quantum state $\rho$ to a real non-negative number that satisfies the properties listed in Table~\ref{Fig:Properties}.
\begin{table}[htb]
\begin{framed}
\centering
\begin{enumerate}[(C1)]
\item
Coherence vanishes for all incoherent states. That is, $C(\delta) = 0$, for all $\delta\in \mathcal{I}$. A stronger requirement claims that (C1') $C(\delta) = 0$, iff $\delta\in \mathcal{I}$.
\item
\emph{Monotonicity}: Coherence should not increase under incoherent operations. Thus, (C2a) $C(\rho) \geq C(\Phi_{\mathrm{ICPTP}}(\rho))$, and (C2b) $C(\rho) \geq \sum_n p_nC(\rho_n)$, where (C2b) is for the case where post-selection is enabled.
\item
\emph{Convexity}: Coherence cannot increase under mixing states, $\sum_e p_eC(\rho_e) \geq C(\sum_e p_e\rho_e)$.
\end{enumerate}
\end{framed}
\caption{Properties that a coherence measure should satisfy.} \label{Fig:Properties}
\end{table}

There are various measures for coherence. Considering the distance measure of two quantum states, the measure of coherence may be defined as the minimum distance from $\rho$ to all incoherent
states in $\mathcal{I}$. Two examples \cite{Baumgratz14} are now presented.

\emph{Relative entropy:}
Here the relative entropy is used as the distance measure.
\begin{equation}\label{eq:relativeentropy}
C_{rel.ent.,J}(\rho)=\min_{\sigma_J \subset \mathcal{I}} S_J(\rho||\sigma_J)=S_J(\rho_{diag})-S_J(\rho)
\end{equation}
where $\rho_{diag}$ only contains diagonal elements of $\rho$.

\emph{$l_1$ norm:}
Another distance measure is a function of the off-diagonal elements of the quantum state. The simplest form is the $l_1$ norm, which is
given by
\begin{equation}\label{eq:l1norm}
C_{l_1,J}(\rho)=\min_{\sigma_J \subset \mathcal{I}}||\rho-\sigma_J||_{l_1}=\sum_{i,j,i\neq j}|\rho_{ij}|
\end{equation}

Besides the distance measures, coherence can be defined in other ways.

\emph{Convex roof:}
According to \cite{Yuan15Coherence}, the intrinsic randomness is also a measure of coherence, therefore we have
\begin{equation}\label{eq:convexroof}
C_J(\rho)=R_J(\rho)=\min_{p_e,\ket{\psi_e}}\sum_e p_e R_J(\ket{\psi_e})
\end{equation}
where $\rho=\sum_e p_e \ket{\psi_e} \bra{\psi_e}$ and $\sum_e p_e=1$, and the minimization runs over all possible decompositions.

\subsection{Quantum entanglement}
\subsubsection{Entanglement framework}
Quantum entanglement describes the nonlocal correlation between different systems. For instance, considering in the bipartite scenario, any product state
\begin{equation}\label{}
  \sigma_{AB} = \rho_A\otimes\rho_B
\end{equation}
has no nonclassical correlation. In addition, a mixture of classical states should also be classical state. Thus, we say that a state is separable when it can be written as
\begin{equation}\label{}
  \sigma_{AB} = \sum_{i}p_i\rho_A^i\otimes\rho_B^i,
\end{equation}
where $p_i\ge0,\forall i, \sum_{i}p_i=1$.

Similar to the framework of coherence, we also need to define classical operations for entanglement. In the same spirit of incoherent operations that do not create coherence from incoherent states, a separable operation is defined such that no entangled state can be generated from separable state. In practice, the operation of local operation and classical communication (LOCC) draws more attention because it has operational meanings. In this case, as a strict subset of separable operations, LOCC are generally referred as the ``free'' operation for entanglement.

Given the definitions of separable states and LOCC operations, we propose measures for entanglement that have the following properties.
\begin{table}[htb]
\begin{framed}
\centering
\begin{enumerate}[(E1)]
\item
$E(\rho^{AB})$ vanishes when $\rho^{AB}$ is separable.
\item
\emph{Monotonicity}: $E(\rho^{AB})$ cannot increase under an LOCC operation, that is, (E2a) $E[\Lambda^{LOCC}(\rho^{AB})]\leq E(\rho^{AB})$. This condition is often replaced by a stronger condition. (E2b) $E(\rho^{AB})$ should not increase on average under LOCC that maps $\rho^{AB}$ to $\rho^{AB}_k$ with probability $p_k$, then $\sum_k p_k E(\rho^{AB}_k)\leq E(\rho^{AB})$.
\item
\emph{Convexity}: $E(\rho^{AB})$ decreases under mixing, $E(\sum_k p_k \rho^{AB}_k) \leq \sum_k p_k E(\rho^{AB}_k)$.
\item
$E(\rho^{AB})$ is invariant under all local unitary operations, that is,  $E(\rho^{AB})=E(U_A \otimes U_B \rho^{AB} U_A^\dag \otimes U_B^\dag)$.
\end{enumerate}
\end{framed}
\caption{Properties that an entanglement measure should satisfy.} \label{Fig:entanglementProperties}
\end{table}

Two widely adopted measures are the relative entropy of entanglement and the entanglement of formation.

\emph{1. Relative entropy of entanglement:}
\begin{equation}\label{eq:relativeentropy3}
E_{rel.ent.}(\rho^{AB})=\min_{\sigma^{AB}} S(\rho^{AB}||\sigma^{AB})
\end{equation}
where the minimization runs over all separable states $\sigma^{AB}$.

\emph{2. Entanglement of formation:}
\begin{equation}\label{eq:relativeentropy3}
E_{EOF}(\rho^{AB})=\min_{p_e,\ket{\psi_e^{AB}}}\sum_e p_e E_{EOF}(\ket{\psi_e^{AB}})
\end{equation}
where the minimization runs over all possible decomposition of $\rho^{AB} = \sum_{e}p_e\ket{\psi_e^{AB}}\bra{\psi_e^{AB}}$ and $E_{EOF}(\ket{\psi_e^{AB}}) = S(\rho^A)$ with $\rho^A$ being the density matrix of system $A$.

\subsubsection{Entanglement witness}
Quantum entanglement plays an important role in the nonclassical phenomena of quantum mechanics. Being the key resource for many tasks in quantum information processing, such as quantum computation \cite{wiesner}, quantum teleportation \cite{tele} and quantum cryptography \cite{bb84,Ekert1991}, entanglement needs to be verified in many scenarios. There are several proposals to witness entanglement and we refer to Ref. \cite{guhne2009} for a detailed review.

A conventional way to detect entanglement, entanglement witness (EW), gives one of two outcomes: `Yes' or `No', corresponding to conclusive result that the state is entangled or fail to draw a conclusion, respectively.  Mathematically, for a given entangled quantum state $\rho$, an Hermitian operator $W$ is called a witness, if $tr[W\rho] < 0$ (output of `Yes') and $tr[W\sigma]\geq0$ (output of `No') for any separable state $\sigma$. Note that there could also exist entangled state $\rho'$ such that $tr[W\rho']\geq0$ (output of `No'). The EW method is shown schematically in Fig.~\ref{Fig:EW0}.
\begin{figure}[hbt]
\centering
\includegraphics[width=8cm]{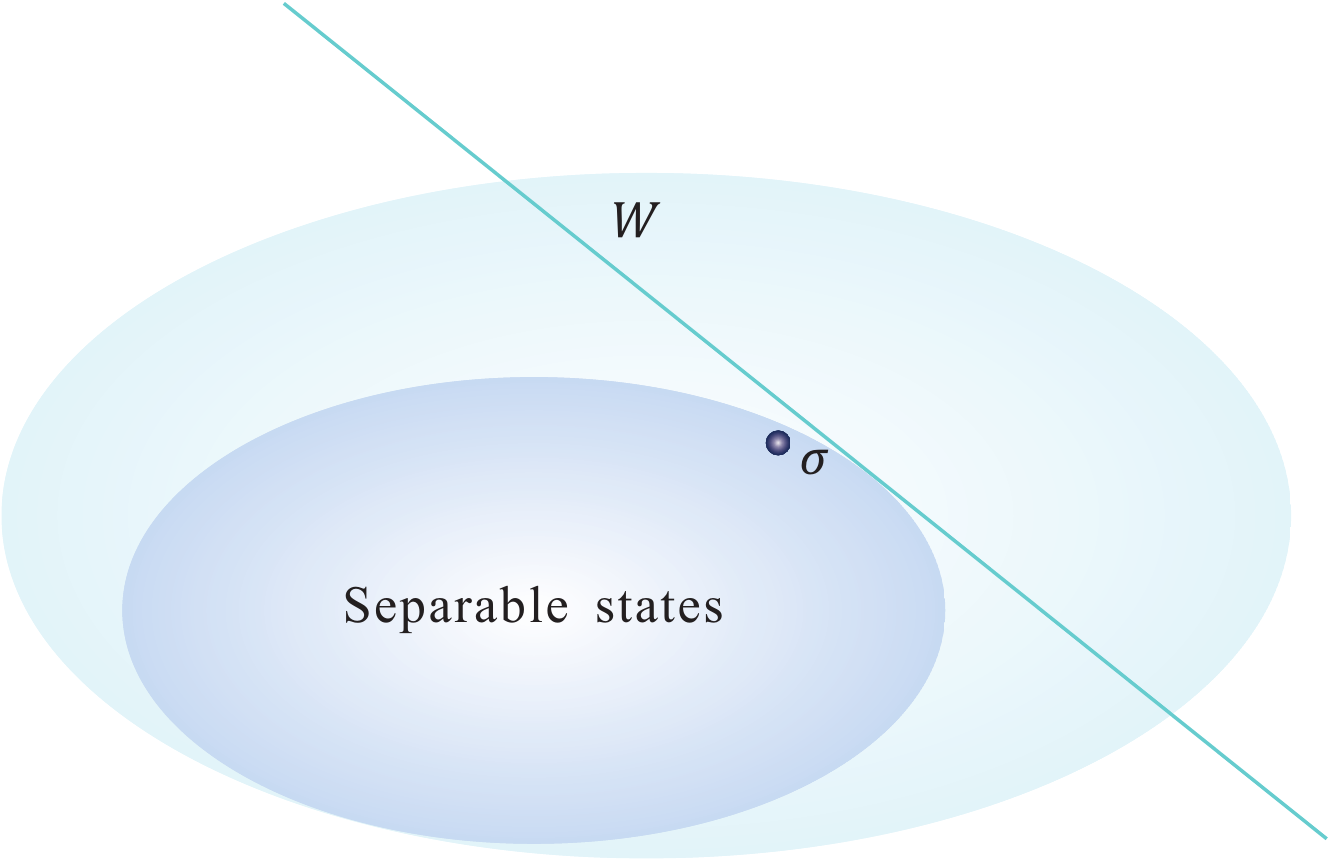}
\caption{Witnessing entanglement via entanglement witness operators.} \label{Fig:EW0}
\end{figure}

In an experimental verification, one can realize the conventional EW with only local measurements by decomposing $W$ into a linear combination of product Hermitian observables \cite{guhne2009}. For example, we can consider a Werner state and the EW
\begin{equation}\label{rhovab}
   \rho^v_{AB}=(1-v)|\Psi^-\rangle\langle\Psi^-|+\frac v4I,
\end{equation}
with $v\in[0,1]$ and $|\Psi^-\rangle=(|01\rangle-|10\rangle)/\sqrt{2}$ and $I$ being the identity matrix. The state is entangled if $v<1/3$, which can be witnessed by the EW,
\begin{equation}\label{wit}
W=\frac{1}{2}I-|\Psi^-\rangle\langle\Psi^-|,
\end{equation}
and its result, $\tr[W\rho^v_{AB}] = (3v-1)/4$. When considering local measurements of Pauli operators, it is easy to verify that
\begin{equation}\label{}
  W = \frac{1}{4}(I\otimes I + \sigma_x\otimes\sigma_x + \sigma_y\otimes\sigma_y + \sigma_z\otimes\sigma_z).
\end{equation}
In experiment, one simply measures local observables of $I\otimes I$, $\sigma_x\otimes\sigma_x$, $\sigma_y\otimes\sigma_y$, $\sigma_z\otimes\sigma_z$ and take the average to get the estimation of $W$.

\section{Selftesting: Bell nonlocality test}\label{Sec:selftesting}
The basic idea of selftesting quantum information processing is to guarantee the quantum advantage with only the observed statistics instead of the implementation device. The key ingredient for fully selftesting is based on the violation of Bell inequalities. Bell test \cite{bell} is motivated to rule out local hidden variable models (LHVMs) \cite{EPR35}. The faithful violation of a Bell inequality assures that the underlying physical process cannot be explained with LHVMs. In quantum information processing, violations of Bell's inequalities are powerful tools that enable device independent tasks, such as quantum key distribution \cite{Mayers98, acin06,masanes2011secure,Vazirani14}, randomness amplification \cite{colbeck2012free, gallego2013full, Dhara14} and generation \cite{Colbeck11, Fehr13, Pironio13, Vazirani12}, entanglement quantification \cite{Moroder13}, and dimension witness \cite{Brunner08}. In this section, we introduce the background of Bell nonlocality test and leave its application to randomness generation in next section. We also leave the discussion of semi-selftesting quantum information in Part III.

\subsection{Clauser-Horne-Shimony-Holt inequality}
One of the best-known Bell inequalities is the Clauser-Horne-Shimony-Holt (CHSH) inequality \cite{CHSH}, which may be expressed in many ways. We study it from a quantum game point of view.

As shown in Fig.~\ref{Fig:BellTest}, two space-like separated parties, Alice and Bob, choose input bit settings $x$ and $y$ at random and output bits $a$ and $b$ based on their inputs and pre-shared quantum ($\rho$) and classical ($\lambda$) resources, respectively. The probability distribution $p(a,b|x,y)$,  obtaining outputs $a$ and $b$ conditioned on inputs $x$ and $y$, is determined by specific strategies of Alice and Bob. By assuming that the input settings $x$ and $y$ are chosen fully randomly and equally likely, the CHSH inequality is defined by a linear combination of the probability distribution $p(a,b|x,y)$ according to
\begin{equation}\label{eq:Bell}
  S = \sum_{a,b,x,y} (-1)^{a\oplus b + x\cdot y}p(a,b|x,y) \leq S_C = 2,
\end{equation}
where the plus operation $\oplus$ is modulo 2, $\cdot$ is numerical multiplication, and $S_C$ is the (classical) bound of the Bell value $S$ for all LHVMs.

\begin{figure}[hbt]
\centering \resizebox{5cm}{!}{\includegraphics{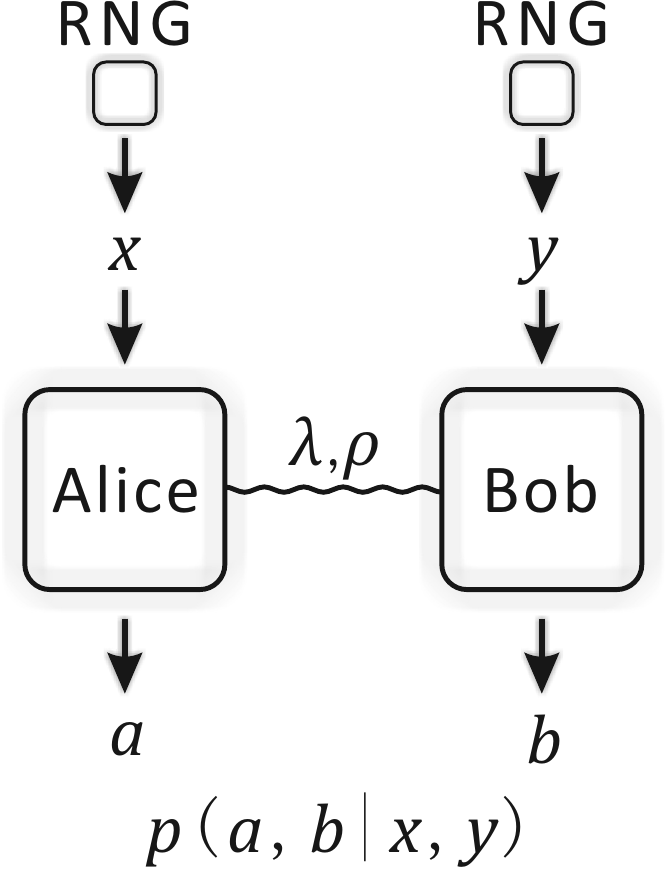}}
\caption{Bipartite Bell inequality. The inputs $x$ and $y$ of Alice and Bob are determined by perfect random number generators (RNGs), which produce uniformly distributed random numbers.} \label{Fig:BellTest}
\end{figure}

An achievable bound for the quantum theory is $S_Q = 2\sqrt{2}$  \cite{cirel1980quantum}. In this case, a violation of the classical bound $S_C$ indicates the need for alternative theories other than LHVMs, such as quantum theory. For general no signalling (NS) theories \cite{prbox}, we denote the corresponding upper bound as $S_{NS} = 4$. It is straightforward to see that $S_{NS} \geq S_Q\geq S_C$.

Different strategies impose different constraints on the probability distribution.
\begin{itemize}
  \item Classical: $p(a,b|x,y) = \sum_\lambda q(\lambda)p(a|x,\lambda)p(b|y,\lambda)$
  \item Quantum: $p(a,b|x,y) = \mathrm{Tr}[\rho_{AB}M_a^x\otimes M_b^y]$
  \item No-signaling: $\sum_a p(a,b|x,y) = \sum_a p(a,b|x',y), \sum_b p(a,b|x,y) = \sum_b p(a,b|x, y')$
\end{itemize}

\subsection{Practical loopholes}
In practice, the conclusion of the violation of a Bell test depends on several assumptions. 
Experimental demonstrations suffer from three major loopholes; a faithful Bell test should close all such loopholes.

\emph{Locality loophole:}
The measurement events of Alice and Bob should be space-like separated. If this condition is not satisfied, Bell's inequality can be violated by signaling even with LHVMs. This loophole can be closed by separating Alice and Bob sufficiently far apart such that the measurement events become space-like separated. In experiment, this loophole is closed in optical systems \cite{Weihs98} and appears nearly closed in atomic systems \cite{hofmann2012heralded}.

\emph{Efficiency loophole:}
The detection efficiency must be greater than a threshold to ensure violation of Bell inequalities without assuming fair sampling. The famous Clauser-Horne (CH) or Eberhard \cite{Eberhard93} test show that the efficiency should be at least 2/3 for each party, which is also proven to be a tight bound \cite{Massar03, Wilms08} for all bipartite Bell tests with two inputs. The efficiency loophole has been closed in different realizations \cite{rowe2001experimental,Christensen13,giustina2013bell}.

\emph{Randomness loophole:}
The inputs $x$ and $y$ should be random and thus cannot be predetermined. Also, we require $x$ and $y$ to be uncorrelated with each other and also come from different runs \cite{Koh12, Pope13}. In experiment, this loophole cannot be closed perfectly, as we can never unconditionally certify the randomness without a faithful Bell test, which in turn requires faithful randomness. Thus, we have to assume the existence of a true random seed. In practice, we can use independent RNGs, such as causally disconnected cosmic photons \cite{Gallicchio14}. Conversely, if we can well characterize the randomness, we can also check whether the input randomness satisfies the requirement \cite{Hall10,Koh12,Koh12,Pope13,Yuan15Coherence,Yuan15CH} that guarantees the conclusion even with imperfectly randomness input.

In experiment, the conclusion of a Bell test is not faithful unless these three major loopholes are closed.
In addition, we must address several technical issues that may also invalidate the Bell test conclusion or make a violation impossible.

\emph{Coincidence-time problem:} If the local detection time depends on the measurement settings, a coincidence time loophole \cite{larsson2004bell} may exist.
  This loophole can be solved by distinguishing each coincidence detection event such that it does not depend on the measurement settings.

\emph{Imperfect devices:} The experiment devices cannot be perfect, which will affect the result of a Bell test.
  \begin{itemize}
    \item Source: The input photon source will differ from the desired source due to practical imperfections. For instance, the photon source may contain multiple photon pairs, which will affect the fidelity of the prepared state.
    \item Dark count: The measurement of the state will be affected by dark counts from environment. A Bell violation will be observed only if the dark count is below a certain threshold.
    \item Misalignment error: In experiment, the measurement may contain misalignment errors that output opposite result. Misalignment error should also be below a certain threshold to guarantee a Bell violation.
  \end{itemize}

\emph{Finite statistics, memory problem:}
In the most general scenario, the measurement devices of Alice and Bob contain a memory such that the outputs of the current run can be conditioned on the inputs and outputs of previous runs \cite{Barrett02}. In this case, we cannot directly obtain the probability distribution. This loophole can only be closed by considering statistics test of a Bell inequality with correlated strategy. Most previous experiments \cite{Christensen13,giustina2013bell} consider asymptotic condition and assume the data to be independent and identically distributed (i.i.d.).

\emph{Nonuniform random inputs:}
Nonuniform random inputs do not affect the CH inequality, which is defined by a linear combination of probability distributions. However, Eberhard's inequality, which is used in practice, should be normalized when the input random bits are not uniform. In this case, we have to consider finite statistics with nonuniform random inputs.

\section{Randomness generation and quantification}\label{Sec:randomness}
In this section, we review the developments of quantum random number generators \cite{Ma2016QRNG}.
\subsection{Randomness generation}
Random numbers play essential roles in many fields, such as, cryptography \cite{Shannon1949OTP}, scientific simulations \cite{metropolis1949monte}, lotteries, and fundamental physics tests \cite{bell1964einstein}. These tasks rely on the unpredictability of random numbers, which generally cannot be guaranteed in classical processes. In computer science, random number generators (RNGs) are based on pseudo-random number generation algorithms \cite{knuth2014art}, which deterministically expand a random seed. Although the output sequences are usually perfectly balanced between 0s and 1s, a strong long-range correlation exists, which can undermine cryptographic security, cause unexpected errors in scientific simulations, or open loopholes in fundamental physics tests \cite{Kofler06,Hall10,Yuan2015CHSH}.

Many researchers have attempted to certify randomness solely based on the observed random sequences. In the 1950s, Kolmogorov developed the \emph{Kolmogorov complexity} concept to quantify the randomness in a certain string \cite{Kolmogorov1998387}. A RNG output sequence appears random if it has a high Kolmogorov complexity. Later, many other statistical tests \cite{marsaglia1996diehard,rukhin2001statistical,kim2004corrections} were developed to examine randomness in the RNG outputs. However, testing a RNG from its outputs can never prevent a malicious RNG from outputting a predetermined string that passes all of these statistical tests. Therefore, true randomness can only be obtained via processes involving inherent randomness.

In quantum mechanics, a system can be prepared in a superposition of the (measurement) basis states, as shown in Fig.~\ref{Fig:superposition}. According to Born's rule, the measurement outcome of a quantum state can be intrinsically random, i.e. it can never be predicted better than blindly guessing. Therefore, the nature of inherent randomness in quantum measurements can be exploited for generating true random numbers.  Within a resource framework, coherence \cite{Baumgratz14} can be measured similarly to entanglement \cite{Bennett96}. By breaking the coherence or superposition of the measurement basis, it is shown that the obtained intrinsic randomness comes from the consumption of coherence. In turn, quantum coherence can be quantified from intrinsic randomness \cite{Yuan15Coherence}.

\begin{figure}
\centering
\includegraphics[width=6cm]{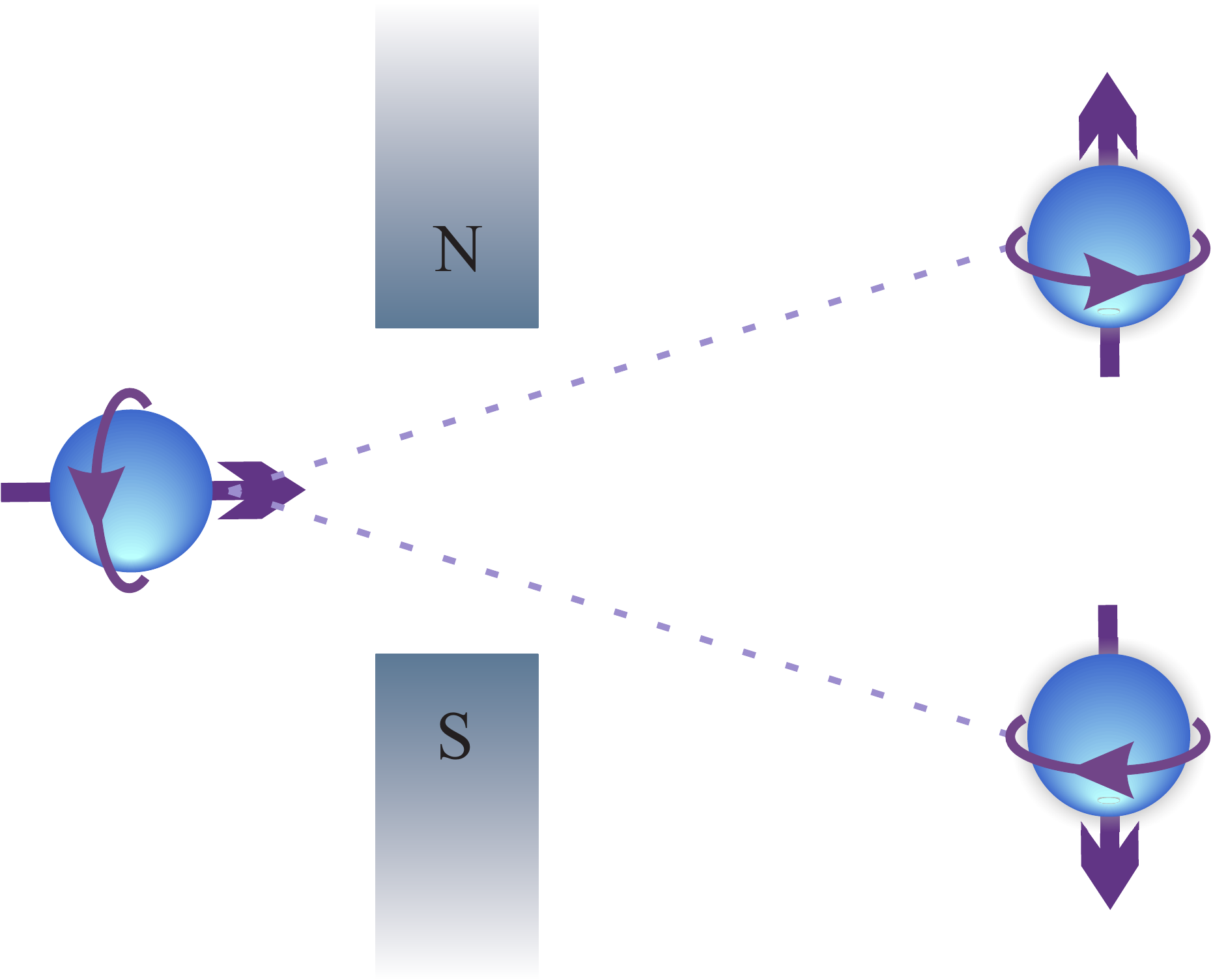}
\caption{Electron spin detection in the Stern-Gerlach experiment. Assume that the spin takes two directions along the vertical axis, denoted by $\ket{\uparrow}$ and $\ket{\downarrow}$. If the electron is initially in a superposition of the two spin directions, $\ket{\rightarrow} = (\ket{\uparrow} + \ket{\downarrow})/\sqrt{2}$, detecting the location of the electron would breaks the coherence and the outcome ($\uparrow$ or $\downarrow$) is intrinsically random.} \label{Fig:superposition}
\end{figure}

A practical QRNG can be developed using the simple process as shown in Fig.~\ref{Fig:superposition}. Based on the different implementations, there exists a variety of practical QRNGs. Generally, these QRNGs are featured for their high generation speed and a relatively low cost. In reality, quantum effects are always mixed with classical noises, which can be subtracted from the quantum randomness after properly modelling the underlying quantum process \cite{ma2013postprocessing}.

The randomness in the practical QRNGs usually suffices for real applications if the model fits the implementation adequately. However, such QRNGs can generate randomness with information-theoretical security only when the model assumptions are fulfilled. In the case that the devices are manipulated by adversaries, the output may not be genuinely random.
For example, when a QRNG is wholly supplied by a malicious manufacturer, who copies a very long random string to a large hard drive and only outputs the numbers from the hard drive in sequence, the manufacturer can always predict the output of the QRNG device.

On the other hand, a QRNG can be designed in a such way that its output randomness does not rely on any physical implementations. True randomness can be generated in a self-testing way even without perfectly characterizing the realisation instruments. The essence of a self-testing QRNG is based on device-independently witnessing quantum entanglement or nonlocality by observing a violation of the Bell inequality \cite{bell1964einstein}. Even if the output randomness is mixed with uncharacterised classical noise, we can still get a lower bound on the amount of genuine randomness based on the amount of nonlocality observed. The advantage of this type of QRNG is the self-testing property of the randomness. However, because the self-testing QRNG must demonstrate nonlocality, its generation speed is usually very low. As the Bell tests require random inputs, it is crucial to start with a short random seed. Therefore, such a randomness generation process is also called randomness expansion.

In general, a QRNG comprises a source of randomness and a readout system. In realistic implementations, some parts may be well characterised while others are not. This motivates the development of an intermediate type of QRNG, between practical and fully self-testing QRNGs, which is called semi-self-testing. Under several reasonable assumptions, randomness can be generated without fully characterising the devices. For instance, faithful randomness can be generated with a trusted readout system and an arbitrary untrusted randomness resource. A semi-self-testing QRNG provides a trade off between practical QRNGs (high performance and low cost) and self-testing QRNGs (high security of certified randomness).

In the last two decades, there have been tremendous development for all the three types of QRNG, trusted-device, self-testing, and semi-self-testing. In fact, there are commercial QRNG products available in the market. A brief summary of representative practical QRNG demonstrations that highlights the broad variety of optical QRNG is presented in Table \ref{Tab:SummaryPractical}. These QRNG schemes will be discussed further in Section \ref{Sec:Practical1} and \ref{Sec:Practical2}. A summary of self-testing and semi-self-testing QRNG demonstrations is presented in Table \ref{Tab:SummaryTheoretical}, which will be reviewed in details in Section \ref{Sec:Self} and \ref{Sec:Semi}.
\begin{table}\footnotesize
\centering
\caption{A brief summary of trusted-device QRNG demonstrations. Detailed description of these schemes can be found in Section \ref{Sec:Practical1} and \ref{Sec:Practical2}. Note that the quality/security of random numbers in different demonstrations may be different. Raw: reported raw generation rate, Refined: reported refined rate, Acquisition: data acquisition by dedicated hardware or commercial oscilloscope, SPD: single photon detector, BS: beam splitter, MCP-PCID: micro-channel-plate-based photon counting imaging detector, PNRD: photon-number-resolving detector, CMOS: complementary
metal-oxide-semiconductor, $-$: no related information found.} \label{Tab:SummaryPractical}
\begin{tabular}{ccccccc}
\hline
Year & Entropy source & Detection & Raw & Refined & Acquisition \\
  \hline
  2000 & Spatial mode \cite{jennewein2000fast} & SPD & 1 Mbps & $-$ & dedicated\\
  2000 & Spatial mode \cite{stefanov2000optical} & SPD & 100 Kbps & $-$ & dedicated \\
  2014& Spatial mode \cite{yan2014multi} & MCP-PCID & 8 Mbps& $-$ & dedicated \\
  2008& Temporal mode \cite{dynes2008high} & SPD & 4.01 Mbps & $-$ & dedicated \\
  2009& Temporal mode \cite{WJAK2009} & SPD & 55 Mbps & 40 Mbps & dedicated \\
  2011& Temporal mode \cite{WLB2011} & SPD & 180 Mbps & 152 Mbps & dedicated \\
  2014& Temporal mode \cite{nie2014practical} & SPD & 109 Mbps & 96 Mbps & dedicated \\
  2010& Photon number \cite{FWN2010} & PNRD  & 50 Mbps & $-$ & dedicated \\
  2011& Photon number \cite{Ren2011} & PNRD  & 2.4 Mbps & $-$ & dedicated \\
  2015& Photon number \cite{ATD2015} & PNRD  & $-$ & 143 Mbps & oscilloscope \\
  2010& Vacuum noise \cite{gabriel2010generator} & Homodyne  & 10 Mbps & 6.5 Mbps & dedicated \\
  2010& Vacuum noise \cite{Yong2010} & Homodyne  & $-$ & 12 Mbps & dedicated \\
  2011& Vacuum noise \cite{symul2011real} & Homodyne  & 3 Gbps & 2 Gbps & dedicated \\
  2010& ASE-intensity noise \cite{WSL2010} & Photo detector  & 12.5 Gbps  & $-$ & dedicated \\
  2011& ASE-intensity noise \cite{li2011scalable} & Photo detector  & 20 Gbps  & $-$ & $-$ \\
  2010& ASE-phase noise \cite{qi2010high} & Self-heterodyne  & 1 Gbps & 500 Mbps & oscilloscope \\
  2011& ASE-phase noise \cite{jofre2011true} & Self-heterodyne  & 1.2 Gbps & 1.11 Gbps & oscilloscope \\
  2012& ASE-phase noise \cite{Xu:QRNG:2012} & Self-heterodyne  & 8 Gbps & 6 Gbps & oscilloscope \\
  2014& ASE-phase noise \cite{YLD2014} & Self-heterodyne  & 80 Gbps & $-$ & oscilloscope \\
  2014& ASE-phase noise \cite{AAJ2014} & Self-heterodyne  & 82 Gbps & 43 Gbps & oscilloscope \\
  2015& ASE-phase noise \cite{nie201568} & Self-heterodyne  & 80 Gbps & 68 Gbps & oscilloscope \\
 \hline
\end{tabular}
\end{table}

\begin{table}
\centering
\caption{A summary of self-testing and semi-self-testing QRNG demonstrations. MDI: measurement device independent, SI: source independent, CV: continuous variable.} \label{Tab:SummaryTheoretical}
\begin{tabular}{ccccc}
\hline
Year & Type & Detection & Speed & Acquisition\\
  \hline
 2010 & Self-testing \cite{Pironio10} & ion-trap & very slow & dedicated \\
 2013 & Self-testing \cite{giustina2013bell} & SPD & 0.4 {bps}& dedicated \\
 2015& SI \cite{Cao2016}&SPD&5 Kbps & dedicated \\
 2015& CV-SI \cite{2015arXiv150907390M} & Homodyne & 1 Gbps& oscilloscope \\
 2015& Self-testing with fixed dimension \cite{Lunghi15}&SPD& 23 bps& dedicated \\
 \hline
\end{tabular}
\end{table}

\subsubsection{Trusted-device QRNG I: single-photon detector} \label{Sec:Practical1}
True randomness can be generated from any quantum process that breaks coherent superposition of states. Due to the availability of high quality optical components and the potential of chip-size integration, most of today's practical QRNGs are implemented in photonic systems. In this survey, we focus on various implementations of optical QRNGs.

A typical QRNG includes an entropy source for generating well-defined quantum states and a corresponding detection system. The inherent quantum randomness in the output is generally mixed with classical noises.  Ideally, the extractable quantum randomness should be well quantified and be the dominant source of the randomness. By applying randomness extraction, genuine randomness can be extracted from the mixture of quantum and classical noise. The extraction procedure is detailed in Methods.

\paragraph{Qubit state}
Random bits can be generated naturally by measuring a qubit\footnote{A qubit is a two-level quantum-mechanical system, which, similar to a bit in classical information theory, is the fundamental unit of quantum information.} $\ket{+}=(\ket{0}+\ket{1})/\sqrt2$ in the $Z$ basis, where $\ket{0}$ and $\ket{1}$ are the eigenstates of the measurement $Z$. For example, Fig.~\ref{Fig:SinglePhoton}~(a) shows a polarization based QRNG, where $\ket{0}$ and $\ket{1}$ denote horizontal and vertical polarization, respectively, and $\ket{+}$ denotes $+45^o$ polarization. Fig.~\ref{Fig:SinglePhoton}~(b) presents a path based QRNG, where $\ket{0}$ and $\ket{1}$ denote the photon traveling via path $R$ and $T$, respectively.

\begin{figure}
\centering
\includegraphics[width=12 cm]{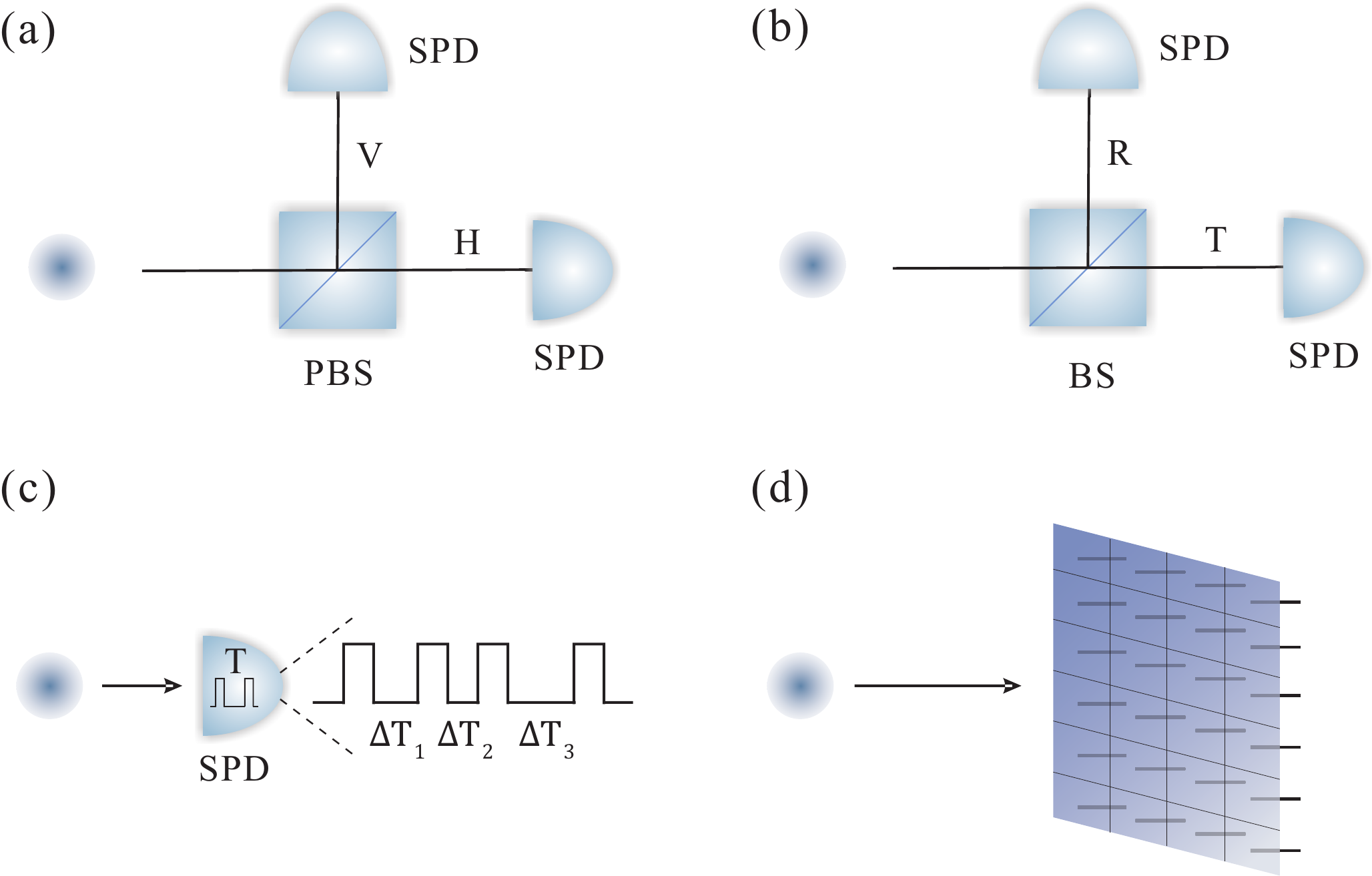}
\caption{Practical QRNGs based on single photon measurement. (a) A
photon is originally prepared in a superposition of horizontal (H) and vertical (V) polarizations, described by $(\ket{H}+\ket{V})/\sqrt{2}$. A polarising beam splitter (PBS) transmits the horizontal and reflects the vertical polarization. For random bit generation, the photon is measured by two single photon detectors (SPDs). (b) After passing through a symmetric beam splitter (BS), a photon exists in a superposition of transmitted (T) and reflected (R) paths, $(\ket{R}+\ket{T})/\sqrt{2}$. A random bit can be generated by measuring the path information of the photon. (c) QRNG based on measurement of photon arrival time. Random bits can be generated, for example, by measuring the time interval, $\Delta t$, between two detection events. (d) QRNG based on measurements of photon spatial mode. The generated random number depends on spatial position of the detected photon, which can be read out by an SPD array.} \label{Fig:SinglePhoton}
\end{figure}

The most appealing property of this type of QRNGs lies on their simplicity in theory that the generated randomness has a clear quantum origin. This scheme was widely adopted in the early development of QRNGs \cite{rarity1994quantum,stefanov2000optical,jennewein2000fast}. Since at most one random bit can be generated from each detected photon, the random number generation rate is limited by the detector's performance, such as dead time and efficiency. For example, the dead time of a typical silicon SPD based on an avalanche diode is tens of ns \cite{Eisaman2011PD}. Therefore, the random number generation rate is limited to tens of Mbps, which is too low for certain applications such as high-speed quantum key distribution (QKD), which can be operated at GHz clock rates \cite{takesue2007quantum,PhysRevX.2.041010}. Various schemes have been developed to improve the performance of QRNG based on SPD.

\paragraph{Temporal mode}
One way to increase the random number generation rate is to perform measurement on a high-dimensional quantum space, such as measuring the temporal or spatial mode of a photon. Temporal QRNGs measure the arrival time of a photon, as shown in Fig.~\ref{Fig:SinglePhoton}~(c). In this example, the output of a continuous-wave laser is detected by a time-resolving SPD. The laser intensity can be carefully controlled such that within a chosen time period $T$, there is roughly one detection event. The detection time is randomly distributed within the time period $T$ and digitized with a time resolution of $\delta_t$. The time of each detection event is recorded as raw data. Thus for each detection, the QRNG generates about $\log_2(T/\delta_t)$ bits of raw random numbers. Essentially, $\delta_t$ is limited by the time jitter of the detector (typically in the order of 100 ps), which is normally much smaller than the detector deadtime (typically in the order of 100 ns) \cite{Eisaman2011PD}.

One important advantage of temporal QRNGs is that more than one bit of random number can be extracted from a single-photon detection, thus improving the random number generation rate. The time period $T$ is normally set to be comparable to the detector deadtime. Comparing to the qubit QRNG, the temporal-mode QRNG alleviates the impact of detection deadtime. For example, if the time resolution and the dead time of an SPD are 100 ps and 100 ns respectively, the generation rate of temporal QRNG is around $\log_2(1000)\times$10 Mbps, which is higher than that of the qubit scheme (limited to 10 Mbps). The temporal QRNGs have been well studied recently \cite{ma2005random,dynes2008high,WJAK2009,WLB2011,nie2014practical}.

\paragraph{Spatial mode}
Similar to the case of temporal QRNG, multiple random bits can be generated by measuring the spatial mode of a photon with a space-resolving detection system. One illustrative example is to send a photon through a $1\times N$ beam splitter and to detect the position of the output photon. Spatial QRNG has been experimentally demonstrated by using a multi-pixel single-photon detector array \cite{yan2014multi}, as shown in Fig.~\ref{Fig:SinglePhoton}~(d).  The distribution of the random numbers depends on both the spatial distribution of light intensity and the efficiency uniformity of the SPD arrays.

The spatial QRNG offers similar properties as the temporal QRNG, but requires multiple detectors. Also, correlation may be introduced between the random bits because of cross talk between different pixels in the closely-packed detector array.

\paragraph{Multiple photon number states}
Randomness can be generated not only from measuring a single photon, but also from quantum states containing multiple photons. For instance, a coherent state
\begin{equation}\label{}
\ket{\alpha} = e^{-\frac{|\alpha|^2}{2}}\sum_{n=0}^\infty \frac{\alpha^n}{\sqrt{n!}}\ket{n},
\end{equation}
is a superposition of different photon-number (Fock) states $\{\ket{n}\}$, where $n$ is the photon number and $|\alpha|^2$ is the mean photon number of the coherent state. Thus, by measuring the photon number of a coherent laser pulse with a photon-number resolving SPD, we can obtain random numbers that follow a Poisson distribution.  QRNGs based on measuring photon number have been successfully demonstrated in experiments \cite{FWN2010,Ren2011,ATD2015}. Interestingly, random numbers can be generated by resolving photon number distribution of a light-emitting diode (LED) with a consumer-grade camera inside a mobile phone, as shown in a recent study \cite{Mobile2014}.

Note that, the above scheme is sensitive to both the photon number distribution of the source and the detection efficiency of the detector. In the case of a coherent state source, if the loss can be modeled as a beam splitter, the low detection efficiency of the detector can be easily compensated by using a relatively strong laser pulse.

\subsubsection{Trusted-device QRNG II: macroscopic photodetector} \label{Sec:Practical2}
The performance of an optical QRNG largely depends on the employed detection device. Beside SPD, high-performance macroscopic photodetectors have also been applied in various QRNG schemes. This is similar to the case of QKD, where protocols based on optical homodyne detection \cite{grosshans2003quantum} have been developed, with the hope to achieve a higher key rate over a low-loss channel. In the following discussion, we review two examples of QRNG implemented with macroscopic photodetector.

\paragraph{Vacuum noise}
In quantum optics, the amplitude and phase quadratures of the vacuum state are represented by a pair of non-commuting operators ($X$ and $P$ with $[X,P]=i/2$), which cannot be determined simultaneously with an arbitrarily high precision \cite{Braunstein05}, i.e. $\langle(\Delta X)^2\rangle\times\langle(\Delta P)^2\rangle\ge1/16$, with $\Delta O$ defined by $O - \langle O\rangle$ and $\langle O\rangle$ denoting the average of $O$. This can be easily visualised in the phase space, where the vacuum state is represented by a two-dimensional Gaussian distribution centered at the origin with an uncertainty of $1/4$ (the shot-noise variance) along any directions, as shown in Fig.~\ref{Fig:Vac}~(a). In principle, Gaussian distributed random numbers can be generated by measuring any field quadrature repeatedly. This scheme has been implemented by sending a strong laser pulse
through a symmetric beam splitter and detecting the differential signal of the two output beams with a balanced receiver \cite{gabriel2010generator,Yong2010,symul2011real}.

\begin{figure}[!hbt]
\centering
\includegraphics[width=11 cm]{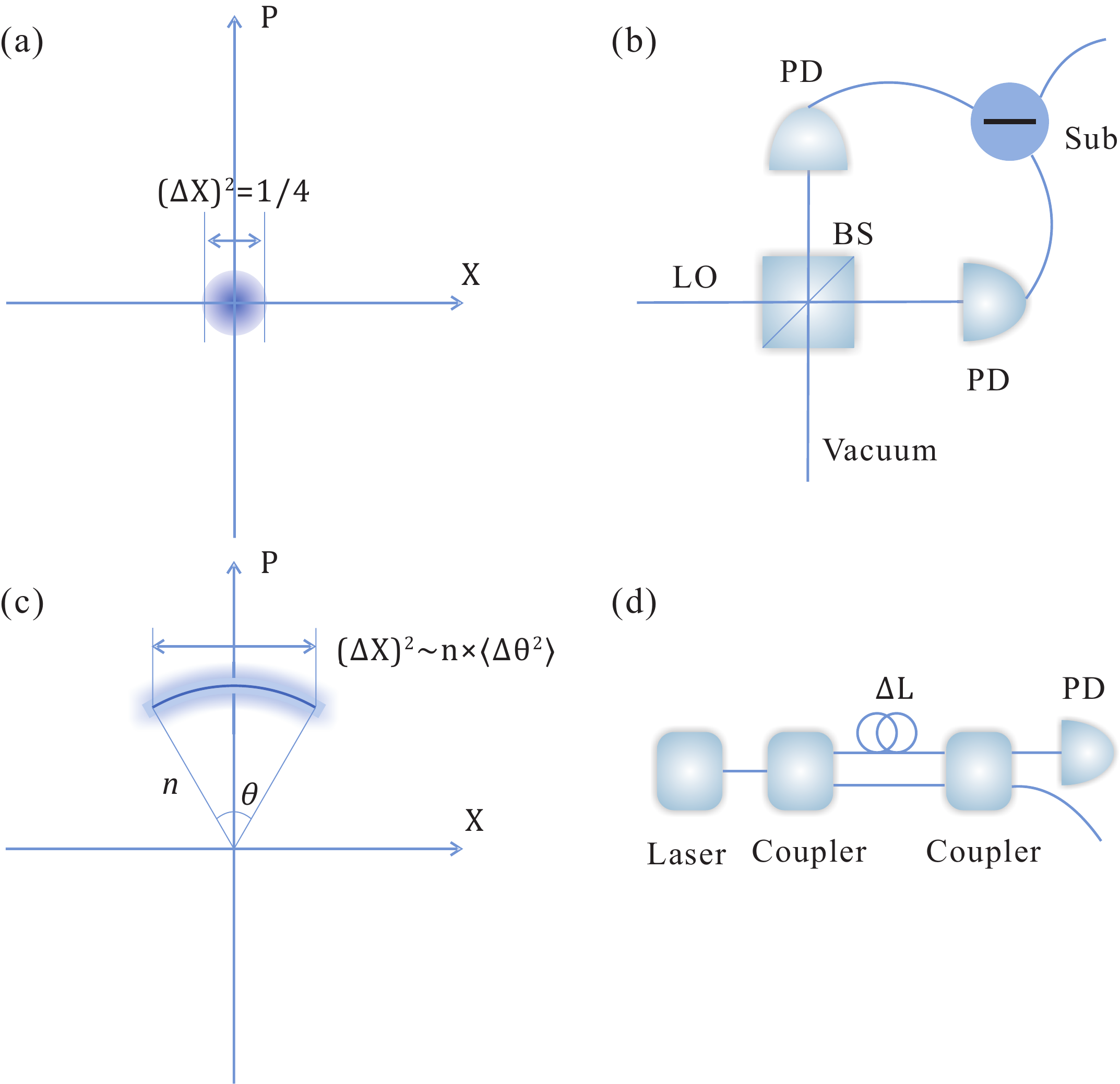}
\caption{QRNGs using macroscopic photodetector. (a) Phase-space representation of the vacuum state. The variance of the $X$-quadrature is 1/4. (b) QRNG based on vacuum noise measurements. The system comprises a strong local oscillator (LO), a symmetric beam splitter (BS), a pair of photon detector (PD), and an electrical subtracter (Sub).  (c) Phase-space representation of a partially phase-randomised coherent state. The variance of the $X$-quadrature is in the order of $n\times\langle\Delta\theta^2\rangle$, where $n$ is the average photon number and $\langle\Delta\theta^2\rangle$ is the phase noise variance. (d) QRNGs based on measurements of laser phase noise. The first coupler splits the original laser beam into two beams, which propagate through two optical fibres of different lengths, thereafter interfering at the second coupler. The output signal is recorded by a photon detector. The extra length $\Delta L$ in one fibre introduces a time delay $T_d$ between the two paths, which in turn determines the variance of the output signal.}
\label{Fig:Vac}
\end{figure}

Given that the local oscillator (LO) is a single-mode coherent state and the detector is shot-noise limited, the random numbers generated in this scheme follow a Gaussian distribution, which is on demand in certain applications, such as Gaussian-Modulated Coherent States (GMCS) QKD \cite{grosshans2003quantum}. There are several distinct advantages of this approach. First, the resource of quantum randomness, the vacuum state, can be easily prepared with a high fidelity. Second, the performance of the QRNG is insensitive to detector loss, which can be simply compensated by increasing the LO power. Third, the field quadrature of vacuum is a continuous variable, suggesting that more than one random bit can be generated from one measurement. For example, 3.25 bits of random numbers are generated from each measurement \cite{gabriel2010generator}.

In practice, an optical homodyne detector itself contributes additional technical noise, which may be observed or even controlled by a potential adversary. A randomness extractor is commonly required to generate secure random numbers. To extract quantum randomness effectively, the detector should be operated in the shot-noise limited region, in which the overall observed noise is dominated by vacuum noise. We remark that building a broadband shot-noise limited homodyne detector operating above a few hundred MHz is technically challenging \cite{okubo2008pulse,chi2011balanced,kumar2012versatile}. This may in turn limit the ultimate operating speed of this type of QRNG.

\paragraph{Amplified spontaneous emission}
To overcome the bandwidth limitation of shot-noise limited homodyne detection, researchers have developed QRNGs based on measuring phase \cite{qi2010high,jofre2011true,Xu:QRNG:2012,AAJ2014,YLD2014,nie201568} or intensity noise \cite{WSL2010,li2011scalable} of amplified spontaneous emission(ASE), which is quantum mechanical by nature \cite{henry1982theory,ma2013postprocessing,zhou2015randomness}.

In the phase-noise based QRNG scheme, random numbers are generated by measuring a field quadrature of \emph{phase-randomized} weak coherent states (signal states). Figure \ref{Fig:Vac}~(c) shows the phase-space representation of a signal state with an average photon number of $n$ and a phase variance of $\langle(\Delta\theta)^2\rangle$. If the average phase of the signal state is around $\pi/2$, the uncertainty of the $X$-quadrature is of the order of  $n\langle(\Delta\theta)^2\rangle$. When $n$ is large, this uncertainty can be significantly larger than the vacuum noise. Therefore, phase noise based QRNG is more robust against detector noise. In fact, this scheme can be implemented with commercial photo-detectors operated above GHz rates.

QRNG based on laser phase noise was first developed using a cw laser source and a delayed self-heterodyning detection system \cite{qi2010high}, as shown in Fig.~\ref{Fig:Vac}~(d). Random numbers are generated by measuring the phase difference of a single-mode laser at times $t$ and $t+T_d$. Intuitively, if the time delay $T_d$ is much larger than the coherence time of the laser, the two laser beams interfering at the second beam splitter can be treated as generated by independent laser sources. In this case, the phase difference is a random variable uniformly distributed in $[-\pi, \pi)$, regardless of the classical phase noise introduced by the unbalanced interferometer itself. This suggests that a robust QRNG can be implemented without phase-stabilizing the interferometer. On the other hand, by phase-stabilizing the interferometer, the time delay $T_d$ can be made much shorter than the coherent time of the laser \cite{qi2010high}, enabling a much higher sampling rate. This phase stabilization scheme has been adopted in a $\ge6$ Gbps QRNG \cite{Xu:QRNG:2012} and a 68 Gbps QRNG demonstration \cite{nie201568}.

Phase noise based QRNG has also been implemented using pulsed laser source, where the phase difference between adjacent pulses is automatically randomized \cite{jofre2011true,AAJ2014,YLD2014}. A speed of 80 Gbps (raw rate as shown in Table1) has been demonstrated \cite{YLD2014}. It also played a crucial role in a recent loophole-free Bell experiment \cite{AAM2015}. Here, we want to emphasize that strictly speaking, none of these generation speeds are real-time, due to the speed limitation of the randomness extraction \cite{ma2013postprocessing}. Although such limitation is rather technical, in practice, it is important to develop extraction schemes and hardware that can match the fast random bit generation speed in the future.

\subsubsection{Self-testing QRNG} \label{Sec:Self}
Realistic devices inevitably introduce classical noise that affects the output randomness, thus causing the generated random numbers depending on certain classical variables, which might open up security issues. To remove this bias, one must properly model the devices and quantify their contributions. In the QRNG schemes described in Section \ref{Sec:Practical1} and Section \ref{Sec:Practical2}, the output randomness relies on the device models \cite{ma2013postprocessing,zhou2015randomness}. When the implementation devices deviate from the theoretical models, the randomness can be compromised. In this section, we discuss \emph{self-testing} QRNGs, whose output randomness is certified independent of device implementations.

\paragraph{Self-testing randomness expansion}
In QKD, secure keys can be generated even when the experimental devices are not fully trusted or characterised \cite{Mayers98,acin06}. Such self-testing processing of quantum information  also occur in randomness generation (expansion). The output randomness can be certified by observing violations of the Bell inequalities \cite{bell1964einstein}, see Fig.~\ref{Fig:Bell}. Under the no-signalling condition \cite{prbox} in the Bell tests, it is impossible to violate Bell inequalities if the output is not random, or, predetermined by local hidden variables.

\begin{figure*}[hbt]
\centering
\resizebox{6cm}{!}{\includegraphics{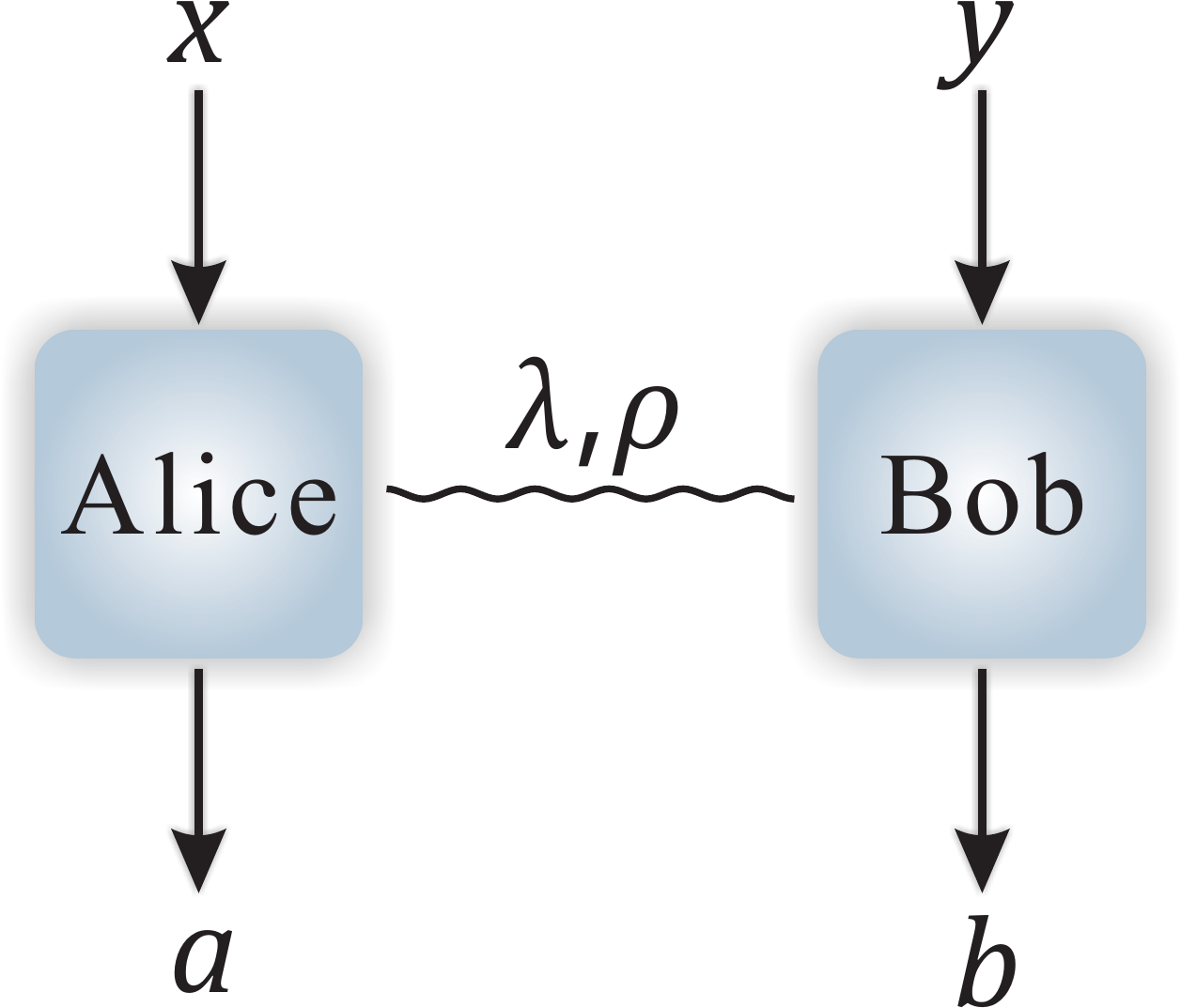}}
\caption{Illustration of a bipartite Bell test. Alice and Bob are two spacelikely separated parties, that output $a$ and $b$ from random inputs $x$ and $y$, respectively. A Bell inequality is defined as a linear combination of the probabilities $p(a,b|x,y)$. For instance, the Clauser-Horne-Shimony-Holt (CHSH) inequality \cite{CHSH} is defined by $S = \sum_{a,b,x,y} (-1)^{a + b + xy}p(a,b|x,y) \leq S_C = 2$, where all of the inputs and outputs are bit values, and $S_C$ is the classical bound for all local hidden-variable models. With quantum settings, that is, performing measurements $M_x^a\otimes M_y^b$ on quantum state $\rho_{AB}$, $p(a,b|x,y) = \mathrm{Tr}[\rho_{AB}M_x^a\otimes M_y^b]$, the CHSH inequality can be violated up to $S_Q = 2\sqrt{2}$. Quantum features (such as intrinsic randomness) manifest as violations of the CHSH inequality.} \label{Fig:Bell}
\end{figure*}

Since Colbeck \cite{Colbeck09, Colbeck11} suggested that randomness can be expanded by untrusted devices, several protocols based on different assumptions have been proposed. For instance, in a non-malicious device scenario, we can consider that the devices are honestly designed but get easily corrupt by unexpected classical noises. In this case, instead of a powerful adversary that may entangle with the experiment devices, we can consider a classical adversary who possesses only classical knowledge of the quantum system and analyzes the average randomness output conditioned by the classical information. Based on the Clauser-Horne-Shimony-Holt (CHSH) inequality \cite{CHSH}, Fehr et al. \cite{Fehr13} and Pironio et al. \cite{Pironio13} proposed self-testing randomness expansion protocols against classical adversaries. The protocols quadratically expands the input seed, implying that the length of the input seed is $O(\sqrt{n}\log_2\sqrt{n})$, where $n$ denotes the experimental iteration number.

A more sophisticated exponential randomness expansion protocol based on the CHSH inequality was proposed by Vidick and Vazirani \cite{Vazirani12}, in which the lengths of the input seed is $O(\log_2n)$. In the same work, they also presented an exponential expansion protocol against quantum adversaries, where quantum memories in the devices may entangle with the adversary. The Vidick-Vazirani protocol against quantum adversaries places strict requirements on the experimental realisation. Miller and Shi \cite{Miller14} partially solved this problem by introducing a more robust protocol. Combined with the work by Chung, Shi, and Wu \cite{Chung14}, they also presented an unbounded randomness expansion scheme. By adopting a more general security proof, Miller and Shi \cite{miller2014universal} recently showed that genuinely randomness can be obtained as long as the CHSH inequality is violated. Their protocol greatly improves the noise tolerance, indicating that an experimental realisation of a fully self-testing randomness expansion protocol is feasible.

The self-testing randomness expansion protocol relies on a faithful realisation of Bell test excluding the experimental loopholes, such as locality and efficiency loopholes. The randomness expansion protocol against classical adversaries is firstly experimentally demonstrated by Pironio et al. \cite{Pironio10} in an ion-trap system, which closes the efficiency loophole but not the locality loophole. To experimentally close the locality loophole, a photonic system is more preferable when quantum memories are unavailable. As the CHSH inequality is minimally violated in an optically realised system \cite{giustina2013bell,Christensen13}, the randomness output is also very small (with min-entropy of $H_{\mathrm{min}} = 7.2\times10^{-5}$ in each run), and the randomness generation rate is $0.4$ {bits/s}. To maximise the output randomness, the implementation settings are designed to maximally violate the CHSH inequality. Due to experimental imperfections, the chosen Bell inequality might be sub-optimal for the observed data. In this case, the output randomness can be optimised over all possible Bell inequalities \cite{Silleras14,Bancal14b}.

Although nonlocality or entanglement certifies the randomness, the three quantities, nonlocality, entanglement, and randomness are not equivalent \cite{acin12}. Maximum randomness generation does not require maximum nonlocal correlation or a maximum entangled state. In the protocols based on the CHSH inequality, maximal violation (nonlocality and entanglement) generates 1.23 bits of randomness. It is shown that 2 bits of randomness can be certified with little involvement of nonlocality and entanglement \cite{acin12}. Furthermore, as discussed in a more generic scenario involving nonlocality and randomness, it is shown that maximally nonlocal theories cannot be maximally random \cite{Torre15}.

\paragraph{Randomness amplification}
In self-testing QRNG protocols based on the assumption of perfectly random inputs, the output randomness is guaranteed by the violations of Bell tests. Conversely, when all the inputs are predetermined, any Bell inequality can be violated to an arbitrary feasible value without invoking a quantum resource. Under these conditions, all self-testing QRNG protocols cease to work any more. Nevertheless, randomness generation in the presence of partial randomness is still an interesting problem. Here, an adversary can use the additional knowledge of the inputs to fake violations of Bell inequalities. The task of generating arbitrarily free randomness from partially free randomness is also called randomness amplification, which is impossible to achieve in classical processes.

The first randomness amplification protocol was proposed by Colbeck and Renner \cite{colbeck2012free}. Using a two-party chained Bell inequality \cite{Pearle70,Braunstein90}, they showed that any Santha-Vazirani weak sources \cite{santha1986generating} (defined in next section), with $\epsilon<0.058$, can be amplified into arbitrarily free random bits in a self-testing way by requiring only no-signaling. A basic question of randomness amplification is whether free random bits can be obtained from arbitrary weak randomness. This question was answered by Gallego et al. \cite{gallego2013full}, who demonstrated that perfectly random bits can be generated using a five-party Mermin inequality \cite{Mermin90} with arbitrarily imperfect random bits under the no-signaling assumption.

Randomness amplification is related to the freewill assumption \cite{Kofler06, Hall10,Barrett11,Koh12,Pope13,putz14,Yuan2015CHSH} in Bell tests. In experiments, the freewill assumption requires the inputs to be random enough such that violations of Bell inequalities are induced from quantum effects rather than predetermined classical processes. This is extremely meaningful in fundamental Bell tests, which aim to rule out local realism. Such fundamental tests are the foundations of self-testing tasks, such as device-independent QKD and self-testing QRNG. Interestingly, self-testing tasks require a faithful violation of a Bell inequality, in which intrinsic random numbers are needed. However, to generate faithful random numbers, we in turn need to witness nonlocality which requires additional true randomness. Therefore, the realisations of genuine loophole-free Bell tests and, hence, fully
self-testing tasks are impossible. Self-testing protocols with securities independent of the untrusted part can be designed only by placing reasonable assumptions on the
trusted part.

\subsubsection{Semi-self-testing QRNGs} \label{Sec:Semi}
Traditional QRNGs based on specific models pose security risks in fast random number generation. On the other hand, the randomness generated by self-testing QRNGs is information-theoretically secure even without characterising the devices, but the processes are impractically slow. As a compromise, intermediate QRNGs might offer a good tradeoff between trusted and self-testing schemes --- realising both reasonably fast and secure random number generation.

As shown in Fig.~\ref{Fig:Semi}, a typical QRNG comprises two main modules, a source that emits quantum states and a measurement device that detects the states and outputs random bits.  In trusted-device QRNGs, both source and measurement devices \cite{ma2013postprocessing,zhou2015randomness} must be modeled properly; while the output randomness in the fully self-testing QRNGs does not depend on the implementation devices.

\begin{figure}[hbt]
\centering
\resizebox{6cm}{!}{\includegraphics{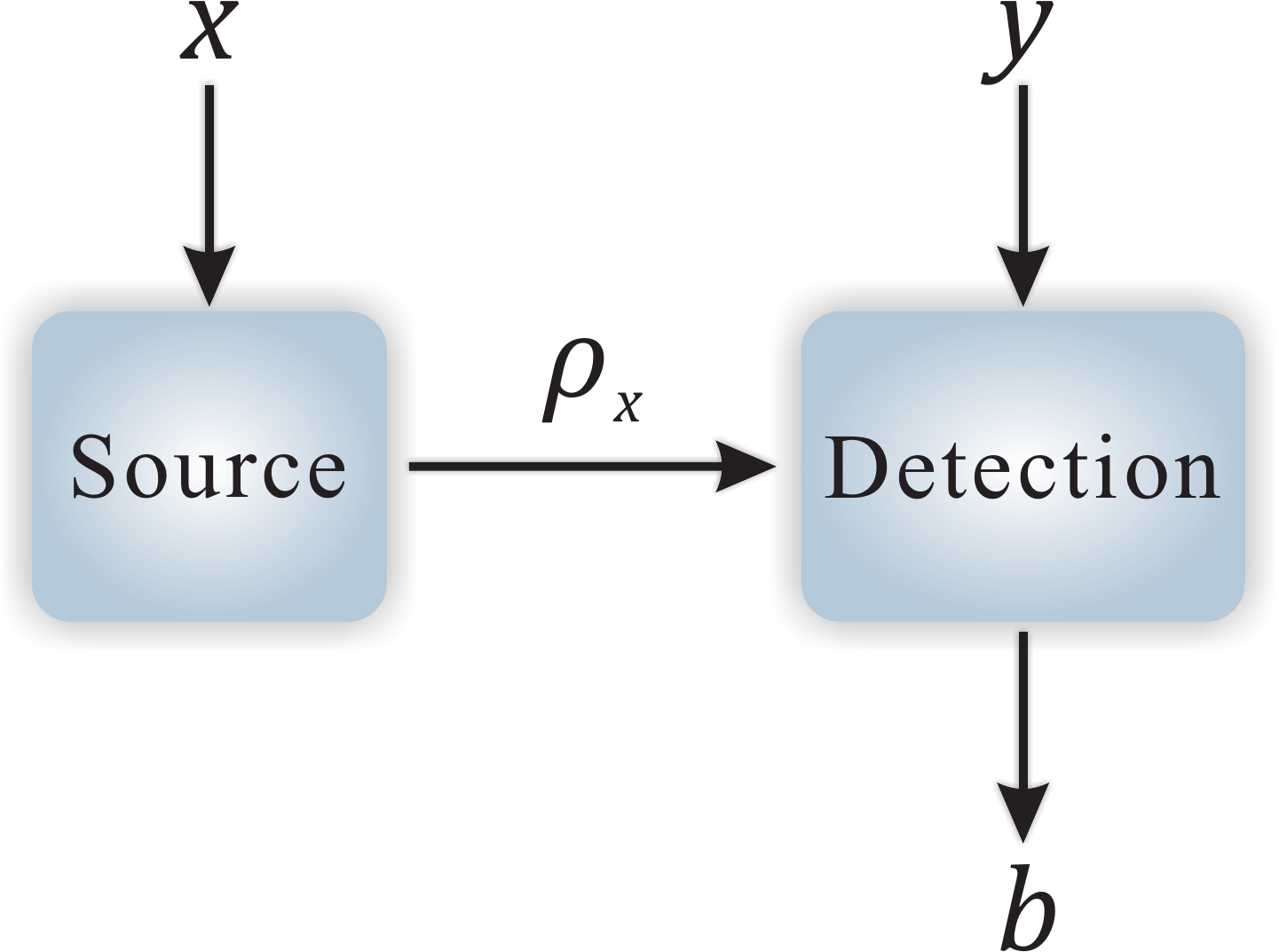}}
\caption{A semi-self-testing QRNG. Conditional on the input setting $x$, the source emits a quantum state $\rho_x$. Conditional on the input $y$, the detection device measures $\rho_x$ and outputs $b$. } \label{Fig:Semi}
\end{figure}

In practice, there exist scenarios that the source (respectively, measurement device) is well characterised, while the measurement device (respectively, source) not. Here, we review the semi-self-testing QRNGs, where parts of the devices are trusted.

\paragraph{Source-independent QRNG}
In source-independent QRNG, the randomness source is assumed to be untrusted, while the measurement devices are trusted. The essential idea for this type of scheme is to use the measurement to monitor the source in real time. In this case, normally one needs to randomly switch among different (typically, complement) measurement settings, so that the source (assumed to be under control of an adversary) cannot predict the measurement ahead. Thus, a short seed is required for the measurement choices.

In the illustration of semi-self-testing QRNG, Fig.~\ref{Fig:Semi}, the source-independent scheme is represented by a unique $x$ (corresponding to a state $\rho_x$) and multiple choices of the measurement settings $y$. In Section \ref{Sec:Practical1}, we present that randomness can be obtained by measuring $\ket{+}$ in the $Z$ basis. However, in a source-independent scenario, we cannot assume that the source emits the state $\ket{+}$. In fact, we cannot even assume the dimension of the state $\rho_x$. This is the major challenge facing for this type of scheme.

In order to faithfully quantify the randomness in the $Z$ basis measurement, first a squashing model is applied so that the to-be-measured state is equivalent to a qubit \cite{BML_Squash_08}. Note that this squashing model puts a strong restriction on measurement devices. Then, the measurement device should occasionally project the input state onto the $X$ basis states, $\ket{+}$ and $\ket{-}$, and check whether the input is $\ket{+}$ \cite{Cao2016}. The technique used in the protocol shares strong similarity with the one used in QKD \cite{Shor2000Simple}. The $X$ basis measurement can be understood as the \emph{phase error estimation}, from which we can estimate the amount of classical noise. Similar to privacy amplification, randomness extraction is performed to subtract the classical noise and output true random values.

The source-independent QRNG is advantageous when the source is complicated, such as in the aforementioned QRNG schemes based on measuring single photon sources \cite{stefanov2000optical,jennewein2000fast,rarity1994quantum}, LED lights \cite{Mobile2014}, and phase fluctuation of lasers \cite{Xu:QRNG:2012}. In these cases, the sources are quantified by complicated or hypothetical physical models. Without a well-characterized source, randomness can still be generated. The disadvantage of this kind of QRNGs compared to fully self-testing QRNGs is that they need a good characterization of the measurement devices. For example, the upper and the lower bounds on the detector efficiencies need to be known to avoid potential attacks induced from detector efficiency mismatch. Also the intensity of light inputs into the measurement device needs to be carefully controlled to avoid attacks on the detectors.

Recently,  a continuous-variable version of the source-independent QRNG is experimentally demonstrated \cite{2015arXiv150907390M}
and achieves a randomness generation rate over 1 Gbps. Moreover, with state-of-the-art devices, it can potentially reach the speed in the order of tens of Gbps, which is similar to the trusted-device QRNGs. Hence, semi-self-testing QRNG is approaching practical regime.

\paragraph{Measurement-device-independent QRNGs}
Alternatively, we can consider the scenario that the input source is well characterised while the measurement device is untrusted. In Fig.~\ref{Fig:Semi}, different inputs $\rho_x$ (hence multiple $x$) are needed to calibrate the measurement device with a unique setting $y$. Similar to the source-independent scenario, the randomness is originated by measuring the input state $\ket{+}$ in the $Z$ basis. The difference is that here the trusted source sends occasionally auxiliary quantum states $\rho_x$, such as $\ket{0}$, to check whether the measurement is in the $Z$ basis \cite{MDIQRNG15}.  The analysis combines measurement tomography with randomness quantification of positive-operator valued measure, and does not assume to know the dimension of the measurement device, i.e., the auxiliary ancilla may have an arbitrary dimension.

The advantage of such QRNGs is that they remove all detector side channels, but the disadvantage is that they may be subject to imperfections in the modeling of the source.  This kind of QRNG is complementary to the source-independent QRNG, and one should choose the proper QRNG protocol based on the experimental devices.

We now turn to two variations of measurement-device-independent QRNGs. First, the measurement tomography step may be replaced by a certain witness, which could simplify the scheme at the expense of a slightly worse performance. Second, similar to the source-independent case, a continuous-variable version of measurement-device-independent QRNG might significantly increase the bit rate. The challenge lies on continuous-variable entanglement witness and measurement tomography.

\paragraph{Other semi-self-testing QRNGs}
Apart from the above two types of QRNGs, there are also some other QRNGs that achieve self-testing except under some mild assumptions. For example, the source and measurement devices can be assumed to occupy independent two-dimensional quantum subspaces \cite{Lunghi15}. In this scenario, the QRNG should use both different input states and different measurement settings. The randomness can be estimated by adopting a \emph{dimension witness} \cite{PRL.112.140407}. A positive value of this dimension witness could certify randomness in this scenario, similar to the fact that a violation of the Bell inequality could certify randomness of self-testing QRNG in Section \ref{Sec:Self}.

\subsubsection{Outlook}
The needs of ``perfect'' random numbers in quantum communication and fundamental physics experiments have stimulated the development of various QRNG schemes, from highly efficient systems based on trusted devices, to the more theoretically interesting self-testing protocols. On the practical side, the ultimate goal is to achieve fast random number generation at low cost, while maintaining high-level of randomness. With the recent development on waveguide fabrication technique \cite{barak2003true}, we expect that chip-size, high-performance QRNGs could be available in the near future. In order to guarantee the output randomness, the underlying physical models for these QRNGs need to be accurate and both the quantum noise and classical noise should be well quantified. Meanwhile, by developing a semi-self-testing protocol, a QRNG becomes more robust against classical noises and device imperfections. In the future, it is interesting to investigate the potential technologies required to make the self-testing QRNG practical. With the new development on single-photon detection, the readout part of the self-testing QRNG can be ready for practical application in the near future. The entanglement source, on the other hand, is still away from the practical regime (Gbps).

On the theoretical side, the study of self-testing QRNG has not only provided means of generating robust randomness, but also greatly enriched our understanding on the fundamental questions in physics. In fact, even in the most recent loophole-free Bell experiment \cite{hensen2015experimental, PhysRevLett.115.250402, PhysRevLett.115.250401,ballance2015hybrid} where high-speed QRNG has played a crucial role, it is still arguable whether it is appropriate to use randomness generated based on quantum theory to test quantum physics itself. Other random resources have also been proposed for loophole-free Bell's inequality tests, such as independent comic photons \cite{Gallicchio14}. It is an open question whether we can go beyond QRNG and generate randomness from a more general theory.

\subsection{Randomness quantification}
Here, we breifly review the quantification of randomness.
\subsubsection{Min-entropy source}
Given the underlying probability distribution, the randomness of a random sequence $X$ on $\{0,1\}^n$ can be quantified by its \emph{min-entropy}
\begin{equation}\label{}
H_{\mathrm{min}} = -\log_2\left(\max_{v\in\{0,1\}^n}\mathrm{Prob}[X=v]\right).
\end{equation}

For example, for a uniform random sequence $X$ on $\{0,1\}^n$, when $\mathrm{Prob}[X=v] = 1/2^n, \forall v\in\{0,1\}^n$, the min-entropy is given $H_{\mathrm{min}} = n$.  As another example, we consider the same random sequence $X$ except that $\mathrm{Prob}[X=0] = 1/2$ and $\mathrm{Prob}[X=v] = 1/2^{n+1}, \forall v\in\{0,1\}^n/{0}$. Although the two sequences looks very similar, the min-entropy for the latter example is much more smaller, $H_{\mathrm{min}} = 1$.
\subsubsection{Santha-Vazirani weak sources \cite{santha1986generating}}
We assume that random bit numbers are produced in the time sequence $x_1, x_2, ..., x_j, ...$. Then, for $0<\epsilon\le 1/2$, the source is called $\epsilon$-free if
\begin{equation}\label{}
  \epsilon \le P(x_j|x_1, x_2, \dots, x_{j-1}, e) \le 1-\epsilon,
\end{equation}
for all values of $j$. Here $e$ represents all classical variables generated outside the future light-cone of the Santha-Vazirani weak sources.

\subsubsection{Randomness extractor}
A RNG typically consists of two components, an entropy source and a randomness extractor \cite{barak2003true}. In a QRNG, the entropy source could be a physical device whose output is fundamentally unpredictable, while the randomness extractor could be an algorithm that generates nearly perfect random numbers from the output of the above preceding entropy source, which can be imperfectly random. The two components of QRNG are connected by quantifying the randomness with min-entropy. The min-entropy of the entropy source is first estimated and then fed into the randomness extractor as an input parameter.

The imperfect randomness of the entropy source can already be seen in the SPD based schemes, such as the photon number detection scheme. By denoting $N$ as the discrimination upper bound of a photon number resolving detector, at most $log_2(N)$ raw random bits can be generated per detection event. However, as the photon numbers of a coherent state source follows a Poisson distribution, the raw random bits follow a non-uniform distribution; consequently, we cannot obtain $log_2(N)$ bits of random numbers. To extract perfectly random numbers, we require a postprocessing procedure (i.e. randomness extractor).

In the coherent detection based QRNG, the quantum randomness is inevitably mixed with classical noises introduced by the detector and other system imperfections. Moreover, any measurement system has a finite bandwidth, implying unavoidable correlations between adjacent samples. Once quantified, these unwanted side-effects can be eliminated through an appropriate randomness extractor \cite{ma2013postprocessing}.

The composable extractor was first introduced in classical cryptography \cite{Canetti2001,canetti2002universally}, and was later extended to quantum cryptography \cite{BenOr:Security:05,Renner:Security:05}. To generate information-theoretically provable random numbers, two typical extractor, the Trevisan's extractor or the
Toeplitz-hashing extractor, are generally employed in practice.

Trevisan's extractor \cite{trevisan2001extractors,raz1999extracting} has been proven secure against quantum adversaries \cite{de2012trevisan}. Moreover, it is a strong extractor (its seed can be reused) and its seed length is polylogarithmic function of the input. Tevisan's extractor comprises two main parts, a one-bit extractor and a combinatorial design. The Toeplitz-hashing extractor was well developed in the privacy amplification procedure of the QKD system \cite{uchida2008fast}. This kind of extractor is also a strong extractor \cite{wegman1981new}. By applying the fast Fourier transformation technique, the runtime of the Toeplitz-hashing extractor can be improved to $O(n\log n)$.

On account of their strong extractor property, both of these extractors generate random numbers even when the random seed is longer than the output length of each run. Both extractors have been implemented \cite{ma2013postprocessing} and the speed of both extractors have been increased in follow-up studies \cite{DBLP:journals/corr/abs-1212-0520,ma2011explicit}, but remain far below the operating speed of the QRNG based on laser-phase fluctuation (68 Gbps \cite{nie201568}). Therefore, the speed of the extractor is the main limitation of a practical QRNG.

\part{Quantumness and randomness}
\chapter{Coherence and randomness}
This chapter introduces the basic quantification and witness methods for quantum coherence. We relate coherence measures with the quantum randomness measured on the computational basis. We refer to Ref.~\cite{Yuan15Coherence, yuan2016interplay} for references of this chapter.
\section{Quantifying quantum randomness}\label{qrand}
\subsection{Quantum randomness against quantum information}
Let us consider a $d$-dimensional Hilbert space and a reference  basis  $I:=\{\ket{i}\} = \left\{\ket{1},\ket{2},\dots,\ket{d}\right\}$. Suppose a projective measurement $\{\ket{i}\bra{i}\}$ is performed on a given quantum state $\rho_A$ accessed by an experimentalist Alice. The measurement outcome has a probability distribution $\{p_i\}, \sum_{i=1}^d p_i=1,  , p_i= \text{Tr}[\rho \ket{i}\bra{i}]\geq 0,  \forall i$. In quantum information theory, a practical quantifier of the total randomness associated to the measurement is given by the Shannon entropy  $H(\{p_i\})_{\rho} =-\sum_i p_i\log(p_i)$. However, the randomness of the measurement is intrinsically twofold: a classical uncertainty due to Alice's ignorance about the system state; and a quantum one due to the coherence of the state in the reference basis. For a mixture of incoherent states $\rho_{{\cal I}}=\sum_i q_i \ket{i}\bra{i},$ the measurement randomness is given by the state mixedness, i.e., a classical source of uncertainty, quantified by the state von Neumann entropy: $H(\{p_i\})_{\rho_{{\cal I}}}=H(\{q_i\})=S(\rho_{{\cal I}})$. On the other hand, for pure states, $\rho_p=\ket{\psi}\bra{\psi}$, the randomness is due to the genuinely quantum overlap between the state and the basis elements: $H(\{p_i\})_{\rho_p}=H(\{|\langle i|\psi\rangle|^2\})$. Here we present an operational  characterization of the quantum randomness for arbitrary coherent mixed states. To be a good measure of quantum uncertainty, a quantity should satisfy the following properties:
\begin{enumerate}
\item
Being nonnegative;
\item
Vanishing if and only if the measurement uncertainty is only due to the state mixedness;
\item
Representing the total uncertainty for pure states;
\item
Being convex \cite{luo2005quantum,luo12,luo03,herbut2005quantum,Baumgratz14}.
\end{enumerate}

We consider the worse case scenario depicted in Fig.\ref{qfigure}, where Alice and Eve share a bipartite system in state $\rho_{AE}$. Alice makes a measurement and obtains outcomes following a probability distribution $\{p_i\}, p_i=\text{Tr}[\rho_A \ket{i}\bra{i}]$. The total randomness associated to the measurement is  $H(\{p_i\})$. The uncertainty of Eve about Alice's measurement outcome is quantified by the conditional entropy $H(\{p_i\}|E)_{\rho_{AE}}$.

\begin{figure}[bht]
\centering
\resizebox{6cm}{!}{\includegraphics[scale=1]{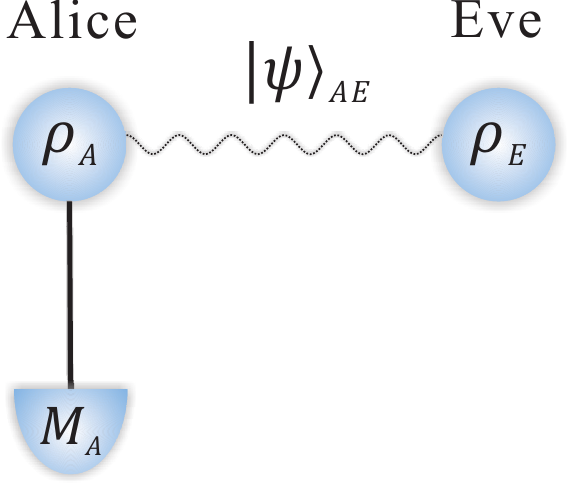}}
\caption{Quantum randomness.  In a bipartite Alice-Eve system described by a pure state $\psi_{AE}$, the quantum randomenss of a measurement performed by Alice on the system in the mixed state $\rho_A$ is given by the amount of uncertainty Eve has on the measurement outcome. Such quantum uncertainty is quantified by the relative entropy of coherence $R_I^Q(\rho_A)$.}\label{qfigure}
\end{figure}

After the Alice's measurement, define the global state by  $\rho_{AE}'$ and Alice's resulted state becomes $\rho_A^{\mathrm{diag}}:=\sum_i p_i \ket{i}\bra{i}$.  Hence, in the best case scenario for Eve, her uncertainty is given by the von Neumann conditional entropy
\begin{eqnarray} \label{Eq:quantumrandomness}
R^Q_I(\rho_A) := \min_{\rho_E} H(\{p_i\}|E)_{\rho_{AE}}=\min_{\rho_E} S(A|E)_{\rho_{AE}'},
\end{eqnarray}
where the optimization runs over all the possible Eve's states such that $\text{Tr}_E(\rho_{AE})=\rho_A$, and the conditional entropy is given by $S(A|E)_{\rho_{AE}'}= S(\rho_{AE}')- S(\rho_E)$.
{It is not hard to see that the best case scenario for Eve is to hold a purification of Alice, $\ket{\psi}_{AE}$. In fact, one can always extend Eve's part to hold a purification of a mixed state $\rho_{AE}$, which will not increase her uncertainty about Alice's measurement outcome.}

When $\rho_A$ is a pure state, then $\ket{\psi}_{AE}$ and hence $\rho_{AE}'=\rho_A^{\mathrm{diag}}\otimes\rho_E$ are both product states. It is easy to verify that Eve's uncertainty corresponds to the total randomness of Alice's measurement:
\begin{equation}\label{Eq:}
R^Q_I(\rho_A) = S(\rho_A^{\mathrm{diag}})=H(\{p_i\})_{\rho_A}.
\end{equation}
When $\rho_A$ is not a pure state, after Alice's measurement, the state is changed to $ \rho_{AE}' = \sum_i p_i\ket{i}_A\bra{i}\otimes\rho_i^E$, where  $\rho_E  = \sum_i p_i\rho_i^E$. In fact, $\rho_i^E={}_A\bra{i}(\ket{\psi}_{AE}\bra{\psi}_{AE})\ket{i}_A/p_i$ is a pure state. The conditional entropy of the post measurement state is given by $S(A|E)_{\rho_{AE}'} = S(\rho_{AE}')- S(\rho_E)$. Using the  equality $S\left(\sum_ip_i\ket{i}\bra{i}\otimes\rho_i\right) = H(p_i) + \sum_i p_i S(\rho_i)$, the conditional entropy is then $S(A|E)_{\rho_{AE}'} = H(p_i) + \sum_ip_iS(\rho_i^E)  - S(\rho_E)$.
Since $H(p_i) = S(\rho_A^{\mathrm{diag}})$, $S(\rho_E) = S(\rho_A)$, and $S(\rho_i^E) = 0, \forall i,$ we have
\begin{equation}\label{Eq:rq}
  R^Q_I(\rho_A) = S(\rho_A^{\mathrm{diag}}) - S(\rho_A).
\end{equation}
It is immediate to observe that the Eve's uncertainty is equal to the relative entropy of coherence
\begin{eqnarray}
R^Q_I(\rho_A)=S(\rho_A||\rho_A^{\mathrm{diag}})=C_R(\rho_A),
\end{eqnarray}
thus satisfying all the requirements for a consistent measure of quantum randomness as well as being a measure of BCP coherence \cite{Baumgratz14}.

{Note that, when considering a tripartite pure state $\ket{\psi_{ABE}}$ and a projective measurement $\{\ket{i}\bra{i}\}$ on system $A$, it is shown \cite{Coles12} that the quantum randomness of the measurement outcome conditioned on system $E$ corresponds the distance between state $\rho_{AB} = \mathrm{tr}_E[\ket{\psi_{ABE}}\bra{\psi_{ABE}}]$ and state $\rho_{AB}'$ after the measurement.  Furthermore, by regarding system $B$ as a trivial system, the analysis in Ref.~\cite{Coles12} also applies to our scenario.}

\subsection{Quantum randomness against classical information}\label{crand}
In the last part, we showed that the quantum randomness of a local measurement can be quantified by the best case uncertainty of a correlated party Eve. Such uncertainty has been quantified by the quantum conditional entropy.  We compare the result with an alternative measure of quantum randomness reported in Ref.~\cite{Yuan15Coherence}.  The setting is for the sake of clarity depicted in Fig.~\ref{cfigure}. The difference is that Eve performs a measurement with probability distribution $\{q^E_i\}, q^E_i=\text{Tr}[\rho_E\ket{e_i'}_E\bra{e_i'}]$ on her own system to predict Alice's measurement outcome. The best case uncertainty is then given by the {\it classical} conditional entropy:
\begin{equation}\label{Eq:classicalrandomness}
  R^C_I(\rho_A) = \min_{\rho_E,\{q^E_i\}} H(\{p_i\}|\{q^E_i\})_{\psi_{AE}},
\end{equation}
where the minimization runs over all the possible Eve's states and measurements.
   When Alice's system is in a pure state $\ket{\psi}_A= \sum_i \sqrt{p_i}\ket{i}$, the probability distributions of $A$ and $E$ are uncorrelated as the the global system is in a tensor product state. Hence, we have $
  R^C_I(\rho_A)=H(\{p_i\}|\{
q^E_i\})_{\psi_{AE}} =  H(\{p_i\})$ for any Eve's strategy. The quantity corresponds to the total randomness as expected.
For an arbitrary mixed state $\rho_A$, it turns out that the Eve's uncertainty on Alice's measurement is given by
\begin{equation}\label{Eq:classicalrandomness2}
  R^C_I(\rho_A) = \min_{\left\{p_i,\ket{\psi_i}_A\right\}} \sum p_i R^C_I\left(\ket{\psi_i}_A\right),
\end{equation}
where the minimization is over all possible decompositions of $\rho_A$. We briefly review the proof here. Given the spectral decomposition  $\rho_A = \sum_i\lambda_i\ket{a_i}\bra{a_i}$, then a purification of $\rho_A$ is $\ket{\psi}_{AE} = \sum_i\sqrt{\lambda_i}\ket{a_i}_A\otimes\ket{e_i}_E$. Here $\{\ket{e_i}_E\}$ is an orthogonal basis of Eve's system. Eve performs a projective measurement $\{\ket{e'_i}_E\}$ on her local system, then based on her measurement outcome $\ket{e'_i}_E$, the Alice's state is
\begin{equation}\label{}
  \ket{\psi_i}_A= \frac{1}{\sqrt{p_i}}\sum_{j}\sqrt{\lambda_{j}}\left|\langle e'_i\ket{e_j}\right|\ket{a_i}_A,
\end{equation}
where $p_i = \sum_{j}{\lambda_{j}}\left|\langle e'_i\ket{e_j}\right|^2$. As the state of Alice is pure for each outcome of Eve, the averaged quantum randomness is $\sum p_i R^C_I\left(\ket{\psi_i}_A\right)$. On the other hand, Eve can choose an arbitrary measurement basis, which determines a decomposition of $\rho_A$, to maximize his prediction success probability. Therefore, the quantum randomness measure should be optimized  over all the possible decompositions of $\rho_A$. When Eve performs a general measurement (POVM), we can always enlarge the system of Eve and  consider a projective measurement, then the proof follows accordingly q.e.d.\\

\begin{figure}[bht]
\centering
\resizebox{6cm}{!}{\includegraphics[scale=1]{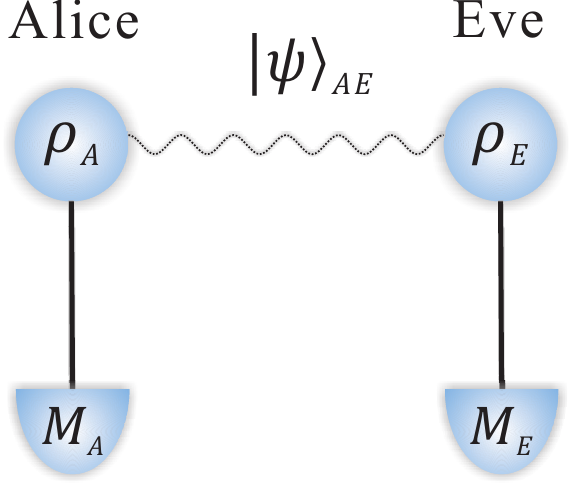}}
\caption{Alternative definition of Quantum randomness.  In a bipartite Alice-Eve system described by a pure state $\psi_{AE}$, the quantum randomness of a measurement performed by Alice on the system in the mixed state $\rho_A$ is given by the minimum amount of uncertainty Eve has on the measurement outcome {\it after performing a measurement on her own systems}. Such quantum uncertainty is quantified by the convex roof measure $R^C_I(\rho_A)$.}\label{cfigure}
\end{figure}

The quantum randomness measure obtained by convex roof extension of the pure state randomness is a measure of BCP coherence as well \cite{Yuan15Coherence}.

\subsubsection{Verifying the properties of $R_z^C$.}
Now we show that the intrinsic randomness $R_Z^C$, defined in Eq.~\eqref{Eq:classicalrandomness2}, satisfies the properties of coherence measures listed in Table.~\ref{Fig:Properties}. That is, the requirements of the measures for quantum coherence and intrinsic randomness are equivalent.

\paragraph{Proof of  (C1)}
In the language of generating randomness, the requirement (C1) in Table.~\ref{Fig:Properties} can be interpreted as that classical states generate no randomness. This is because that an incoherent state $\delta$, defined in Eq.~\eqref{Eq:sigma}, can be understood as a statistical mixture of classical states. We can easily verify that $R_Z^C(\delta) = 0$, since $R_Z^C(\delta) \leq  \sum_{i=1}^d p_iR_Z^C(\ket{i}\bra{i}) = 0$ from Eq.~\eqref{Eq:sigma} and $R_Z^C(\rho) \geq 0$ by definition. The stronger requirement (C1') implies that any non-classical states, which cannot be represented in the form of Eq.~\eqref{Eq:sigma}, could always be used to generate intrinsic randomness. Thus, this result answers why nonzero intrinsic randomness always indicates `quantumness' as discussed above. To prove that $R_Z^C(\rho)$ satisfies (C1'), consider a state $\rho\notin \mathcal{I}$ that has $R_Z^C\left(\rho\right) = 0$. From the definition of $R_Z^C$, there exists a decomposition $\rho = \sum_e p_e \ket{\psi_e}\bra{\psi_e}$ such that $R_Z^C(\ket{\psi_e}\bra{\psi_e}) = 0$ for all $e$. As any pure state with zero randomness is in the basis $I$, we have $\ket{\psi_e} = \ket{i_e} \in I$, and  $\rho = \sum_e p_e \ket{i_e}\bra{i_e}$, which belongs to the set $\mathcal{I}$, which is a contradiction. We can also show that the upper bound of its intrinsic randomness is given by $R_Z^C\left(\rho\right) \leq \log_2d$. The maximally coherent state $\ket{\Psi_d}$, defined in Eq.~\eqref{Eq:Psid}, has the largest intrinsic randomness.

\paragraph{Proof of (C2)}
The requirement (C2) implies a monotonicity property of incoherent operations. In the corresponding randomness picture, incoherent operations can be understood as classical operations that map one zero intrinsic randomness (classical) state to another one. An interpretation of (C2a) is that such classical operations should not increase the randomness of a given state. While (C2b) requires that the randomness cannot increase on average when probabilistic strategies are considered. Let us quickly check why (C2b) is true for the pure state case. For a pure state $\rho$, the randomness measure $R_Z^C(\rho)$ equals the relative entropy of coherence $C_{\mathrm{rel, ent}}(\rho)$, whose monotonicity has been proved \cite{Baumgratz14}. That is, we have
\begin{equation}\label{Eq:Monotonicitypure}
  R_Z^C\left(\ket{\psi}\right) \geq \sum_n p_nR_Z^C\left(\ket{\psi_n}\right),
\end{equation}
where $\ket{\psi_n} = K_n\ket{\psi}/\sqrt{p_n}$, and $p_n = \mathrm{Tr}\left[K_n\ket{\psi}\bra{\psi}\right]$. This is because for   a pure state $\rho$,  the intrinsic randomness $R_Z^C(\rho)$ equals the relative entropy coherence measure $C_{\mathrm{rel, ent}}(\rho)$ \cite{Baumgratz14}, whose monotonicity has already been proved.

For a general mixed state $\rho$, suppose that the optimal decomposition that achieves the minimum in Eq.~\eqref{Eq:classicalrandomness2} is given by $\rho = \sum_e p_e \ket{\psi_e}\bra{\psi_e}$. Then, we have
\begin{equation}\label{}
  R_Z^C\left(\rho\right) = \sum_e p_e R_Z^C(\ket{\psi_e})
\end{equation}
Now suppose that the incoherent operation defined in the main text is acted on $\rho$. What we need to prove is that
\begin{equation}\label{}
\begin{aligned}
  R_Z^C\left(\rho\right) \geq \sum_n p_nR_Z^C(\rho_n).
\end{aligned}
\end{equation}
where $  \rho_n = {K_n\rho K_n^\dag}/{p_n}$ and $p_n = \mathrm{Tr}\left[ K_n\rho K_n^\dag\right]$.
As $\rho = \sum_e p_e \ket{\psi_e}\bra{\psi_e}$, we have
\begin{equation}\label{}
\begin{aligned}
  \rho_n &= \frac{K_n\rho K_n^\dag}{p_n} \\
  &= \sum_e \frac{p_e}{p_n} {K_n\ket{\psi_e}\bra{\psi_e} K_n^\dag} \\
  &= \sum_e \frac{p_e}{p_n}p_{en}\rho_{en}
\end{aligned}
\end{equation}
where, we denote $p_{en} = \mathrm{Tr}[K_n\ket{\psi_e}\bra{\psi_e} K_n^\dag]$, $\rho_{en} = {K_n\ket{\psi_e}\bra{\psi_e} K_n^\dag}/{p_{en}}$, and we have $p_n = \sum_e p_ep_{en}$. Then, we can finish the proof
\begin{equation}\label{}
\begin{aligned}
  R_Z^C\left(\rho\right) &= \sum_e p_e R_Z^C(\ket{\psi_e}) \\
  &\geq \sum_{e} p_e \sum_np_{en}R_Z^C(\rho_{en})\\
   &= \sum_{n} p_n\sum_e\frac{p_e p_{xn}}{p_n}R_Z^C(\rho_{en})\\
  &\geq \sum_{n} p_nR_Z^C\left(\sum_e\frac{p_e p_{en}}{p_n}\rho_{en}\right)\\
   &= \sum_n p_nR_Z^C(\rho_n),
    \end{aligned}
\end{equation}
where the first inequality is based on the conclusion for pure states in Eq.~\eqref{Eq:Monotonicitypure} and the last inequality is due to the convexity of $R_Z^C$.

\paragraph{Proof of (C3)}
The convexity property (C3) can be understood as a requirement on the randomness generation process. In other words, the randomness cannot increase on average by statistically mixing several states. With the convex roof definition of $R_Z^C(\rho)$, given in Eq.~\eqref{Eq:classicalrandomness2}, we can easily verify the convexity property (C3). The proof follows directly by considering a specific decomposition of $\rho = \sum_n p_n\rho_n$ in (C3).
Note that, the property (C2a) can be derived when (C2b) and (C3) are fulfilled, thus we also prove (C2a) for $R_Z^C(\rho)$.

In summary, we prove that the intrinsic randomness $R_Z^C(\rho)$ indeed measures the strength of coherence. A state with stronger coherence would therefore indicate larger randomness in measurement outcomes, and vice versa.

\subsection{Qubit example}
Here, we derive the intrinsic randomness formula of qubit state. We denote the Pauli matrices by $\sigma_i, \sigma_x, \sigma_y, \sigma_z$. When measured in the $\sigma_z$ basis, the intrinsic randomness for pure a qubit state $\ket{\psi} = \alpha \ket{0}+ \beta \ket{1}$ is given by
\begin{equation}\label{}
  R_z^C(\ket{\psi}) = H(|\alpha|^2) = H(|\beta|^2),
\end{equation}
where $H(p) = p\log p + (1-p)\log(1-p)$. If we define $n_x = \bra{\psi}\sigma_x\ket{\psi} = \alpha^*\beta + \alpha\beta^*$ and $n_y = \bra{\psi}\sigma_y\ket{\psi} = -i\alpha^*\beta + i\alpha\beta^*$, then it is easy to check that
\begin{equation}\label{Eq:RZ}
R_{z}^C(\ket{\psi})= H\left(\frac{1+\sqrt{1-n_x^2 - n_y^2}}{2}\right).
\end{equation}

For a general mixed state $\rho$, we can follow the  method for  deriving the entanglement of formation \cite{Wootters98}. In this case, we need to first define $\ket{\tilde{\psi}} = \sigma_x\ket{\psi^*} = \beta^*\ket{0}+  \alpha^*\ket{1}$, and the coherent concurrence by
\begin{equation}\label{}
  C_z(\ket{\psi}) = |\langle\psi|\tilde{\psi}\rangle| = 2|\alpha\beta|.
\end{equation}
Then it is easy to check that
\begin{equation}\label{Eq:RZ}
R_{z}^C(\ket{\psi})= H\left(\frac{1+\sqrt{1-C_z^2}}{2}\right).
\end{equation}
The randomness $R_{I}^C\left(\rho\right)$ can be obtained according to Eq.~\eqref{Eq:RZ} by first calculating the coherent concurrence. Follow the method of deriving the entanglement of formation, the  $C_z$ value can be obtained by $ C_z =  |\sqrt{\eta_1}-\sqrt{\eta_2}|$, where $\eta_1$ and $\eta_2$ are the eigenvalues of the matrix $M = \rho \sigma_x \rho^* \sigma_x$. In the Bloch sphere representation, the value of $C_z$ of a quantum state $\rho = (\sigma_i + n_x \sigma_x + n_y\sigma_y +n_z \sigma_z )/2$ can be calculated by
\begin{equation}\label{Eq:concurrence}
  C_z = \sqrt{n_x^2 + n_y^2}.
\end{equation}

Compared to the $l_1$ norm coherence measure $C_{l_1}$ \cite{Baumgratz14}, which is defined by the sum of the off-diagonal elements
\begin{equation}\label{}
C_{l_1}(\rho) = \sum_{i\neq j}|\rho_{ij}|,
\end{equation}
one can easily check that $C_{l_1}(\rho)$ equals the concurrence $C_z$ for the qubit case. This is because
\begin{equation}\label{}
\begin{aligned}
C_{l_1}(\rho) &= |\bra{0}\rho\ket{1}| + |\bra{1}\rho\ket{0}| \\
&= |\frac{1}{2}(n_x - in_y)| + |\frac{1}{2}(n_x + in_y)|\\
&= \sqrt{n_x^2 + n_y^2}.
\end{aligned}
\end{equation}
We conjecture that the coherence concurrence can be generalized to an arbitrary high dimensional space by following a similar method to that used for the entanglement concurrence \cite{Rungta01, Audenaert01, badziag2002concurrence}.

\subsection{Comparison between the two randomness measures}
Let us compare  the two quantities $R^C_I(\rho_A) , R^Q_I(\rho_A)$ in a simple example about  a qubit system.  In the Bloch sphere representation, $\rho_A = (I+\vec{n}\cdot\vec{\sigma})/2$, where $\vec{n} = (n_x,n_y,n_z)$ and $\vec{\sigma} = (\sigma_x,\sigma_y, \sigma_z)$ are the Pauli matrices. Supposing that the measurement basis is the $\sigma_Z$ eigenbasis, which is denoted by $\{\ket{0},\ket{1}\}$, then we obtain
\begin{eqnarray}\label{Eq:ef}
R_{z}^C(\rho_A)&=& H\left(\frac{1+\sqrt{1-n_x^2 - n_y^2}}{2}\right)\\
R_{z}^Q(\rho_A)&=& H\left(\frac{n_z + 1}{2}\right) - H\left(\frac{|n| + 1}{2}\right),\nonumber
\end{eqnarray}
where $|n| = \sqrt{n_x^2+n_y^2+n_z^2}$ and $H$ is the binary entropy.
Specifically, for the state $\rho_A(v) = v\ket{+}\bra{+} + \frac{1 - v}{2}\mathbb{I}$, where $\ket{+} = (\ket{0}+\ket{1})/2, v\in[0,1], \vec{n}(v) = (v, 0, 0)$, we have
\begin{equation}\label{Eq:ef}
\begin{aligned}
R_{z}^C(\rho_A)&= H\left(\frac{1+\sqrt{1-v^2}}{2}\right),\\
R_{z}^Q(\rho_A)&= 1 - H\left(\frac{v + 1}{2}\right).
\end{aligned}
\end{equation}
In Fig.~\ref{fig:RQC}, we plot the quantum randomness versus the mixing parameter $v$. As expected, the quantum randomness measure $R_z^Q$ obtained through the a fully quantum picture is  smaller than  $R_z^C$, which is  derived by the measurement-based method, while they both vanish when the state is incoherent, and are equal to the Shannon entropy in the pure state case.

\begin{figure}[bht]
\centering
\resizebox{12cm}{!}{\includegraphics[scale=1]{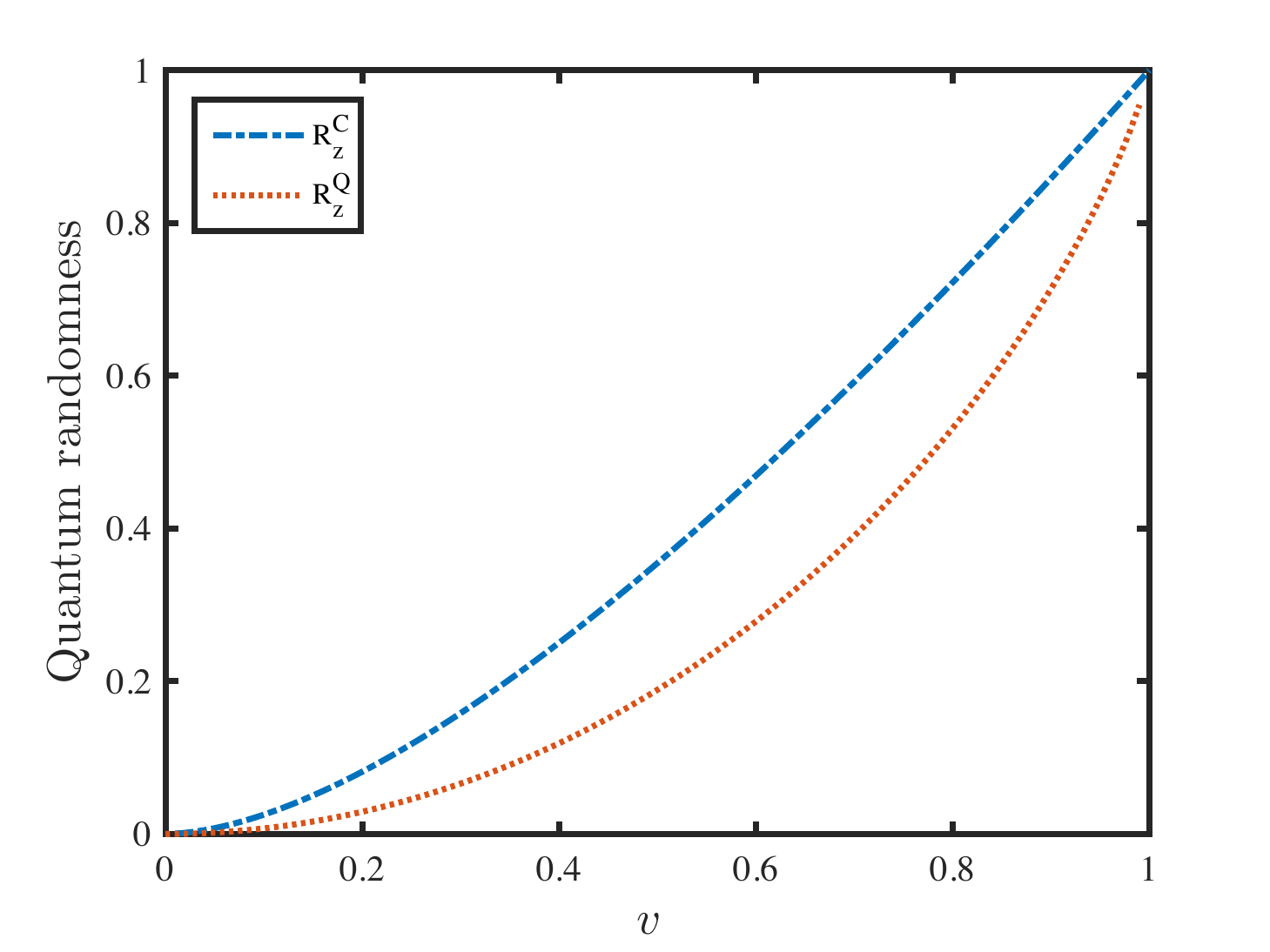}}
\caption{Comparison of the measures of quantum randomness $R_z^Q$ (red dotted line) and $R_z^C$ (blue dot-dashed line) in the qubit state $\rho_A(v)=v\ket{+}\bra{+} + \frac{1 - v}{2}\mathbb{I}$ versus the mixing parameter $v$.}\label{fig:RQC}
\end{figure}

\section{Coherence or randomness distillation}
When Alice performs a projective measurement $P_I$ on $N$ identical pure states $\ket{\psi} = \sum_i a_i\ket{i}$, she will obtain $N$ i.i.d.~random variables $A_1, A_2, \dots, A_N$. For the state $\ket{\psi}$ that is not maximally coherent, the randomness of the measurement outcomes is biased. Then, as shown in Fig.~\ref{Fig:Extractor}(a), Alice can perform a randomness extraction process to transform the $N$ biased random numbers to $l \approx NR_Z^C(\ket{\psi})$ almost uniformly distributed random bits.

\begin{figure}[hbt]
\centering \resizebox{10cm}{!}{\includegraphics{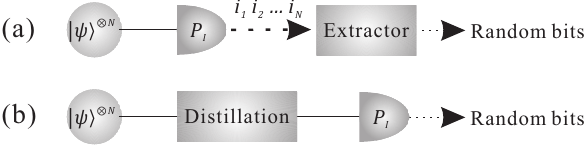}}
\caption{Random number extraction and coherence distillation. The randomness extraction process can be replicated by first distilling the coherence of the quantum state. Measurement outcomes will directly produce uniformly random bits.} \label{Fig:Extractor}
\end{figure}

We show in Fig.~\ref{Fig:Extractor}(b) that the extraction can be equivalently performed before measurement. Now, the extraction becomes a quantum procedure, which we call \emph{quantum extraction}. Considering the equivalence between intrinsic randomness and quantum coherence, quantum extraction can be regarded as a procedure of \emph{coherence distillation}. This concept resembles the distillation procedure of another (more popular) quantumness measure---entanglement \cite{Bennett96}.

With quantum extraction, we can first distil the input state $\ket{\psi} = \sum_i a_i\ket{i}$ into the maximally coherent state $\ket{\Psi_2} = (\ket{0} + \ket{1})/\sqrt{2}$. Then, we can directly obtain uniformly distributed random bits by measuring the maximally coherent state. For $N$ copies of $\ket{\psi}$, it is shown in Supplementary Materials that we can asymptotically obtain $l$ copies of $\ket{\Psi_2}$, where $l$ and $N$ satisfy the following condition,
\begin{equation}\label{Eq:CohDistill}
{l}/{N} \approx R_Z^C\left(\ket{\psi}\right).
\end{equation}
Taking a pure qubit input state as an example, the distillation procedure is summarized as follows.
\begin{enumerate}
\item
Prepare $N$ copies of qubit state $\ket{\psi}^{\otimes N} = \left(\alpha \ket{0} + \beta \ket{1}\right)^{\otimes N}$, which can be binomially expanded on the computational basis. There are $N+1$ distinct coefficients, $\beta^N, \alpha^{1}\beta^{N-1}, \dots, \alpha^N$, corresponding to different subspaces that have the same number of $\ket{0}$ or $\ket{1}$.

\item
Perform a projection measurement to distinguish between those subspaces. For the $k$th subspace, which has coefficient $\alpha^{N-k}\beta^{k}$, the measurement probability is given by $p_k = {N\choose k}|\alpha|^{2(N-k)}|\beta|^{2k}$. The resulting quantum state of the $k$th outcome corresponds to a maximally coherent state $\ket{\Psi_{D_k}}$  of dimension $D_k = {N\choose k}$.

\item
Suppose that $2^r\leq D_k< 2^{r+1}$, then we can directly project onto the $2^r$ subspace and convert to $r$ copies of $\ket{\Psi_2}$ as desired.
\end{enumerate}

To see why $r/N$ equals the randomness of $\ket{\psi}$ on average, we only need to take account of the operations that cause a loss of coherence. As shown in Supplementary Materials, the only two projection measurements lose negligible amount of coherence, thus we asymptotically have $NR_Z^C(\ket{\psi}) \approx r$.

{In Appendix~\ref{App:coherence}, we further extend the definition of distillable coherence to mixed quantum states. Compared to the definition of the regulated entanglement of formation \cite{hayden2001asymptotic}, we also define coherence of formation and conjecture that it equals the regulated intrinsic randomness measure,
\begin{equation}\label{Eq:RNrho}
  {R_Z^C}^\infty(\rho) = \lim_{N\rightarrow\infty}\frac{ R_Z^C\left(\rho^{\otimes N}\right)}{N}.
\end{equation} }

\subsection{Comparison with entanglement}
As shown in Table~\ref{coherenceEntanglement}, there exist strong similarities between the frameworks of coherence and entanglement (see also Ref.~\cite{Streltsov15}), our study can be regarded as an extension of the convex roof measure from entanglement to coherence. Similar to the case of EOF, as a convex roof measure for coherence, we expect our proposed measure to play an important role in the research of coherence.

\begin{table*}[hbt]\footnotesize
\centering
\caption{Comparing  the frameworks of coherence and entanglement. DI: device-independent; MDI: measurement-device-independent; QKD: quantum key distribution; QRNG: quantum random number generation.}\label{ps}
\begin{tabular}{ccc}
\hline
Properties & Coherence & Entanglement \\
\hline
Classical operation & Inherent operation \cite{Baumgratz14} & LOCC \cite{Bennett96}\\
Classical state & Incoherent state, Eq.~\eqref{Eq:sigma} & Separable state\\
Distance measure & $C_{\mathrm{rel, ent}}(\rho)$, Eq.~\eqref{eq:relativeentropy} & Relative entropy distance \cite{Vedral98}\\
Convex roof measure & $R_Z^C(\rho)$, Eq.~\eqref{Eq:classicalrandomness2} & EOF \cite{Bennett96, Hill97, Wootters98}\\
Distillation & Coherence distillation (Methods) & Entanglement distillation \cite{Rains98, Horodecki09}\\
Formation (cost) & Coherence formation & Entanglement cost \cite{hayden2001asymptotic, Horodecki09}\\
Foundation tests & Further research direction & Nonlocality tests \cite{bell, CHSH}  \\
Interconvertibility &  \cite{Du15a, Du15b} & Deterministic \cite{Nielsen99}, stochastic \cite{Dur00}\\
Catalysis effect & Further research direction & Entanglement catalysis \cite{Jonathan99, Eisert00}\\
Witness & Further research direction & Entanglement witness (EW) \\
DI applications & Further research direction & DIQKD \cite{Mayers98, acin2007device}, DIQRNG \cite{Vazirani12}\\
MDI applications & Further research direction & MDIQKD \cite{Lo2012MDI, Braunstein12}, MDIEW \cite{Branciard13, Yuan14}\\
\hline
\end{tabular}\label{coherenceEntanglement}
\end{table*}

For further research directions, it is interesting to extend the framework of entanglement to coherence. An incomplete list of comparison between the two are shown in Table~\ref{coherenceEntanglement}. For instance, it is interesting to see whether $C_{\mathrm{rel, ent}}(\rho)$ and $R_Z^C(\rho)$ are the unique lower and upper bounds of all coherence measures after regularization, and whether they can coincide.
Another interesting and related question is that of quantifying the coherence for an unknown quantum state, similar to the task of using an entanglement witness for quantification. The coherence measure $R_{I}(\rho)$ given in Eq.~\eqref{Eq:classicalrandomness2} ensures the true randomness when measuring a state $\rho$ in the $I$ basis. Such a technique can be utilized to construct a semi self-testing quantum random number generator. A straightforward way to do this is to first perform tomography on the to-be-measured state $\rho$ and then estimate the randomness of the $I$ basis measurement outcomes according to Eq.~\eqref{Eq:classicalrandomness2}. As the coherence measure $R_{I}(\rho)$ quantifies the output randomness in a measurement, our result can also be applied in other randomness generation scenarios \cite{colbeck2009quantum, gabriel2010generator, Xu:QRNG:2012, Vazirani12}.

 {The definition of coherence and intrinsic randomness is based on a specific computational basis. In this perspective, the quantum feature can be quantified by the superposition strength on the measurement basis. Alternatively, we can define similar quantumness as the ability of measurements. For an arbitrary pure quantum state, if we can choose the measurement basis that is complementary to the state, quantum feature similar to coherence can also be maximally revealed. The definitions of coherence based on the property of quantum state with a given measurement basis and the property of measurement is similar to the relationship between the pictures of Schrodinger and Heisenberg. The current definition of coherence thus follows from the routine of the Schrodinger¡¯s picture.}

 {We also investigate intrinsic randomness without specifying a measurement basis. In this case, we consider the scenario that an optimal measurement basis is chosen to maximize the output randomness. Here, we do not assume that the choice of measurement basis is secret from Eve¡¯s point of view. Thus, we still have to consider the minimization of intrinsic randomness on the chosen measurement basis. In this case, this basis-independent intrinsic randomness can be defined as $R(\rho) = \max_I R_Z^C\left(\rho\right)$. In the qubit example, we show that the basis-independent intrinsic randomness is related to the purity of a quantum state. We thus demonstrate that intrinsic randomness can be used to quantify other quantum features.
}

\section{Basis independent randomness and coherence}
Here, we quantify the intrinsic randomness under a different scenario. When the measurement basis has not been specified, we can still consider the intrinsic randomness. In this case, Alice can choose an optimal measurement basis to maximize the output randomness. Here, we do not assume that the choice of measurement basis is secret from Eve¡¯s point of view. Thus, we still have to consider the minimization of intrinsic randomness on the chosen measurement basis. In this case, this basis-independent intrinsic  randomness can be defined as
\begin{equation}\label{Eq:arbirary}
  R(\rho) = \max_I R_I\left(\rho\right).
\end{equation}
Here, the maximization is over all possible projective measurement basis $I$ and the randomness measure could be either $R_I^C$ or $R_I^Q$. We do not consider general POVM measurement for Alice as for it generally requires ancillary quantum states which might introduce coherence and hence randomness. The new definition $R(\rho)$ still represents the existence of quantum effects. That is, nonzero $R(\rho)$ will always indicate the existence of quantumness, although $R(\rho)$ does not quantify the coherence of quantum states in this instance, since the coherence is defined on a specific basis.


As an example, we consider the basis independent randomness with randomness measure $R_I^C$ for qubit state $\rho$.
In this example, we also give a direct expression of $R(\rho)$ for a qubit state $\rho$. As the randomness measure $R_I^C$ is a function of $C_I$ according to Eq.~\eqref{Eq:RZ}, we can thus similarly define $C = \max_IC_I$. In the following, we will focus on calculating $C$ and the basis independent randomness $R(\rho)$ is a direct function of it.

An arbitrary measurement basis $I$ can be considered as a unitary transformation of the measurement basis on the original $\delta_z$ basis.
Equivalently, we can suppose that the measurement basis is unchanged, while a unitary transformation acts upon the to-be-measured quantum state. Suppose that the original state $\rho$ has a spectral decomposition given by
\begin{equation}\label{}
  \rho = \lambda_1 \ket{\Psi_1}\bra{\Psi_1} + \lambda_2 \ket{\Psi_2}\bra{\Psi_2}.
\end{equation}
Thus, a unitary transformation on the state $\rho$ would only transform the eigenstates $\ket{\Psi_1}$ and $\ket{\Psi_2}$ to another basis and leave the eigenvalues $\lambda_1$ and $\lambda_2$ unchanged. In this case, we can still work on the $\delta_z$ basis, and obtain
\begin{equation}\label{}
  R(\rho) = \max_{\ket{\Psi_1}, \ket{\Psi_2}} R(\rho = \lambda_1 \ket{\Psi_1}\bra{\Psi_1} + \lambda_2 \ket{\Psi_2}\bra{\Psi_2}).
\end{equation}
Now, the maximization over all possible measurement basis is equivalently achieved over all possible eigenstates $\ket{\Psi_1}$ and $\ket{\Psi_2}$ with given eigenvalues $\lambda_1$ and $\lambda_2$.

We assume that $\left| \Psi_1 \right \rangle=a\left|0 \right \rangle+b\left| 1\right \rangle$ and $\left| \Psi_2 \right \rangle=b^*\left|0 \right \rangle-a^*\left| 1\right \rangle$, where the normalized coefficients $a,b$ are complex numbers yielding $a^*a+b^*b=1$ and $\left|0 \right \rangle$ and $\left|1 \right \rangle$ are eigenstates of $\delta_z$.
Then a general density matrix $\rho$ can be represented by
\begin{equation}
\rho  = \left(
\begin{array}{cc}
\lambda_1 \left|a\right|^2+\lambda_2 \left|b\right|^2 & (\lambda_1-\lambda_2)ab^* \\
(\lambda_1-\lambda_2)a^*b & \lambda_1 \left|b\right|^2+\lambda_2 \left|a\right|^2 \\
\end{array}
\right).
\end{equation}

To get $C$, we need to calculate $M$, which can be written as
\begin{equation}
\begin{aligned}
M =& \left(
\begin{array}{cc}
(\lambda_1-\lambda_2)ab^* & \lambda_1 \left|a\right|^2+\lambda_2 \left|b\right|^2 \\
\lambda_1 \left|b\right|^2+\lambda_2 \left|a\right|^2 & (\lambda_1-\lambda_2)a^*b \\
\end{array}
\right)\cdot
\\
&
\left(
\begin{array}{cc}
(\lambda_1-\lambda_2)a^*b & \lambda_1 \left|a\right|^2+\lambda_2 \left|b\right|^2 \\
\lambda_1 \left|b\right|^2+\lambda_2 \left|a\right|^2 & (\lambda_1-\lambda_2)ab^* \\
\end{array}
\right)
\end{aligned}
\end{equation}
By denoting
\begin{equation}
\begin{aligned}
 A&=(\lambda_1-\lambda_2)\left|a\right|^2 \left|b\right|^2 \\
 B&=(\lambda_1 \left|a\right|^2+\lambda_2 \left|b\right|^2)(\lambda_1 \left|b\right|^2+\lambda_2 \left|a\right|^2) \\
 C&=2(\lambda_1-\lambda_2)ab^*(\lambda_1 \left|a\right|^2+\lambda_2 \left|b\right|^2) \\
 D&=2(\lambda_1-\lambda_2)a^*b(\lambda_1 \left|b\right|^2+\lambda_2 \left|a\right|^2)
\end{aligned}
\end{equation}
The matrix $M$ can be denoted by
\begin{equation}
M = \left(
\begin{array}{cc}
A+B & C \\
D & A+B \\

\end{array}
\right)
\end{equation}
And $\eta_1$ and $\eta_2$ will be the two roots of the equation
\begin{equation}
\eta^2-2(A+B)\eta+(A+B)^2-CD=0
\end{equation}

Now we can calculate the square of $C(\rho)$ by
\begin{equation}
\begin{aligned}
C(\rho)^2 &=(\sqrt{\eta_1}-\sqrt{\eta_2})^2\\
 &=\eta_1+\eta_2-2\sqrt{\eta_1\eta_2} \\
 &=2(A+B)-2\sqrt{(A+B)^2-CD}
\end{aligned}
\end{equation}
Here we notice that $CD=4AB$, hence
\begin{equation}
C(\rho)^2= 2(A+B)-2\left|A-B\right|
\end{equation}
where $A-B=-\lambda_1\lambda_2(\left|a\right|^2+\left|b\right|^2)^2 \leq 0$, thus
\begin{equation}
\begin{aligned}
C(\rho)&=\sqrt{4A}=\left|\lambda_1-\lambda_2\right|(2\left|a\right|\left|b\right|)\\
&\leq \left|\lambda_1-\lambda_2\right|(\left|a\right|^2+\left|b\right|^2)\\
&=\left|\lambda_1-\lambda_2\right|
\end{aligned}
\end{equation}

Notice that, the equal sign is taken for a measurement basis complementary to the eigenbasis of $\rho$. For $\rho = (I + n_x \delta_x + n_y\delta_y +n_z \delta_z )/2$, we have that $\lambda_1 = (1+n)/2$ and $\lambda_2 = (1-n)/2$, where $n = \sqrt{n_x^2 + n_y^2 + n_z^2}$. Therefore, we have
\begin{equation}
C(\rho)= n.
\end{equation}

Thus this value can be physically related to the degree of how mixed the state $\rho$ is.

\chapter{Quantum Bernoulli Factory}
Bernoulli factory \cite{Asmussen92, RSA:RSA20333} is a simple yet interesting task in classical randomness processing. This chapter discusses quantum Bernoulli factory \cite{dale2015provable} and show its fundamental difference from classical Bernoulli factory. The key difference lie on the coherent superposition of states. To demonstrate this difference, we also present a theoretical protocol and an experiment verification with superconducting qubits \cite{Yuan2016Bernoulli}.
\section{Theoretical protocol}
Coherent superposition of different states, coherence, is a peculiar feature of quantum mechanics that distinguishes itself from Newtonian theory. In different scenarios, coherence exhibits as various quantum resources, such as entanglement \cite{Horodecki09}, discord \cite{Modi12}, and single-party coherence \cite{Baumgratz14}. In many quantum information tasks, the common resource leading to quantum advantage is multipartite quantum correlations. For instance, entanglement plays a crucial role in quantum key distribution \cite{bb84,Ekert1991}, teleportation \cite{Teleport1993}, and computation \cite{Shor97,Grover96}. While the essence of multipartite correlation originates from coherent superposition, it is natural to expect the essence of quantum advantage to also originate from coherence. This raises a fundamental question: Can quantum advantage be obtained without using multipartite correlations?

In randomness generation, it has been shown that coherence is the essential resource for generating true random numbers \cite{Yuan15Coherence}. It is thus natural to expect coherence to be a
resource for displaying quantum advantages in certain randomness related tasks. Remarkably, in a recent work by Dale \emph{et al.}, a rather simple task of randomness processing is proposed to show that coherence yields a provable quantum advantage over classical stochastic physics \cite{dale2015provable}. In this randomness processing task, a classical coin, see Fig.~\ref{fig:coin}(a), corresponds to a classical machine that produces independent and identically distributed random variables where each one has the binary values, head (0) and tail (1). A coin is called $p$-coin if the probability of producing a head is $p$, where $p\in[0,1]$. Given an unknown $p$-coin, an interesting question is whether one can construct an $f(p)$-coin, where $f(p)$ is a function of $p$ and $f(p)\in[0,1]$. Such construction processing is called a Bernoulli factory \cite{Asmussen92, RSA:RSA20333}.

\begin{figure}[bht]
\centering
\resizebox{8cm}{!}{\includegraphics[scale=1]{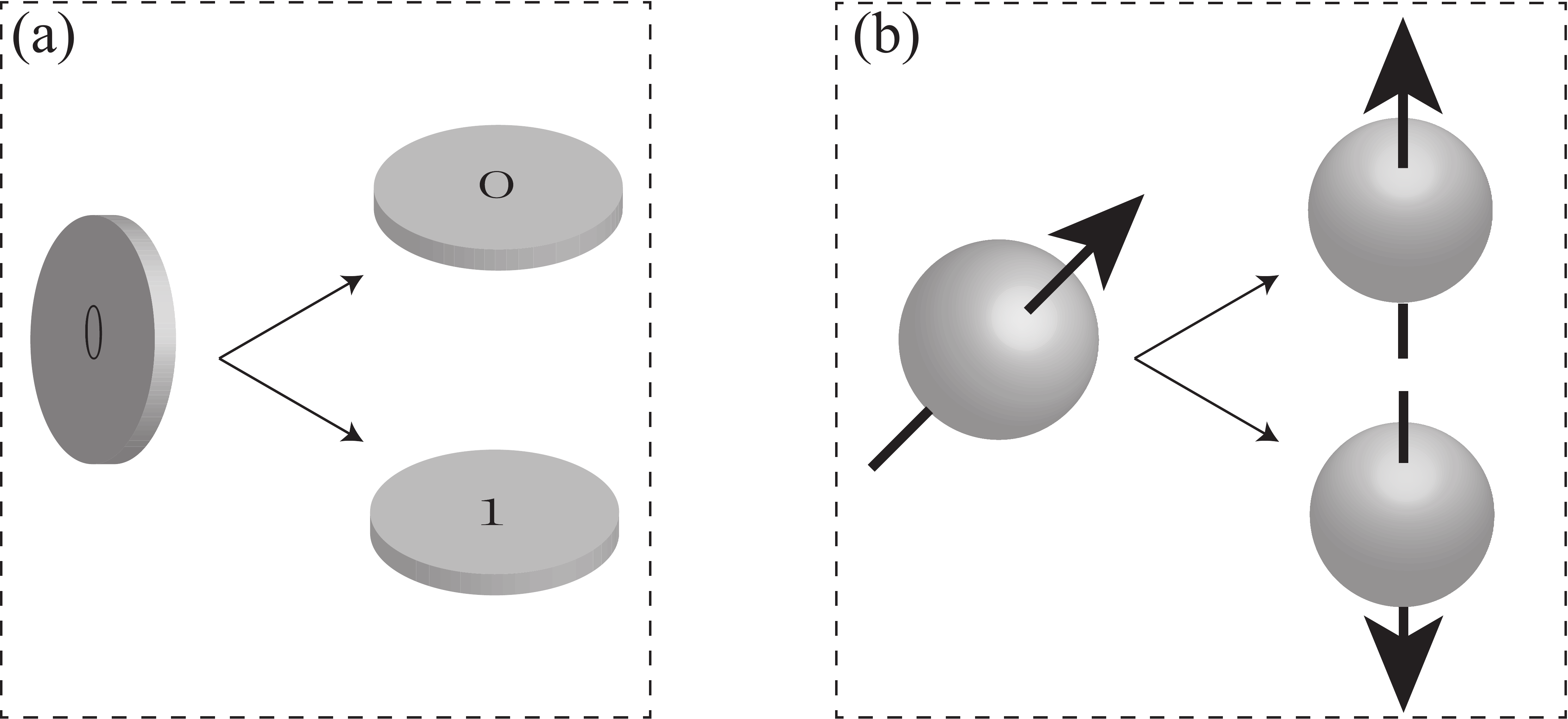}}
\caption{Classical and quantum coin. For a given $p$ value, (a) classical and (b) quantum $p$-coin corresponds to two different ways of encoding $p$, see Eqs.~\eqref{Eq:classicalcoin} and \eqref{Eq:quantumcoin}, respectively. The key difference lies in whether there is coherence in the computational basis.}\label{fig:coin}
\end{figure}

Let us take $f(p)=1/2$ for example, which was solved by von Neumann with a rather simple but heuristic strategy \cite{von195113}. Flip the $p$-coin ($p\ne0$) twice. If the outcomes are the same start over; otherwise, output the second coin value as the $1/2$-coin output. Therefore, the function of $f(p)=1/2$ can be constructed from an arbitrary unknown $p$-coin. As a generalization, a natural question involves which kind of function $f(p)$ can be constructed from an unknown $p$-coin. This classical Bernoulli factory problem was solved by Keane and O'Brien \cite{Keane94}. Generally speaking, a necessary condition for $f(p)$ being constructible is that $f(p)\neq 0$ or $1$ when $p\in(0,1)$. The function $f(p)=1/2$ satisfies this condition, while there are many other examples that violate it. For instance, surprisingly, the simple ``probability amplification'' function $f(p)=2p$ \footnote{A complete definition is: $f(p)=2p$ when $p\in[0,1/2]$ and $f(p)=2(1-p)$ when $p\in(1/2,1]$.} does not satisfy the constructible condition, where we have $f(1/2)=1$. Therefore, there is no classical method to construct an $f(p)=2p$-coin.

In the language of quantum mechanics, a $p$-coin corresponds to a machine that outputs identically mixed qubit states,
\begin{equation}\label{Eq:classicalcoin}
  \rho_C = p\ket{0}\bra{0} + (1-p)\ket{1}\bra{1},
\end{equation}
where $p\in[0,1]$, and $Z = \{\ket{0},\ket{1}\}$ is the computational basis denoting head and tail, respectively. As $p$ is generally unknown, we can regard $\rho_C$ as a classical way of encoding an unknown parameter $p$. A measurement in the $\{\ket{0},\ket{1}\}$ basis would output a head or a tail with a probability according to $p$ and $1-p$, respectively. On the other hand, a quantum way of encoding $p$, see Fig.~\ref{fig:coin}(b), can be a coherent superposition of $\ket{0}$ and $\ket{1}$, i.e., $\rho_Q = \ket{p}\bra{p}$ with
\begin{equation}\label{Eq:quantumcoin}
  \ket{p} = \sqrt{p}\ket{0}+\sqrt{1-p}\ket{1}.
\end{equation}
Following the nomenclature in Ref.~\cite{dale2015provable}, we call such a quantum coin a \emph{quoin}. It is straightforward to see that a $p$-coin can always be constructed from a $p$-quoin by measuring it in the $Z$ (computational) basis. Thus, classically constructible (via coins) $f(p)$ functions are also quantum mechanically constructible (via quoins), while a really interesting question is whether the set of quantum constructible functions (via quantum a Bernoulli factory) is strictly larger than the classical set.


In Ref.~\cite{dale2015provable}, Dale \emph{et al.} have theoretically proved the necessary and sufficient conditions for $f(p)$ being quantum constructible. Specifically, they show that there are functions, for instance $f(p)=2p$, which are impossible to construct classically, but can be efficiently realized in the presence of $p$-quoins. Therefore, they provide a positive answer to this problem where quantum resources are strictly superior to classical ones. The protocol for generating the $f(p)=2p$ function relies on Bell state measurement on two quoins, which essentially establish entanglement between the two quoins.

Now, we are interested in seeing whether such a quantum advantage persists even when multipartite correlations, such as entanglement, are absent. Thus, we only allow single qubit operations. Without two-qubit operations, it turns out that constructing the $f(p)=2p$ function will require many copies of qubits defined in Eq.~\eqref{Eq:quantumcoin} and the convergence could be poor. In this chapter, we propose another function that is impossible with classical means but feasible with only limited number of single-qubit operations.


\paragraph{The protocol for quantum Bernoulli factory}
Here, we analyze a classically impossible $f(p)$ function defined by
\begin{equation}\label{Eq:fp}
  f(p) = 4p(1-p).
\end{equation}
For $p=1/2$, we have $f(p)=1$ which means that this function is classically unachievable. On the other hand, it is straightforward to check that the $f(p)$-function satisfies the requirements for beign quantum constructible \cite{dale2015provable}. Given a $p$-quoin, we explicitly present an efficient protocol for generating an $f(p)$-coin as follows.

\begin{enumerate}[Step 1]
\item
\textbf{Generate a $p$-coin}:

When measuring a $p$-quoin, as given by Eq.~\eqref{Eq:quantumcoin}, in the $Z$ basis, the probabilities of obtaining $0$ and $1$ are $p$ and $1-p$, respectively.

\item
\textbf{Generate a $q$-coin}, where $q = \left[1+2\sqrt{p(1-p)}\right]/2$:

When measuring a $p$-quoin in the $X=\{(\ket{0}+\ket{1})/\sqrt{2},(\ket{0}-\ket{1})/\sqrt{2}\}$ basis, the probabilities of obtaining $(\ket{0}+\ket{1})/\sqrt{2}$ and $(\ket{0}-\ket{1})/\sqrt{2}$ are $\left[1+2\sqrt{p(1-p)}\right]/2$ and $\left[1-2\sqrt{p(1-p)}\right]/2$, respectively.

\item
\textbf{Construct an $m$-coin from a $p$-coin}, where $m=2p(1-p)$: toss the $p$-coin twice, output head if the two tosses are different and tail otherwise.

The probability of output two different tossing result is
\begin{equation}\label{}
m=\Pro{head}\Pro{tail} + \Pro{tail}\Pro{head} = 2p(1-p).
\end{equation}
Similarly, one can construct an $n$-coin from a $q$-coin, where $n=2q(1-q)=1/2-2p(1-p)$.

\item
\textbf{Construct an $s$-coin from an $m$-coin, where $s=m/(m+1)$}: toss the $m$-coin twice, if the first toss is tail then output tail; otherwise if the second toss is tail, output head; otherwise, repeat this step.

Denote the probability of outputting head and tail by $\Pro{H}$ and $\Pro{T}$, respectively, then,
\begin{equation}\label{}
\begin{aligned}
\Pro{H} &= \Pro{head}\Pro{tail} + \Pro{head}(1-\Pro{tail})\Pro{H}\\
& = m(1-m) + m^2\Pro{H}.
\end{aligned}
\end{equation}
Solving this equation, we have
\begin{equation}\label{}
s=\Pro{H} = \frac{m}{m+1}.
\end{equation}
Similarly, one can construct a $t$-coin from an $n$-coin, where $t=n/(n+1)$.

\item
\textbf{Construct an $f(p)=4p(1-p)$-coin}: first toss the $s$-coin and then the $t$-coin. If the first toss is head and the second toss is tail, then output head; if the first toss is tail and the second toss is head, then output tail; otherwise repeat this step.

Denote the probability of outputting head and tail by $\Pro{H}$ and $\Pro{T}$, respectively, then,
\begin{equation}\label{}
\begin{aligned}
  \Pro{H} =& \Pro{s-head}\Pro{t-tail} + (1 - \Pro{s-head}\Pro{t-tail}
  \\& - \Pro{s-tail}\Pro{t-head})\Pro{H}\\
  = &s(1-t) + \left[1 - s(1-t) - (1-s)t\right]\Pro{H}\\
  = &\frac{m}{m+1}\left(1-\frac{n}{n+1}\right)  \\
  &+\Pro{H}\left[1 - \frac{m}{m+1}\left(1-\frac{n}{n+1}\right) - \left(1-\frac{m}{m+1}\right)\frac{n}{n+1}\right]\\
    = &\frac{m}{(m+1)(n+1)} + \frac{mn + 1}{(m+1)(n+1)} \Pro{H}
\end{aligned}
\end{equation}
Solving this equation, we have
\begin{equation}\label{}
\begin{aligned}
  \Pro{H} &= \frac{m}{m+n} \\
  & = \frac{2p(1-p)}{2p(1-p)+1/2-2p(1-p)}\\
  &= 4p(1-p).
\end{aligned}
\end{equation}
\end{enumerate}

In our protocol, generating the $q$-coin, where
\begin{equation}\label{eq:qcoins}
   q= \frac{1}{2}\left[1+2\sqrt{p(1-p)}\right],
\end{equation}
is an essential nonclassical step. In fact, the only additionally required coin for constructing all quantum constructible $f(p)$-coins is the $h_a(p)$-coin,
\begin{equation}\label{}
  h_a(p) = \left(\sqrt{p(1-a)} + \sqrt{a(1-p)}\right)^2,
\end{equation}
which can be obtained by measuring the quoin in the $\{\sqrt{1-a}\ket{0}+\sqrt{a}\ket{1}, \sqrt{a}\ket{0}-\sqrt{1-a}\ket{1}\}$ basis. In our case, we set $a=1/2$. Here, one can see that entanglement is not necessary to quantum Bernoulli factory.

In the protocol, the first two steps involve quantum devices, where quoins are measured in the $Z$ and $X=\{(\ket{0}+\ket{1})/\sqrt{2},(\ket{0}-\ket{1})/\sqrt{2}\}$ bases, respectively, to obtain the $p$- and $q$-coins. The following steps (step 3-5) are classical processing of the $p$- and $q$-coins. The rigourous derivation of the classical steps can be found in the Appendix. Comparing to the $f(p)=2p$ function, our protocol converges much faster, which results in a higher fidelity for the realization.

In practice, owing to experimental imperfections, we cannot realize perfect $p$-quoins and perform ideal measurements to get perfect $p$- and $q$-coins. Thus, in reality, we cannot realize exact $f(p)$-coins, especially, we cannot get $f(p) = 1$ when $p=1/2$. Following previous studies \cite{nacu2005fast, Mossel05, thomas2011practical}, we employ a truncated function
\begin{equation}\label{Eq:truncated}
  f_{t} = \min \{f, 1-\epsilon\},
\end{equation}
with $\epsilon$ describing the imperfections. When $\epsilon$ is nonzero, the truncated function of  $f=4p(1-p)$ falls in the classical Bernoulli factory and hence can be constructed via $p$-coins. However, the number of classical coins $N$ required to construct $f(p)$ scales poorly with $\epsilon$, see Appendix for more details. In the experiment, we need to implement high fidelity state preparation and measurement to reduce $\epsilon$ as small as possible in order to faithfully demonstrate the quantum advantage.

In the following, we focus on the preparation and measurement of the $p$-quoin, and how to construct an $f(p)=4p(1-p)$ coin via necessary classical processing. Here, we emphasize that the quantum circuit to realize the operations should be independent of $p$. In demonstration, we fix the measurement setting and prepare $p$-quoins for various $p$ values.


\section{Experimental realization}
We choose a superconducting qubit system to prepare $p$-quoins. Superconducting quantum systems have made tremendous progress in the last decade, including realizing long coherence times, showing great stability with fast and precise qubit manipulations, and demonstrating high fidelity quantum non-demolition (QND) qubit measurement. Thus, it makes a perfect candidate for our test.

\subsection{Experiment setup}
In our experiment, we employ the so-called `circuit quantum electrodynamics architecture' \cite{Wallraff}. A superconducting transmon qubit (our quoin) is located in a waveguide trench and dispersively couples to two 3D cavities~\cite{Kirchmair,Vlastakis,SunNature} as shown in Fig.~\ref{fig:device}. The transmon qubit has a transition frequency of $\omega_q/2\pi=5.577$ GHz, an anharmonicity $\alpha_q/2\pi=-246$ MHz, an energy relaxation time $T_1=9~\mu$s, and a Ramsey time $T_2^*=7~\mu$s. The larger cavity has a resonant frequency of $\omega_c/2\pi=7.292$ GHz and a decay rate of $\kappa/2\pi=3.62$ MHz, which provides a fast way of reading out the qubit state through their strong dispersive interaction with a dispersive shift $\chi/2\pi=-4.71$~MHz. As we focus on exhibiting quantum advantage solely with a single quantum system, the smaller cavity with a higher resonant frequency is not used and remains in a vacuum state. This higher frequency cavity can potentially be used as another $p$-quoin in future experiments \cite{Leghtas2}. In this case, joint measurement can be performed on two $p$-quoins, which may save the resource. For now, we focus on single-qubit operations.

\begin{figure}[b]
\centering
\resizebox{10cm}{!}{\includegraphics[scale=1]{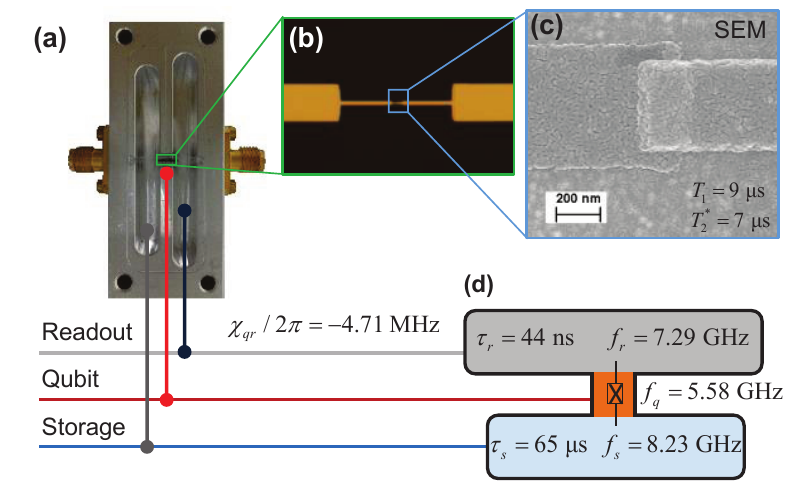}}
\caption{Experimental setup. (a) Optical image of a transmon qubit located in a trench, which dispersively couples to two 3D Al cavities. (b) Optical image of the single-junction transmon qubit. (c) Scanning electron microscope image of the Josephson junction. (d) Schematic of the device with the main parameters. In our experiment, the higher frequency cavity is not used and always remains in vacuum, which can be used as another $p$-quoin in future experiments \cite{Leghtas2}. Note that the highlighted boxes in (a) and (b) are not to scale and are intended for illustrative purposes only.}
\label{fig:device}
\end{figure}

The output of the readout cavity is connected to a Josephson parametric amplifier (JPA)~\cite{Hatridge,Roy}, operating in a double-pumped mode~\cite{Kamal,Murch} as the first stage of amplification between the readout cavity, at a base temperature of 10~mK, and the high electron mobility transistor, at 4~K. To minimize pump leakage into the readout cavity and achieve a longer $T_2^*$ dephasing time, we operate the JPA in a pulsed mode. The readout pulse width has been optimized to 180~ns with a few photons in order to have a high signal-to-noise ratio. This JPA allows a high-fidelity single-shot readout of the qubit state. The overall readout fidelity of the qubit measured for the ground state $\ket{0}$ when initially prepared at $\ket{0}$ by a post-selection is 0.996, demonstrating the high QND nature of the readout, while the fidelity for the excited state $\ket{1}$ is slightly lower, 0.943 (see Appendix). The loss of both fidelities is predominantly limited due to the $T_1$ process during both the waiting time of the initialization measurement (300~ns) and the qubit readout time (180~ns).

Due to stray infrared photons and other background noise, our qubit has an excited state population of about $8.5\%$ in the steady state. The high QND qubit measurement allows us to eliminate these imperfections by performing an initialization measurement to purify the qubit by only selecting the ground state for the following experiments \cite{Riste}. The measurement pulse sequences for preparing quoins can be found in the Appendix. It is worth mentioning that our superconducting system always yields a detection result once the measurement is performed, which is very challenging for other implementations, such as lossy photonic systems.

We apply an on-resonant microwave pulse to rotate the qubit to an arbitrary angle $\theta$ along the $Y$-axis, $R_\theta^Y = \exp(-i\sigma_y\theta/2)$, where $\sigma_y$ is the Pauli matrix, for a preparation of any $p = \cos^2(\theta/2)$-quoins. We  choose a gaussian envelope pulse truncated to $4\sigma=24$~ns for the rotation operations. We also use the so-called ``derivative removal by adiabatic gate" \cite{Motzoi} technique to minimize qubit leakage to higher levels outside the computational space. A randomized benchmark calibration~\cite{Knill2008,Ryan2009,Magesan2012,BarendsNature} shows that the $R_{\pi/2}^{Y}$ gate fidelity itself is about 0.998, mainly limited by the qubit decoherence. The final measurement for the quoins is along either the $Z$-axis or the $X$-axis. The measurement along the $X$-axis is realized by applying an extra $R_{\pi/2}^{-Y}$ rotation (Hadamard transformation) followed by a $Z$-basis measurement.

The readout property of the qubit is first characterized as shown in Fig.~\ref{fig:ReadoutProperty}. The smaller cavity has a resonant frequency of $\omega_s/2\pi=8.229$ GHz and remains in vacuum all the time. Because we always purify our qubit initial state to the ground state $\ket{0}$ and use pulses with DRAG~\cite{Motzoi} to minimize the leakage to levels higher than the first excited state $\ket{1}$, we do not distinguish the levels higher than $\ket{1}$ in the readout. We thus adjust the phase between the JPA readout signal and the pump such that $\ket{0}$ and $\ket{1}$ states can be distinguished with optimal contrast. Figure~\ref{fig:ReadoutProperty}a shows the histogram of the qubit readout. The histogram is clearly bimodal and well-separated. A threshold $V_{th}=0$ is chosen to digitize the readout signal.

\begin{figure}[bht]
\centering
\resizebox{10cm}{!}{\includegraphics[scale=1]{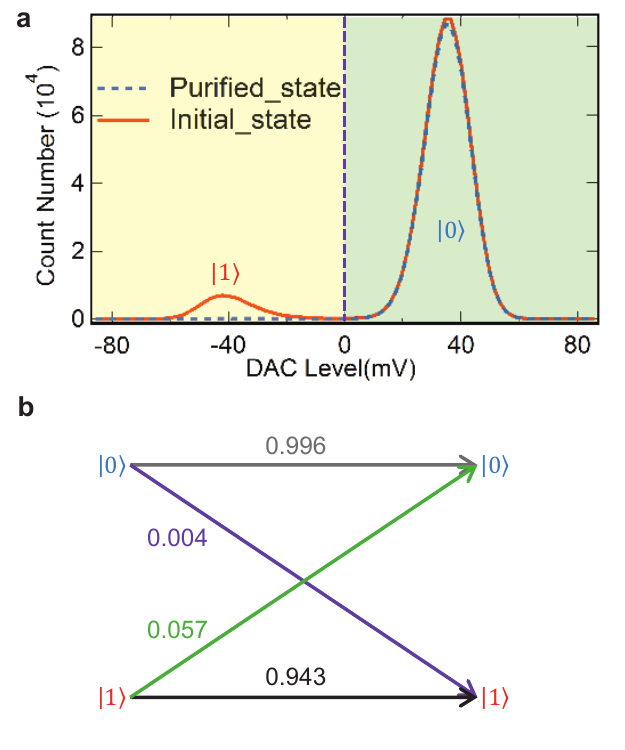}}
\caption{Readout properties of the qubit. The phase between the JPA readout signal and the pump has been adjusted such that $\ket{0}$ and $\ket{1}$ states can be distinguished with optimal contrast. a) Bimodal and well-separated histogram of the qubit readout. A threshold $V_{th}=0$ has been chosen to digitize the readout signal. Solid line is for an initial measurement showing about 8.5\% $\ket{1}$ state, while dashed line is for a second measurement after initially selecting $\ket{0}$ state. The disappearance of $\ket{1}$ state demonstrates both a high purification and high quantum non-demolition measurement of the qubit. b) Basic qubit readout matrix. The loss of fidelity predominantly comes from the $T_1$ process during both the waiting time after the initialization measurement and the qubit readout time.}
\label{fig:ReadoutProperty}
\end{figure}

Due to stray infrared photons or other background noises, our qubit has an excited state population of about $8.5\%$ in the steady state (solid histogram in Fig.~\ref{fig:ReadoutProperty}a). In order to eliminate these excited states for the quoin experiments, a high quantum non-demolition qubit measurement M1 is performed to allow a qubit purification by only selecting $\ket{0}$ state (see Fig.~\ref{fig:quantum circuit})~\cite{Riste}. We wait 300~ns for the readout photons to leak out before the preparation of the qubit to arbitrary superposition states through an on-resonant microwave pulse with various amplitudes. After a purification to $\ket{0}$ state, the following measurement gives a probability of 0.996 of $\ket{0}$ state (dashed histogram in Fig.~\ref{fig:ReadoutProperty}a), demonstrating the high quantum non-demolition nature of the qubit measurement. Figure~\ref{fig:ReadoutProperty}b shows the basic qubit readout properties. The readout fidelity of the qubit measured at $\ket{1}$ state while initially prepared at $\ket{1}$ state by a measurement is 0.943. The loss of both fidelities is predominantly limited due to the $T_1$ process during both the waiting time of the initialization measurement (300~ns) and the qubit readout time (180~ns).

\begin{figure}[bht]
\centering
\resizebox{10cm}{!}{\includegraphics[scale=1]{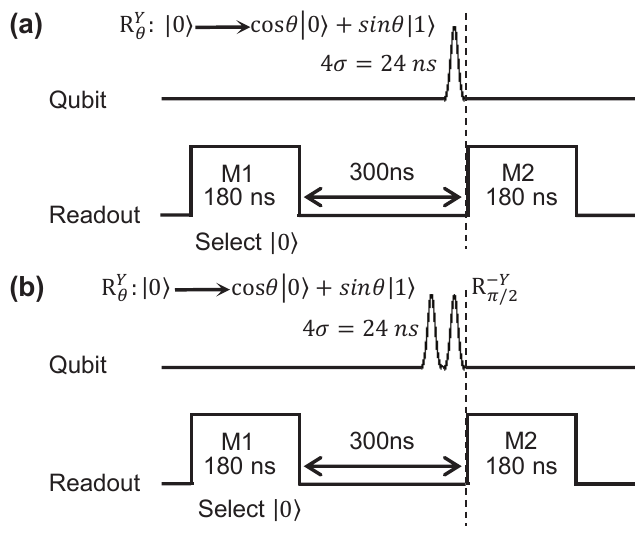}}
\caption{Experimental pulse sequences for the preparation of quoins and the measurements in the $Z$ (a) and $X$ (b) bases. An initial measurement M1 is firstly performed to purify the qubit to the ground state $\ket{0}$. The rotation of the qubit is realized by applying an on-resonance microwave pulse with various amplitudes. The measurement is always performed in the $Z$-basis. The measurement in the $X$-basis is realized by performing an extra $R_{\pi/2}^{-Y}$ pre-rotation. The phase of this extra pre-rotation is chosen to minimize the effect from qubit decoherence during the measurement for the case of $p=0.5$, which is most sensitive to the final qubit readout accuracy.}
\label{fig:quantum circuit}
\end{figure}


\subsection{Results}
The experimental pulse sequences for the quoins with state preparations are shown in Fig.~\ref{fig:quantum circuit}. The measurement is always performed in the $Z$ basis. The $X$-basis measurement is realized by performing an extra $R_{\pi/2}^{-Y}$ rotation before the Z-basis measurement. The phase of this extra pre-rotation is chosen to minimize the effect from qubit decoherence during the measurement.
In our experiment, the $q$-coins as defined in Eq.~\eqref{eq:qcoins} are implemented, which are also classically impossible \footnote{The function $q(p)$ is defined in Eq.~\eqref{eq:qcoins}. It is straightforward to check to see that $q(1/2)=1$, indicating that it is classically impossible} when regarded as a function of $q$.
The $q$-coin corresponds to the qubits of $\ket{\psi} = \cos(\theta/2)\ket{0} +(\sin\theta/2)\ket{1}$. For different $\theta$, our experiment results are listed in Table~\ref{table:results}.

\begin{table}[hbt]\centering
\caption{Experiment results. $\theta$ is the angle of the quoin state; $N_p$ is the total number of $p$-quoins prepared. About half of the prepared $p$-coins are measured in the $X$ basis to prepare the $q$-coin. $q_{\mathrm{th}}$ is the theoretically estimated value based on the estimation of $p$; $q_{\mathrm{exp}}$ is the experimentally estimated value from the obtained $q$-coins; $f_{\mathrm{th}}(p)$ is the theoretically estimated value from the estimation of $p$; $f_{\mathrm{exp}}(p)$ is the experimentally estimated value from the obtained $f(p)$-coins; $N_{f(p)}$ is the number of $f(p)$-coins obtained.}
\begin{tabular}{cccccccc}
  \hline
$\theta$ & $p$ &$N_{p}$& $q_{\mathrm{th}}$&$q_{\mathrm{exp}}$  & $f_{\mathrm{th}}(p)$ & $f_{\mathrm{exp}}(p)$ &$N_{f(p)}$\\
\hline
$0^\circ$&$0.996$&$2.06*10^{7}$&$0.563$&$0.504$&$0.016$&$0.014$&$8.66*{10^5}$\\
$15^\circ$&$0.979$&$1.87*10^{6}$&$0.644$&$0.628$&$0.083$&$0.081$&$8.29*{10^5}$\\
$30^\circ$&$0.929$&$2.07*10^{7}$&$0.756$&$0.746$&$0.262$&$0.257$&$1.05*{10^6}$\\
$45^\circ$&$0.850$&$2.08*10^{7}$&$0.857$&$0.847$&$0.509$&$0.495$&$1.25*{10^6}$\\
$60^\circ$&$0.748$&$2.08*10^{7}$&$0.934$&$0.924$&$0.754$&$0.731$&$1.07*{10^6}$\\
$75^\circ$&$0.630$&$1.99*10^{6}$&$0.983$&$0.974$&$0.933$&$0.901$&$8.90*{10^5}$\\
$90^\circ$&$0.502$&$2.09*10^{7}$&$1.000$&$0.990$&$1.000$&$0.965$&$8.85*{10^5}$\\
$105^\circ$&$0.375$&$2.18*10^{7}$&$0.984$&$0.974$&$0.938$&$0.905$&$9.74*{10^5}$\\
$120^\circ$&$0.258$&$2.08*10^{7}$&$0.938$&$0.926$&$0.766$&$0.737$&$1.06*{10^6}$\\
$135^\circ$&$0.157$&$2.19*10^{7}$&$0.864$&$0.849$&$0.530$&$0.508$&$1.32*{10^6}$\\
$150^\circ$&$0.080$&$2.09*10^{7}$&$0.772$&$0.749$&$0.296$&$0.283$&$1.08*{10^6}$\\
$165^\circ$&$0.033$&$2.08*10^{7}$&$0.678$&$0.633$&$0.126$&$0.120$&$9.45*{10^5}$\\
  \hline
\end{tabular}\label{table:results}
\end{table}

We also plot the experiment result of the $q$-coins in Fig.~\ref{fig:result}(a) and the result of the $f(p) = 4p(1-p)$-coins by following the protocol in Fig.~\ref{fig:result}(b). The experimentally realized values of $q_{\mathrm{exp}}$ and $f_{\mathrm{exp}}(p)$ are sampled from the observed coins, which match well with the theoretical predictions. By implementing state preparation, operation and measurement with high fidelities, we are able to achieve $q_{\mathrm{exp}}(1/2) = 0.990$ and $f_{\mathrm{exp}}(1/2)=0.965$, which can be well modeled by the truncated function defined in Eq.~\eqref{Eq:truncated} with $\epsilon = 0.010$ and $\epsilon = 0.035$, respectively.

\begin{figure}[bht]
\centering
{\resizebox{8cm}{!}{\includegraphics[scale=1]{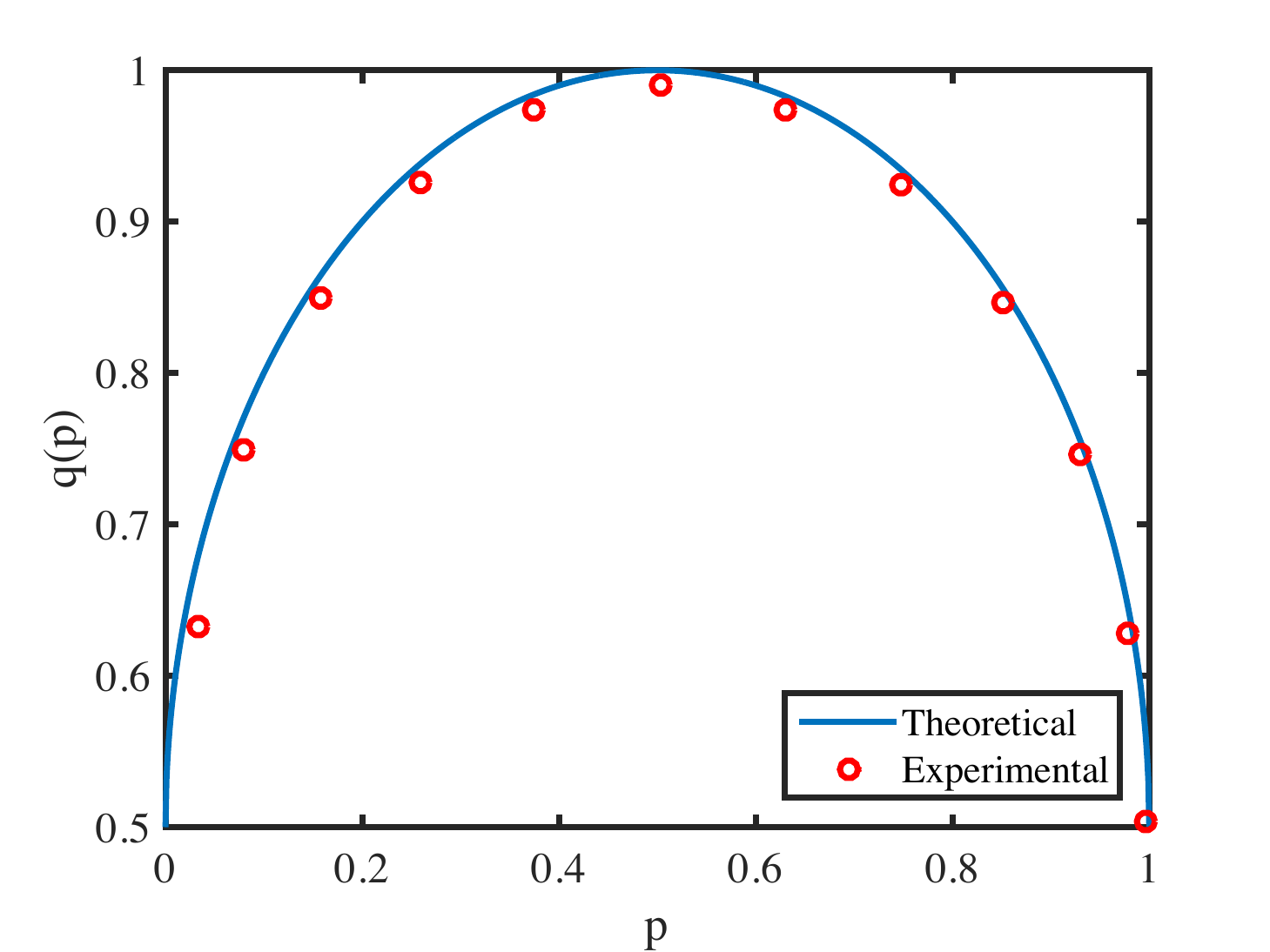}}}
{\resizebox{8cm}{!}{\includegraphics[scale=1]{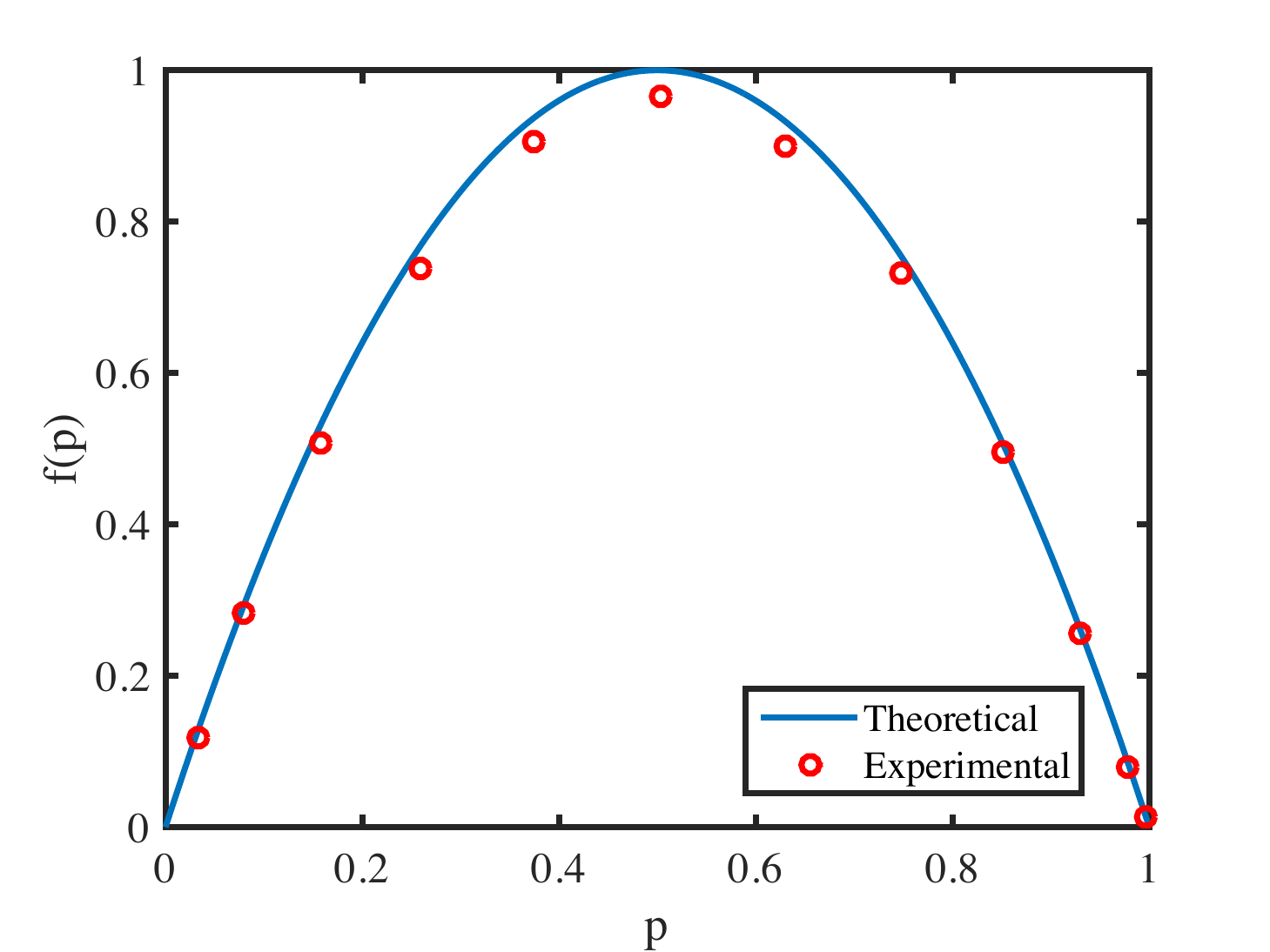}}}
\caption{Theoretical and experimental results for the (a) $q$-coin and (b) $f(p)=4p(1-p)$-coin. 
Here, the number of experiment data for the $p$-quoins is in the order of $10^7$ and the number for the $f(p)$-coin is in the order of $10^6$. On average, we need about 20 $p$-quoins to construct a $f(p)$-coin. The standard deviations of $p$, $q$, and $f(p)$ are in the order of $10^{-4}$, thus are not plotted in the figure.}\label{fig:result}
\end{figure}

A randomized benchmarking experiment~\cite{Knill2008,Ryan2009,Magesan2012,BarendsNature} is performed to determine the fidelity of the $\pi/2$ gate around the $Y$ axis, $R_{\pi/2}^{Y}$, which is the most critical gate for the quoin measurement. The randomized gates used in this experiment are chosen from the single-qubit Clifford group. This group contains 24 rotation gates which are composed from rotations around the $X$ and $Y$ axes using the generators: $\{I, +X, +Y, \pm X/2, \pm Y/2\}$. The reference curve is measured after applying sequences of $m$ random Clifford gates, while the $Y/2$ curve is realized after applying sequences that interleave $R_{\pi/2}^{Y}$ with $m$ random Clifford gates. Each sequence is followed by a recovery Clifford gate in the end right before the final measurement. The number of random sequences of length $m$ in our experiment is chosen to be $k=100$. Both curves are fitted to $F=Ap^m+B$ with different sequence decay $p$. The reference decay indicates the average error of the single-qubit gates, while the ratio of the interleaved and reference decay gives the specific gate fidelity. The experiment results are displayed in Fig.~\ref{fig:RB}. The data point is the average of the sequence fidelities of the $k=100$ sample sequences, and the error bar shows the standard deviation of the sample. Each random sequence is measured over 10,000 times to get the sequence fidelity whose error could be neglected. As a result, the average single-qubit gate error $r_{s}=r_{ref}/1.875=(1-p_{ref})/2/1.875=0.0014$, and the $R_{\pi/2}^{Y}$ gate error $r_{Y/2}=(1-p_{int}/p_{ref})/2=0.0013$. The dashed lines indicate a gate fidelity of 0.998 and 0.997 respectively. Therefore, the $R_{\pi/2}^{Y}$ gate fidelity in our experiment is greater than 0.998, and the uncertainty in the gate fidelity is typically 7e-5, determined by bootstrapping.



\begin{figure}[bht]
\centering
\resizebox{10cm}{!}{\includegraphics[scale=1]{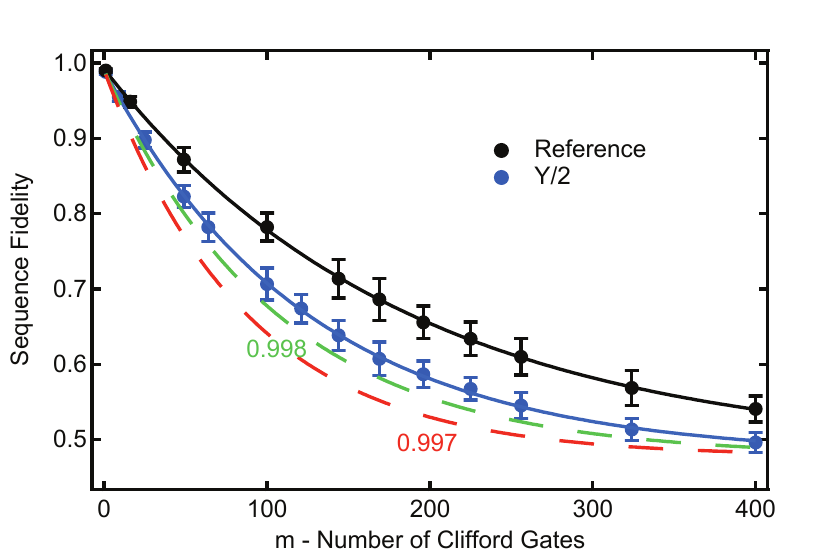}}
\caption{Randomized benchmarking measurement for $R_{\pi/2}^{Y}$ gate fidelity. The reference curve is measured after applying sequences of $m$ random Clifford gates, while the $Y/2$ curve is realized after applying sequences that interleave $R_{\pi/2}^{Y}$ with $m$ random Clifford gates. Each sequence is followed by a recovery Clifford gate in the end right before the final measurement. The number of random sequences of length $m$ in our experiment is $k=100$. Both curves are fitted to $F=Ap^m+B$ with different sequence decay $p$. The data point is the average of the sequence fidelities of the $k=100$ sample sequences, and the error bar shows the standard deviation of the sample. The average single-qubit gate error $r_{s}=r_{ref}/1.875=(1-p_{ref})/2/1.875=0.0014$, and the $R_{\pi/2}^{Y}$ gate error $r_{Y/2}=(1-p_{int}/p_{ref})/2=0.0013$. The dashed lines indicate a gate fidelity of 0.998 and 0.997 respectively.}
\label{fig:RB}
\end{figure}

\section{Simulation of Experiment data}
\subsection{The truncated function}\label{Sec:trun}
Here, we show how to construct the truncated function
\begin{equation}\label{Eq:truncated}
  f_t(p) = \min\{4p(1-p), 1-\epsilon_1\}, \epsilon_1=0.035
\end{equation}
from $p$-coin by classical means. The protocol works as follows.
\begin{enumerate}[(i)]
\item
Toss the $p$-coin twice, if the outputs are different, then output head otherwise output tail. This achieves $g(p)=2p(1-p)$-coin with two $p$-coins.
\item
Apply Theorem 1 in Ref.~\cite{nacu2005fast}, which gives $h(p)= \min\{2p, 1- 2\epsilon_1'\}$, and perform the composition $h(g(p))=\min\{4p(1-p), 1- 2\epsilon_1'\}$. Let $\epsilon_1'= 0.0175$, the desired function is obtained.
\end{enumerate}
Now, we calculate the number of $p$-coins needed in step (ii). By Theorem 1 in Ref.~\cite{nacu2005fast}, the probability that more than $n$ $p$-coins are needed is bounded by
\begin{equation}\label{Eq:pn}
 \begin{aligned}
 P(N>n) \le & \frac{\sqrt{2}}{\epsilon_1' (\sqrt{2}-1)}\sqrt{\frac{2}{n}} e^{-2\epsilon_1'^2 n}\\
 & + 72(1- e^{-2 \epsilon_1'^2})^{-1} e^{-2\epsilon_1'^2 n}+ 4\cdot  2^{-n/9}.
 \end{aligned}
\end{equation}
With large $n$ and small $\epsilon_1'=0.0175$, we can approximate Eq.~\eqref{Eq:pn} by
\begin{equation}\label{Eq:}
\begin{aligned}
P(N>n) \lessapprox &\frac{4.8}{\epsilon_1' \sqrt{n}}e^{-2\epsilon_1'^2 n} + \frac{36}{\epsilon_1'^2}e^{-2\epsilon_1'^2 n}\\
\approx & \frac{36}{\epsilon_1'^2}e^{-2\epsilon_1'^2 n}.
\end{aligned}
\end{equation}
This bound is nontrivial only if the right hand side is less and equal to 1, that is,
\begin{equation}\label{}
  n \approx \frac{-1}{2\epsilon_1'^2}\ln\left(\frac{\epsilon_1'^2}{36}\right) \approx 1.9 \times 10^4.
\end{equation}

Thus combining (i) and (ii), the number of $p$-coins needed to simulate the $f_t(p)$ function is more than $2\times 1.9\times 10^4= 3.8\times 10^4$. Note that, Eq.~\eqref{Eq:pn} provides only an upper bound to the probability distribution, there may exists more efficient protocols that requires less number usages of $p$-coins. 

%

\subsection{Simulation of the $q$-coin}
Here, we show how to simulate the truncated function of $q$ coin,
\begin{equation}\label{}
  q_t(p)=\min  \left\{\frac{1}{2}\left[1+2\sqrt{p(1-p)}\right], 1-\epsilon_3\right\},
\end{equation}
with $\epsilon_3=0.01$. To do so, we first construct the truncated coin $f_t(p)$ defined in Eq.~\eqref{Eq:truncated}.  Then we can simulate $q_t(p)$ with the $f_t(p)$-coin by applying the following protocol,
\begin{enumerate}
  \item Apply a square root function of $f_t(p)$, which gives a $\sqrt{f_t(p)}$-coin.
  \item Toss the 1/2-coin and the $\sqrt{f_t(p)}$-coin, output tail if both tosses are tail.
\end{enumerate}
Then, it is straightforward to check that the following coin is prepared
\begin{equation}\label{}
\begin{aligned}
  Q_t(p)&=\frac{1}{2}\left[1+\sqrt{f_t(p)}\right]\\
  & = \min  \left\{\frac{1}{2}\left[1+2\sqrt{p(1-p)}\right], \frac{1}{2}\left[1+\sqrt{1-\epsilon_1}\right]\right\},
\end{aligned}
\end{equation}
which coincides with the $q_t(p)$-coin if we let
\begin{equation}\label{}
  1-\epsilon_3 = \frac{1}{2}\left[1+\sqrt{1-\epsilon_1}\right].
\end{equation}
In this case, we have $\epsilon_1 = 0.04$. To simulate the $f_t(p)$-coin, we can follow the protocol in Sec.~\ref{Sec:trun}, which costs more than $4\times10^4$ number of $p$-coins on average for each $f_t(p)$-coin. The square root function of $f_t(p)$ can be constructed by following the method from Ref.~\cite{Mossel05} or the one presented in Ref.~\cite{dale2015provable}. On average, more than $10$ coins are needed for constructing the square root function. Therefore, more than $4\times10^5$ number of $p$-coins are necessary for the construction of the truncated function $q_t(p)$.

The classical Bernoulli factory cannot produce exact $q$- and $f(p) = 4p(1-p)$-coins with finite number of usages of $p$-coins. In practice, the implemented function may deviate from the desired one due to device imperfections. In this case, the practically realized coins may be constructible with classical means, though the number of classical coins required may increase drastically with decreasing deviation. Focusing on the truncated function defined in Eq.~\eqref{Eq:truncated}, we present a classical protocol for simulating the experiment data $f_{\mathrm{exp}}(p)$ with $\epsilon = 0.035$. It is shown that, on average, more than $10^4$ classical $p$-coins are required for constructing the truncated function, which is much larger than the average number of quoins (about $20$) used in our protocol \footnote{Strictly speaking, the quantum advantage demonstrated here is a weaker version of the one mentioned in the beginning of the Letter, where the function is classically impossible to construct.}. For the $q$-coin, as the deviation is smaller, the classical simulation is even harder. In the Appendix, we show that more than $10^5$ classical coins are needed for the truncated function, while our quantum protocol only requires one quoin.

From the experimental perspective, the small deviation $f_{exp}(1/2)$ from unity in the ideal case is dominated by qubit decoherence. With better qubit coherence times of $T_1, T_2\sim 100~\mu$s achieved recently~\cite{Rigetti}, we expect the deviation of $f_{exp}(p)$ from $f_{th}(p)$ to be an order of magnitude lower. In future, a more accurate quantum Bernoulli factory can be realized and the classical simulation will eventually become intractable.

It is noteworthy that entanglement can be exploited to save resource in the quantum Bernoulli factory, which provides an extra advantage for randomness processing \cite{dale2015provable}. Extending our implementation to multi-qubit systems can verify this extra quantum advantage. When considering practical imperfections, multiple qubit operation generally has a lower fidelity of measurement. Balancing between the saving of resource and decoherence due to multiple-qubit interactions, it is interesting to see whether multipartite correlation can have extra advantage in practice. As we are focusing on proving the advantage only with coherence, we leave such extension and discussion to future works.

\part{Quantumness and selftesting}
\chapter{Measurement-device-independent entanglement witness}
A conventional way to detect entanglement is via entanglement witness (EW). Practical imperfections can affect the correctness of the witness conclusion. This chapter introduces the measurement device independent entanglement witness (MDIEW) method \cite{Branciard13}. We show a time-shift attack to conventional EW and how the MDIEW scheme be immune to such attacks. We also show an experimental realization of the MDIEW scheme \cite{Yuan14}.




\section{Time-shift attack}
In this section, we show the time-shift attack to conventional entanglement witness. By controlling the arriving time of the photon, we show that the measurement efficiencies mismatches can be exploited to attack conventional EW.
\subsection{EW and device imperfections}
Mathematically, for a given entangled quantum state $\rho$, an Hermitian operator $W$ is called a witness, if $\tr[W\rho] < 0$ (output of `Yes') and $\tr[W\sigma]\geq0$ (output of `No') for any separable state $\sigma$.
Focusing on the bipartite scenario, a general illustration of the conventional EW is shown in Fig.~\ref{fig:EWAMDIEW}(a), where two parties, Alice and Bob, each receives one component of a bipartite state $\rho_{AB}$ from an untrusted third party Eve. They want to verify whether $\rho_{AB}$ is entangled or not, by performing local operations and measurements on $\rho_{A}=\tr_B[\rho_{AB}]$ and $\rho_{B}=\tr_A[\rho_{AB}]$. The correctness of such witness relies on implementation details of $W$. An unfaithful implementation of $W$, say, due to device imperfections, would render the witness results unreliable. For example, the measurement devices used by Alice and Bob might possibly be manufactured by another untrusted party, who could collaborate with Eve and deliberately fabricate devices to make the real implementation $W' = W+\delta W$ be deviated from $W$, such that $W'$ is not a witness any more,
\begin{equation}\label{MDIEW:Def:SuccessAttack}
  \tr[W'\sigma]<0<\tr[W\sigma].
\end{equation}
That is, with the deviated witness $W'$, a separable state $\sigma$ could be identified as an entangled one, which is more likely to happen when $\tr[W\sigma]$ is near zero.

\begin{figure}[bht]
\centering
\resizebox{8cm}{!}{\includegraphics[scale=1]{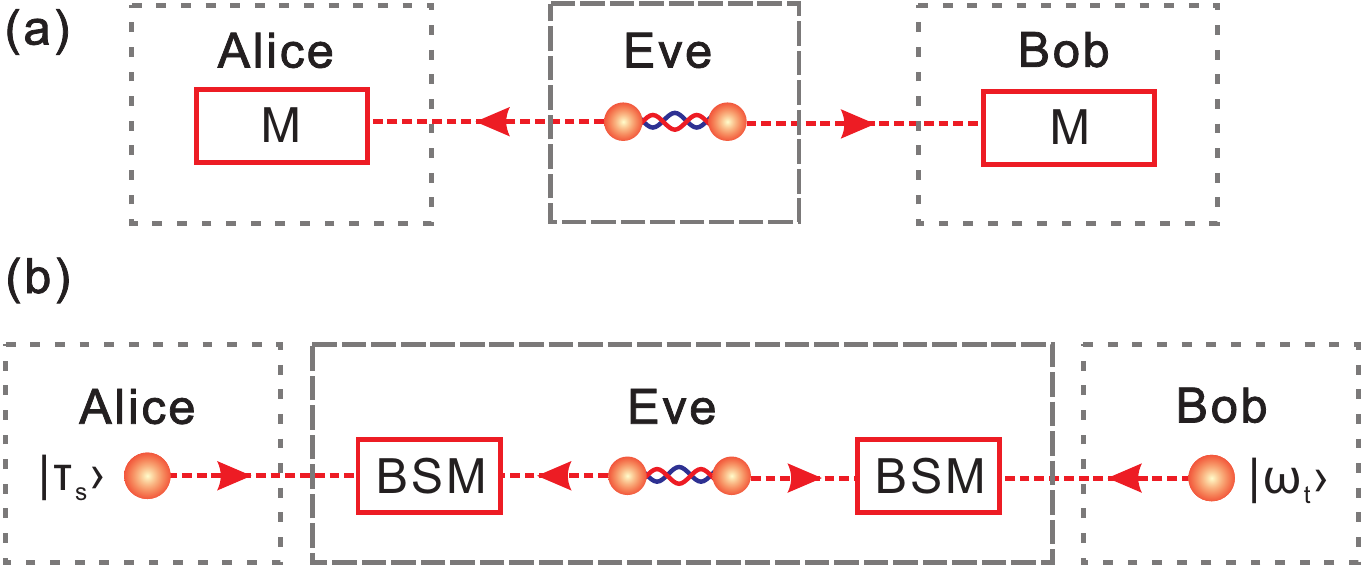}}
\caption{(a) Conventional EW setup, where Alice and Bob perform local measurements separately and collect information to decide whether the input state is entangled or not. (b) Measurement-device-independent (MDI) EW setup, where Alice and Bob each prepares an ancillary state and a third party Eve performs Bell state measurements (BSMs) on the ancillary states and the to-be-witnessed bipartite state. Based on the choices of Alice and Bob's ancillary states and the BSM results, they can judge whether the input state is entangled or not.}\label{fig:EWAMDIEW}
\end{figure}

There is a strong similarity between EW and quantum key distribution (QKD) where an entanglement-breaking channel would cause insecurity \cite{Lutkenhaus04}. Roughly speaking, it is crucial for Alice and Bob to prove that entanglement can be preserved in a secure QKD channel. From this point of view, there exists correlation between the security of QKD and the success of EW. For the varieties of attacks in QKD, such as time-shift attack \cite{qi07} and fake-state attack \cite{Makarov06}, one may also find similar detection loopholes in the conventional EW process. Originated from this analogy, we construct a time-shift attack that manipulates the efficiency mismatch between detectors used in an EW process. Under this attack, any state could be witnessed to be entangled, even if the input state is separable. By this example, we demonstrate that there do exist loopholes in the conventional EW procedure.




\subsection{Time-shift attack}
Originated from quantum cryptography \cite{qi07}, takes advantage of efficiency mismatch of the measurement devices. As shown in Fig.~\ref{TSAS}(a), typically two detectors are used on each side of Alice and Bob. By controlling the single-photon-counting modules (SPCMs) and coincidence gate, Eve is able to enlarge the efficiency mismatch and hence manipulate the EW result.

\begin{figure}[bht]
\centering
\resizebox{10cm}{!}{\includegraphics[scale=1]{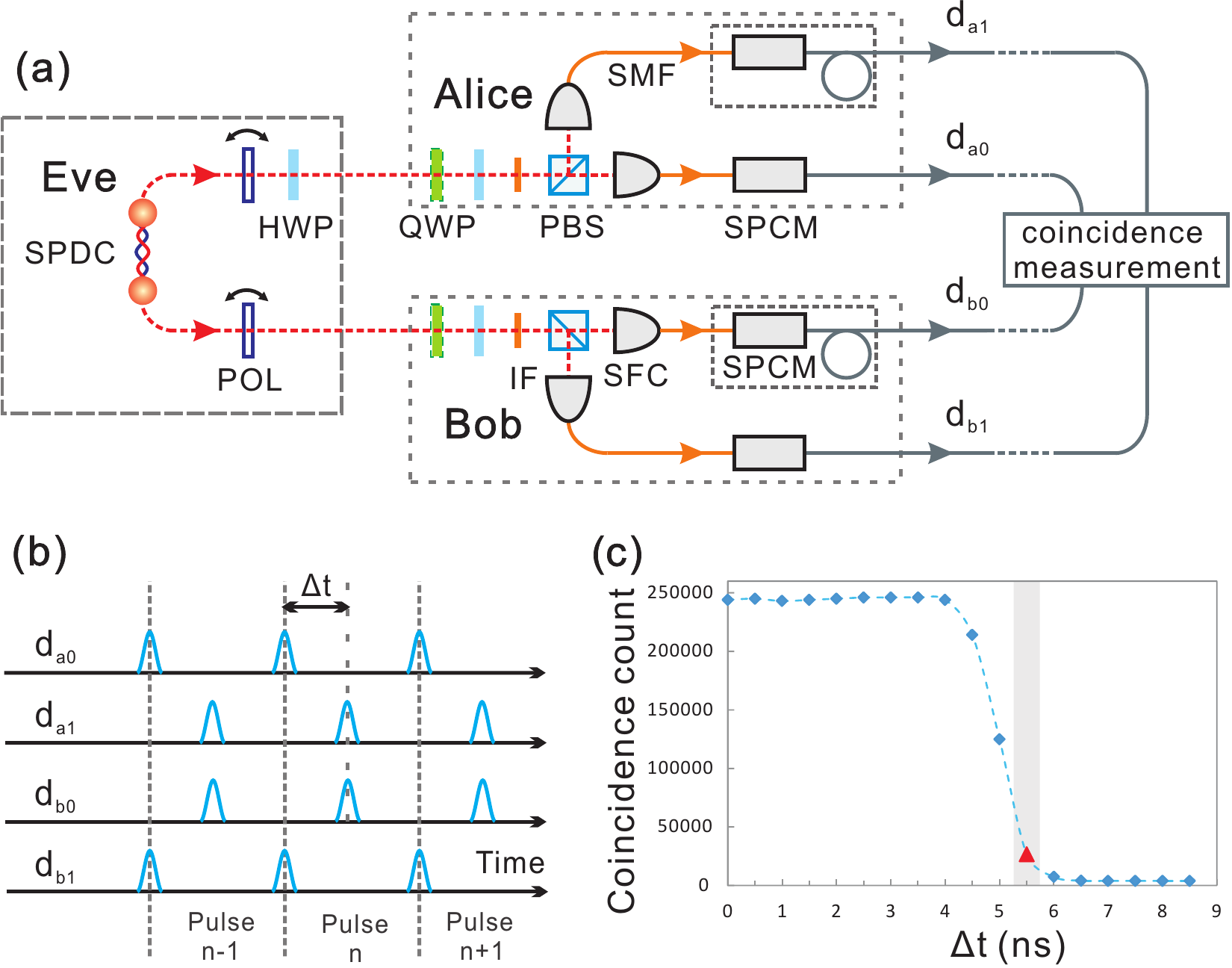}}
\caption{Time shift attack to the  conventional EW. (a) Experimental setup of the time-shift attack. Photon pairs are generated by SPDC using a femtosecond pump laser with a central wavelength of 390 nm and a repetition frequency of 80 MHz. POL: polarizer, HWP: half-wave plate, QWP: quarter-wave plate, IF: interference filter with 780 nm central wavelength, PBS: polarizing beam splitter, SFC: single-mode fiber coupler, SMF: single-mode fiber, SPCM: single-photon-counting module, some with extra internal delay lines.
(b) Synchronization between SPCMs.  Build-in delay lines enable Eve to shift the output signals $d_{a1}$ and $d_{b0}$ by $\Delta t$.
(c) Coincidence count versus time delay, where the time window is set to 4 ns. All data points are measured for 2 seconds, and time-shift attack is implemented with $\Delta t=5.50 \pm0.24$ ns, which corresponds to the grey area.}\label{TSAS}
\end{figure}

To implement this attack, we choose a conventional witness,
\begin{equation*}\label{eq:w}
  W = \frac{1}{2}I-|\Psi^-\rangle\langle\Psi^-|,
\end{equation*}
for bipartite states in the form of
\begin{equation}\label{eq:rho}
\rho^v_{AB}=(1-v)|\Psi^-\rangle\langle\Psi^-|+\frac{v}{2}(|HH\rangle\langle HH| + |VV\rangle\langle VV|),
\end{equation}
where H (V) denotes the horizontal (vertical) polarization of the single photons and $|\Psi^{-}\rangle = (|HV\rangle - |VH\rangle)/\sqrt{2}$ is a Bell state. By decomposing $W$ into a linear combination of product Pauli matrices,
\begin{equation}\label{}
  W = \frac{1}{4}\left(\sum_{i = 0}^3\sigma_i\otimes\sigma_i\right),
\end{equation}
the EW can be realized by local measurements,
we can decompose $W$ to
\begin{equation*}\label{eq:expEW}
Tr[W\rho_{AB}]=\frac{1}{4}\left( {1 + \left\langle {{\sigma _x}{\sigma _x}} \right\rangle  + \left\langle {{\sigma _y}{\sigma _y}} \right\rangle  + \left\langle {{\sigma _z}{\sigma _z}} \right\rangle } \right).
\end{equation*}
That is, to identify the entanglement, Alice and Bob just have to each analyze the qubit state in three bases separately. When the bipartite state is projected to the positive (negative) eigenstates of $\sigma_x\sigma_x$, $\sigma_y\sigma_y$, and $\sigma_z\sigma_z$, it will contribute positively (negatively) to the witness result $Tr[W\rho_{AB}]$. For example, when measuring ${\sigma _x}{\sigma_x}$, Alice and Bob will both project the input state to the eigenstates of $\sigma_x$, $\sigma_x^+$ or $\sigma_x^-$, with corresponded eigenvalues of $+1$ or $-1$, respectively, and obtain probabilities $\left\langle {{\sigma _x^\pm}{\sigma _x^\pm}} \right\rangle$. Then the value of $\left\langle {{\sigma _x}{\sigma _x}} \right\rangle$ is defined as $\left\langle {{\sigma _x^+}{\sigma _x^+}} \right\rangle + \left\langle {{\sigma _x^-}{\sigma _x^-}} \right\rangle - \left\langle {{\sigma _x^+}{\sigma _x^-}} \right\rangle - \left\langle {{\sigma _x^-}{\sigma _x^+}} \right\rangle$.  From Eve's point of view, she wants to convince Alice and Bob that the bipartite state is entangled, that is, $Tr[W\rho_{AB}]<0$. Thus, her objective is to suppress the positive contributions of $Tr[W\rho_{AB}]$, such as $\left\langle {{\sigma _x^+}{\sigma _x^+}} \right\rangle$ and $\left\langle {{\sigma _x^-}{\sigma _x^-}} \right\rangle$ for ${\sigma _x}{\sigma_x}$ measurement, by manipulating the coincidence rate between SPCMs, equivalently enlarging the detector efficiency mismatch. In this case, from Alice and Bob's point of view, the real implemented witness $W'$ is deviated from the desired one $W$, and satisfies Eq.~\eqref{MDIEW:Def:SuccessAttack}.

To realize the attack, we exploit the time mismatch of the two single-photon-counting modules (SPCMs) such that one detector is more efficient than the other. In this case, the real implementation ($W'$) is deviated from the original design witness $W$. In the attack Eve can suppress the positive contributes of the witness result $Tr[W\rho_{AB}]$ to let the witness result $Tr[W'\rho_{AB}]$ be negative by adjusting the time mismatch. For example, when measuring ${\sigma _x}{\sigma_x}$, Alice and Bob will project the input state to the eigenstates of $\sigma_x$, that is $\sigma_x^+$ and $\sigma_x^-$, corresponding to positive and negative eigenvalue respectively, and obtain probabilities $\left\langle {{\sigma _x^\pm}{\sigma _x^\pm}} \right\rangle$. Then the value of $\left\langle {{\sigma _x}{\sigma _x}} \right\rangle$ is defined as
\begin{equation}\label{}
  \left\langle {{\sigma _x}{\sigma _x}} \right\rangle = \left\langle {{\sigma _x^+}{\sigma _x^+}} \right\rangle + \left\langle {{\sigma _x^-}{\sigma _x^-}} \right\rangle - \left\langle {{\sigma _x^+}{\sigma _x^-}} \right\rangle - \left\langle {{\sigma _x^-}{\sigma _x^+}} \right\rangle.
\end{equation}
The probabilities $\left\langle {{\sigma _x^\pm}{\sigma _x^\pm}} \right\rangle$ is measured from coincidence counts $N_A^\pm N_B^\pm$ of detectors, that is
\begin{equation}\label{}
  \left\langle {{\sigma _x^\pm}{\sigma _x^\pm}} \right\rangle = \frac{N_A^\pm N_B^\pm}{\sum N_A^\pm N_B^\pm}.
\end{equation}
If the positive coincidence counts are all suppressed, that is $N_A^+ N_B^+ = N_A^- N_B^- = 0$, then the outcome of $\left\langle {{\sigma _x}{\sigma _x}} \right\rangle$ is
\begin{equation}\label{}
  \left\langle {{\sigma _x}{\sigma _x}} \right\rangle = - \left\langle {{\sigma _x^+}{\sigma _x^-}} \right\rangle - \left\langle {{\sigma _x^-}{\sigma _x^+}} \right\rangle = -\frac{N_A^+ N_B^-}{\sum N_A^\pm N_B^\pm} - \frac{N_A^- N_B^+}{\sum N_A^\pm N_B^\pm} = -1.
\end{equation}

Similarly, the all the other local measurements $\left\langle {{\sigma _y}{\sigma _y}} \right\rangle$ and  $\left\langle {{\sigma _z}{\sigma _z}} \right\rangle $ become $-1$ by suppressing positive coincidence counts, which gives a witness result of
\begin{equation}\label{}
  Tr[W'\rho_{AB}] = -\frac{1}{2}
\end{equation}
for any state $\rho_{AB}$.

In our experiment demonstration, we only suppress the positive coincidence counts to $10.9(1)\%$ instead of neglecting all of them to make a wrong witness result of a separable state to be entangled.

In our experiment, as shown in Fig.~\ref{TSAS}(a), by encoding qubits in the polarization of photons, the bipartite state $(\vert HH\rangle_{ab}+\vert VV\rangle_{ab})/\sqrt 2$ is generated via spontaneous parametric down conversion (SPDC). Two adjustable POLs are used to disentangle the initial state and project it to $\vert HH \rangle_{ab}$ and $\vert VV \rangle_{ab}$ with equal probabilities, corresponding to the separable state with $v = 1$ in Eq.~\eqref{eq:rho}. After a $45^\circ$ HWP, the to-be-witnessed two-qubit system is prepared in the state of $\rho_{AB} = \left( {{{\left| {HV} \right\rangle }}\left\langle {HV} \right| + {{\left| {VH} \right\rangle }}\left\langle {VH} \right|} \right)/2$.
Then Alice and Bob each performs polarization analysis on a qubit from the bipartite state using waveplates, PBSs and SPCMs, and guides the electronic signals from the SPCMs into a coincidence gate.

As shown in Fig.~\ref{TSAS}b, in the time-shift attack, Eve controls the delay lines in the detection systems and the time window of the coincidence gate, and hence manipulates the time-dependent coincidence counting rates between detectors $d_{a0}$ and $d_{b0}$, $d_{a1}$ and $d_{b1}$. Hence, she can suppress the positive contributions of measurements $\left\langle {{\sigma _x}{\sigma _x}} \right\rangle , \left\langle {{\sigma _y}{\sigma _y}} \right\rangle $ and $\left\langle {{\sigma _z}{\sigma _z}} \right\rangle $. In our demonstration, by setting proper parameters, we let the positive contributions drop to 10.9(1) $\%$ of their original values.
Since this attack would not affect the negative contributions of $Tr[W\rho_{AB}]$, the experimental outcomes for $\left\langle {{\sigma _x}{\sigma _x}} \right\rangle , \left\langle {{\sigma _y}{\sigma _y}} \right\rangle $ and $\left\langle {{\sigma _z}{\sigma _z}} \right\rangle $ become negative as expected.  Finally, Alice and Bob obtain a witness of $\rho_{AB} $ be $ \tr\left[ {W'\rho_{AB} } \right] =  - 0.379\left( 4 \right)$,  although the input state $\rho_{AB} $  is, in fact, separable. By changing the $\Delta t$ to a larger value, one can even obtain a fake result as that from a maximal entangled state. Thus, a separable bipartite state could be wrongly witnessed to be entangled when Eve is able to manipulate the detection system. It is not hard to see that for any state $\rho$, Eve can perform a similar attack and trick Alice and Bob that it is entangled.


\section{The MDIEW scheme}
Recently, Lo et al. \cite{Lo2012MDI} proposed an measurement-device-independent (MDI) QKD method, which is immune to all hacking strategies on detection. Due to the similarity between QKD and EW, one would also expect that there exist EW schemes without detection loopholes. Meanwhile, a nonlocal game is proposed to distinguish any entangled state from all separable states \cite{Buscemi12}. Inspired by this game, Branciard et al. \cite{Branciard13} proposed an MDIEW method, where they proved that there always exists an MDIEW for any entangled state with untrusted measurement apparatuses.

As shown in Fig.~\ref{fig:EWAMDIEW}(b), Alice and Bob want to identify whether a given bipartite state, prepared by an untrusted party Eve, is entangled or not without trusting measurement devices. To do so, Alice (Bob) prepares an ancillary state $\tau_s$ ($\omega_t$) and sends it along with the to-be-witnessed bipartite state to a willing participant, who can be assumed to be Eve again in the worst case scenario. Eve performs two Bell-state measurements (BSMs) on the two ancillary states and the bipartite state. Then, she announces to Alice and Bob the BSMs results, based on which they will witness the entanglement of the bipartite state. In MDIEW, it is guaranteed that a separable state will never be wrongly identified as an entangled one, even if Eve maliciously makes wrong measurements and/or announces unfaithful information.

Measurement-device-independent entanglement witness (MDIEW) provides means to witness entanglement of a quantum state without trusting measurement devices \cite{Branciard13}. The idea of MDIEW is inspired from the MDI quantum key distribution (MDIQKD) \cite{Lo2012MDI}. As proved in Ref.~\cite{Branciard13}, there always exists an MDIEW for any quantum state $\rho$, as one can always construct MDIEW based on the conventional witness $W$ which exists for any quantum state (we refer to \cite{guhne2009} for details of conventional entanglement witness). In the following, we will design a MDIEW scheme and apply it to a type of bipartite quantum states in the form of
\begin{equation}\label{rhovab}
   \rho^v_{AB}=(1-v)|\Psi^-\rangle\langle\Psi^-|+\frac v2(|00\rangle\langle00| + |11\rangle\langle11|),
\end{equation}
with $v\in[0,1]$ and $|\Psi^-\rangle=(|01\rangle-|10\rangle)/\sqrt{2}$. The state is entangled if $v<1/2$, which can be witnessed by a conventional EW,
\begin{equation}\label{wit}
W=\frac{1}{2}I-|\Psi^-\rangle\langle\Psi^-|,
\end{equation}
and its result, $\tr[W\rho^v_{AB}] = (2v-1)/2$.

Practically, the conventional EW can be realized with only local measurements by decomposing $W$ into a linear combination of product Hermitian observables. In the bipartite scenario of Alice and Bob, they only need to perform local measurements to decide the entanglement of quantum states. In contrast, MDIEW requires Alice (Bob) to prepare another ancillary state ${\tau_s}$ (${\omega_t}$) and perform Bell-state measurements (BSMs) on the to be witnessed state and the ancillary state. Based on the choice of the ancillary states, labeled by $s$ and $t$, and the measurement outcomes, labeled by $a$ and $b$, MDIEW is defined as
\begin{equation}\label{Supp:MDIEWJ}
 J(\rho_{AB}) = \sum_{s,t}\beta^{a,b}_{s,t}p(a, b|\tau_s, \omega_t).
\end{equation}
That is, $\rho_{AB}$ is entangled while $J(\rho_{AB})<0$ and for any separable state $\sigma_{AB}$, we have $J(\sigma_{AB})\geq0$.
Here the probabilities $p(a, b|\tau_s, \omega_t)$ are obtained from performing two BSMs on the to be witnessed state $\rho_{AB}$ and the ancillary states ${\tau_s}$ and ${\omega_t}$. That is,
\begin{equation}\label{eq:P}
  p(a, b|\tau_s, \omega_t) = Tr[(M_a\otimes M_b)(\tau_s\otimes\rho_{AB}\otimes\omega_s)],
\end{equation}
where $M_a$ and $M_b$ represent BSMs performed by Alice and Bob with outcome $a$ and $b$, respectively. In Eq.~\eqref{Supp:MDIEWJ}, the coefficient $\beta^{a,b}_{s,t}$ is determined by the choice of ancillary states, measurement outcomes and the conventional witness $W$. In the experiment, as only two $|\Phi^+\rangle=(|00\rangle + |11\rangle)/{\sqrt{2}}$ and $|\Phi^-\rangle=(|00\rangle - |11\rangle)/{\sqrt{2}}$ out of four BSM outcomes are recorded, we consider the outcomes of $a$ and $b$ to be $+$ and $-$, which refer to  $|\Phi^-\rangle$ and $|\Phi^-\rangle$, respectively. There are four kind of $\beta^{a,b}_{s,t}$, depending on different values of $a$ and $b$. In the following, we will design $\beta^{a,b}_{s,t}$ for our MDIEW.

The case of $a =+$ and $b = +$ is considered in Ref.~\cite{Branciard13}. Decompose a conventional EW as a linear combination of product Hermitian operators, $\{\tau_s\otimes\omega_t$\},
\begin{equation}\label{decom1}
  W^{}=\sum_{s,t}\beta^{++}_{s,t}\tau_s^T\otimes\omega_t^T,
\end{equation}
where the superscript $T$ means matrix transpose. In the corresponding MDIEW, Alice and Bob prepare their ancillary states into $\{\tau_s\}$ and $\{\omega_t\}$, respectively. According to Eq.~\eqref{eq:P}, $p(+, +|\tau_s, \omega_t)$ is obtained by projecting the joint states $\tr_B[\rho_{AB}]\otimes\tau_s$ and $\tr_A[\rho_{AB}]\otimes\omega_t$ to the maximally entangled states $|\Phi^+_{AA}\rangle=(|00\rangle + |11\rangle)/{\sqrt{2}}$ and $|\Phi^+_{BB}\rangle=(|00\rangle + |11\rangle)/\sqrt{2}$, respectively. Then it is easy to show that the relation between MDIEW and the conventional EW is
\begin{equation}\label{}
  J(\rho_{AB}) = \tr[W^{}\rho_{AB}]/4,
\end{equation}
which equals $({2v-1})/{8}$ using Eq.~\eqref{rhovab} and \eqref{wit}.

\paragraph{MDIEW with two measurement outcomes}
In our work, we also consider other BSM outcomes. For example, if Alice and Bob get outcomes $a = -$ and $b = -$, then $\beta^{--}_{s,t}$ is calculated similarly as Eq.~\eqref{decom1} by decomposing $W$,
\begin{equation}\label{}
W^{}=\sum_{s,t}\beta^{--}_{s,t}\tilde{\tau}_s^T\otimes\tilde{\omega}_t^T,
\end{equation}
where $\langle j|\tilde{\tau}|i\rangle = (-)^{i+j}\langle j|{\tau}|i\rangle$ and $\langle j|\tilde{\omega}|i\rangle = (-)^{i+j}\langle j|{\omega}|i\rangle$. By redefining the basis that $W$ is decomposed, $\{\tilde{\tau}\otimes\tilde{\omega}\}$, the ancillary states prepared by Alice and Bob are still $\{\tau_s\}$ and $\{\omega_t\}$. In this case, $p(-, -|\tau_s, \omega_t)$ is obtained by projecting the joint states $\tr_B[\rho_{AB}]\otimes\tau_s$ and $\tr_A[\rho_{AB}]\otimes\omega_t$ to the maximally entangled states $|\Phi^-_{AA}\rangle=(|00\rangle - |11\rangle)/{\sqrt{2}}$ and $|\Phi^-_{BB}\rangle=(|00\rangle - |11\rangle)/\sqrt{2}$, respectively.

With a similar manner, one can also decompose $W$ for the cases of $a =+$ and $b = -$, $a=-$ and $b = +$. All the four cases of $a$ and $b$ are summarized in Table \ref{Tab:Wexpansion}.

\begin{table}[hbt]
\centering
\caption{Decomposition of $W$ based on different measurement outcomes.} \label{Tab:Wexpansion}
  \begin{tabular}{c c c}
    \hline
    $M_{AA}$ & $M_{BB}$ & $W$\\
    \hline
    $|\Phi^+_{AA}\rangle=\frac{|0\rangle\otimes|0\rangle + |1\rangle\otimes|1\rangle}{\sqrt{2}}$ & $|\Phi^+_{BB}\rangle=\frac{|0\rangle\otimes|0\rangle + |1\rangle\otimes|1\rangle}{\sqrt{2}}$ & $W^{}=\sum_{s,t}\beta^{++}_{s,t}{\tau}_s^T\otimes{\omega}_t^T$\\
    $|\Phi^-_{AA}\rangle=\frac{|0\rangle\otimes|0\rangle - |1\rangle\otimes|1\rangle}{\sqrt{2}}$ & $|\Phi^-_{BB}\rangle=\frac{|0\rangle\otimes|0\rangle - |1\rangle\otimes|1\rangle}{\sqrt{2}}$ & $W^{}=\sum_{s,t}\beta^{--}_{s,t}\tilde{\tau}_s^T\otimes\tilde{\omega}_t^T$\\
    $|\Phi^+_{AA}\rangle=\frac{|0\rangle\otimes|0\rangle + |1\rangle\otimes|1\rangle}{\sqrt{2}}$ & $|\Phi^-_{BB}\rangle=\frac{|0\rangle\otimes|0\rangle - |1\rangle\otimes|1\rangle}{\sqrt{2}}$ & $W^{}=\sum_{s,t}\beta^{+-}_{s,t}{\tau}_s^T\otimes\tilde{\omega}_t^T$\\
    $|\Phi^-_{AA}\rangle=\frac{|0\rangle\otimes|0\rangle - |1\rangle\otimes|1\rangle}{\sqrt{2}}$ & $|\Phi^+_{BB}\rangle=\frac{|0\rangle\otimes|0\rangle + |1\rangle\otimes|1\rangle}{\sqrt{2}}$ & $W^{}=\sum_{s,t}\beta^{-+}_{s,t}\tilde{\tau}_s^T\otimes{\omega}_t^T$\\
    \hline
  \end{tabular}
\end{table}

Next, we need to calculate the coefficients $\beta_{s,t}^{\pm\pm}$ and the corresponding probabilities $p(\pm,\pm|\tau_s,\omega_t)$ for given ancillary quantum states $\{\tau_s\}$ and $\{\omega_t\}$. Define $\sigma_0 = I$ and $\sigma_1, \sigma_2, \sigma_3$ to be the Pauli matrices. Then let $\tau_s$ and $\omega_s$ both be the eigenstates of $\sigma_s$ with eigenvalues of 1. That is, $\tau_0 = \omega_0 = I/2$, $\tau_s = \omega_s = (I + \sigma_s)/{2}$ for $s = 1, 2, 3$. By decomposing $W$ into $\{\tau_s^T\otimes\omega_t^T\}$ and $\{\widetilde{\tau}_s^T\otimes\widetilde{\omega}_t^T\}$, we find that the coefficients $\beta^{ab}_{st}$ and the probabilities $p(a,b|\tau_s,\omega_t)$ of the two cases $++$ and $--$ are the same, and those of $+-$ and $+-$ are the same.


In the cases of $++$ and $--$, the coefficients are given by
\begin{equation}
[\beta^{++}_{st}] = [\beta^{--}_{st}] = \left[
  \begin{array}{cccc}
    4 & -1 & -1 & -1 \\
    -1 & 1 & 0 & 0 \\
    -1 & 0 & 1 & 0 \\
    -1 & 0 & 0 & 1 \\
  \end{array}
\right],
\end{equation}
with corresponding probabilities of
\begin{equation}\label{}
  p(+,+|\tau_s,\omega_t) = p(-,-|\tau_s,\omega_t)=
  \left[
    \begin{array}{cccc}
      1/16 & 1/16 & 1/16 & 1/16 \\
      1/16 & v/16 & 1/16 & 1/16 \\
      1/16 & 1/16 & v/16 & 1/16 \\
      1/16 & 1/16 & 1/16 & v/8 \\
    \end{array}
  \right].
\end{equation}
There are ten nonzero terms in the coefficient matrix, so ten different ancillary inputs ($\tau_s, \omega_t$) are required. In practice, it is possible to reduce the number of inputs by introducing two other states $\tau_4 = \frac{I+(\sigma_x+\sigma_y+\sigma_z)/\sqrt{3}}{2}$ and $\omega_4 = \frac{I+(\sigma_x+\sigma_y+\sigma_z)/\sqrt{3}}{2}$. In this case, we have another decomposition of $W$ with coefficients of
\begin{equation}
[\beta^{++}_{st}] = [\beta^{--}_{st}] = \left[
  \begin{array}{ccccc}
    2\sqrt3-2 & 0 & 0 & 0 & -\sqrt3 \\
    0 & 1 & 0 & 0 & 0\\
    0 & 0 & 1 & 0 & 0 \\
    0 & 0 & 0 & 1  & 0\\
    -\sqrt3 & 0 & 0 & 0  & 0\\
  \end{array}
\right].
\end{equation}
In this setting, only six ancillary sets are required (comparing to ten in the original construction). As a result, we derive the coefficients and probabilities in Eq.~\eqref{Supp:MDIEWJ} for outcomes $++$ and $--$, as shown in Table \ref{Tab:Coeffppmm}.

\begin{table}[hbt]\tiny
\centering
\caption{Coefficients and probabilities for MDIEW with outcomes $++$ and $--$. Note that when $\beta=0$, the corresponding probability $p$ is irrelevant.} \label{Tab:Coeffppmm}
\begin{tabular}{c | c c c c c}
  \hline
   & $\tau_0= I/2$ & $\tau_1 = \frac{I+\sigma_x}{2}$ & $\tau_2 = \frac{I+\sigma_y}{2}$ & $\tau_3 = \frac{I+\sigma_z}{2}$ & $\tau_4 = \frac{I+(\sigma_x+\sigma_y+\sigma_z)/\sqrt{3}}{2}$  \\
   \hline
  $\omega_0 = I/2$ & $\beta = 2\sqrt3-2, p = \frac{1}{16}$ & $\beta = 0$ & $\beta = 0$ & $\beta = 0$ & $\beta = -\sqrt3, p = \frac{1}{16}$ \\
  $\omega_1 = \frac{I+\sigma_x}{2}$ & $\beta = 0$ & $\beta = 1, p = \frac{v}{16}$ & $\beta = 0$ & $\beta = 0$ & $\beta = 0$ \\
  $\omega_2 = \frac{I+\sigma_y}{2}$ & $\beta = 0$ & $\beta = 0$ & $\beta = 1, p = \frac{v}{16}$ & $\beta = 0$ & $\beta = 0$ \\
  $\omega_3 = \frac{I+\sigma_z}{2}$ & $\beta = 0$ & $\beta = 0$ & $\beta = 0$ & $\beta = 1, p = \frac{v}{8}$ & $\beta = 0$ \\
  $\omega_4 = \frac{I+(\sigma_x+\sigma_y+\sigma_z)/\sqrt{3}}{2}$ & $\beta = -\sqrt3, p = \frac{1}{16}$ & $\beta = 0$ & $\beta = 0$ & $\beta = 0$ & $\beta = 0$ \\
  \hline
\end{tabular}
\end{table}

Similarly, for the other two cases of outcomes $+-$ and $-+$, the coefficients are
\begin{equation}
[\beta^{-+}_{st}] = [\beta^{+-}_{st}] = \left[
  \begin{array}{cccc}
    0 & 1 & 1 & -1 \\
    1 & -1 & 0 & 0 \\
    1 & 0 & -1 & 0 \\
    -1 & 0 & 0 & 1 \\
  \end{array}
\right]
\end{equation}
with corresponding probabilities of
\begin{equation}\label{}
  p(+,-|\tau_s,\omega_t) = p(-,+|\tau_s,\omega_t)=
  \left[
    \begin{array}{cccc}
      1/16 & 1/16 & 1/16 & 1/16 \\
      1/16 & (2-v)/16 & 1/16 & 1/16 \\
      1/16 & 1/16 & (2-v)/16 & 1/16 \\
      1/16 & 1/16 & 1/16 & v/8 \\
    \end{array}
  \right].
\end{equation}
when using the ancillary states $\tau_0 = \omega_0 = I/2$, $\tau_s = \omega_s = (I + \sigma_s)/{2}$ for $s = 1, 2, 3$. Similarly, we can define $\tau'_4 = \frac{I+(-\sigma_x-\sigma_y+\sigma_z)/\sqrt{3}}{2},\, \omega'_4 = \frac{I+(-\sigma_x-\sigma_y+\sigma_z)/\sqrt{3}}{2}$ so that another decomposition of $W$ is derived,
\begin{equation}
[\beta^{+-}_{st}] = [\beta^{+-}_{st}] = \left[
  \begin{array}{ccccc}
    2\sqrt3+2 & 0 & 0 & 0 & -\sqrt3 \\
    0 & -1 & 0 & 0 & 0\\
    0 & 0 & -1 & 0 & 0 \\
    0 & 0 & 0 & 1  & 0\\
    -\sqrt3 & 0 & 0 & 0  & 0\\
  \end{array}
\right]
\end{equation}
Again, in this setting, only six measurements are required. The coefficients and probabilities of outcomes $+-$ and $-+$ are shown in Table \ref{Tab:Coeffpmpm}.

\begin{table}[hbt]\tiny
\centering
\caption{Coefficients and probabilities for MDIEW with outcomes $+-$ and $-+$. Note that when $\beta=0$, the corresponding probability $p$ is irrelevant.} \label{Tab:Coeffpmpm}
\begin{tabular}{c | c c c c c}
  \hline
   & $\tau_0= I/2$ & $\tau_1 = \frac{I+\sigma_x}{2}$ & $\tau_2 = \frac{I+\sigma_y}{2}$ & $\tau_3 = \frac{I+\sigma_z}{2}$ & $\tau'_4 = \frac{I+(-\sigma_x-\sigma_y+\sigma_z)/\sqrt{3}}{2}$  \\
   \hline
  $\omega_0 = I/2$ & $\beta = 2\sqrt3+2, p = \frac{1}{16}$ & $\beta = 0$ & $\beta = 0$ & $\beta = 0$ & $\beta = -\sqrt3, p = \frac{1}{16}$ \\
  $\omega_1 = \frac{I+\sigma_x}{2}$ & $\beta = 0$ & $\beta = -1, p = \frac{2-v}{16}$ & $\beta = 0$ & $\beta = 0$ & $\beta = 0$ \\
  $\omega_2 = \frac{I+\sigma_y}{2}$ & $\beta = 0$ & $\beta = 0$ & $\beta = -1, p = \frac{2-v}{16}$ & $\beta = 0$ & $\beta = 0$ \\
  $\omega_3 = \frac{I+\sigma_z}{2}$ & $\beta = 0$ & $\beta = 0$ & $\beta = 0$ & $\beta = 1, p = \frac{v}{8}$ & $\beta = 0$ \\
  $\omega'_4 = \frac{I+(-\sigma_x-\sigma_y+\sigma_z)/\sqrt{3}}{2}$ & $\beta = -\sqrt3, p = \frac{1}{16}$ & $\beta = 0$ & $\beta = 0$ & $\beta = 0$ & $\beta = 0$ \\
  \hline
\end{tabular}
\end{table}

Although each of the four cases above defines an MDIEW, we can combine four of them as one to enhance the successful probability of MDIEW,
\begin{equation} \label{eq:newJ}
\begin{aligned}
J &= \frac{1}{4}\sum_{s,t}^{a, b}\beta^{a, b}_{s,t}p(a, b|\tau_s, \omega_t)\\
&= \frac{1}{4}\sum_{s,t}( \beta^{++}_{s,t}p(+, +|\tau_s, \omega_t) + \beta^{+-}_{s,t}p(+, -|\tau_s, \omega_t) + \beta^{-+}_{s,t}p(-, +|\tau_s, \omega_t) +\beta^{--}_{s,t}p(-, -|\tau_s, \omega_t))
\end{aligned}
\end{equation}
By doing this, we improve the efficiency of experiments by four times comparing to the original proposal \cite{Branciard13}.

To witness entanglement for the bipartite states defined in Eq.~\eqref{rhovab} with MDIEW defined in Eq.~\eqref{eq:newJ}, in total eight different ancillary state pairs should be prepared, and the results are summarized in Table \ref{decom}.
\begin{table}[hbt]
\centering
  \caption{Our MDIEW in the form of Eq.~\eqref{eq:newJ} for the bipartite states defined in Eq.~\eqref{rhovab}.}\label{decom}
  \begin{tabular}{cccccc}
    \hline
    $\tau_s$ & $\omega_t$ & $\beta^{++}_{st}$ & $p(+,+|\tau_s, \omega_t)$ & $\beta^{+-}_{st}$ & $p(+,-|\tau_s, \omega_t)$\\
    \hline
    $I/2$ & $I/2$& $2\sqrt{3}-2$& $1/16$ & $2\sqrt{3}+2$& $1/16$\\
    $\frac{I+\sigma_x}{2}$ & $\frac{I+\sigma_x}{2}$& $1$& $v/16$& $-1$& $(2-v)/16$\\
    $\frac{I+\sigma_y}{2}$ & $\frac{I+\sigma_y}{2}$& $1$& $v/16$& $-1$& $(2-v)/16$\\
    $\frac{I+\sigma_z}{2}$ & $\frac{I+\sigma_z}{2}$& $1$& $v/8$& $1$& $v/8$\\
    $I/2$ & $\frac{I+(\sigma_x+\sigma_y+\sigma_z)/\sqrt{3}}{2}$& $-\sqrt{3}$& $1/16$ & 0 &- \\
    $\frac{I+(\sigma_x+\sigma_y+\sigma_z)/\sqrt{3}}{2}$ & $I/2$& $-\sqrt{3}$& $1/16$ &0&- \\
    $I/2$ & $\frac{I+(-\sigma_x-\sigma_y+\sigma_z)/\sqrt{3}}{2}$&0&- & $-\sqrt{3}$& $1/16$\\
    $\frac{I+(-\sigma_x-\sigma_y+\sigma_z)/\sqrt{3}}{2}$ & $I/2$&0&- & $-\sqrt{3}$& $1/16$\\
    \hline
  \end{tabular}
\end{table}

\section{Experimental realization}
\subsection{Experiment setup}
Our experimental setup for MDIEW is shown in Fig.~\ref{expsetup}, where a six-photon interferometry is utilized. The to-be-witnessed bipartite state $\rho^v_{34}$, defined in Eq.~\eqref{eq:rho}, is encoded in the photon pair 3 and 4. Photon pairs 1, 2 and 5, 6 are used to prepare the ancillary input states $|\tau_s\rangle_2$ and $|\omega_t\rangle_5$, respectively. In our work, various bipartite states $\{\rho^v_{34}\}$, from maximally entangled to separable, are prepared and tested with the MDIEW. The bipartite state $\rho^v_{34}$ is firstly prepared in the Bell state $\left| {{\Phi ^ - }} \right\rangle _{34} = \left( {{\left| {HH} \right\rangle } - \left| {VV} \right\rangle } \right)/\sqrt 2$ via a Bell-state synthesizer \cite{yao12}. As the coherence length of photons is limited by the interference filtering, two 2-mm BBO crystals in each arm result in a relative phase delay between horizontal and vertical polarization components and cause polarization decoherence. Different $v$ can be selected by the ``state selector" \cite{White01}. They satisfy the relation of
\begin{equation}\label{eq:vtheta}
v  = {\cos ^2}\left( {2\theta } \right),
\end{equation}
where $\theta$ is the angle of the fast axis of the selector HWP.

\begin{figure}[hbt]
\centering
\resizebox{12cm}{!}{\includegraphics[scale=1]{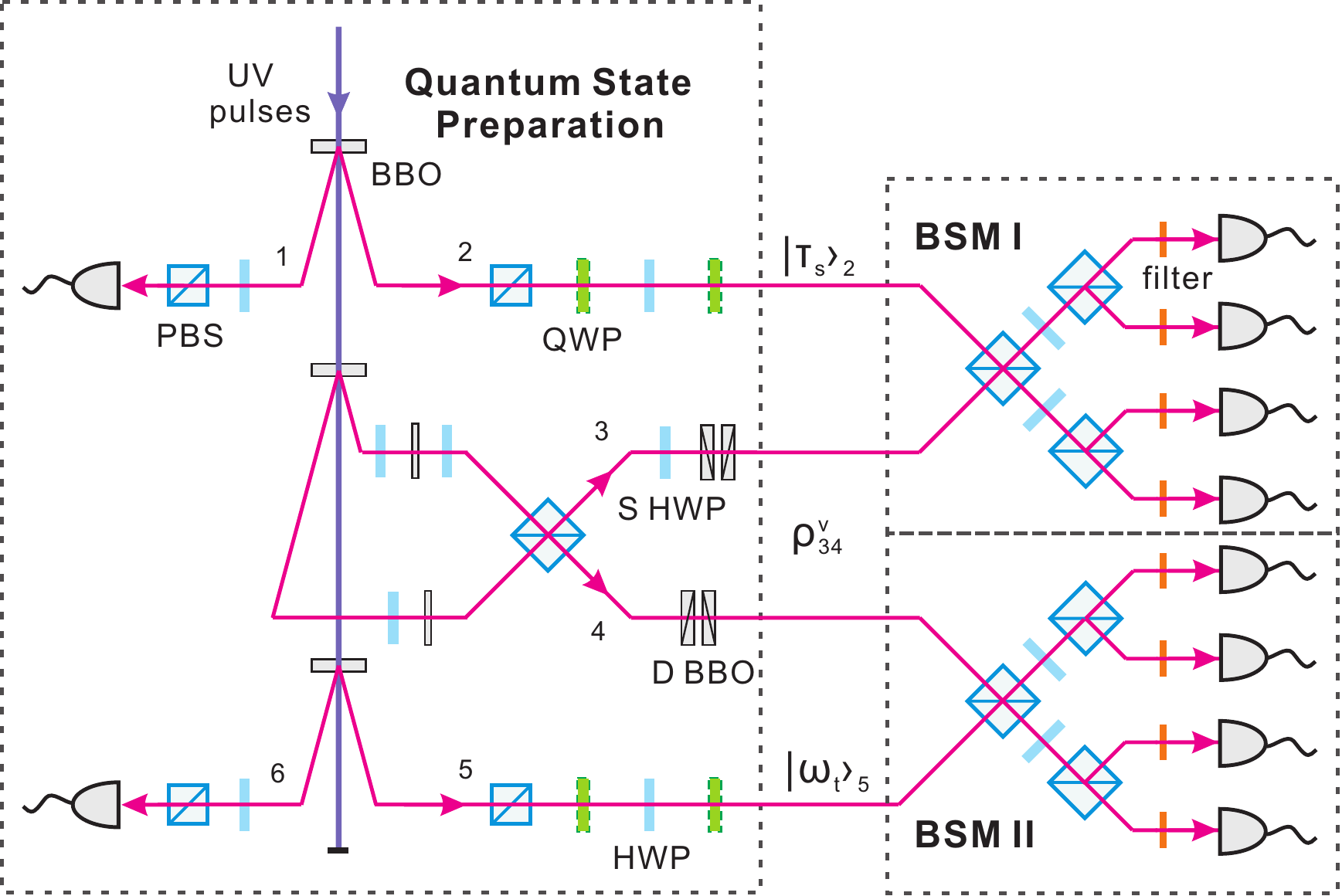}}
\caption{Experimental setup for the MDIEW. The photon pairs are generated by type-II SPDC in 2-mm $\beta$-barium-borate (BBO) crystals. The pulsed pump laser has a central wavelength of 390 nm and a repetition rate of 76 MHz. To prepare the desired state~\eqref{eq:rho}, two 2-mm decoherer BBOs (D BBO) on each side with fast axis setting at $0^\circ $ (up) and $180^\circ$ (down) to reduce the spatial walk-off effect. By changing the angle $\theta$ of the selector HWP (S HWP), the desired state~\eqref{eq:rho} is prepared with $ v=cos^2(2 \theta)$. Heralded photons 2 and 5 are triggered by the detections of photon 1 and 6, respectively. Waveplates are used to rotate the polarizations to encode photons 2 and 5 to the desired states, $\left| {{\tau _s}}\right\rangle _2$ and $\left| {{\omega _t}} \right\rangle _5$. The BSM module is composed of three PBSs and two HWPs at $22.5^\circ $. All photons are filtered by narrow-band filters (with $\lambda_{FWHM}$  = 2.8 nm for BSM I and $\lambda _{FWHM}$  = 8.0 nm for BSM II) and then coupled into single-mode fibers which connect to SPCMs.
}\label{expsetup}
\end{figure}

In the experiment, eight ancillary state pairs $\{\tau_s, \omega_t\}$ are prepared. The states are encoded by tunable waveplates (one HWP sandwiched by two QWPs), which can realize arbitrary single-qubit unitary transformation. Different from directly polarization measurement in the conventional EW, the analysis of MDIEW is completed by BSMs on $\rho_{3}^v\otimes|\tau_s\rangle\langle\tau_s|_2$ and $\rho_{4}^v\otimes|\omega_t\rangle\langle\omega_t|_5$, with two, $|\Phi^\pm\rangle=(|HH\rangle \pm |VV\rangle)/{\sqrt{2}}$,  out of four outcomes been collected.

\begin{figure}[hbt]
\centering
\resizebox{12cm}{!}{\includegraphics[scale=1]{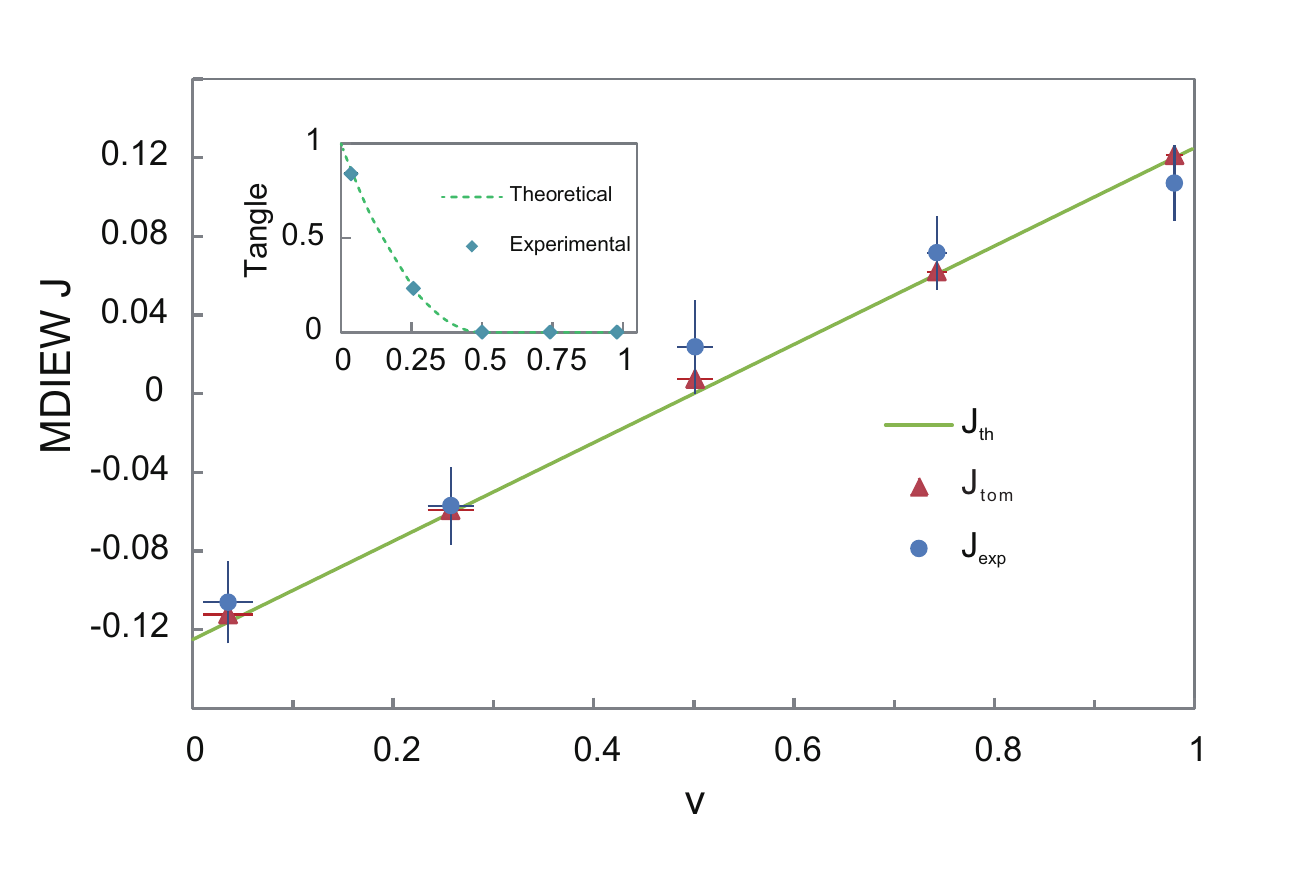}}
\caption{MDIEW values are compared for three cases. The theoretical results ($J_{th}$, solid line) are calculated for the states $\rho_{AB}^v$ with different values of $v$ in Eq.~\eqref{eq:rho}. The tomography results ($J_{tom}$, triangle points) are evaluated for the states $\rho_{34}^v$ after performing tomography on the to-be-witnessed bipartite state. Each point of the experimental results ($J_{exp}$, circular points) is measured from a 16-hour experiment. Vertical error bars indicate one standard deviation and horizontal error bars of the fitting values $v$ from state tomography are described in Supplemental Materials. The inset shows theoretical and experimental values of tangle for input states $\rho_{34}^v$.} \label{resmdi}
\end{figure}

As defined in Eq.~\eqref{eq:newJ}, we obtain the experimental results $J_{exp}^v$ as shown in Fig.~\ref{resmdi}. In comparison, we also plot $J_{th}(\rho_{AB}^v)$ for all values of $v$. Recall that in the aforementioned time-shift attack demonstration, the conclusion from the conventional witness is entangled for $v=1$, whereas here we show that our MDIEW result is 0.107 $\pm$ 0.019 and does not conclude an entangled state. One can see that our MDIEW is immune to this attack.

Furthermore, we perform tomography on the to-be-witnessed bipartite states $\{\rho_{34}^{v}\}$. The results of the density matrices are shown in Fig.~\ref{denvt}. The corresponding $v$ are set by the angle $\theta$ of the selector HWP given in Eq.~\eqref{eq:vtheta}, which is consistent with our fitting values as shown in Supplemental Materials. We evaluate the MDIEW results, Eq.~\eqref{eq:newJ}, from the results of the state tomography $J_{tom}$ as shown in Fig.~\ref{resmdi}. Meanwhile, to quantify the entanglement of the bipartite states $\{\rho_{34}^{v}\}$, we adopt the measure of tangle \cite{Wootters98}, which can be directly calculated from tomography results. When the tangle goes to zero, the bipartite state becomes a separable state. As shown in the insert of Fig.~\ref{resmdi}, no entanglement exists when $v$ grows beyond $1/2$. Such phenomenon is related to the ``sudden death of entanglement" \cite{PhysRevLett.93.140404}.

\begin{figure*} [hbt]
\centering
\resizebox{16cm}{!}{\includegraphics[scale=1]{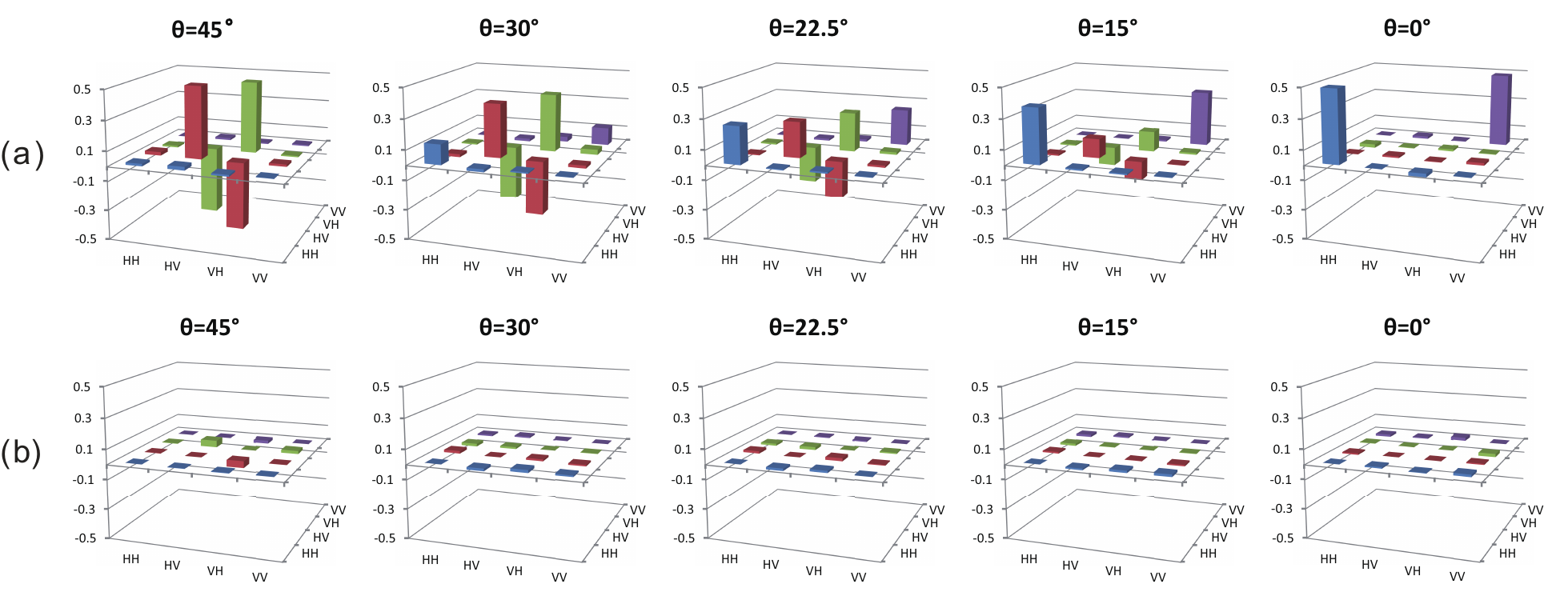}}
\caption{Tomography of the bipartite state $\rho^v_{34}$. Density matrices are constructed through tomography and over 250,000 coincidence detection events are obtained for each plot. Depending on the angle $\theta$ of the state selector defined in Eq.~\eqref{eq:vtheta}, various states ${\rho _{34}^v}$ are prepared. (a) Real part of the density matrices $\rho^v_{34}$. (b) Imaginary part of the density matrices $\rho^v_{34}$.
}\label{denvt}
\end{figure*}

\subsection{Experiment result}\label{appTomography}
\subsubsection{Tomography}
In the experiment, we prepare the to-be-witnessed bipartite states $\rho_{AB}^v$ in the form of Eq.~\eqref{rhovab} with different values $v$. To verify whether the prepared states $\rho^v_{34}$ is close to the desired ones $\rho_{AB}^v$, their density matrices are reconstructed via quantum tomography with $v$ controlled by the angle $\theta$ of the selector HWP, as shown in Eq.~(4) in Main Text. Then we fit the value $v$ by the measured density matrixes $\rho^v_{34}$ to the desired states $\rho_{AB}^v$. As shown in Eq.~\eqref{rhovab}, $\rho_{AB}^v$ contains only real numbers, we can infer $v$ from the real part of $\rho^v_{34}$, and the imaginary parts are supposed to be near zero.

The parameter $v$ can be derived from the real-part of matrix $\rho^v_{34}$. For each matrix elements of $\rho^v_{34}$, $\rho_{11}, \rho_{22}, \rho_{33}, \rho_{44}$, and $\rho_{23}$ ($\rho_{32}$ is identical to $\rho_{23}$), one can estimate $v$, as shown in Table \ref{vfit1}. Accordingly, the average value of $v$ and its error bar are evaluated. As one can see that the experimental results agree the theoretical results well.

\begin{table}[hbt]
\caption{Error estimation of density matrix, real non-zero parts} \label{vfit1}
\begin{tabular}{cccccccccc}
    \hline
    & &\multicolumn{5}{c}{$v_{experiment}$} \\
    \cline{3-7}
     $\theta$ & $v_{theory}$ & $v_{\rho_{11}}$ & $v_{\rho_{22}}$ & $v_{\rho_{33}}$ & $v_{\rho_{44}}$ & $v_{\rho_{23}}$ & $\overline{v_{exp}}$ & $\delta\overline{v_{exp}}$ & $\delta v_{exp}$\\
    \hline
    45&    0&	0.0196 &	0.0228 &	0.0064 &	0.0258 	&0.0290 	&0.0207 	&0.0039 	&0.0087\\
    30& 0.25&	0.2580 &	0.2538 	&0.2426 	&0.2686 	&0.2644 	&0.2575 	&0.0045 	&0.0101\\
    22.5&  0.5	&0.4944 &	0.4820& 	0.4824 	&0.5230 	&0.5108 	&0.4985 	&0.0081 	&0.0180\\
    15& 0.75&	0.7298 	&0.7198 	&0.7280 	&0.7718 	&0.7620 	&0.7423 	&0.0103 	&0.0231\\
    0&    1&	0.9680 &	0.9818 &	0.9222 &	0.9684 	&0.9822 	&0.9645 	&0.0110 	&0.0246\\
    \hline
  \end{tabular}
\end{table}


\subsubsection{Tangle}
To quantify the entanglement of quantum states, we adopt the measure of tangle \cite{Wootters98}. For a 2-qubit state, $\rho_{AB}$, one can evaluate its tangle by the following steps.
\begin{enumerate}
\item
Define a non-Hermitian matrix
  \begin{equation}\label{}
    R = \rho_{AB}\Sigma\rho_{AB}^T\Sigma,
  \end{equation}
  where $\rho^T_{AB}$ is the transpose of $\rho_{AB}$, and the ``spin flip matrix $\Sigma$'' is defined as
  \begin{equation}\label{}
     \Sigma= \left[
  \begin{array}{ccccc}
0 &   0 &    0 &  -1\\		
0 &   0      &   1& 0	\\	
0  & 1  &   0        &  0	\\
-1&   0  & 0  & 0\\
  \end{array}
\right];
  \end{equation}

\item
Calculate the eigenvalues of $R$, and arrange them in decreasing order, $\lambda_1 \geq \lambda_2 \geq \lambda_3 \geq \lambda_4$;

\item
The concurrence of $\rho_{AB}$ is defined as
\begin{equation}\label{}
    C = Max\{0, \sqrt{\lambda_1}- \sqrt{\lambda_2}-\sqrt{\lambda_3}-\sqrt{\lambda_4}\};
\end{equation}

\item
The tangle is defined as
\begin{equation}\label{}
   tangle = C^2.
\end{equation}
\end{enumerate}

The tangle of a bipartite state is a measure of entanglement. If the tangle is zero, then the bipartite state $\rho_{AB}$ must be a separable state. For states defined in Eq.~\eqref{rhovab}, we can calculate the corresponding tangle. By following the aforementioned steps, we first calculate the four eigenvalues, $0, (1-v)^2, v^2/4, v^2/4$. For $v>2/3$, we have $v^2/4 > (1-v)^2$ and hence $tangle = C^2 = 0$. For $2/3\geq v$, we have $v^2/4 \leq (1-v)^2$ and hence $\sqrt{(1-v)^2}- 2\sqrt{v^2/4} = 1-2v$. Therefore, $C = 0$ for $v \geq 1/2$ and $C = 1-2v$ for $v<1/2$,
\begin{equation}\label{}
       tangle(\rho^v_{AB})= \left\{
  \begin{array}{cc}
(1-2v)^2 &  v<1/2 \\		
0 &   v\geq1/2. \\	
  \end{array}
  \right.
\end{equation}
The fitting value of $v$ from state tomography and the tangles are shown in Table \ref{ps}.


\begin{table}[hbt]\centering
  \caption{The tangle values of the input states by tomography.}\label{ps}
  \begin{tabular}{cccccc}
  \hline
    $\theta_{exp}$& $v_{theory}$& $v_{exp}$ &  $v_{error}$& tangle($\rho^v_{34}(\theta)$) & $\textrm{tangle}_{\textrm{error}}$\\
    \hline
    $45^\circ$& 0 & 0.021 &	0.009 &	0.840 & 0.001\\
    $30^\circ$& 0.25 & 0.257 & 0.010 & 0.233 & 0.001\\
    $22.5^\circ$& 0.5 & 0.499 &	0.018 &	0.000 & 0\\
    $15^\circ$& 0.75 & 0.742 &	0.023 &	0.000 & 0\\
    $0^\circ$& 1 & 0.965 &	0.025 &	0.000 & 0\\
    \hline
  \end{tabular}\label{vtangle}
\end{table}

\chapter{Reliable and robust entanglement witness}
This chapter introduces the reliable and robust problem in entanglement witness. The reliable problem can be overcome by the MDIEW scheme. While we show in this chapter how the robust problem can also be resolved \cite{PhysRevA.93.042317}.

\section{Reliable and robust problem in EW}




In reality, EW implementation may suffer from two problems. The first one is \emph{reliability}. That is, one might conclude unreliable results due to imperfect experimental devices. In this case, the validity of the EW result depends on how faithful one can implement the measurements according to the witness $W$. If the realization devices are not well calibrated, the practically implemented observable $W'$ may deviate from the original theoretical design $W$, see Fig.~\ref{fig:Witnesscw} as an example, which can even be not a witness. That is, there may exist some separable states $\sigma$, such that $\mathrm{Tr}[\sigma W'] < 0 \le \mathrm{Tr}[\sigma W]$. Practically, by exploiting device imperfections, an attack has been experimentally implemented for an entanglement witness procedure \cite{Yuan14}. In cryptographic applications, such problem is regarded as a loophole, where one mistakes separable states to be entangled ones. For instance, in QKD, this would indicate that an adversary successfully convinces the users Alice and Bob to share keys which they think are secure but are eavesdropped.  Such problem is solved by the measurement-device-independent QKD scheme \cite{Lo2012MDI}, inspired by the time-reversed entanglement-based scheme \cite{Biham:1996:Quantum,Inamori:TimeReverseEPR:2002,Braunstein12}. Branciard et al.~applied a similar idea to EW and proposed the measurement-device-independent entanglement witness (MDIEW) scheme \cite{Branciard13}, in which entanglement can be witnessed without assuming the realization devices. The MDIEW scheme is based on an important discovery that any entangled state can be witnessed in a nonlocal game with quantum inputs \cite{Buscemi12}. In the MDIEW scheme, it is shown that an arbitrary conventional EW can be converted to be an MDIEW, which has been experimentally tested \cite{Yuan14}.

\begin{figure}[bht]
\centering
\resizebox{8cm}{!}{\includegraphics[scale=1]{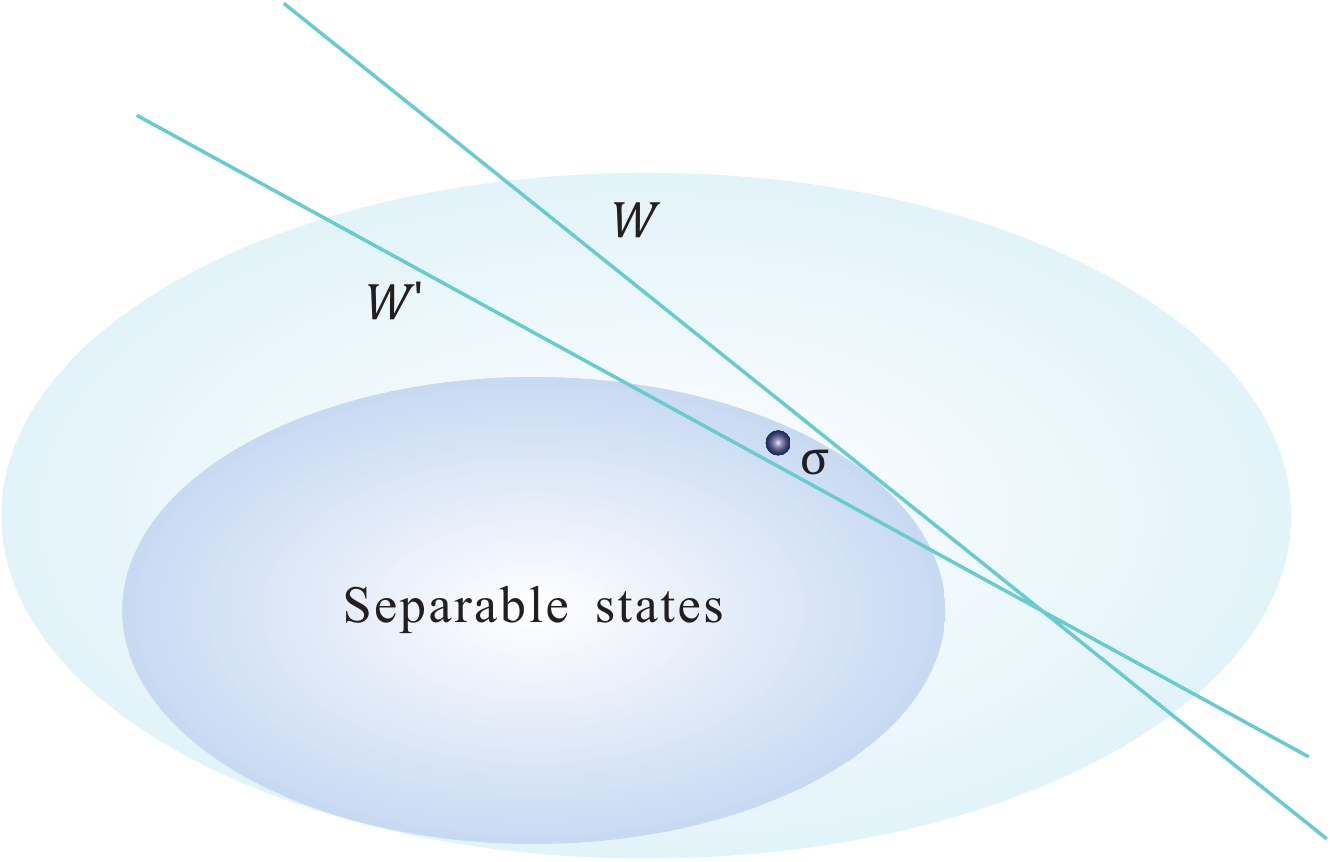}}
\caption{Entanglement witness and the reliability problem. }\label{fig:Witnesscw}
\end{figure}

The second problem lies on the \emph{robustness} of EW implementation. Since each (linear) EW can only identify certain regime of entangled states, a given EW is likely to be ineffective to detect entanglement existing in an unknown quantum state. While a failure of detecting entanglement is theoretically acceptable, in practice, such failure may cause experiment to be highly inefficient. In fact, a conventional EW can only be designed optimal when the quantum state has been well calibrated, which, on the other hand, generally requires to run quantum state tomography.
Practically, when the prepared state can be well modeled, one can indeed choose the optimal EW to detect its entanglement. Since a full tomography requires exponential resources regarding to the number of parties, EW plays as an important role for detecting well modeled entanglement, which would generally fail for an arbitrary unknown state.
In a way, this problem becomes more serious in the MDIEW scenario, where the measurement devices are assumed to be uncharacterized and even untrusted. In this case, the implemented witness, which may although be designed optimal at the first place, can become a bad one which merely detects no entanglement. However, the observed experimental data may still have enough information for detecting entanglement.
Therefore, the key problem we are facing here is that given a set of observed experimental data, what is the best entanglement detection capability one can achieve. 

In detecting quantum nonlocality, a similar problem is to find the optimal Bell inequality for the observed correlation, which can be solved efficiently with linear programming \cite{boyd2004convex}.
Regarding to our problem, we essentially need to optimize over all entanglement witness to draw the best conclusion of entanglement with the same experiment data, as shown in Fig.~\ref{fig:witness2}(a). As the set of separable states is not a polytope, this problem cannot be solved by linear programming. Generally speaking, it is proved that the problem of accurately finding such an optimal witness is NP-hard \cite{ben1998robust}. However, if certain failure probability is tolerable, we show in this work that this problem can be efficiently solved. That is, if we admit a probability less than $\epsilon$ to detect a separable state to be entangled, we show that the optimal entanglement witness can be efficiently found. As the optimization step can be effectively conducted as post-processing, our scheme does not pose extra burdens to experiments compared to the original MDIEW scheme. In this case, our result can be directly applied in practice.




\begin{figure*}[bht]
\centering
\resizebox{14cm}{!}{\includegraphics[scale=1]{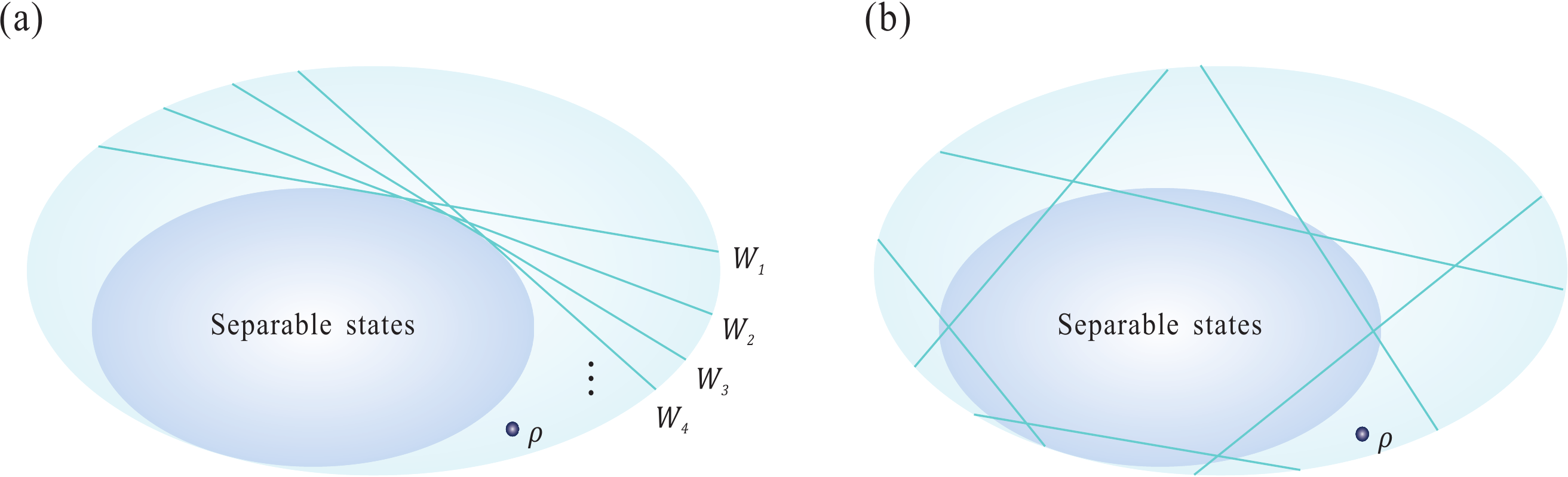}}
\caption{Optimization of entanglement witnesses. (a) To get the optimal witness of an unknown entangled state $\rho$, one has to run over all possible witnesses. Intuitively, this is done by scanning over all witnesses that are \emph{tangent} to the set of separable states. (b) The optimization can be efficiently done if certain failure probability can be tolerated.}\label{fig:witness2}
\end{figure*}


\section{Reliable entanglement witness}\label{Sec:MDIEW}
The reliability problem can be overcome by the MDIEW scheme. For self-consistency, we will breifly review the MDIEW scheme.
\subsection{Nonlocal game}
Before, we first discuss about nonlocal games with classical and quantum inputs as shown in Fig.~\ref{fig:nonlocal}. In a classical nonlocal game, classical random inputs $x$ and $y$ are given to two spacelikely separated users Alice and Bob, who perform measurement on pre-shared entangled state $\rho_{AB}$ and output $a$ and $b$, respectively. According to the probability distribution $p(a,b|x,y)$, a Bell inequality can be defined by
\begin{equation}\label{eq:Bell}
I = \sum_{a,b,x,y} \beta_{a,b}^{x,y}{p}(a,b|x,y) \geq I_C,
\end{equation}
where $I_C$ is a bound for all separable state $\sigma_{AB}$. A violation of the inequality can be considered as a witness for entanglement. As the Bell test does not assume measurement detail, witnessing entanglement by Bell test is device independent. However, as the conclusion is so strong such that the implementation is self-testing, not all entangled states can be witnessed in such a way \cite{Werner89, Barrett02}. Furthermore, the requirement of a faithful Bell test is very high, which makes such a witnesses impractical. For instance, the minimum efficiency required is $2/3$ for all Bell tests with binary inputs and outputs \cite{Massar03, Wilms08}.
On the other hand, if we can trust the measurement, a Bell test essentially becomes an EW. Although such method is able to detect all entangled state and is easy to realize, this scheme is not measurement-device-imperfection-tolerant.


\begin{figure}[bht]
\centering
\resizebox{10cm}{!}{\includegraphics[scale=1]{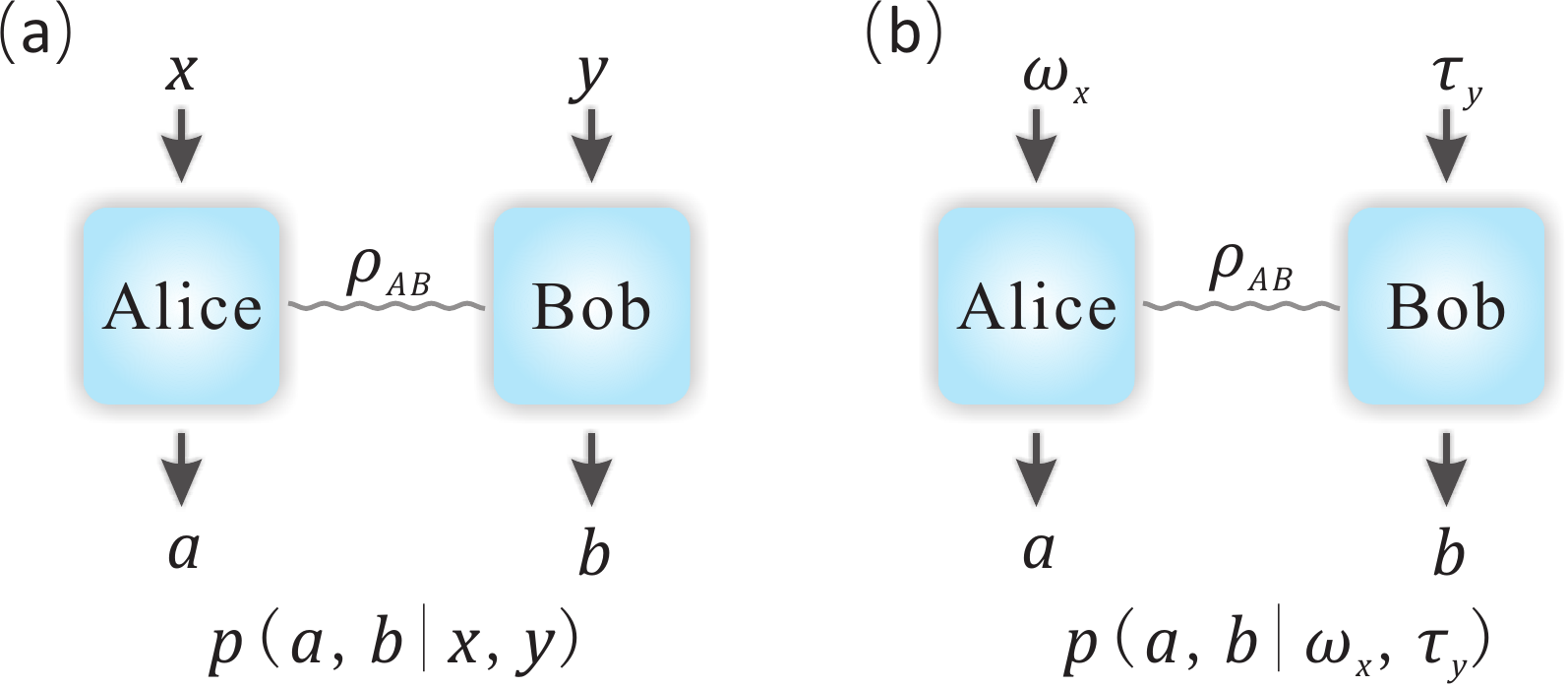}}
\caption{Bipartite nonlocal game with classical and quantum inputs. (a) Nonlocal game with classical inputs. Based on the classical inputs $x$ and $y$, Alice and Bob perform local measurement on the pre-shared entangled state $\rho_{AB}$, and get classical outputs $a$ and $b$, respectively. A linear combination of the probability distribution $p(a,b|x,y)$ defines a Bell inequality as shown in Eq.~\eqref{eq:Bell}. 
(b) Nonlocal game with quantum inputs. The quantum inputs of Alice and Bob are respectively $\omega_x$ and $\tau_y$. It is shown \cite{Buscemi12} that any entangled quantum states can be witnessed with a certain nonlocal game with quantum inputs. Equivalently, if we consider that Alice and Bob each prepares an ancillary state and a third party Eve performs the measurement, this setup also corresponds to the case of MDIEW. }\label{fig:nonlocal}
\end{figure}

In the seminal work \cite{Buscemi12}, Buscemi introduces the concepts of nonlocal games with quantum inputs. Denote the inputs of Alice and Bob by $\omega_x$ and $\tau_y$, then an inequality similar to Bell inequality can be defined by
\begin{equation}\label{eq:QuantumBell}
J = \sum_{a,b,x,y} \beta_{a,b}^{x,y}{p}(a,b|\omega_x,\tau_y) \geq J_C,
\end{equation}
where $J_C$ is also the bound for all separable state $\rho_{AB}$. As the quantum inputs can be indistinguishable, it is proved that all entangled states can violate a certain inequality \cite{Buscemi12}. If we consider the input states are faithfully prepared by Alice and Bob, then such nonlocal game with quantum inputs can be considered as an MDIEW \cite{Branciard13}. Moreover, as shown below, there is no detection efficiency limit for such a test.

\subsection{MDIEW}
The nonlocal game presented in Ref.~\cite{Buscemi12} can be considered as a reliable entanglement witness method, which does not witness separable state as entangled with arbitrary implemented measurement. This nonlocal game is thus an MDIEW, i.e., $J\ge 0$ for all separable states and $J$ can be negative if Alice and Bob share entangled state. Furthermore, the statement that $J\ge 0$ for all separable states is independent of the implementation of the measurement. In Ref.~\cite{Branciard13}, the authors put this statement into more concrete and practical framework. They show that, for an arbitrary conventional EW, there is a corresponded MDIEW. Below, we will quickly show how to derive MDIEWs from conventional EWs.



Focus on the bipartite scenario with Hilbert space $\mathcal{H}_A\otimes \mathcal{H}_B$, with dimensions $\mathrm{dim}\mathcal{H}_A=d_A$ and $\mathrm{dim}\mathcal{H}_B=d_B$. For a bipartite entangled state $\rho_{AB}$ defined on $\mathcal{H}_A\otimes \mathcal{H}_B$, we can always find a conventional entanglement witness $W$ such that $\mathrm{Tr}[W\rho_{AB}] < 0$ and $\mathrm{Tr}[W\sigma_{AB}]\ge 0$ for any separable state $\sigma_{AB}$. Suppose $\{\omega_x^{\mathrm{T}}\}$ and $\{\tau_y^{\mathrm{T}}\}$ to be two bases for Hermitian operators on $\mathcal{H}_A$ and  $\mathcal{H}_B$, respectively. Thus, we can decompose $W$ on the basis $\{\omega_x^{\mathrm{T}}\otimes\tau_y^{\mathrm{T}}\}$ by
\begin{equation}\label{eq: decomposition}
	W = \sum_{x,y} \beta^{x,y}\omega_x^{\mathrm{T}} \otimes \tau_y^{\mathrm{T}},
\end{equation}
where $\beta^{x,y}$ are real coefficients and the transpose is for later convenience. Notice that, owing to the completeness of the set of density matrices, we further require $\{\omega_x\}$ and $\{\tau_y\}$ to be density matrices. In addition, the decomposition of Hermitian operators is not unique which varies with different $\{\omega_x\}$ and $\{\tau_y\}$.

With a conventional EW decomposed in Eq.~\eqref{eq: decomposition}, an MDIEW can be obtained by
\begin{equation}\label{eq: simplifiedBell}
	J = \sum_{x,y} \beta^{x,y}_{1,1} {p}(1,1|\omega_x,\tau_y)
\end{equation}
where $\beta^{x,y}_{1,1} =  \beta^{x,y}$ and ${p}(1,1|\omega_x,\tau_y)$ is the probability of outputting $(a=1,b=1)$ with input states $(\omega_x,\tau_y)$. In the MDIEW design, Alice (Bob) performs Bell state measurement on $\rho_A$ ($\rho_B$) and $\omega_x$ ($\tau_y$). The probability distribution ${p}(1,1|\omega_x,\tau_y)$ is thus obtained by the probability of projecting onto the maximally entangled state $\ket{\Phi_{AA}^+} =1/\sqrt{d_A}\sum_i \ket{ii}$ and $\ket{\Phi_{BB}^+} = 1/\sqrt{d_B} \sum_j \ket{jj}$.

As shown in Ref.~\cite{Branciard13}, $J$ is linearly proportional to the conventional witness with ideal measurement,
\begin{equation}\label{}
  J = \mathrm{Tr}[W\rho_{AB}] / d_A d_B.
\end{equation}
Thus, $J$ defined in Eq.~\eqref{eq: simplifiedBell} witnesses entanglement. Furthermore, it can be proved that such a witness is independent of the measurement devices. That is, even if the measurement devices are imperfect, $J$ is always non-negative for all separable states and hence no separable state will be mistakenly witnessed to be entangled. We refer to Ref.~\cite{Branciard13} for a rigorous proof.

Theoretically, the MDIEW scheme prevents identifying separable states to be entangled. Such a reliable MDIEW has been experimentally demonstrated lately \cite{Yuan14}. In practice, however, such a scheme can be inefficient, meaning that it witnesses very few entangled states despite that the observed data could actually provide more information. This is because, in the MDIEW procedure, one first chooses a conventional EW and realize in an MDI way. The conventional EW is chosen based on an empirical estimation of the to-be-witnessed state, thus it may not be able to witness the state for an ill estimation. Furthermore, even if the conventional EW is optimal at the first place, the measurement imperfection will make it sub-optimal in practice. Especially, when the input states $\{\omega_x\otimes\tau_y\}$ is complete, a specific witness may not be able to detect entanglement. With complete information, a natural question is whether we can obtain maximal information about entanglement, i.e., get the optimal estimation of MDIEW.

\section{Robust MDIEW}\label{Sec:OMDIEW}
Now, we present a method to optimize the MDIEW given a fixed observed experiment data $p(1,1|\omega_x,\tau_y)$. Before digging into the details, we compare the problem to a similar one in nonlocality. In the nonlocality scenario, a Bell inequality is used as a witness for quantumness, see Eq.~\eqref{eq:Bell}. In practice, the Bell inequality may not be optimal for the observed probability distribution $p(a,b|x,y)$. As the probability distribution of classical correlation forms a polytope, one can run a linear programming to get an optimal Bell inequality for $p(a,b|x,y)$. While, in our case, the probability distribution $p(1,1|\omega_x,\tau_y)$ with separable states is only a convex set but no-longer a polytope. Thus, our problem cannot be solved directly with linear programming.

\subsection{Problem formulation}
Let us start with formulating the optimization problem. Informally, our problem can be described as follows,
\begin{flushleft}
\emph{Problem (informal): find an optimal witness for the observed probability distribution $p(1,1|\omega_x,\tau_y)$.}
\end{flushleft}
According to Eq.\eqref{eq: simplifiedBell}, the witness value is defined by a linear combination of $p(1,1|\omega_x,\tau_y)$ with coefficient $\beta^{x,y}$. To witness entanglement, the coefficient $\beta^{x,y}$ must lead to a witness as defined in Eq.~\eqref{eq: decomposition}. In addition, as we can always assign $2\beta^{x,y}$ to double a violation, we require a trace normalization of the witness $W$ by
\begin{equation}\label{}
  \mathrm{Tr}[W] = 1.
\end{equation}
Under this normalization, the optimal entanglement witness $W$ \cite{Lewenstein00} for a given state $\rho$ is defined by the solution to the minimization
\begin{equation}\label{}
  \min \mathrm{Tr}[W \rho].
\end{equation}
Generally speaking, the minimum value, i.e., maximum violation, of the entanglement witness makes the result more robust to experimental errors and statistical fluctuations. Furthermore, a larger violation of entanglement witness can also help for a larger estimation of entanglement measures \cite{eisert2007quantitative}.

Therefore, the problem can be expressed as
\begin{flushleft}
\emph{Problem (formal): For a given probability distribution ${p}(1,1|\omega_x,\tau_y)$, minimize
\begin{equation}\label{eq:Jqminimum}
J(\beta^{x,y}) = \sum_{x,y} \beta^{x,y}{p}(1,1|\omega_x,\tau_y)
\end{equation}
over all $\beta^{x,y}$ satisfying
\begin{equation}\label{eq:constraintsss}
\sum_{x,y}\beta^{x,y}\mathrm{Tr}\left[\sigma_{AB}(\omega_x^{\mathrm{T}}\otimes\tau_y^{\mathrm{T}})\right]\ge0,
\end{equation}
for any separable state $\sigma_{AB}$ and
\begin{equation}\label{}
  \mathrm{Tr}\left[\sum_{x,y} \beta^{x,y}\omega_x^{\mathrm{T}} \otimes \tau_y^{\mathrm{T}}\right] = 1.
\end{equation}
}
\end{flushleft}


Contrary to the optimization of Bell inequality, we can see that this problem is much more complex. When the measurements are implemented faithfully, it is easy to verify that $p(1,1|\omega_x,\tau_y) = \mathrm{Tr}[(\omega_x\otimes\tau_y)\rho_{AB}]/\sqrt{d_Ad_B}$, where $\rho_{AB}$ is the state measured. Therefore, finding the optimal $\beta^{x,y}$ is equivalent to find the optimal entanglement witness $W = \sum_{x,y}\beta^{x, y}\omega_x^{\mathrm{T}}\otimes\tau_y^{\mathrm{T}}$ for state $\rho_{AB}$. A possible solution to this problem is to try all entanglement witnesses to find the optimal one, see Fig.~\ref{fig:witness2}. However, it is proved that the problem of accurately finding such an optimal witness is NP-hard \cite{ben1998robust}. Thus, our problem is also intractable for the most general case.

\subsection{$\epsilon$-level optimal EW}
The key for the problem being intractable is that there is no efficient way to characterize an arbitrary entanglement witness. In the bipartite case, an operator is an witness if and only if
\begin{equation}\label{eq:}
\mathrm{Tr}[\sigma_{AB}W]\ge0
\end{equation}
for any separable state $\sigma_{AB}$. As $\sigma_{AB}$ can always be decomposed as a convex combination of separable states as $\ket{\psi}_A\ket{\phi}_B$, the condition can be equivalently expressed as
\begin{equation}\label{eq:witnessconstraints}
\bra{\psi}_A\bra{\phi}_BW\ket{\psi}_A\ket{\phi}_B\ge0,
\end{equation}
for any pure states $\ket{\psi}_A$ and $\ket{\phi}_B$. The constraints for a witness $W$ are very difficult to describe in the most general case, which makes our problem hard.

While, this problem can be resolved if we allow certain failure errors. A Hermitian operator $W_\epsilon$ is defined as an $\epsilon$-level entanglement witness, when
\begin{equation}\label{eq:epsilonwitness}
\mathrm{Prob}\left\{\mathrm{Tr}[\sigma W_\epsilon]<0|\sigma \in S\right\} \le \epsilon,
\end{equation}
where $S$ is the set of separable states. That is, the operator $W_\epsilon$ has a probability less than $\epsilon$ to detect a randomly selected separable quantum state to be entangled. Intuitively, $\epsilon$ can be regarded as a failure error probability. We refer to Ref.~\cite{Brandao04} for a rigorous definition.  It is shown that the $\epsilon$-level optimal EW can be found efficiently for any given entangled state $\rho$. In particular, constrained on $\mathrm{Tr}[W_\epsilon] = 1$ and $W_\epsilon$ to be an $\epsilon$-level EW, one can run a semi-definite programming (SDP) to minimize $\mathrm{Tr}[W_\epsilon\rho]$.

\subsection{Solution}

Following the method proposed in Ref.~\cite{Brandao04}, we can solve the minimization problem given in Eq.~\eqref{eq:Jqminimum} by allowing a certain failure probability $\epsilon$. First, we relax the constraint given in Eq.~\eqref{eq:constraintsss}. Instead of requiring being non-negative for all separable states, we randomly generate $N$ separable states $\{\ket{\psi}_A^i\ket{\phi}_B^i\}$ and require that
\begin{equation}\label{eq:witnessconstraints2}
\sum_{x,y}\beta^{x,y}\langle\omega_x^{\mathrm{T}}\otimes \tau_y^{\mathrm{T}}\rangle^i\ge0, \forall i\in\{1,2,\dots,N\},
\end{equation}
where $\langle\omega_x^{\mathrm{T}}\otimes \tau_y^{\mathrm{T}}\rangle^i = \bra{\psi}_A^i\bra{\phi}_B^i \omega_x^{\mathrm{T}}\otimes \tau_y^{\mathrm{T}}\ket{\psi}_A^i\ket{\phi}_B^i$. Then the problem can be expressed as

\begin{flushleft}
\emph{Problem ($\epsilon$-level): given a probability distribution $p(1,1|\omega_x,\tau_y)$, minimize
\begin{equation}\label{eq:minimum}
J(\beta^{x,y}) = \sum_{x,y} \beta^{x,y}{p}(1,1|\omega_x,\tau_y)
\end{equation}
over all $\beta^{x,y}$ satisfying
\begin{equation}\label{eq:witnessconstraints3}
\sum_{x,y}\beta^{x,y}\langle\omega_x^{\mathrm{T}}\otimes \tau_y^{\mathrm{T}}\rangle^i\ge0, \forall i\in\{1,2,\dots,N\},
\end{equation}
for $N$ randomly generated separable states $\{\ket{\psi}_A^i\ket{\phi}_B^i\}$ and
\begin{equation}\label{eq:}
\sum_{x,y} \beta^{x,y}\mathrm{Tr}\left[\omega_x^{\mathrm{T}} \otimes \tau_y^{\mathrm{T}}\right] = 1.
\end{equation}
}
\end{flushleft}

This problem can be converted to an SDP solvable problem when we re-express the inequality of numbers in Eq.~\eqref{eq:witnessconstraints3} by an inequality of matrices. To do so, we only need to notice that Eq.~\eqref{eq:witnessconstraints} is equivalent to require that
\begin{equation}\label{eq:witnessconstraints2}
W_B = \bra{\psi}_AW_\epsilon\ket{\psi}_A\ge0, \forall \ket{\psi}_A,
\end{equation}
where $W_B\ge0$ indicates that $W_B$ has non-negative eigenvalues.
Therefore, we only need to generate $N$ states $\ket{\psi}_A^i$, for $i = 1, 2, \dots, N$, and the problem is
\begin{flushleft}
\emph{Problem ($\epsilon$-level, SDP): given a probability distribution ${p}(1,1|\omega_x,\tau_y)$, minimize
\begin{equation}\label{eq:minimum}
J(\beta^{x,y}) = \sum_{x,y} \beta^{x,y}{p}(1,1|\omega_x,\tau_y)
\end{equation}
over all $\beta^{x,y}$ satisfying
\begin{equation}\label{eq:witnessconstraints4}
\sum_{x,y}\beta^{x,y}\bra{\psi}_A^i\omega_x^{\mathrm{T}}\ket{\psi}_A^i\tau_y^{\mathrm{T}}\ge0, \forall i\in\{1,2,\dots,N\},
\end{equation}
for $N$ randomly generated states $\{\ket{\psi}_A^i\}$ and
\begin{equation}\label{eq:}
\sum_{x,y} \beta^{x,y}\mathrm{Tr}\left[\omega_x^{\mathrm{T}} \otimes \tau_y^{\mathrm{T}}\right] = 1.
\end{equation}
}
\end{flushleft}
In practice, we can run an SDP to solve this problem. According to Ref.~\cite{Brandao04, calafiore2005uncertain}, to get the $\epsilon$-level witness with probability at least $1-\beta$, the number of random states $N$ should be at least $r/(\epsilon\beta) - 1$. Here $r$ is the number of optimization variables, i.e., coefficients $\beta$, and $\beta$ can be understood as the failure probability of the minimization program. It is worth to remark that the problem can be similarly solved in the multipartite case.


\subsection{Example}\label{Sec:Example}
In this section, we show explicit examples about how the witness becomes non-optimal in the MDI scenario and how this problem can be resolved by running the optimizing program.

Suppose the to-be-witnessed state is a two-qubit Werner state \cite{Werner89}:
\begin{equation}\label{eq:werner}
\rho^v_{AB} = v\ket{\Psi^-}\bra{\Psi^-} + (1-v)I/4,
\end{equation}
where $\ket{\Psi^-} = 1/\sqrt{2}(\ket{01}-\ket{10})$ and $I$ is the identity matrix. The designed entanglement witness for the Werner states is
\begin{equation}\label{eq:}
W = \frac{1}{2}I - \ket{\Psi^-}\bra{\Psi^-}.
\end{equation}
As $\mathrm{Tr}[W\rho^v_{AB}] = (1-3v)/4$, $\rho^v_{AB}$ is entangled for $v>1/3$ and separable otherwise.

As shown in Ref.~\cite{Branciard13}, we can choose the input set by
\begin{equation}\label{eq:}
\omega_x = \sigma_x\frac{I+\vec{n}\cdot\vec{\sigma}}{2}\sigma_x, \tau_y = \sigma_y\frac{I+\vec{n}\cdot\vec{\sigma}}{2}\sigma_y, x,y = 0,\dots, 3
\end{equation}
where $\vec{n} = (1,1,1)/\sqrt{3}$, $\vec{\sigma} =
(\sigma_1,\sigma_2,\sigma_3)$ is the Pauli matrices, and $\sigma_0 = I$. According to Eq.~\eqref{eq: decomposition}, the witness can be decomposed on the basis of $\{\omega_x\otimes\tau_y\}$ with coefficient $\beta^{x,y}$ given by
\begin{equation}\label{eq:}
\beta^{x,y} =\left\{\begin{array}{cc}
               \frac{5}{8},& \mathrm{if} \, x = y, \\
               -\frac{1}{8},& \mathrm{if} \, x \ne y.
             \end{array}\right.
\end{equation}
And the MDIEW value is given by
\begin{equation}\label{eq:}
J = \frac{5}{8}\sum_{x=y}p(1,1|\omega_x,\tau_y) - \frac{3}{8}\sum_{x\ne y}p(1,1|\omega_x,\tau_y).
\end{equation}

In the ideal case, the probability distribution $p(1,1|\omega_x,\tau_y)$ is obtained by projecting onto maximally entangled states, that is,
\begin{equation}\label{eq:}
p(1,1|\omega_x,\tau_y) = \mathrm{Tr}[\left(M_A \otimes M_B\right)\times (\omega_x \otimes \rho_{AB} \otimes \tau_y)]
\end{equation}
where $M_A = \ket{\Phi_{AA}^+}\bra{\Phi_{AA}^+}$ and $M_B = \ket{\Phi_{BB}^+}\bra{\Phi_{BB}^+}$. While, in practice, there may exist imperfection in measurement. For instance, we consider that Alice's measurement is perfect while Bob's measurement is instead
\begin{equation}\label{eq:}
M_B' = \ket{\Phi_{BB}^-}\bra{\Phi_{BB}^-},
\end{equation}
where $\ket{\Phi_{BB}^-} = 1/\sqrt{2}(\ket{00} - \ket{11})$. In the case of quantum key distribution, projecting onto $\ket{\Phi_{BB}^-}$ can be regarded as a phase error.

As shown in Fig.~\ref{fig:OMDIEW}, we plot the MDIEW  and the optimized MDIEW results. For the original MDIEW result, as Bob's measurement is incorrect, no Werner state given in Eq.~\eqref{eq:werner} can be witnessed to be entangled. Although, by optimizing over all possible entanglement witness, we show that $\rho^v_{AB}$ is entangled as long as $v>1/3$. In this case, the optimized MDIEW can detect all entangled Werner states. In our program, we set $N = 1000$ and we can see from Fig.~\ref{fig:OMDIEW} that no separable state is falsely identified as entangled.

\begin{figure}[bht]
\centering
\resizebox{10cm}{!}{\includegraphics[scale=1]{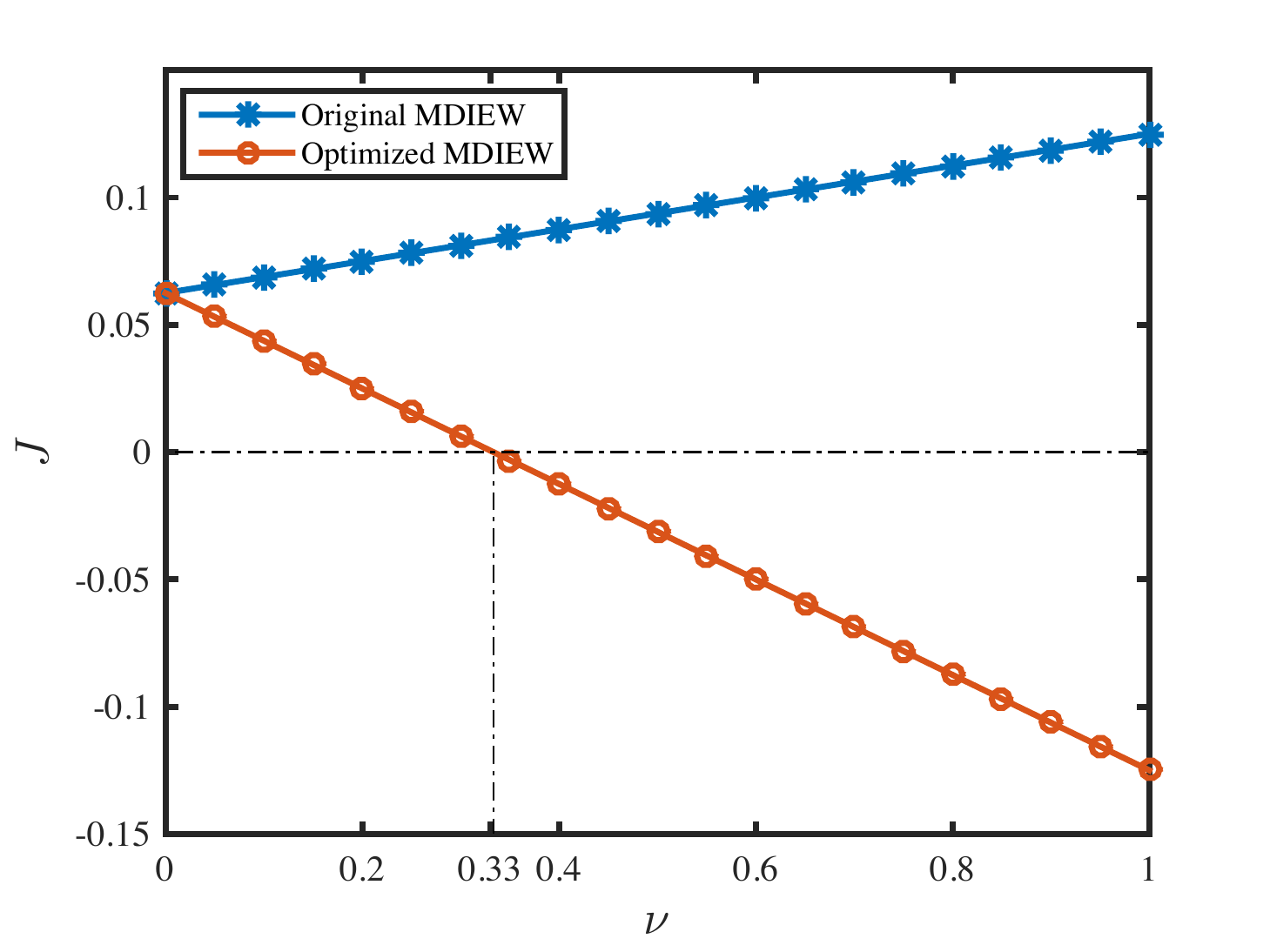}}
\caption{Simulation results of the original and optimized MDIEW protocol. The to be witness state is the two-qubit Werner state defined in Eq.~\eqref{eq:werner}. Here, we consider that Alice projects onto $\ket{\Phi_{AA}^+}$ and Bob projects onto $\ket{\Phi_{BB}^-}$. In this case, the original MDIEW cannot detect entanglement, while the optimized MDIEW protocol detects all entangle Werner states.}\label{fig:OMDIEW}
\end{figure}





The optimization program finds the optimal $\epsilon$-level optimal EW $W_\epsilon$, which as its name indicates, has a probability less than $\epsilon$ to detect an separable state to be entangled. To get a smaller $\epsilon$, one can use a larger number $N$ of random states. In this case, the $\epsilon$ can be regarded as the statistical fluctuation which is inversely related to the number of trials $N$.
On the one hand, to efficiently get the optimal witness $W_\epsilon$, one has to introduce a nonzero failure error $\epsilon$; On the other hand, we can always add an extra term to the EW to eliminate $\epsilon$, i.e.,
\begin{equation}\label{}
  W = W_\epsilon + \alpha I,
\end{equation}
where $\alpha$ is chosen to be the minimum value such that $W$ is an entanglement witness. To efficiently find $\alpha$, one can make use of the technique similar to Ref.~\cite{Sperling13}, in which, EW can be systematically constructed.

\part{Randomness and selftesting}
\chapter{Randomness Requirement on CHSH Bell Test in the Multiple Run Scenario}
This chapter investigates the randomness assumption in Bell test. Specifically, we discuss the randomness requirement such that quantum mechanics can have a violation of the Clauser-Horne-Shimony-Holt (CHSH) inequality \cite{Yuan15CHSH}.

\section{Randomness Requirement}\label{Sec:Randomness}
\subsection{Randomness loophole}
Historically, Bell tests \cite{bell} are proposed for distinguishing quantum theory from local hidden variable models (LHVMs)  \cite{EPR35}. In a general picture, a Bell  test involves multiple parties who randomly choose inputs and generate outputs with pre-shared physical resources. Based on the probability distributions of inputs and outputs, an inequality, called Bell's inequality, is defined. A Bell test is meaningful when all LHVMs satisfy the Bell's inequality; while in quantum mechanics, such inequality can be violated via certain quantum settings. Experimentally observing a violation of any Bell's inequality would show that LHVMs are not sufficient to describe the world, and other theories, such as the quantum mechanics, are demanded.

Here, we focus on the bipartite scenario and investigate one of the most well-known Bell tests, the CHSH inequality \cite{CHSH}. As shown in Fig.~\ref{Fig:BellTest}(a), two space-like separated parties, Alice and Bob, randomly choose input bit settings $x$ and $y$ and generate outputs bits $a$ and $b$ based on their inputs and pre-shared quantum ($\rho$) and classical ($\lambda$) resources, respectively. The probability distribution $p(a,b|x,y)$,  obtaining outputs $a$ and $b$ conditioned on inputs $x$ and $y$, are determined by specific strategies of Alice and Bob. By assuming that the input settings $x$ and $y$ are chosen fully randomly and equally likely, the CHSH inequality is defined by a linear combination of the probability distribution $p(a,b|x,y)$ according to
\begin{equation}\label{eq:Bell}
  S = \sum_{a,b,x,y} (-1)^{a\oplus b + x\cdot y}p(a,b|x,y) \leq S_C = 2,
\end{equation}
where the plus operation $\oplus$ is modulo 2, $\cdot$ is numerical multiplication, and $S_C$ is the (classical) bound of Bell value $S$ for all LHVMs. Similarly, there is an achievable bound $S_Q = 2\sqrt{2}$ for the quantum theory \cite{cirel1980quantum}. In this case, a violation of the classical bound $S_C$ indicates the need for alternative theories other than LHVMs, such as quantum theory. For general no signalling (NS) theories \cite{prbox}, denote the corresponded upper bound as $S_{NS} = 4$. It is straightforward to see that $S_{NS} \geq S_Q\geq S_C$.

\begin{figure*}[thb]
\centering
\resizebox{14cm}{!}{\includegraphics[scale=1]{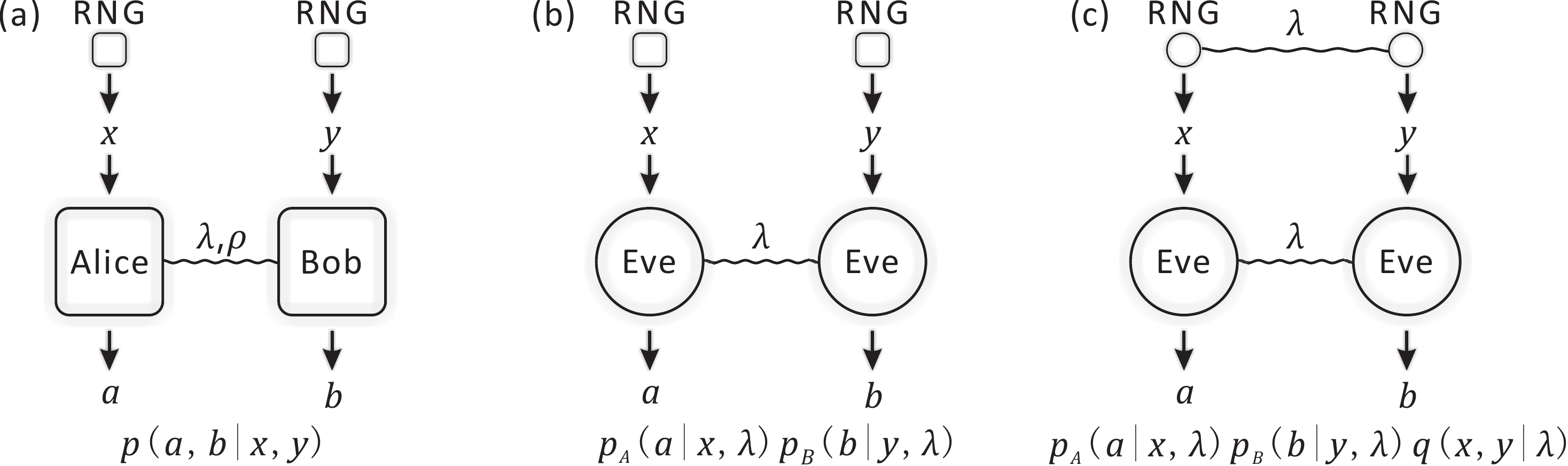}}
\caption{Bell tests in the bipartite scenario. (a) The inputs of Alice and Bob, $x$ and $y$,  are decided by perfect random number generators (RNGs), which produce uniformly distributed random numbers; (b) The measurement devices are controlled by an adversary Eve through local hidden variables $\lambda$; (c) The input random numbers are also controlled by the same local hidden variable $\lambda$, which is accessible to Eve.}\label{Fig:BellTest}
\end{figure*}

The violation of Bell's inequality not only acts as a test for fundamental laws of physics, but has varieties of  applications in modern quantum information tasks. For instance, observing violations of Bell's inequalities can be applied in device independent tasks, such as quantum key distribution \cite{Mayers98, acin06,masanes2011secure,Vazirani14}, randomness amplification \cite{colbeck2012free, gallego2013full, Dhara14} and generation \cite{Colbeck11, Fehr13, Pironio13, Vazirani12}, entanglement quantification \cite{Moroder13}, and dimension witness \cite{Brunner08}.
Security proofs of these tasks are generally independent of the realization devices or correctness of quantum theory, but relies on violating a Bell's inequality. For instance, consider the devices of Alice and Bob as black boxes. In this case, assume, in the worse scenario, that an adversary Eve, instead of Alice and Bob, performs measurements as shown in Fig.~\ref{Fig:BellTest}(b). Because the two parties are space-like separated, the probability distribution generated in this way is always within the scope of LHVMs, that is, $p(a,b|x,y) = p(a|x,\lambda)p(b|y,\lambda)$, where $\lambda$ is a hidden variable that is controlled by Eve. Therefore, Eve cannot fake a violation of any Bell tests, which intuitively explains the security of the device independent tasks.

Since the first experiment in the early 1980s \cite{Aspect1982PhysRevLett.49.91}, lots of lab demonstrations of the CHSH inequality have been presented. These experiment results show explicit violations of the LHVMs bound $S_C$, and meanwhile, suffer from a few technical and inherent loopholes, which might invalidate the conclusions. Two well-known technical obstacles are due to the locality loophole and the detection efficiency loophole, which can be closed with more delicately designed experiments and developed instruments
in varieties of experiment systems, including optic systems \cite{Christensen13,giustina2013bell},  superconducting systems \cite{ansmann2009violation}, ionic systems \cite{rowe2001experimental}, and atomic systems \cite{hofmann2012heralded}.
In contrast to the technical loopholes, there also exists an inherent loophole that cannot be closed completely in any Bell test --- the input settings may not be chosen randomly. In the worst case, the inputs can be all predetermined, which makes it possible to violate the Bell inequalities even with LHVMs. In this case, witnessing a violation of a Bell's inequality does not imply the demand for non-LHVM theories and such Bell test cannot be used for the device independent tasks either. On the other hand, without the quantum theory or violation of Bell's inequalities, one cannot get provable randomness. Therefore, the assumption of true input randomness are indispensable in Bell tests because one cannot prove or disprove its existence.

Practically, the case of not fully random input settings corresponds to the scenario where the input settings are partially controlled by an adversary Eve, who wants to convince Alice and Bob a violation of Bell's inequality with classical settings. In this case, Eve is able to \emph{simultaneously} control the input settings and measurement devices, as shown in Fig.~\ref{Fig:BellTest}(c). We model the imperfect randomness by assuming that the input settings $x$ and $y$ are chosen according to some probability distribution $q(x,y|\lambda)$, conditioned on the same local variable $\lambda$ which is available to the adversary Eve. Now, the probability distribution $p(a,b|x,y)$ of LHVMs are defined by
\begin{equation}\label{eq:randomness}
p(a,b|x,y) = \frac{\sum_{\lambda} p_A(a|x,\lambda)p_B(b|y,\lambda)q(x,y|\lambda)q(\lambda)}{q(x,y)},
\end{equation}
where $q(\lambda)$ is the prior  probability distribution of $\lambda$, and $q(x,y)= \sum_\lambda q(x,y|\lambda)q(\lambda)$ is the observed average probability of the input settings $x$ and $y$. Notice that $q(\lambda)$ is normalized by restricting $\sum_\lambda q(\lambda) = 1$.
Now, the CHSH $S$ value  under the classical strategy given in Fig.~\ref{Fig:BellTest}(c)  can be rewritten according to
\begin{equation}\label{Eq:BellFree}
  S = 4\sum_{\lambda}\sum_{a,b,x,y} (-1)^{a\oplus b + x\cdot y}p_A(a|x,\lambda)p_B(b|y,\lambda)q(x,y|\lambda)q(\lambda),
\end{equation}
where we additionally require the observed probability of choosing $x$ and $y$ to be uniform, that is, $q(x,y)=1/4, \forall x, y$.

Notice that, in the extreme (deterministic) case where $q(x,y|\lambda) = 0$ or $1$ for all $x$, $y$, the local hidden variable $\lambda$ deterministically controls the input settings. Then Eve is able to violate Bell tests to an arbitrary value with LHVMs. On the other hand, if Eve has no control of the input settings where $q(x,y|\lambda) = 1/4$ for all $x$, $y$, she cannot fake a violation at all.
Therefore, a meaningful question to ask is how one can assure that a violation of the CHSH inequality is not caused by Eve's attack on imperfect input randomness. That is, we want to know what the requirement of the input randomness is to guarantee that an observed violation truly stems from quantum effects. 


\subsection{Randomness requirement in Bell test}
Let us start with quantifying the input randomness. Here, we make use of the randomness parameter $P$ adopted in Ref.~\cite{Koh12} to fulfill such an attempt, other tools such as the Santha-Vazirani source \cite{santha1986generating} may work similarly. The parameter $P$ is defined to be the maximum probability of choosing the inputs conditioned on the hidden variable $\lambda$,
\begin{equation}\label{eq:randomness}
P = \max_{x,y,\lambda}q(x,y|\lambda).
\end{equation}
With this definition, the larger $P$ is, the less input randomness, the more information about the inputs Eve has, and the easier for her to fake a quantum violation with LHVMs. In the CHSH test, $P$ takes values in the regime of $[1/4,1]$. When $P=1$, it represents the case that Eve has whole information of Alice and Bob's inputs, that is, Eve can always correctly infer the values of $x$ and $y$ by accessing the local hidden variable $\lambda$. When $P = 1/4$, it corresponds to the case of complete randomness, where the adversary have no additional information on the inputs compared to naive guess. Note that the definition of $P$ essentially follows the min-entropy, which is widely used to quantify randomness of a random variable $X$ in information theory, $H_{min}=-\log\left[\max_x prob(X=x)\right]$.

Intuitively, given complete randomness where $P = 1/4$, the value $S$ with LHVMs are bounded by $S_C$ as shown in Eq.~\eqref{eq:Bell}; while given the most dependent (on $\lambda$) randomness where $P=1$, the value $S$ with LHVMs could reach the mathematical maximum, $S_{NS}$ in the CHSH test. Then it is interesting to check the maximal $S$ value for $P\in(1/4,1)$ with LHVMs. In this work, we are interested in when the adversary can fake a quantum violation given certain randomness $P$. We thus exam the lower bound $P_Q$ of $P$ such that the Bell test result can reach the quantum bound $S_{Q}$ with an optimal LHVM. This lower bound $P_Q$ puts a minimal randomness requirement in a Bell test experiment. Only if the freedom of choosing inputs satisfies $P<P_Q$, can one claim that the Bell test is free of the randomness loophole.

Recently, lots of efforts have been spent on investigating such requirement of randomness needed to guarantee the correctness of Bell tests  \cite{Hall10,Barrett11,Hall11,Koh12,Pope13,Thinh13,putz14}. These works analyze under different conditions. One condition is about whether the input settings of the same party are dependent or not  in different runs. We call it \emph{single run}, referring to the case that the input settings of Alice (Bob) are correlated for different runs, and \emph{multiple run} referring to otherwise. The other condition is about whether the random inputs of Alice and Bob are correlated or not. Conditioned on these different assumptions of the input randomness, the lower bound $P_Q$ that allows LHVMs to saturate the quantum bound $S_Q$ in the CHSH Bell test is summarized in Table~\ref{table:Violation}.

\begin{table}[hbt]
\centering
\caption{The lower bound for randomness parameter $P$ defined in Eq~\eqref{eq:randomness} that allows the CHSH value $S$, defined in Eq.~\eqref{Eq:BellFree}, to reach the quantum bound $S_Q$ by LHVMs in the CHSH test under different conditions.}
\begin{tabular}{ccc}
  \hline
  &Correlated inputs& Uncorrelated inputs\\
   \hline
  Single Run&0.285 \cite{Hall10,Koh12}&0.354 \cite{Koh12}\\
  Multiple Run&0.258 \cite{Pope13}&$\leq0.264$ (Our Work)\\
  \hline
\end{tabular}
\label{table:Violation}
\end{table}

In the single run scenario, the optimal strategies for Eve reach $S = 24P-4$ and $S = 8P$ in the case that Alice's and Bob's input settings are correlated and uncorrelated, respectively \cite{Hall10,Koh12}. To achieve the maximum quantum violation $S_Q = 2\sqrt{2}$, the critical randomness requirement is shown in Table~\ref{table:Violation}. It is worth mentioning that if one has randomness $P \geq P_{NS} = 1/3$ and $P\geq P_{NS} = 1/2$ for the case of correlated and uncorrelated inputs, respectively, Eve is able to recover arbitrary NS correlations.

In a more realistic scenario, the multiple run case, the input settings of Alice (Bob) are dependent in different runs. Now, suppose the inputs may correlate for each  $N$ sequent runs, where $N = 1$ stands for the single run case, and $N>1$ for the multiple run case. For each unit of $N$ runs, denote $x_j$ ($y_j$) and $a_j$ ($b_j$) to be the input and output of Alice (Bob) for the $j$th run, where $j = 1, 2, \dots, N$, respectively. In the multiple run scenario,  correlations of the inputs of each $N$ runs can be represented by
\begin{equation}\label{}
q(x_1, x_2, \dots, x_N, y_1, y_2, \dots, y_N|\lambda)¡£
\end{equation}
Therefore, similar to the definition of Eq.~\eqref{Eq:BellFree}, the $S$ value with LHVMs in the multiple run case can be defined by
\begin{equation}\label{Eq:multipleRun}
\begin{aligned}
S = \frac{4}{N}\sum_{j=1}^{N}\sum_\lambda\sum_{a_j, b_j, x_j, y_j} (-1)^{a_j\oplus b_j + x_j\cdot y_j}
p_A(a_j|x_j,\lambda)p_B(b_j|y_j,\lambda)q(\mathbf{x},\mathbf{y}|\lambda)q(\lambda),
\end{aligned}
\end{equation}
where the index $j$ denotes the $j$th run, and $\mathbf{x} = (x_1, x_2, \dots, x_N)$, $\mathbf{y} = (y_1, y_2, \dots, y_N)$. Notice that we only consider the  correlations of inputs in the unit of $N$ runs, which is not the total number of runs in experiment. To get an accurate estimation of the $S$  value defined in Eq.~\eqref{Eq:multipleRun}, one also need to perform the $N$ runs multiple times similar to the single run case.

In the multiple run scenario, as an extension of Eq.~\eqref{eq:randomness}, the input randomness parameter is defined according to
\begin{equation}\label{eq:randomnessM}
P = \left(\max_{\mathbf{x},\mathbf{y},\lambda}q(\mathbf{x},\mathbf{y}|\lambda)\right)^{1/N}.
\end{equation}
It is quite straightforward that the adversary is  easier to fake a violation of a Bell test with LHVMs with increasing number of correlation $N$ of the inputs. This is because the adversary can take advantage of additional dependence of the inputs in different runs. It has been shown that with randomness $P \geq P_{Q} \approx 0.258$, Eve is able to fake the maximum quantum violation $S_Q$ \cite{Pope13} with the number of input correlation $N$ goes to infinity. This result \cite{Pope13} lower bounds $P_Q$ for all finite $N$, and thus puts a very strict requirement on the input randomness  to guarantee a faithful CHSH test.

A meaningful remaining question is thus to consider the multiple run but uncorrelated scenario. As all Bell experiments must run many times to sample the probability distribution, it is reasonable and also practical to consider a joint attack by Eve. On the other hand, the uncorrelated assumption is also reasonable when the inputs of Alice and Bob are independent even conditioned on $\lambda$, that is, $q_A(x|\lambda,y) = q_A(x|\lambda)$ and $q_B(y|\lambda,x) = q_B(y|\lambda)$. Equivalently, the probability of the inputs are required to be factorizable,
\begin{equation}\label{Eq:Uncorrelated}
  q(x,y|\lambda) = q_A(x|\lambda)q_B(y|\lambda).
\end{equation}
This factorizable (uncorrelated) condition constrains the power of Eve in controlling or inferring the inputs of Alice and Bob. A general distribution $q(x,y|\lambda)$  requires Eve to jointly control the instruments that Alice and Bob use to generate random inputs. In the case when the experiment instruments of Alice and Bob are manufactured independently or the inputs are determined by sources causally disconnected from each other, such as cosmic photons \cite{Gallicchio14}, the inputs $x$ and $y$ can be assumed to be independent to each other conditioned on the hidden variable $\lambda$. That is, Eve can only control each of the input settings independently according to Eq.~\eqref{Eq:Uncorrelated}.

In the multiple run and uncorrelated scenario, the $S$ value with LHVMs is defined by
\begin{equation}\label{Eq:BellFinal}
  S = \frac{4}{N}\sum_{j=1}^{N}\sum_\lambda\sum_{a_j, b_j, x_j, y_j} (-1)^{a_j\oplus b_j + x_j\cdot y_j}p_A(a_j|x_j,\lambda)p_B(b_j|y_j,\lambda)q_A(\mathbf{x}|\lambda)q_B(\mathbf{y}|\lambda)q(\lambda).
\end{equation}
Our purpose is to investigate the optimal attack of CHSH test with restricted randomness input $P$.
Therefore we want to maximize Eq.~\eqref{Eq:BellFinal} with the constraint of Eq.~\eqref{eq:randomnessM}. In particular, we are interested to see when this maximal value can reach $S_Q=2\sqrt2$.

\section{Single run case}\label{Sec:single}
We first review the optimal strategy in the single run scenario \cite{Koh12} to get an intuition behind the optimal attack of the adversary. Hereafter, we mainly focus on the scenario that Alice and Bob's inputs are uncorrrelated as defined in Eq.~\eqref{Eq:Uncorrelated}. Thus, what we want is  to maximize  the $S$ value,
\begin{equation}\label{Eq:}
  S = \sum_\lambda q(\lambda) S_\lambda,
\end{equation}
where
\begin{equation}\label{Eq:Slambda}
  S_\lambda = 4\sum_{a,b,x,y} (-1)^{a\oplus b + x\cdot y}p_A(a|x,\lambda)p_B(b|y,\lambda)q_A(x|\lambda)q_B(y|\lambda),
\end{equation}
with restricted randomness $P$, given in Eq.~\eqref{eq:randomness}.

Since any probabilistic LHVM, that is, $p_A(a|x,\lambda)p_B(b|y,\lambda)$, could be realized by a convex combination of deterministic ones \cite{Fine82}, it is therefore sufficient to only consider deterministic LHVMs. Due to the symmetric definition of the CHSH inequality, we only need to consider a specific strategy of $p_A(0|x, \lambda) = p_B(0|y, \lambda) = 1$, and $p_A(1|x, \lambda) = p_B(1|y, \lambda) = 0$ for some given $\lambda$, and all the other ones works similarly. By substituting the special strategy into Eq.~\eqref{Eq:Slambda}, we get
\begin{equation}\label{}
  S_\lambda = 4\left[q_A(0)q_B(0) + q_A(0)q_B(1) + q_A(1)q_B(0)- q_A(1)q_B(1)\right].
\end{equation}
Suppose $P_A = \max_{x,\lambda}\{q_A(x|\lambda)\}$, $P_B = \max_{y,\lambda}\{q_B(x|\lambda)\}$, and hence $P = P_AP_B$, $S_\lambda$ can be maximized to
\begin{equation}\label{}
  S_\lambda \leq  4\left[1 - 2(1-P_A)(1-P_B)\right] = 8(P_A + P_B - P) - 4.
\end{equation}
Given $P$, $S_\lambda$ is upper bounded by
\begin{equation}\label{}
  S_\lambda \leq 8P,
\end{equation}
where the equality holds when $P_B = 1/2$ and $P_A = 2P$. Thus, the optimal strategy with LHVMs is $S= 8P$.
Note that, when the input settings are fully random, $P=1/4$, the optimal strategy of LHVMs  is $S = 2$, which recovers the original LHVMs bound $S_C$. It is easy to see that, to saturate the quantum bound $S_Q = 2\sqrt{2}$, the randomness should be at least $P_Q = S_Q/8 = \sqrt{2}/4\approx0.354$, as shown in Table \ref{table:Violation}.

In the single run case, we only need to consider one specific deterministic strategy of $p(a,b|x,y)$ due to the symmetric definition of the CHSH inequality. We also take advantage of this property in the derivation of the multiple run case. In addition, we can see that the optimal strategy of LHVMs is to choose $x$ or $y$ fully randomly and the other one as biased as possible. This biased optimal strategy is counter-intuitive since the adversary do not need to control the inputs of both parties, but only those of one party.
We show that this counter-intuitive feature does not hold in the optimal strategy in the multiple run case.

\section{Multiple run case}\label{Sec:Result}
Now we consider the multiple run scenario with uncorrelated input randomness. That is, optimizing Eve's LHVM strategy Eq.~\eqref{Eq:BellFinal} with constraints defined in Eq.~\eqref{eq:randomnessM}. Similar to the single run case, from the symmetric argument, we can also solely consider one specific deterministic strategy, that is, $p_A(0|x, \lambda) = p_B(0|y, \lambda) = 1$, and $p_A(1|x, \lambda) = p_B(1|y, \lambda) = 0$.  Given the probabilities of Alice's and Bob's inputs, $q_A(\mathbf{x}|\lambda)$, $q_B(\mathbf{y}|\lambda)$, the $S$ value, defined in Eq.~\eqref{Eq:BellFinal}, for this specific strategy labeled with $\lambda$ is given by
\begin{equation}\label{Eq:CHSHmultiple}
  S_\lambda = 4\left(1 - \frac{2}{N}\sum_{\mathbf{x},\mathbf{y}\in\{0,1\}^N}\mathbf{x}\cdot\mathbf{y}q_A(\mathbf{x}|\lambda)q_B(\mathbf{y}|\lambda)\right),
\end{equation}
where $\cdot$ is vector inner product.
Our attempt is therefore to maximize Eq.~\eqref{Eq:CHSHmultiple} with constraints
\begin{equation}\label{Eq:constraintsx}
q_A(\mathbf{x}|\lambda)q_B(\mathbf{y}|\lambda) \leq P^N,
\end{equation}
for all $q_A(\mathbf{x}|\lambda)$ and $q_B(\mathbf{y}|\lambda)$.

Since in the single run scenario, the optimal strategy requires only one party with biased conditional probability, we first analyze the case with only Alice's inputs biased and Bob's inputs uniformly distributed. Then we investigate the case where the inputs of both parties are biased. We can see that the one party biased strategy is not optimal in the multiple run case, even when $N=2$.

\subsection{One party Biased}
In the case when Eve only (partially) controls one of the inputs, say Alice's, the probability of Alice's input string $q_A(\mathbf{x}|\lambda)$ is biased and Bob's input string is uniformly distributed, that is,
\begin{equation}\label{eq:1biasedqB}
  q_B(\mathbf{y}|\lambda) = \frac{1}{2^N}.
\end{equation}
The randomness is characterized by Eq.~\eqref{eq:randomnessM}, after substituting Eq.~\eqref{eq:1biasedqB}, 
\begin{equation}\label{Eq:Pm0}
  P = \frac{P_A}{2},
\end{equation}
where $P_A$ is defined by $P_A = \max_{\lambda,\mathbf{x}}q_A(\mathbf{x}|\lambda)^{1/N}$. Then, the $S$ value, defined in Eq.~\eqref{Eq:CHSHmultiple}, becomes
\begin{equation}\label{}
  S_\lambda = 4\left(1 - \frac{1}{N2^{N-1}}\sum_{\mathbf{x},\mathbf{y}\in\{0,1\}^N}\mathbf{x}\cdot\mathbf{y}q_A(\mathbf{x}|\lambda)\right).
\end{equation}

Denote the number of bit $1$ in an $N$ string $\mathbf{a}$ as $L_1(\mathbf{a})$. Given the number of bit $1$ in $\mathbf{x}$, $k_A = L_1(\mathbf{x})$,  we can sum over $\mathbf{y}$,
\begin{equation}\label{}
  \sum_{\mathbf{y}\in\{0,1\}^N}\mathbf{x}\cdot\mathbf{y} = \sum_{j = 1}^{k_A} 2^{N-k_A}j{k_A\choose j} = 2^{N-1}k_A,
\end{equation}
and group the summation of $\mathbf{x}$ according to $k_A$,
\begin{equation}\label{}
  S_\lambda = 4\left(1 - \frac{1}{N}\sum_{k_A= 0}^{N}\sum_{L_1(\mathbf{x}) = k_A}q_{A}(\mathbf{x}|\lambda)k_A\right),
\end{equation}
One only need to consider the LHVMs whose probabilities of $q_A(\mathbf{x}|\lambda)$ with the same $k_A$ are the same. Otherwise, we can always take an average of $q_A(\mathbf{x}|\lambda)$ with the same $k_A$ without increasing the randomness parameter $P$. Thus we can rewrite $S_\lambda$ as
\begin{equation}\label{Eq:optBiased}
  S_\lambda = 4\left(1 - \frac{1}{N}\sum_{k_A= 0}^{N}q_{k_A}(\mathbf{x}|\lambda){N\choose k_A}k_A\right),
\end{equation}
with normalization requirement
\begin{equation}\label{Eq:conBiased}
  \sum_{k_A= 0}^{N}q_{k_A}(\mathbf{x}|\lambda){N\choose k_A} = 1,
\end{equation}
and constraints defined in Eq.~\eqref{Eq:constraintsx}.


The optimization of Eq.~\eqref{Eq:optBiased} can be solved efficiently via linear programming. Intuitively, to maximize $S_\lambda$ with given $P$ defined in Eq.~\eqref{Eq:Pm0},  we can simply assign $q_{k_A}(\mathbf{x}|\lambda)$ that has large $k_A$ be 0 and that has small $k_A$  be $(2P)^N$. Suppose there exists an integer $l$ such that $P$ can be written as
\begin{equation}\label{Eq:PP}
  P = \frac{1}{2}\left[\sum_{k_A= 0}^{l}{N\choose k_A}\right]^{-1/N},
\end{equation}
then, Eq.~\eqref{Eq:optBiased} can be rewritten as
\begin{equation}\label{PP0}
  S = 4\left[1 - \frac{1}{N}\sum_{k_A= 0}^{N}\frac{1}{2}\left(\sum_{k_A= 0}^{l}{N\choose k_A}\right)^{-1/N}{N\choose k_A}k_A\right].
\end{equation}
For a general case where an integer $l$ cannot be found satisfying Eq.~\eqref{Eq:PP}, we can first find an integer $l$ such that,
\begin{equation}\label{Eq:PP1}
  \frac{1}{2}\left[\sum_{k_A= 0}^{l + 1}{N\choose k_A}\right]^{-1/N}< P \leq \frac{1}{2}\left[\sum_{k_A= 0}^{l}{N\choose k_A}\right]^{-1/N}.
\end{equation}
Then we can assign $q_{k_A}(\mathbf{x}|\lambda)$ to be
\begin{equation}\label{Eq:PP2}
  q_{k_A}(\mathbf{x}|\lambda) = \left\{
  \begin{array}{cc}
    (2P)^N & k_A\leq l \\
    \frac{\left[1 - \sum_{k_A= 0}^{l}(2P)^N{N\choose k_A}\right]^{-1/N}}{{N\choose l+1}} & k_A = l + 1 \\
    0 & k_A > l+1
  \end{array}
  \right.
\end{equation}

For finite $N$, one can numerically solve the problem according to Eq.~\eqref{Eq:PP2}. As shown in Fig.~\ref{Fig:OneBiased}, the optimal strategy for $N = 1, 10, 100$ are calculated. With increasing $N$, the optimal value $S$ increases and hence a valid Bell test requires a smaller $P$ (more randomness).

\begin{figure}[thb]
  \centering
  \resizebox{12cm}{!}{\includegraphics{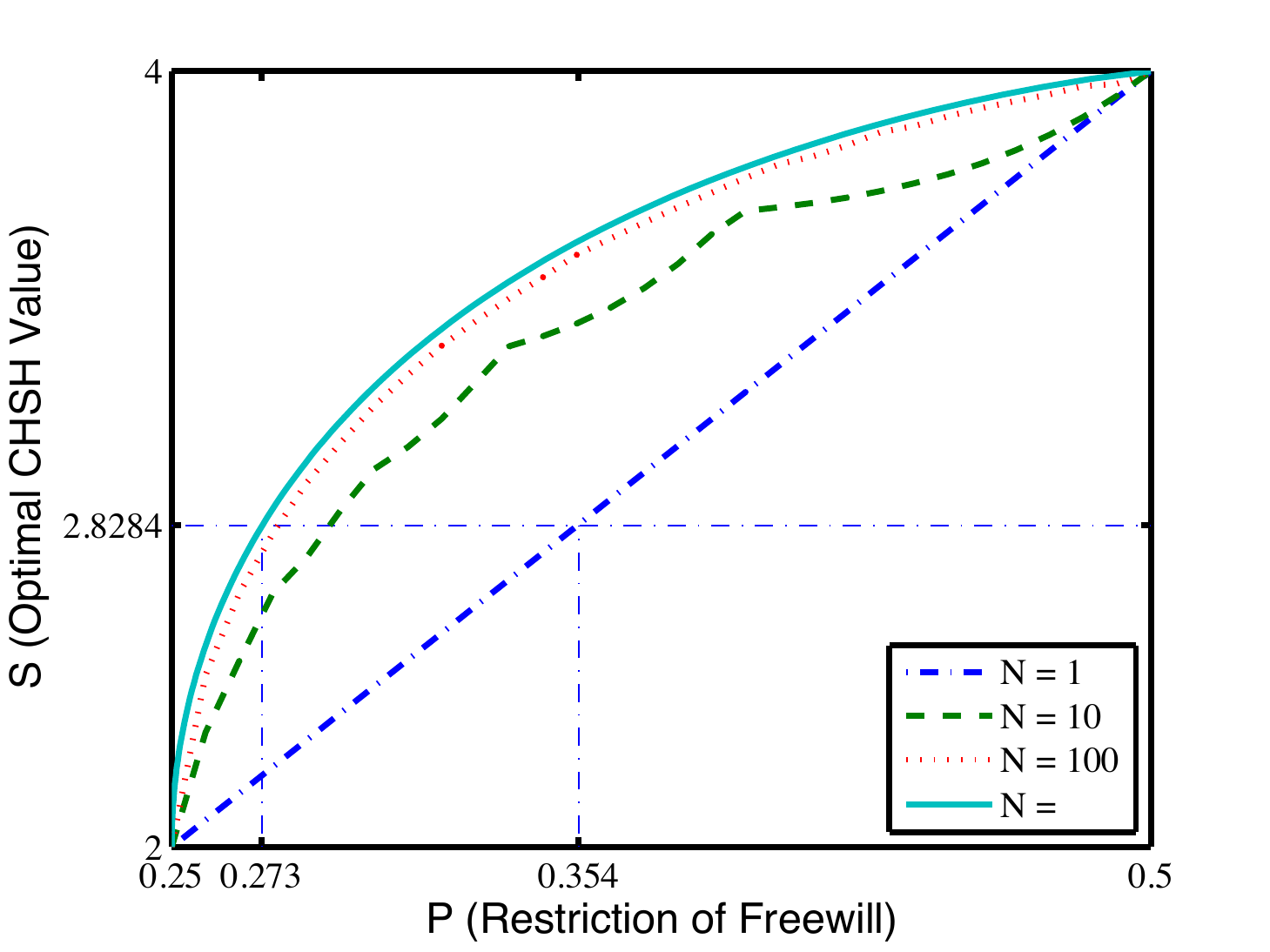}}
  \caption{(Color online) Optimal values of the CHSH test for different randomness $P$ with various rounds $N$ based on only Alice's inputs biased when conditioned on the hidden variable $\lambda$. The solid line is the optimal strategy for $N\rightarrow\infty$, which upper bounds all finite $N$ rounds.  Note that the curve is not smooth for finite runs $N$ because the optimal strategy $q_{k_A}$ defined in Eq.~\eqref{Eq:PP2} jumps on $l$.  With $N$ grows larger, the curve tends to be smoother.  }\label{Fig:OneBiased}
\end{figure}

In the case of $N\rightarrow\infty$, we can derive an analytic bound for all finite $N$ strategies. By following the technique used in Ref.~\cite{Pope13}, we first  estimate $P$ defined in Eq.~\eqref{Eq:PP1} with the limit of $N\rightarrow\infty$ by,
\begin{equation}\label{Eq:PBound}
  \lim_{N\rightarrow\infty}P = \frac{1}{2}\bar{l}^{\bar{l}}(1-\bar{l})^{1-\bar{l}},
\end{equation}
where $\bar{l} = l/N$, and similarly $S$ by,
\begin{equation}\label{Eq:Sinfty}
    \lim_{N\rightarrow\infty}S =  4-4\bar{l}.
\end{equation}
Then we can substitute Eq.~\eqref{Eq:Sinfty} into Eq.~\eqref{Eq:PBound}, and get a relation between optimized  $S$  value and the corresponding randomness parameter $P$,
\begin{equation}\label{Eq:bound}
  P=\frac{1}{2}\left(\frac{4-S}{4}\right)^{(4-S)/4}\left(\frac{S}{4}\right)^{(S/4)}.
\end{equation}

By substituting the quantum bound $S_Q = 2\sqrt{2}$ into Eq.~\eqref{Eq:bound}, we can get the critical randomness requirement to be $P_Q \approx 0.273$. Note that, although Eve only control Alice's input settings, she can still fake a quantum violation with sufficiently low randomness, which is lower than the single run case even when Alice's and Bob's inputs are correlated. Thus we show that the randomness is more demanded for the conditions of multiple/single run compared to the correlation between Alice and Bob.


\subsection{Both parties biased}
Now we consider a general attack, where Eve controls both inputs of Alice and Bob. In this case, we need to optimize Eq.~\eqref{Eq:CHSHmultiple} with constraints defined in Eq.~\eqref{Eq:constraintsx}.
Similarly, we group the summation of $\mathbf{x}$ and $\mathbf{y}$ according to the corresponded number of bit $1$, $k_A = L_1(\mathbf{x})$ and $k_B = L_1(\mathbf{y})$,
\begin{equation}\label{Eq:CHSH2}
  S_\lambda = 4\left(1 - \frac{2}{N}\sum_{k_A,k_B = 0}^{N}\sum_{L_1(x) = k_A}\sum_{L_1(y) = k_B}q_{A}(\mathbf{x}|\lambda)q_{B}(\mathbf{y}|\lambda)\mathbf{x}\cdot\mathbf{y}\right).
\end{equation}
Now, if we assume that $q_{A}(\mathbf{x}|\lambda)$ ($q_{B}(\mathbf{y}|\lambda)$) has the same value for equal $k_A$ ($k_B$), we can sum over $\mathbf{x}$ and $\mathbf{y}$  for given $k_A$ and $k_B$,
\begin{eqnarray}
  \sum_{k_A,k_B}\mathbf{x}\cdot\mathbf{y} &=& {N \choose k_A}\sum_{j = \max\{1, k_A + k_B - N\}}^{\min\{k_A,k_B\}}j{k_A \choose j}{N - k_A \choose k_B - j}\nonumber \\
  &=& {N \choose k_A}k_A{N-1 \choose k_B-1}\nonumber \\
    &=& \frac{k_Ak_B}{N}{N \choose k_A}{N \choose k_B}.
\end{eqnarray}
We can then get the $S$ value,
\begin{equation}\label{Eq:slambda}
  S_\lambda = 4\left(1 - \frac{2}{N^2}\sum_{k_A,k_B = 0}^{N}q_{k_A}(\mathbf{x}|\lambda){N \choose k_A}q_{k_B}(\mathbf{y}|\lambda){N \choose k_B}{k_Ak_B}\right),
\end{equation}
with the constraints of $q_A(\mathbf{x}|\lambda)$ and $q_B(\mathbf{y}|\lambda)$,
\begin{eqnarray}\label{Eq:constrainty}
\sum_{k_A = 1}^{N} q_{k_A}(\mathbf{x}|\lambda){N\choose k_A} &=& 1,\nonumber \\
\sum_{k_B = 1}^{N} q_{k_B}(\mathbf{y}|\lambda){N\choose k_B} &=& 1.
\end{eqnarray}

It is worth mentioning that the assumption that $q_{A}(\mathbf{x}|\lambda)$ ($q_{B}(\mathbf{y}|\lambda)$) takes the same value for equal $k_A$ ($k_B$) is not obviously equivalent to the original optimization problem defined in Eq.~\eqref{Eq:CHSH2}. We thus take this step as an additional assumption, and conjecture it to be true for certain cases.

The problem defined in Eq.~\eqref{Eq:slambda} with constraints of Eq.~\eqref{Eq:constrainty} cannot be solved by linear programming directly, as for the nonlinear terms $q_{k_A}(\mathbf{x}|\lambda)q_{k_B}(\mathbf{y}|\lambda)$.  However, we can still optimize it with similar methods used in the previous section.
 Define the maximum randomness on each side
\begin{eqnarray}\label{}
    P_A &=& [\max_{\lambda,\mathbf{x}}q_{k_A}(\mathbf{x}|\lambda)]^{1/N}, \nonumber\\
    P_B &=& [\max_{\lambda,\mathbf{y}} q_{k_B}(\mathbf{y}|\lambda)]^{1/N}.
\end{eqnarray}
To maximize $S_\lambda$, we can first optimize Alice's side, $q_{k_A}$, and then Bob's side $q_{k_B}$. By doing so, it is not hard to see that $S_\lambda$ is maximized by assigning $q_{k_A}$ that has small number of $k_A$ to be $P_A$ and that has large number of $k_A$ to be 0, and similarly for $q_{k_B}$.
Thus we need to first find $l_A$ and $l_B$ for Alice and Bob, such that,
\begin{eqnarray}\label{Eq:PAB}
  \left[\sum_{k_A= 0}^{l_A + 1}{N\choose k_A}\right]^{-1/N}&< &P_A \leq \left[\sum_{k_A= 0}^{l_A}{N\choose k_A}\right]^{-1/N}\nonumber\\
    \left[\sum_{k_B= 0}^{l_B + 1}{N\choose k_B}\right]^{-1/N}&< &P_B \leq \left[\sum_{k_B= 0}^{l_B}{N\choose k_B}\right]^{-1/N}.
\end{eqnarray}
Then we can assign $q_{k_A}(\mathbf{x}|\lambda)$ and $q_{k_B}(\mathbf{y}|\lambda)$ to be
\begin{eqnarray}\label{Eq:qqqq}
  q_{k_A}(\mathbf{x}|\lambda) &=& \left\{
  \begin{array}{cc}
    (P_A)^N & k_A\leq l_A \\
    \frac{\left[1 - \sum_{k_A= 0}^{l_A}P_A^N{N\choose k_A}\right]^{-1/N}}{{N\choose l_A+1}} & k_A = l_A + 1 \\
    0 & k_A > l_A+1
  \end{array}
  \right.\nonumber\\
      q_{k_B}(\mathbf{y}|\lambda) &=& \left\{
  \begin{array}{cc}
    (P_B)^N & k_B\leq l_B \\
    \frac{\left[1 - \sum_{k_B= 0}^{l_B}P_B^N{N\choose k_B}\right]^{-1/N}}{{N\choose l_B+1}} & k_B = l_B + 1 \\
    0 & k_B > l_B+1
  \end{array}
  \right.
\end{eqnarray}
to optimize $S_\lambda$ defined in Eq.~\eqref{Eq:slambda}.

For finite $N$, we can also numerically solve the optimization problem defined in Eq.~\eqref{Eq:slambda}. As shown in Fig.~\ref{Fig:Biased}. The value $S$ increases with the number of runs $N$, thus the strategy with infinite rounds puts a bound on the strategy with finite rounds.
\begin{figure}[htb]
  \centering
  \resizebox{12cm}{!}{\includegraphics{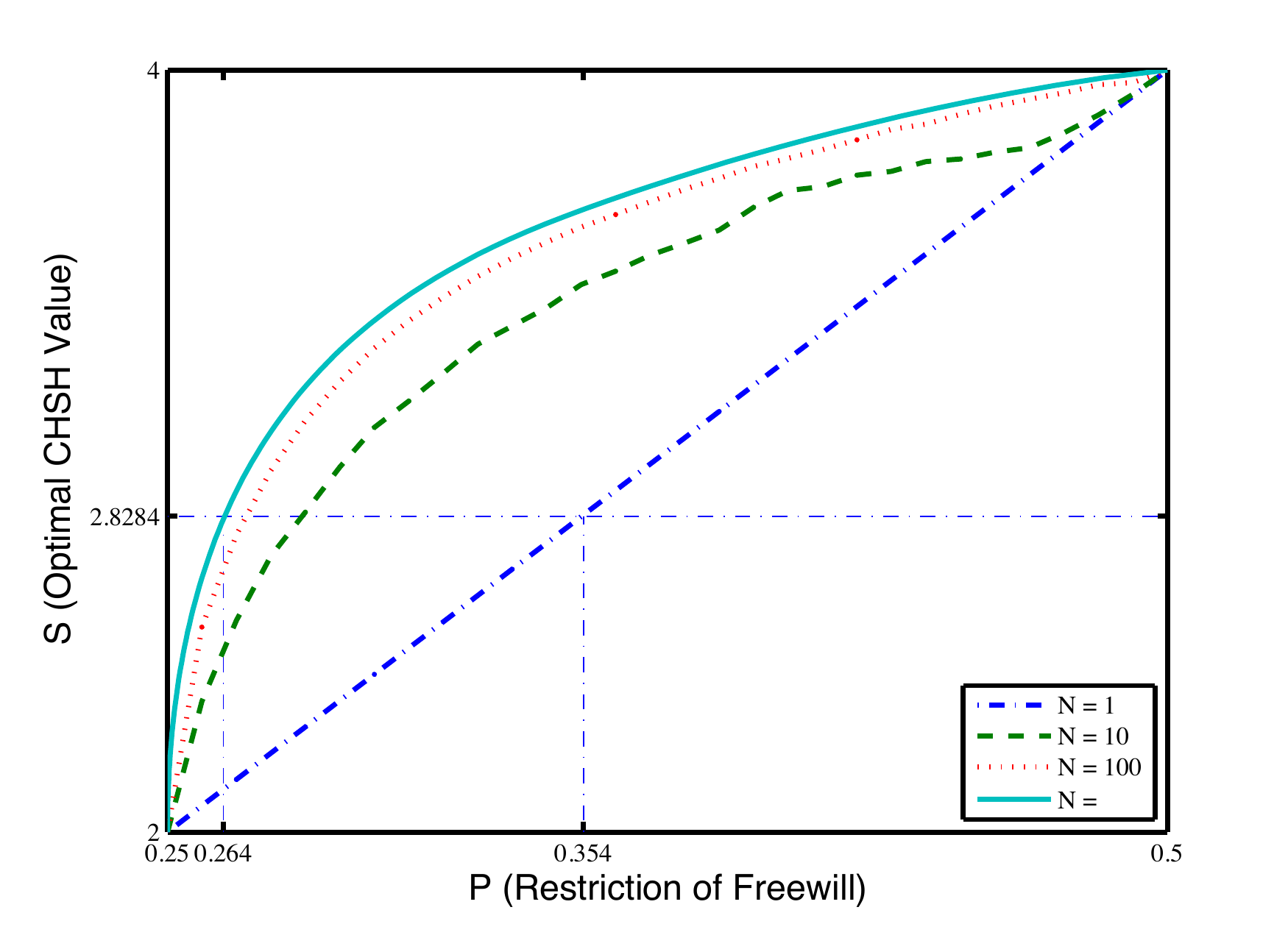}}\\
  \caption{(Color online) Possible optimal values of the CHSH test for different randomness $P$ with various rounds $N$ based on uncorrelated inputs of Alice and Bob. The solid line corresponds the strategy for $N\rightarrow\infty$, which upper bounds all finite $N$ cases. The curves are not smooth for finite $N$ as for similar reasons like in the one party biased case, and it tends to be smooth with $N\rightarrow\infty$.}\label{Fig:Biased}
\end{figure}

In the case of  $N\rightarrow\infty$, we can also find analytical relation between optimized $S$ and the corresponded $P$. Similarly, we first estimate $P_A$ and $P_B$ defined in Eq.~\eqref{Eq:PAB} with the limit of $N\rightarrow\infty$  by
\begin{eqnarray}\label{}
    \lim_{N\rightarrow\infty} P_A &=  & \bar{l}_A^{\bar{l}_A}(1-\bar{l}_A)^{1-\bar{l}_A}, \nonumber\\
    \lim_{N\rightarrow\infty} P_B &=  & \bar{l}_B^{\bar{l}_B}(1-\bar{l}_B)^{1-\bar{l}_B},
\end{eqnarray}
where $\bar{l}_A = l_A/N$ and $\bar{l}_B = l_B/N$, and $S$ according to
\begin{equation}\label{}
    S =  4 - 8\bar{l}_A\bar{l}_B.
\end{equation}

As we still have to optimize over all possible $P_A$ and $P_B$ that satisfies $P_AP_B = P$, we cannot get a direct analytic formula like in Eq.~\eqref{Eq:bound}, while we can still numerically solve  and plot it in Fig.~\ref{Fig:Biased}.
To reach a maximum quantum violation $S_Q = 2\sqrt{2}$ with a LHVM, the randomness is required to be $P\geq P_Q \approx 0.264$, which is larger than the case where Eve only control's Alice's input.

\subsection{Discussion}
We take an additional assumption in the derivation of the both parties biased case, thus the obtained bound $P_Q \approx 0.264$ is still an upper bound of a general optimal attack for the case of $N$ goes to infinity. As we already know, the randomness requirement for the worst case, that is, multiple run with Alice and Bob's inputs correlated, is strictly bounded by $P_Q \approx 0.258$ \cite{Pope13}. Thus, we know that the tight $P_Q$ for the case of multiple run but Alice and Bob uncorrelated should lie in the interval of $[0.258, 0.264]$.

To gain intuition why we take the additional assumption, first notice that what we want is to minimize the average contribution of $\mathbf{x}\cdot \mathbf{y}$ in Eq.~\eqref{Eq:CHSH2}. 
In our case, where $P$ is near 1/4, $q_A(\mathbf{x|}\lambda)$ and $q_B(\mathbf{y|}\lambda)$ can be regarded as an approximately flat distribution.
On average, the $\mathbf{x}$ ($\mathbf{y}$) contains less number of 1s will contribute more to $S$, which means we should assign the corresponded probability $q_A(\mathbf{x|}\lambda)$ ($q_A(\mathbf{y|}\lambda)$) bigger in order to maximize $S$. As $q_A(\mathbf{x|}\lambda)$ ($q_A(\mathbf{y|}\lambda)$) is upper bounded by $P_A$ ($P_B$), an intuitive optimal strategy is then to let $q_A(\mathbf{x|}\lambda)$ ($q_A(\mathbf{y|}\lambda)$) be $P_A$ ($P_B$) for $\mathbf{x}$ ($\mathbf{y}$) contains less number of 1s, and be 0 for the ones contains more number of 1s. As $q_A(\mathbf{x|}\lambda)$ ($q_A(\mathbf{y|}\lambda)$) should also satisfy the normalization condition (Eq.~\eqref{Eq:constrainty}), we can simply follow the strategy defined in Eq.~\eqref{Eq:qqqq} to realize the intuition, which on the other hand satisfies the assumption we take.
Follow the above intuition, we conjecture the assumption to be true for certain cases of $N$.  That is, for finite $N$, we conjecture it to be true when equalities are taken in Eq.~\eqref{Eq:PAB} for both $P_A$ and $P_B$.

On the other hand, we want to emphasize that for a finite $N$, the assumption will not generally hold in the optimal strategy if the equalities in Eq.~\eqref{Eq:PAB} are not fulfilled. For example, if the probability of $l_A+1$ and $l_B+1$  in Eq.~\eqref{Eq:qqqq} is not 0 but very small, we should not take all $q_A(\mathbf{x}|\lambda)$ and $q_B(\mathbf{y}|\lambda)$ equally as $q_{k_A}$ and $q_{k_B}$, especially for the case of
$L_1(\mathbf{x}) = l_A+1$ and $L_1(\mathbf{y}) = l_B+1$ , respectively. In fact, there do exists a cleverer  assignment of $q_A(\mathbf{x}|\lambda)$ and $q_B(\mathbf{y}|\lambda)$ such that only $\mathbf{x}$ and $\mathbf{y}$ that gives small $\mathbf{x}\cdot\mathbf{y}$ get probability instead of all of $\mathbf{x}$ and $\mathbf{y}$ that $L_1(\mathbf{x}) = l_A+1$ and $L_1(\mathbf{B}) = l_B+1$.
However, with increasing runs $N$, this kind of clever attack stops working as for the equalities can be more approximately satisfied with larger $N$. Therefore, we also conjecture the assumption to be true for all possible $P$ with $N$ goes to infinity.

As we can see, our obtained $P_Q\approx 0.264$ is already very close to the worst case value that is $0.258$, we can therefore conclude that the multiple run correlation is already a strong resource for the adversary, no matter whether the inputs of Alice and Bob are correlated or not.
In addition, as we know that the bound $P_Q$ for the most loose case, that is, single run and Alice Bob uncorrelated, is given to be $0.354$ \cite{Koh12}, we also suggest that the key loophole of the input randomness is the correlation between multiple runs instead of correlation of Alice and Bob.


Considering that Eve controls only Alice's input but leaves Bob's input uniformly distributed, we found the least randomness Eve need to control to fake a quantum violation is $P_Q \approx 0.273$. And the least randomness  required when controlling both Alice and Bob is $P_Q \le 0.264$. By comparing the results to the ones listed in Table.~\ref{table:Violation}, we conclude that the key randomness loophole is due to the correlation between multiple runs. As the randomness requirement which considers multiple run attack is not easy to fulfill in real experiments, we thus suggest the experiments to rule correlations of the input settings from different runs. To guarantee the securities of the device independent tasks, we also suggest that one should check whether there are correlations between random inputs from different runs.

For further research, we are interested to know whether there exists Bell inequalities that suffers less from the randomness loophole. By assuming different kinds of assumptions, the randomness requirement behaves different. For example, it is interesting to investigate the scenario where the input settings are uncorrelated with the measurement devices by assuming the manufactures are different. That is, there are two uncorrelated hidden variables in Fig.~\ref{Fig:BellTest}(c), controlling the input settings and measurement devices independently. Moreover, recently, by considering a nonzero lower bound for the input random probability $p(x,y|\lambda)$, $\mathrm{P\ddot{u}tz}$ et al. show a Bell inequality which suffers from very little randomness loophole \cite{putz14}. That is, no adversary can fake a quantum violation as long as the lower bound of $p(x,y|\lambda)$ is nonzero regardless of its upper bound $P$ defined in Eq.~\eqref{eq:randomness}. Therefore, it is interesting to investigate the multiple run randomness requirement of the CHSH inequality with additional assumptions.

\chapter{Clauser-Horne Bell test with imperfect random inputs}
This chapter investigates general randomness requirement for the Clauser-Horne (CH) inequality \cite{Yuan15CH}. We consider for different conditions. In addition, our method applies for general Bell inequalities.

\section{General randomness requirement}
In the bipartite scenario, a general Bell test involves two remotely separated parties, Alice and Bob, who receive random inputs $x$ and $y$ and produce outputs $a$ and $b$, respectively. Based on the probability distribution $\tilde{p}_{AB}(a,b|x,y)$ of the outputs conditioned on the inputs, Bell's inequality can be defined by a linear combination of $\tilde{p}_{AB}(a,b|x,y)$ according to
\begin{equation}\label{eq:Bell}
J = \sum_{a,b,x,y} \beta_{a,b}^{x,y}\tilde{p}_{AB}(a,b|x,y) \leq J_C,
\end{equation}
where $J_C$ is a bound for all  local hidden variable models (LHVMs), meaning that, any LHVM cannot violate any Bell's inequality. Now, we consider the case that the inputs are not fully random. That is, the inputs $x$ and $y$ depend on some local hidden variable, denoted as $\lambda$, as shown in Fig.~\ref{Fig:BellTest}.

\begin{figure}[thb]
\centering
\resizebox{4cm}{!}{\includegraphics[scale=1]{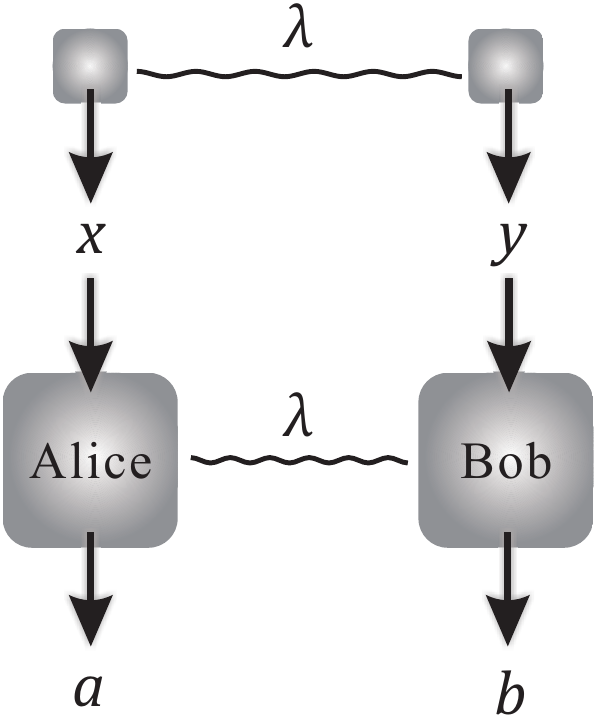}}
\caption{Bell tests in a bipartite scenario. In general, the inputs depend on some local hidden variable $\lambda$. The local hidden variables that control the inputs and the devices may be different. While, we can still denote these two local hidden variables with a single one denoted as $\lambda$.}\label{Fig:BellTest}
\end{figure}

The input randomness can be quantified by the dependence of the inputs conditioned on $\lambda$. Suppose the inputs $x$ and $y$ are chosen according to \emph{a priori} probability $p(x,y|\lambda)$, the input randomness can be measured by its upper and lower bounds,
\begin{equation}\label{eq:randomness}
\begin{aligned}
P = \max_{x,y,\lambda}p(x,y|\lambda),\\
Q = \min_{x,y,\lambda}p(x,y|\lambda).
\end{aligned}
\end{equation}
As an example, for the CH test, where the inputs are binary, the upper and lower bounds are in the range of $[1/4, 1]$ and $[0, 1/4]$, respectively. Focusing on the upper bound $P$, when it equals $1$, it represent the case that the local hidden variable $\lambda$  deterministically decides at least one input. When $P = 1/4$, this corresponds to the case that the inputs are fully random. Similarly, we can see how the lower bound $Q$ characterizes the input randomness. In many of previous works \cite{Hall10,Hall11,Koh12,Pope13, Thinh13, Yuan15CHSH}, only the upper bound $P$ is considered. It is recently noted in Ref.~\cite{putz14} that the lower bound $Q$ also plays an important role in analysis. We thus consider both the upper and lower bounds as quantifications of the input randomness.


With binary inputs, we can consider a symmetric case where $P = 1/4 + \delta$ and $Q = 1/4 - \delta$. In other words, we can quantify the input randomness by its deviation from  a unform distribution, quantified by $\delta$,
\begin{equation}\label{Eq:deltasource}
  \delta = \max_{x,y,\lambda} \left|p(x,y|\lambda) - \frac{1}{4}\right|.
\end{equation}
Note that all our following results apply for asymmetric cases (with arbitrary $P$ and $Q$) as well.

When the input settings are determined by $p(x,y|\lambda)$, the observed probability $\tilde{p}_{AB}(a,b|x,y)$ of outputs conditioned on inputs is given by
\begin{equation}\label{Eq:pabxy}
   \tilde{p}_{AB}(a,b|x,y) = \frac{\sum_\lambda \tilde{p}_{AB}(a,b|x,y,\lambda)p(x,y|\lambda)q(\lambda)}{p(x,y)},
\end{equation}
where $q(\lambda)$ is the priori probability of $\lambda$, $p(x,y) = \sum_\lambda p(x,y|\lambda)q(\lambda)$ is the averaged probability of choosing $x$ and $y$, and $\tilde{p}_{AB}(a,b|x,y,\lambda)$ is the strategy of Alice and Bob conditioned on $\lambda$. Then, the Bell's inequality defined in Eq.~\eqref{eq:Bell} should be rephrased by
\begin{equation}\label{Eq:BellFree}
\begin{aligned}
J &= \sum_{x,y} \frac{1}{p(x,y)}\sum_\lambda\sum_{a,b} \beta_{a,b}^{x,y}\tilde{p}_{AB}(a,b|x,y,\lambda)p(x,y|\lambda)q(\lambda) \\
  &\leq J_C.
\end{aligned}
\end{equation}

In this work, we are interested in how LHVMs can fake a violation of Bell's inequality with imperfect input randomness. Thus, we can also set the strategy $\tilde{p}_{AB}(a,b|x,y,\lambda)$ of deciding the outputs based on the inputs by $\tilde{p}_A(a|x,\lambda)\tilde{p}_B(b|y,\lambda)$, and the Bell value with a LHVM is given by
\begin{equation}\label{Eq:Bellvalue}
\begin{aligned}
  J^{\mathrm{LHVM}} = \frac{1}{p(x,y)}\sum_\lambda\sum_{a,b,x,y} \beta_{a,b}^{x,y}\tilde{p}_A(a|x,\lambda)\tilde{p}_B(b|y,\lambda)p(x,y|\lambda)q(\lambda).
\end{aligned}
\end{equation}
Here, for simplicity, we assume that $p(x,y)$ is independent of $x$ and $y$.
What we are interested is to maximize $J^{\mathrm{LHVM}}(P, Q)$ with LHVMs. From another point of view, we want to establish the Bell's inequality when imperfectly random inputs are considered. Any breach of these bounds (using quantum settings) would rule out LHVMs and in favor of quantum mechanics. Suppose the quantum bound to Eq.~\eqref{Eq:BellFree} is denoted by $J_Q$, then we are especially interested to see the condition of $P$ and $Q$ such that $J^{\mathrm{LHVM}}(P, Q) < J_Q$. In experiment, such condition is the necessary condition for a valid Bell test. For a specific observed violation $J_{\mathrm{obs}}$ and input randomness characteristics $P$ and $Q$, it witnesses non-local feature only if the Bell value satisfies $J^{\mathrm{LHVM}}(P, Q) < J_{\mathrm{obs}}$.

In varieties of previous works \cite{Hall10,Barrett11,Hall11,Koh12,Pope13,Thinh13, Yuan15CHSH}, randomness requirements for the CHSH inequality are analyzed. In this work, we focus on another inequality --- the CH inequality and consider in general scenarios. For instance, many previous work \cite{Hall10,Hall11,Koh12,Pope13, Yuan15CHSH} assumes the underlying probability distribution $\tilde{p}_{AB}(a,b|x,y)$ to satisfy the no-signaling (NS) \cite{prbox} condition. However, in real experiment, the probability distribution $\tilde{p}_{AB}(a,b|x,y)$ may behave signaling due to statistical fluctuation, devices imperfection, or other possible interventions by the adversary Eve. We thus also consider the general case where $\tilde{p}_{AB}(a,b|x,y)$ can be signaling \footnote{It is worth to remark that even if $\tilde{p}_{AB}(a,b|x,y)$ can be signaling, we still assume that the locality loophole is closed. The possibility of signaling comes from the fact that partial knowledge of the inputs are known. In practice, Alice and Bob cannot transmit signal with such signaling probability distribution.}. In addition, we consider the case that the random inputs of Alice and Bob are factorizable. In this case, the input randomness can be written as
\begin{equation}\label{Eq:Uncorrelated}
  p(x,y|\lambda) = p_A(x|\lambda)p_B(y|\lambda).
\end{equation}
This factorizable assumption is reasonable in some practical scenarios, where the experiment devices that determine the input settings are from independent manufactures or the randomness generation events are also spacelikely separated. For example, if the inputs are determined by cosmic photons that are causally disconnected from each other \cite{Gallicchio14}, the input randomness can be reasonably assumed to be factorizable.

\section{CH inequality}\label{Sec:CH}
In this section, we will investigate the randomness requirement of the CH inequality under different conditions, including whether $\tilde{p}_{AB}(a,b|x,y)$ is signaling or NS, and whether the factorizable condition is satisfied or not.
\subsection{CH inequality with LHVMs}
The CH inequality is defined in the bipartite scenario, where the input settings $x$ and $y$ and the outputs $a$ and $b$ are all bits. Based on the probability distribution that obtains a specific measurement outcome, for instance $00$, the CH inequality is defined according to
\begin{equation}\label{eq:CH}
\begin{aligned}
  J_{\mathrm{CH}}  &= \tilde{p}_{AB}(0,0) + \tilde{p}_{AB}(0,1) +\tilde{p}_{AB}(1,0) \\
  &-\tilde{p}_{AB}(1,1) -\tilde{p}_{A}(0) -\tilde{p}_{B}(0)\leq 0,
\end{aligned}
\end{equation}
where we omit the outputs $a$ and $b$ and define $\tilde{p}_{A}(x)$ ($\tilde{p}_B(y)$) to be the probability of detecting $0$ with input setting $x$ ($y$) by Alice (Bob), and $\tilde{p}_{AB}(x,y)$ the probability of coincidence detection $00$ for both sides with input settings $x$ and $y$ for Alice and Bob, respectively.
To satisfy the general definition of a Bell inequality as shown in Eq.~\eqref{eq:Bell}, the single party probabilities $\tilde{p}_{A}(0)$ and  $\tilde{p}_{B}(0)$ need to be properly defined by coincidence detection probabilities. For instance, we can either define $\tilde{p}_{A}(0)$ by the detection probabilities with input $(x = 0, y = 0)$, or $(x = 0, y = 1)$, or a convex mixture. This arbitrary definition vanishes when the NS condition is satisfied.

In experiment realization, one has to run the CH test multiple times, for instance, $N$, to determine the probabilities in Eq.~\eqref{eq:CH}. Denote the coincidence counts by $C_{AB}$  and single counts by $S_{A(B)}$, we can then write
\begin{equation}\label{eq:CH2}
\begin{aligned}
  J_{\mathrm{CH}}  &= \frac{C_{AB}(0,0)}{N_{AB}(0,0)} + \frac{C_{AB}(0,1)}{N_{AB}(0,1)} +\frac{C_{AB}(1,0)}{N_{AB}(1,0)}\\
   &-\frac{C_{AB}(1,1)}{N_{AB}(1,1)} -\frac{S_{A}(0)}{N_{A}(0)} -\frac{S_{B}(0)}{N_{B}(0)}.
\end{aligned}
\end{equation}
Here, $N_{AB}(x,y)$ denotes the total number of trials with input setting $x$ and $y$, and $N_{A(B)}$ the number of trials with input setting $x$ ($y$) of Alice (Bob).

When the input settings are chosen truly randomly,  the CH Bell value $J^{\mathrm{LHVM}}_{\mathrm{CH}}$ with LHVM is always non-positive. While quantum theory could maximally violate it to be $J_Q= (\sqrt{2} - 1)/2\approx0.207$.
If the measurement settings $x$, $y$ are additionally determined by some hidden variable $\lambda$ by probability distribution $p(x, y|\lambda)$, we show in the following that the CH inequality could be violated even with LHVMs.

With a general LHVM strategy defined in Eq.~\eqref{Eq:Bellvalue}, each term in the CH value in Eq.~\eqref{eq:CH2} can be described by
\begin{equation}\label{}
\begin{aligned}
C_{AB}(x,y)&=N\sum_\lambda \tilde{p}_{A}(x,\lambda)\tilde{p}_{B}(y,\lambda)p(x,y|\lambda)q(\lambda)\\
N_{AB}(x,y)&=N\sum_\lambda p(x,y|\lambda)q(\lambda)\\
S_{A(B)}(0)&=N\sum_\lambda \tilde{p}_{A(B)}(0,\lambda)\left(p(0,0|\lambda)+p(0,1|\lambda)\right)q(\lambda)\\
N_{A(B)}(0)&=N\sum_\lambda \left(p(0,0|\lambda)+p(0,1|\lambda)\right)q(\lambda)\\
\end{aligned}
\end{equation}
Here, we adopt a specific realization of the single counts by taking an average of the observed value. For instance, the single detection probability $p_{A}(0)$ is defined to be a mean of the single detection probabilities with input $(x = 0, y = 0)$ and $(x = 0, y = 1)$.

Besides, in order to convince Alice and Bob that the input settings $x$ and $y$ are chosen freely, Eve has to impose that the averaged probability distributions of the input settings are uniformly random. Then, we can assume $p(x,y)$ to be $1/4$,
\begin{equation}\label{eq:CHReq}
\begin{aligned}
  N_{AB}(x,y) = N\sum_\lambda p(x,y|\lambda)q(\lambda) =  N/4, \forall x, y.
\end{aligned}
\end{equation}
In real experiments, the input probability can be arbitrary, where our result can still apply with certain modifications on normalization. With the normalization condition Eq.~\eqref{eq:CHReq}, the CH value with LHVMs strategies is given by
\begin{equation}\label{Eq:CHLHVM}
  J^{\mathrm{LHVM}}_{\mathrm{CH}} = 4\sum_{\lambda} q(\lambda) J_\lambda
\end{equation}
with $J_\lambda$  defined by
\begin{equation}\label{}
\begin{aligned}\label{}
  J_\lambda &= \tilde{p}_{A}(0,\lambda)\tilde{p}_{B}(0,\lambda)p(0,0|\lambda) + \tilde{p}_{A}(0,\lambda)\tilde{p}_{B}(1,\lambda)p(0,1|\lambda) \\ &+\tilde{p}_{A}(1,\lambda)\tilde{p}_{B}(0,\lambda)p(1,0|\lambda)
-\tilde{p}_{A}(1,\lambda)\tilde{p}_{B}(1,\lambda)p(1,1|\lambda) \\ &-\tilde{p}_{A}(0,\lambda)(p(0,0|\lambda)+p(0,1|\lambda))/2 \\ &-\tilde{p}_{B}(0,\lambda)(p(0,0|\lambda)+p(1,0|\lambda))/2.
\end{aligned}
\end{equation}
With the randomness parameter defined in Eq.~\eqref{eq:randomness}, our target is to maximize $J^{\mathrm{LHVM}}_{\mathrm{CH}}$ defined in Eq.~\eqref{Eq:CHLHVM} for given randomness input $P$ and $Q$ under constraints in Eq.~\eqref{eq:CHReq}.

\subsection{General strategy (attack)}
In this part, we consider a general strategy (attack) where no additional assumption is imposed. It is worth mentioning that with the following method, we can essentially convert the optimization problem over all LHVMs into a clearly defined mathematical problem. In the CH example, we show an explicit solution to this mathematical problem. A general solution to this type of mathematical problems will provide a general solution to the problem of imperfect randomness in Bell test.

Note that the optimization of Eq.~\eqref{Eq:CHLHVM} requires to optimize over the strategy of Alice and Bob, $\tilde{p}_A(x,\lambda)$ and $\tilde{p}_B(y,\lambda)$, and also the strategy of deciding the inputs, $p(x,y|\lambda)$, which also satisfies the constraints defined in Eq.~\eqref{eq:CHReq}. Here, we first analyze how to optimize the strategy of Alice and Bob.

Because all probabilistic LHVM strategies can be realized with a convex combination of deterministic strategies, it is sufficient to just consider deterministic strategies, i.e., $\tilde{p}_A(x),\tilde{p}_B(y)\in\{0,1\}$ for the optimization. Conditioned on different values of $\tilde{p}_A(x)$ and $\tilde{p}_B(y)$, 16 possible values of $J_\lambda$ are listed in Table~\ref{table:CHattack1}, where we omit the $\lambda$ for simple notation hereafter.
\begin{table*}[hbt]\label{table:CHattack1}\tiny
\centering
\caption{The value of $J_\lambda$ with deterministic strategy.}
\begin{tabular}{cccccc}
  \hline
  &&\multicolumn{4}{c}{$(\tilde{p}_B(0),\tilde{p}_B(1))$}\\
   &&$(0,0)$&$(0,1)$&$(1,0)$&$(1,1)$\\
   \hline
     \multirow{4}{*}{$(\tilde{p}_A(0),\tilde{p}_A(1))$}&$(0,0)$&$0$&$0$&$-(p(0,0)+p(1,0))/2$&$-(p(0,0)+p(1,0))/2$\\
    &$(0,1)$&$0$&$-p(1,1)$&$(p(1,0)-p(0,0))/2$&$(p(1,0)-p(0,0))/2-p(1,1)$\\
    &$(1,0)$&$-(p(0,0)+p(0,1))/2$&$(p(0,1)-p(0,0))/2$&$-(p(0,1)+p(1,0))/2$&$(p(0,1)-p(1,0))/2$\\
    &$(1,1)$&$-(p(0,0)+p(0,1))/2$&$(p(0,1)-p(0,0))/2-p(1,1)$&$(p(1,0)-p(0,1))/2$&$(p(1,0)+p(0,1))/2-p(1,1)$\\
  \hline
\end{tabular}
\label{table:CHattack1}
\end{table*}
Note that, for given $p(x,y|\lambda)$, we should choose the optimal strategy of $\tilde{p}_A(x)$ and $\tilde{p}_B(y)$ that maximize $J_\lambda$. Thus we here only consider the possible optimal strategies as listed in Table~\ref{table:CHattack2}. We refer to Appendix~\ref{App:proof1} for rigorous proof of why we only consider the possible optimal strategies.

\begin{table}[hbt]\label{table:CHattack2}
\centering
\caption{Possible strategies for letting $J_\lambda$ be positive. }
\begin{tabular}{cc}
  \hline
  $(\tilde{p}_A(0),\tilde{p}_A(1),\tilde{p}_B(0),\tilde{p}_B(1))$&$J_\lambda$\\
  \hline
(0,1,1,0)&$(p(1,0)-p(0,0))/2$\\
(0,1,1,1)&$(p(1,0)-p(0,0))/2-p(1,1)$\\
(1,0,0,1)&$(p(0,1)-p(0,0))/2$\\
(1,0,1,1)&$(p(0,1)-p(1,0))/2$\\
(1,1,0,1)&$(p(0,1)-p(0,0))/2-p(1,1)$\\
(1,1,1,0)&$(p(1,0)-p(0,1))/2$\\
(1,1,1,1)&$(p(1,0)+p(0,1))/2-p(1,1)$\\
  \hline

\end{tabular}
\label{table:CHattack2}
\end{table}

As the strategies of $(\tilde{p}_A(0), \tilde{p}_A(1),\tilde{p}_B(0),\tilde{p}_B(1)) = (0,1,1,0)$ and $(\tilde{p}_A(0), \tilde{p}_A(1),\tilde{p}_B(0),\tilde{p}_B(1)) = (1,0,0,1)$ are always better than the strategies of $(\tilde{p}_A(0), \tilde{p}_A(1),\tilde{p}_B(0),\tilde{p}_B(1)) = (0,1,1,1)$ and $(\tilde{p}_A(0), \tilde{p}_A(1),\tilde{p}_B(0),\tilde{p}_B(1)) = (1,1,0,1)$, respectively, we can always replace the later strategies with the former ones without affecting $p(x,y)$ but achieving a larger $J_\lambda$.
For simple notation, we denote $p(i,j)$ by $p_{2*i + j}$ hereafter, thus the possible deterministic strategies for $J_\lambda$ are in the following set
\begin{equation}\label{eq:chq}
\begin{aligned}
  \left\{\frac{p_2-p_0}{2}, \frac{p_1-p_0}{2}, \frac{p_1-p_2}{2}, \frac{p_2-p_1}{2}, \frac{p_2+p_1}{2}-p_3\right\}.
\end{aligned}
\end{equation}

Because there are only five possible strategies of Alice and Bob, we can also consider that there are only five different strategies of choosing the input settings. The intuition is that, for the input settings that using the same strategies of Alice and Bob, for instance, $J_\lambda = (p_2 - p_0)/2$, we can always take an average of the different strategies of $p(x,y|\lambda)$ without decreasing $J_\lambda$. We refer to Appendix~\ref{App:proof1} for a rigorous proof.
Therefore, we label $\lambda_j$ to be the $j$th strategy of choosing the input settings and $J^{\mathrm{LHVM}}_{\mathrm{CH}}$ can be rewritten in the following way,
\begin{equation}\label{eq:Proof}
\begin{aligned}
  &J^{\mathrm{LHVM}}_{\mathrm{CH}}/4 \\
  &= q(\lambda_1)(p_2(\lambda_1)-p_0(\lambda_1))/2 + q(\lambda_2)(p_1(\lambda_2)-p_0(\lambda_2))/2\\
  &+q(\lambda_3)(p_1(\lambda_3)-p_2(\lambda_3))/2+q(\lambda_4)(p_2(\lambda_4)-p_1(\lambda_4))/2\\
  &+q(\lambda_5)[(p_2(\lambda_5)+p_1(\lambda_5))/2-p_3(\lambda_5)].
\end{aligned}
\end{equation}
The constraints of $q(\lambda)$ and $p(\lambda)$ are given by
\begin{equation}\label{eq:constraints}
\begin{aligned}
  &\sum_j q(\lambda_j) p_i(\lambda_j) = 1/4, \forall i \\
  &\sum_i p_i(\lambda_j) = 1, \forall j, \\
  &\sum_j q(\lambda_j) = 1,\\
   &Q\le p_i(\lambda_j)\le P, \forall i,j.\\
\end{aligned}
\end{equation}

Furthermore, we can denote the coefficient of $q(\lambda_j)p_i(\lambda_j)$ by $\beta_{ij}$ as shown in Table~\ref{Table:coe}. Then $J^{\mathrm{LHVM}}_{\mathrm{CH}}$ can be expressed by
\begin{equation}\label{Eq:mathJlambda}
  J^{\mathrm{LHVM}}_{\mathrm{CH}} = 4\sum_{ij} \beta_{ij} q(\lambda_j)p_i(\lambda_j),
\end{equation}

\begin{table}[hbt]
\centering
\caption{The coefficient $\beta_{ij}$ of $q(\lambda_j)p_i(\lambda_j)$ in the expression of $J^{\mathrm{LHVM}}_{\mathrm{CH}}$ of the CH inequality.}
\begin{tabular}{cccccc}
  \hline
   &$\lambda_1$&$\lambda_2$&$\lambda_3$&$\lambda_4$&$\lambda_5$\\
   \hline
   $p_0$&$-\frac{1}{2}$&$\frac{1}{2}$&$0$&$0$&$0$\\
   $p_1$&$0$&$\frac{1}{2}$&$\frac{1}{2}$&$-\frac{1}{2}$&$\frac{1}{2}$\\
   $p_2$&$\frac{1}{2}$&$0$&$-\frac{1}{2}$&$\frac{1}{2}$&$\frac{1}{2}$\\
   $p_3$&$0$&$0$&$0$&$0$&$-1$\\
  \hline
\end{tabular}\label{Table:coe}
\end{table}

The solution to this optimization problem is shown in Appendix~\ref{app:CH1}. Based on the value of $P$ and $Q$, we give the optimal CH value $J^{\mathrm{LHVM}}_{\mathrm{CH}}$ with LHVMs by
\begin{equation}\label{eq:JCH}
  J^{\mathrm{LHVM}}_{\mathrm{CH}}(P,Q) =
  \left\{
  \begin{array}{cc}
    \frac{5}{2}(4P-1) & 3P+Q\le 1,\\
    1-4Q & 2P+Q\ge \frac{3}{4},\\
     4P-2Q-\frac{1}{2} & \mathrm{else},
  \end{array}
  \right.
\end{equation}
and plot it in Fig.~\ref{Fig:plot0}. Note that when $P$ is greater than ${3}/{8}$, the value of $J^{\mathrm{LHVM}}_{\mathrm{CH}}$ is independent of $P$. Hence, we only plot the situation where $P$ is less than 3/8.
\begin{figure}[!ht]
    \centering
 \resizebox{9cm}{!}{\includegraphics[scale=1]{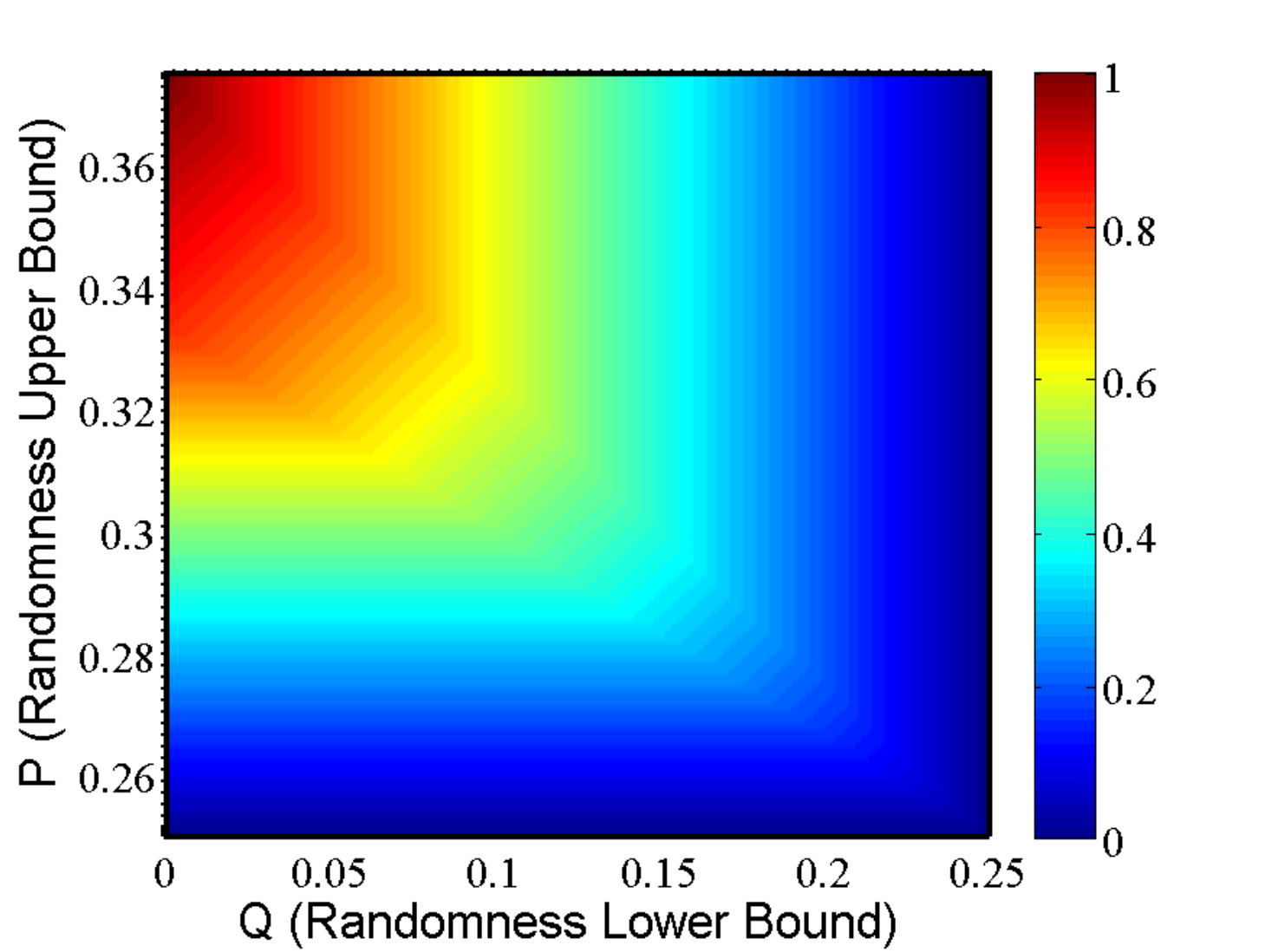}}
\caption{(Color online) The CH value $J^{\mathrm{LHVM}}_{\mathrm{CH}}$ as a function of $P$ and $Q$, according to Eq.~\eqref{eq:JCH}. }\label{Fig:plot0}
\end{figure}

In addition, we can also investigate the optimal CH value $J^{\mathrm{LHVM}}_{\mathrm{CH}}$ with input randomness quantified as in Eq.~\eqref{Eq:deltasource}. It is easy to check that $2P + Q\geq 3/4$, and the optimal CH value $J^{\mathrm{LHVM}}_{\mathrm{CH}}$ is thus
\begin{equation}\label{}
  J^{\mathrm{LHVM}}_{\mathrm{CH}}(\delta) = 4\delta.
\end{equation}

\subsection{Result}

\subsubsection{factorizable condition}
Here, we consider the optimal LHVMs strategy in the case where the probability of the input settings are  factorizable, as defined in Eq.~\eqref{Eq:Uncorrelated}.

Following a similar derivation, we show in the Appendix~\ref{app:fac} that the optimal CH value $J^{\mathrm{LHVM, Fac}}_{\mathrm{CH}}$ with LHVMs under factorizable condition is
\begin{equation}\label{}
  J^{\mathrm{LHVM, Fac}}_{\mathrm{CH}}(P,Q) =
  \left\{
  \begin{array}{cc}
    (4P-1) & P+Q\le \frac{1}{2},\\
    1-4Q & P+Q> \frac{1}{2}.\\
  \end{array}
  \right.
\end{equation}
We show the optimal value of  $J^{\mathrm{LHVM, Fac}}_{\mathrm{CH}}$  in Fig.~\ref{Fig:plot1}.
\begin{figure}[!ht]
    \centering
 \resizebox{9cm}{!}{\includegraphics[scale=1]{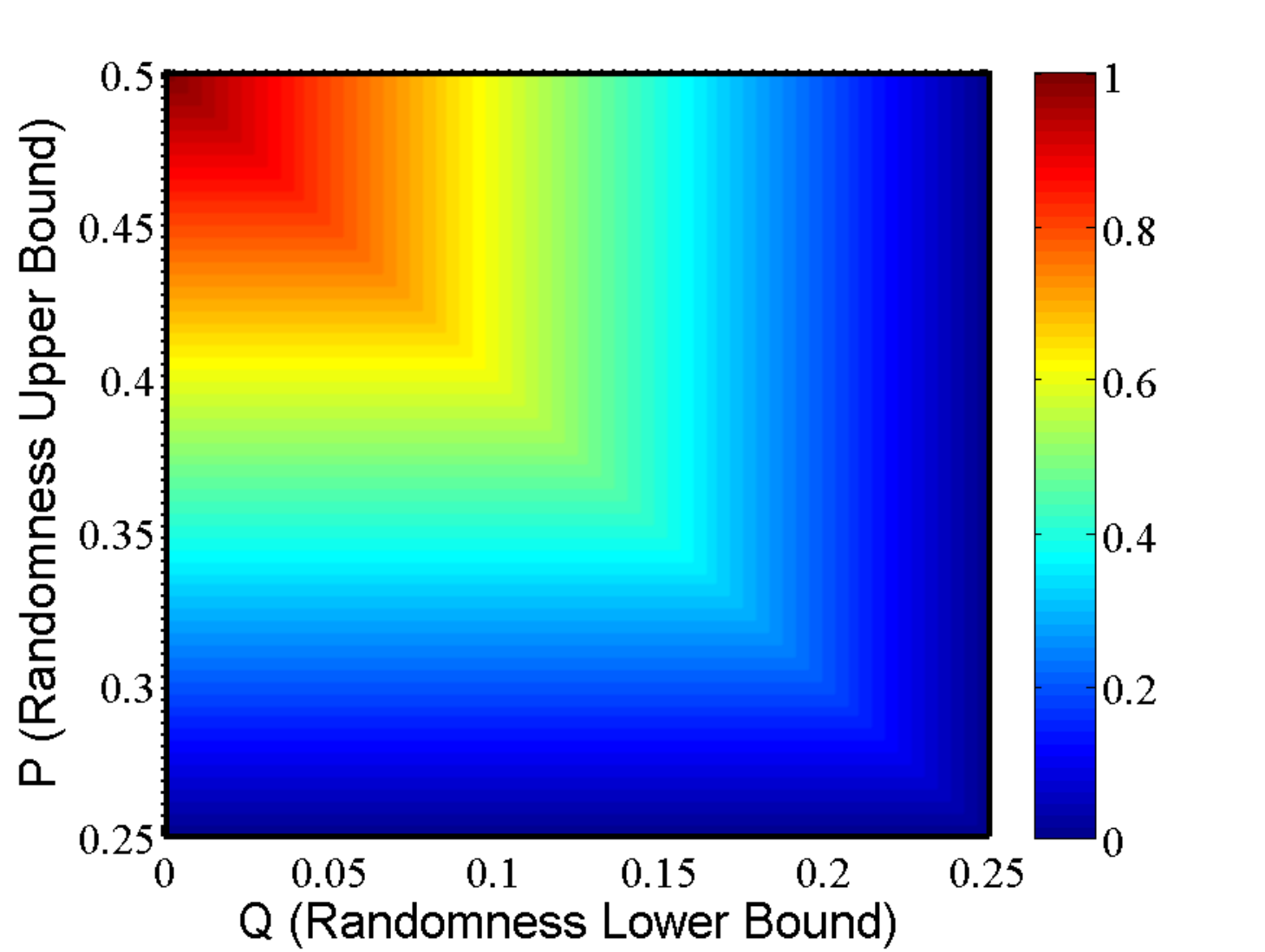}}
\caption{(Color online) The CH value $J^{\mathrm{LHVM, Fac}}_{\mathrm{CH}}$ as a function of $P$ and $Q$ with the factorizable condition Eq.~\eqref{Eq:Uncorrelated}. }\label{Fig:plot1}
\end{figure}

When we quantify $P$ and $Q$ by $P = 1/4 + \delta$ and $Q = 1/4 - \delta$, the fomular can be rewritten by
\begin{equation}\label{}
  J^{\mathrm{LHVM, Fac}}_{\mathrm{CH}}(\delta) = 4\delta.
\end{equation}

It is interesting to note that the factorizable condition does not affect the optimal CH value $J^{\mathrm{LHVM}}_{\mathrm{CH}}(\delta)$ when the input randomness is quantified as in Eq.~\eqref{Eq:deltasource}. The quantum bound $J_Q$ is given by $(\sqrt{2}-1)/2$, thus we can see that $\delta$ should  at least be less than $0.051$ for all CH experiment realizations.

\subsubsection{NS condition}\label{Sec:CHSH}
In addition, we consider the scenario where the probability distribution $\tilde{p}_{AB}(a,b|x,y)$ defined in Eq.~\eqref{Eq:pabxy} satisfies the NS condition, which adds a constraint on $\tilde{p}_{AB}(a,b|x,y)$. That is, the probability of output $a$ ($b$) only relies on the input $x$ ($y$) independently of the input from the other party. To be more specific, NS requires $\tilde{p}_{AB}(a,b|x,y)$ to satisfy
 \begin{equation}\label{eq:nosignaling}
 \begin{aligned}
&\sum_b \tilde{p}_{AB}(a,b|x,y)=\sum_b \tilde{p}_{AB}(a,b|x,y')\equiv \tilde{p}_A(a|x), ~ \forall a,x,y,y'\\
&\sum_a \tilde{p}_{AB}(a,b|x,y)=\sum_a \tilde{p}_{AB}(a,b|x',y)\equiv \tilde{p}_B(b|y). ~\forall b,x,x',y
\end{aligned}
\end{equation}
We can follow the above derivation by imposing an additional NS constraint, which makes the problem even more complex.

Instead, we note that the CHSH inequality and the CH inequality are equivalent under NS, that is,
\begin{equation}\label{Eq:CHCHSH}
  J_{\mathrm{CH}}^{\mathrm{NS}} = \frac{1}{4}(J_{\mathrm{CHSH}} - 2),
\end{equation}
which we refer to Appendix~\ref{app:equiva} for a rigorous proof. As the CHSH inequality is defined with strong symmetry, we  solve the optimization problem with the CHSH inequality. We should note that we essentially take the NS condition into account when deriving the equivalence between the CH and CHSH inequality.

Based on the general definition of Bell's inequalities in Eq.~\eqref{eq:Bell}, the coefficients of the CHSH inequality is defined by
\begin{equation}\label{}
    \beta_{a,b,x,y}^{\mathrm{CHSH}} = (-1)^{xy + a + b},
\end{equation}
that is,
\begin{equation}\label{}
  J_{\mathrm{CHSH}} = \sum_{x,y,a,b\in\{0,1\}}(-1)^{xy + a + b}\tilde{p}(a, b|x, y).
\end{equation}

When considering LHVMs strategies with imperfect input randomness, the CHSH value can be written by
\begin{equation}\label{}
    J_{\mathrm{CHSH}}^{\mathrm{LHVM}} = 4\sum_\lambda q_\lambda J_\lambda,
\end{equation}
where
\begin{equation}\label{}
  J_\lambda = \sum_{x,y,a,b}(-1)^{xy + a + b}\tilde{p}_A(a|x,\lambda)\tilde{p}_B(b|y,\lambda)p(x,y|\lambda).
\end{equation}

Following a similar method described above, we first consider deterministic strategies, i.e., $\tilde{p}_A(a|x),\tilde{p}_B(b|y)\in\{0,1\}$ for the reason that any probabilistic LHVM could be realized with convex combination of deterministic ones.
Denote $p(i,j)$ as $p_{2*i + j}$, it is easy to show that the possible optimal deterministic strategies for $J_\lambda$ are
\begin{equation}
\begin{aligned}
   \{p_0+p_1+p_2-p_3&, p_0+p_1+p_3-p_2, \\
   p_0+p_2+p_3-p_1&, p_1+p_2+p_3-p_0\},
\end{aligned}
\end{equation}
and the constraints can also be described by Eq.~\eqref{eq:constraints}. Following a similar argument, we only need to consider four different types strategies of choosing the input settings. Thus $J^{\mathrm{LHVM}}_{\mathrm{CHSH}}$ can be given by
\begin{equation}\label{chsh}
\begin{aligned}
  &J^{\mathrm{LHVM}}_{\mathrm{CHSH}}/4 \\
  &= q(\lambda_1)(p_0(\lambda_1)+p_1(\lambda_1)+p_2(\lambda_1)-p_3(\lambda_1))\\
  &+q(\lambda_2)(p_0(\lambda_2)+p_1(\lambda_2)+p_3(\lambda_2)-p_2(\lambda_2))\\
  &+q(\lambda_3)(p_0(\lambda_3)+p_2(\lambda_3)+p_3(\lambda_3)-p_1(\lambda_3))\\
  &+q(\lambda_4)(p_1(\lambda_4)+p_2(\lambda_4)+p_3(\lambda_4)-p_0(\lambda_4)).
\end{aligned}
\end{equation}
With a symbolic notation in Eq.~\eqref{Eq:mathJlambda}, we can also present the coefficient $\beta_{ij}$ in a matrix, as shown in Table~\ref{Table:coeCHSH}.

\begin{table}[hbt]
\centering
\caption{The coefficient of $p_i(\lambda_j)$ in the expression of $J^{\mathrm{LHVM}}_{\mathrm{CHSH}}$ of the CHSH inequality.}
\begin{tabular}{ccccc}
  \hline
   &$\lambda_1$&$\lambda_2$&$\lambda_3$&$\lambda_4$\\
   \hline
   $p_0$&$1$&$1$&$1$&$-1$\\
   $p_1$&$1$&$1$&$-1$&$1$\\
   $p_2$&$1$&$-1$&$1$&$1$\\
   $p_3$&$-1$&$1$&$1$&$1$\\
  \hline
\end{tabular}\label{Table:coeCHSH}
\end{table}

We solve the optimization problem in Appendix~\ref{app:CHSH}. Based on the value of $P$ and $Q$, we give the optimal CHSH value $J^{\mathrm{LHVM}}_{\mathrm{CHSH}}$ with LHVMs by
\begin{equation}\label{}
  J^{\mathrm{LHVM}}_{\mathrm{CHSH}}(P,Q) =
  \left\{
  \begin{array}{cc}
    24P-4 & 3P+Q\le 1,\\
    4-8Q & 3P+Q\ge 1.
  \end{array}
  \right.
\end{equation}
Then the optimal CH value $J^{\mathrm{LHVM, NS}}_{\mathrm{CH}}$ with LHVMs under NS is
\begin{equation}\label{}
  J^{\mathrm{LHVM, \mathrm{NS}}}_{\mathrm{CH}}(P,Q) =
  \left\{
  \begin{array}{cc}
    6P-3/2 & 3P+Q\le 1,\\
    1/2-2Q & 3P+Q\ge 1.
  \end{array}
  \right.
\end{equation}
We show the optimal value of  $J^{\mathrm{LHVM, \mathrm{NS}}}_{\mathrm{CH}}$  in Fig.~\ref{Fig:plot2}.
\begin{figure}[!ht]
    \centering
 \resizebox{9cm}{!}{\includegraphics[scale=1]{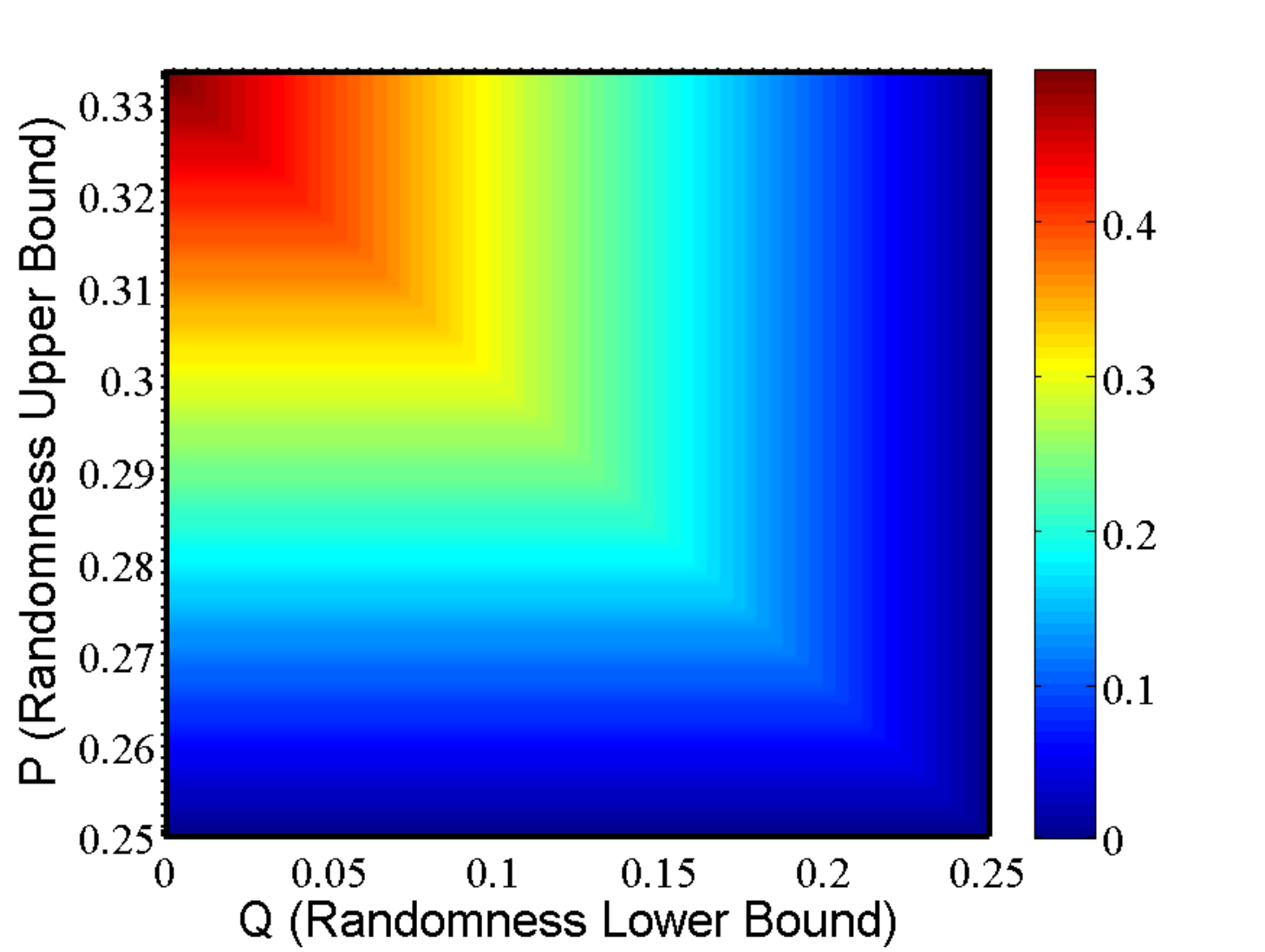}}
\caption{(Color online) The CH value $J^{\mathrm{LHVM, \mathrm{NS}}}_{\mathrm{CH}}$ as a function of $P$ and $Q$ under NS condition Eq.~\eqref{eq:nosignaling}. }\label{Fig:plot2}
\end{figure}

If we quantify the input randomness by its deviation from uniform distribution as defined in Eq.~\eqref{Eq:deltasource}, the optimal CH value $J^{\mathrm{LHVM,NS}}_{\mathrm{CH}}(\delta)$  is given by
\begin{equation}\label{}
  J^{\mathrm{LHVM}}_{\mathrm{CH,NS}}(\delta) = 2\delta.
\end{equation}
The quantum bound for the CH inequality $J_Q$ is $(\sqrt{2}-1)/2$, thus  $\delta$ should be less than $(\sqrt{2} - 1)/4\approx 0.104$ for all experiment realizations.

\subsubsection{NS condition and factorizable}
At last, we consider the probability distribution $\tilde{p}_{AB}(a,b|x,y)$ to be NS and the input randomness $p(x,y|\lambda)$ is factorizable. The optimization of the CHSH inequality is solved in Appendix~\ref{app:NSFac}, and the result is,
\begin{equation}\label{}
  J^{\mathrm{LHVM,Fac}}_{\mathrm{CHSH}}(P,Q) =
  \left\{
  \begin{array}{cc}
    8P & P+Q \le \frac{1}{2},\\
    4-8Q & P+Q> \frac{1}{2}.\\
  \end{array}
  \right.
\end{equation}
Then the optimal CH value $J^{\mathrm{LHVM, NS, Fac}}_{\mathrm{CH}}$ with LHVMs under NS and factorizable condition is
\begin{equation}\label{}
  J^{\mathrm{LHVM, NS, Fac}}_{\mathrm{CH}}(P,Q) =
  \left\{
  \begin{array}{cc}
    2P-1/2 & P+Q \le \frac{1}{2},\\
    1/2-2Q & P+Q> \frac{1}{2}.\\
  \end{array}
  \right.
\end{equation}
We show the optimal value of  $J^{\mathrm{LHVM, NS, Fac}}_{\mathrm{CH}}$  in Fig.~\ref{Fig:plot3}.
\begin{figure}[!ht]
    \centering
 \resizebox{9cm}{!}{\includegraphics[scale=1]{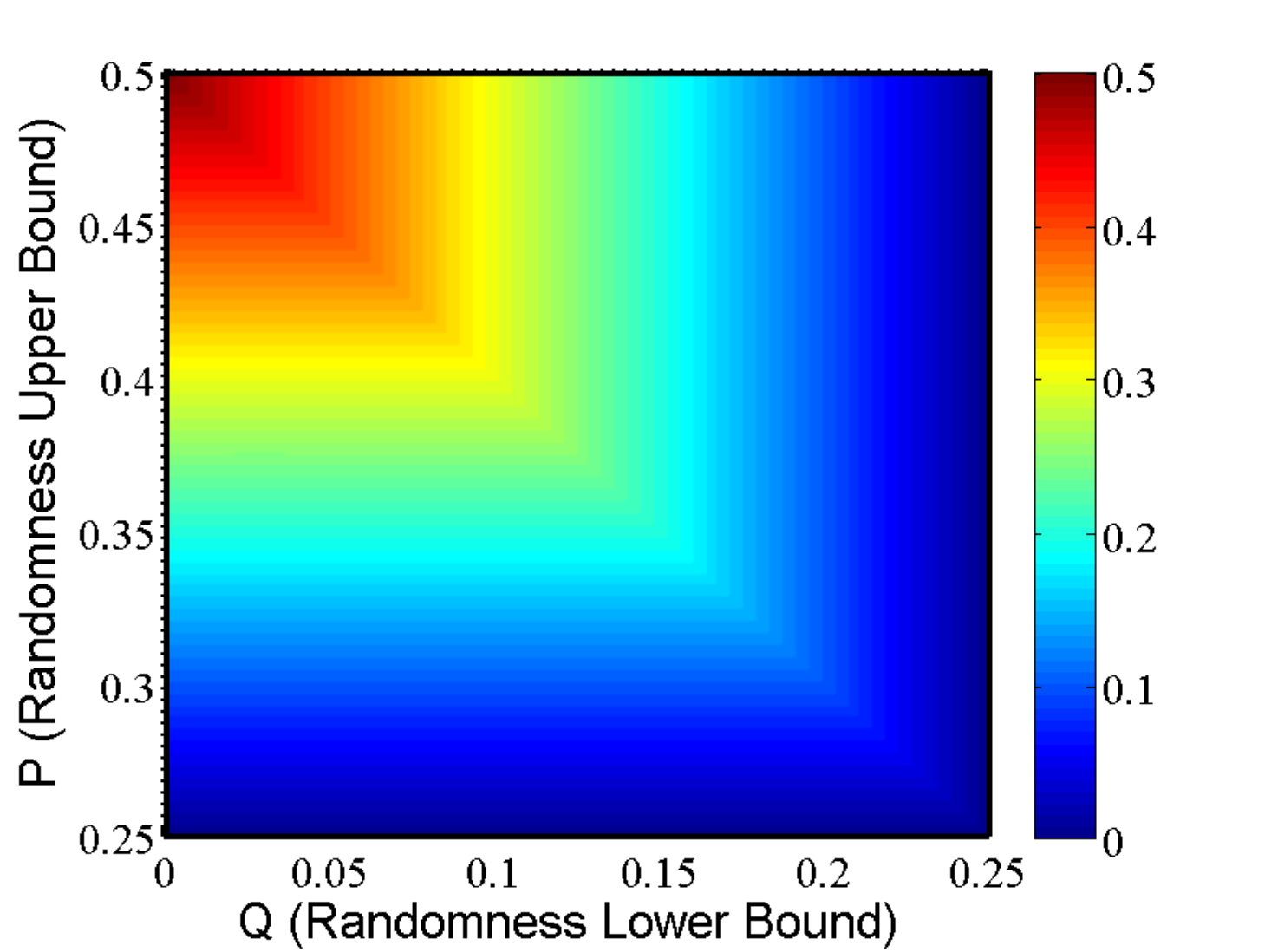}}
\caption{(Color online) The CH value $J^{\mathrm{LHVM, NS, Fac}}_{\mathrm{CH}}$ as a function of $P$ and $Q$ under factorizable Eq.~\eqref{Eq:Uncorrelated} and NS Eq.~\eqref{eq:nosignaling} conditions.} \label{Fig:plot3}
\end{figure}

When the input randomness is quantified as in Eq.~\eqref{Eq:deltasource}, where $P = 1/4+\delta$ and $Q = 1/4-\delta$, we have
\begin{equation}\label{}
  J^{\mathrm{LHVM, NS, Fac}}_{\mathrm{CH}}(\delta) = 2\delta.
\end{equation}
Again, it is interesting to note that the factorizable condition does not affect the optimal CH value $J^{\mathrm{LHVM,NS}}_{\mathrm{CH}}(\delta)$ when the input randomness is quantified as in Eq.~\eqref{Eq:deltasource}.

\subsubsection{Comparison}
Let us compare the results of the CH values $J^{\mathrm{LHVM}}_{\mathrm{CH}}$ under different conditions. For the maximal quantum violation $J_Q=(\sqrt{2}-1)/2$, we calculate the critical values of $Q$ and $P$ such that $J^{\mathrm{LHVM}}_{\mathrm{CH}}(P, Q)=J_Q$ and plot them in Fig.~\ref{Fig:PQ}. When $Q$ is small, the optimal CH value $J^{\mathrm{LHVM}}_{\mathrm{CH}}(P, Q)$ depends only on $P$. In this case, the critical values of $P$ for the signaling, signaling+fac, NS, and NS+fac are 0.207, 0.302, 0.285, 0.354, respectively. Thus, we can see that the factorizable condition puts a stronger requirement for $P$ compared to the NS condition. On the other hand, when $Q$ is large, the optimal CH value $J^{\mathrm{LHVM}}_{\mathrm{CH}}(P, Q)$ depends only on $Q$ instead. In this case, the critical values of $Q$ for the signaling and NS condition are 0.198 and 0.146, respectively. It is interesting to note that when both $P$ and $Q$ is large, the optimal CH value $J^{\mathrm{LHVM}}_{\mathrm{CH}}(P, Q)$ is independent on the factorizable condition.
\begin{figure}[!ht]
    \centering
    \resizebox{10cm}{!}{\includegraphics[scale=1]{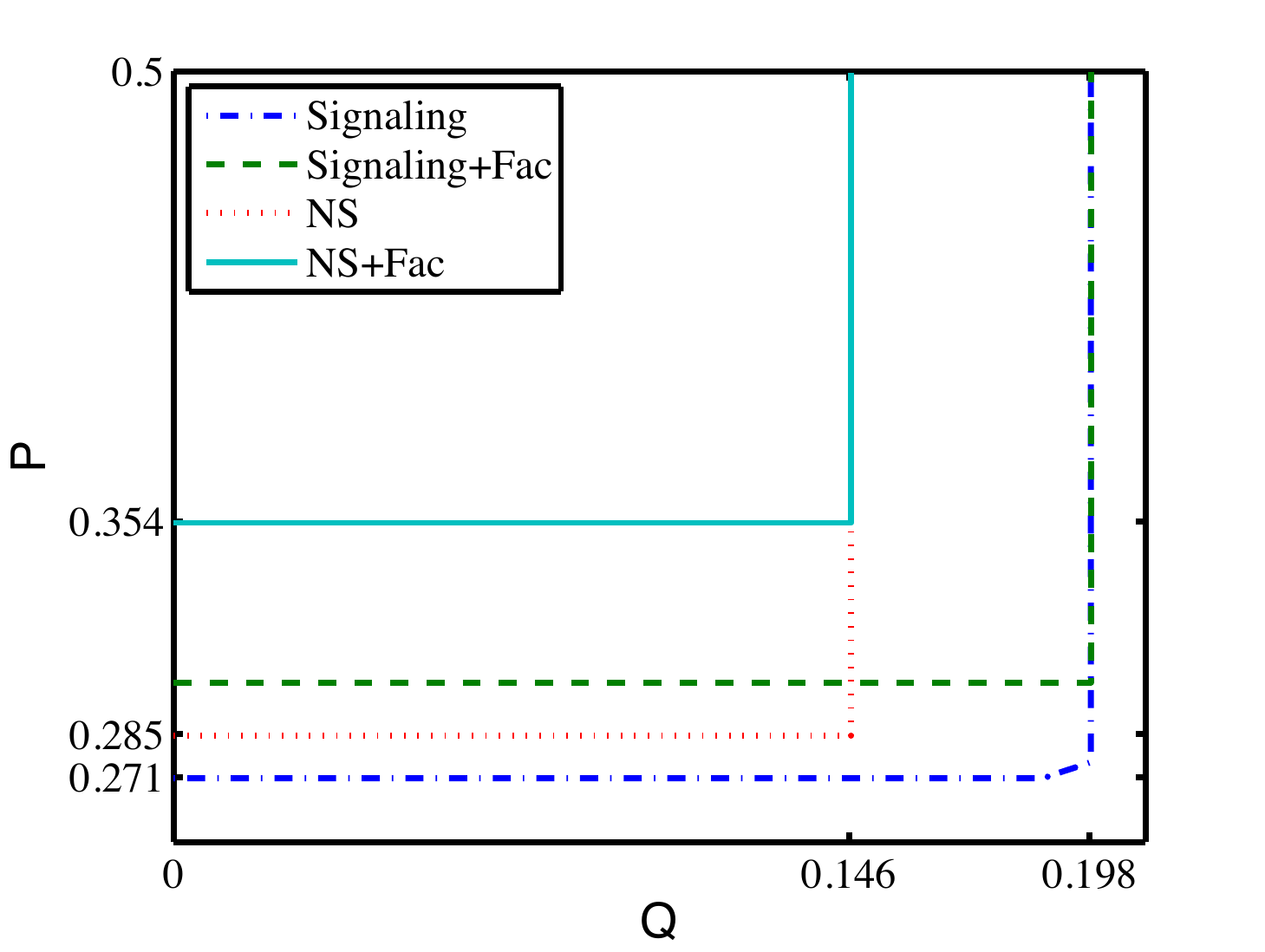}}
    \caption{The critical value of $Q$ and $P$ such that the CH value $J^{\mathrm{LHVM}}_{\mathrm{CH}}(P,Q)$ equals the maximal quantum value $J_Q=(\sqrt{2}-1)/2$.}\label{Fig:PQ}
\end{figure}

Besides, if we make use of the quantification method defined in Eq.~\eqref{Eq:deltasource}, we have already noticed that the optimal CH value $J^{\mathrm{LHVM}}_{\mathrm{CH}}(\delta)$ is independent on the factorizable condition. Here, we compare $J^{\mathrm{LHVM}}_{\mathrm{CH}}(\delta)$ between the signaling and NS condition as shown in Fig.~\ref{Fig:Jdelta}. For the maximal quantum violation $J_Q=(\sqrt{2}-1)/2$, we calculate the critical values of $\delta$ for the signaling and NS condition to be 0.052 and 0.104, respectively.

\begin{figure}[!ht]
    \centering
    \resizebox{10cm}{!}{\includegraphics[scale=1]{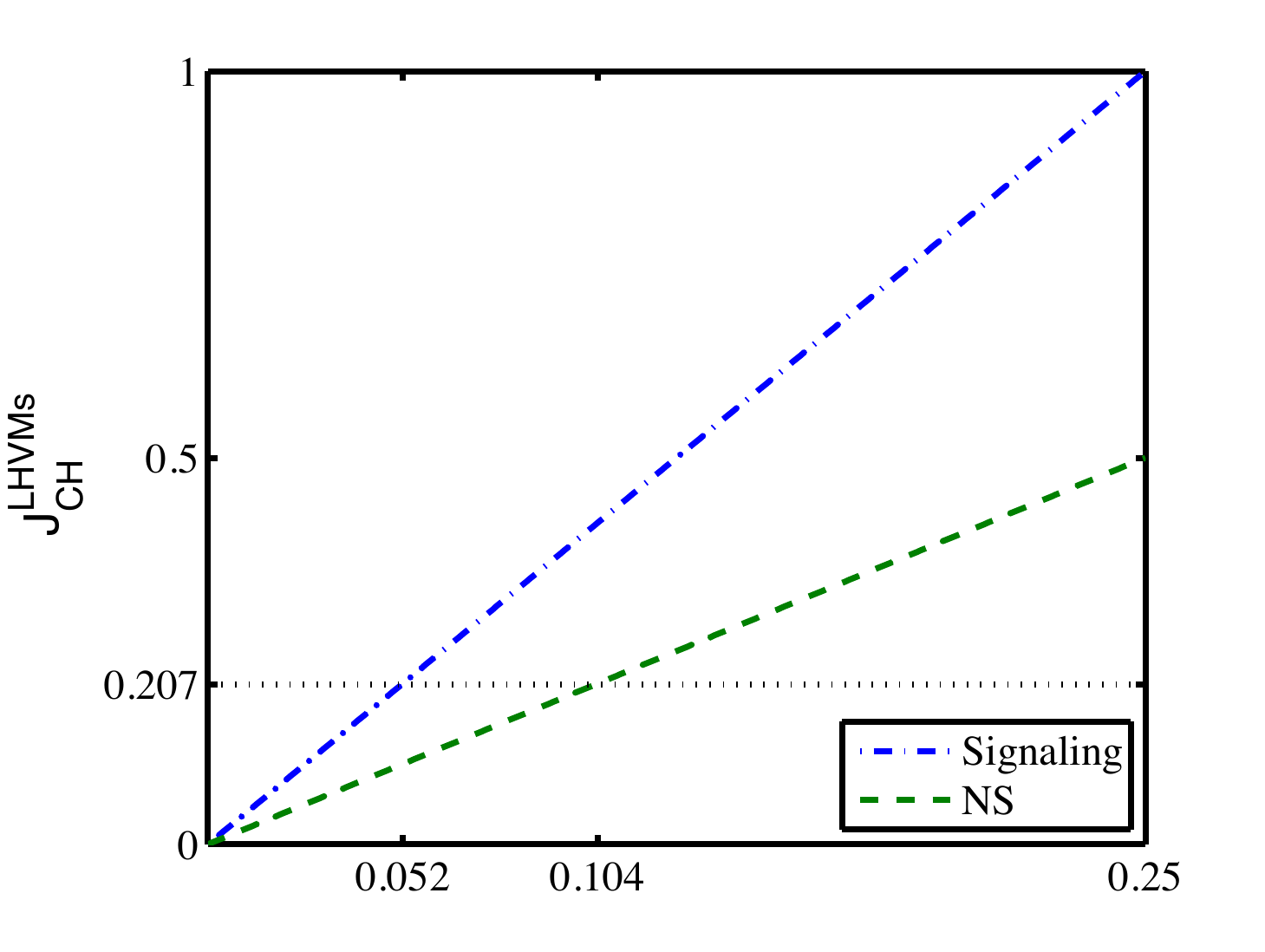}}
    \caption{The CH value $J^{\mathrm{LHVM}}_{\mathrm{CH}}(\delta)$ under different conditions.}\label{Fig:Jdelta}
\end{figure}


When we quantify the input randomness as Eq.~\eqref{Eq:deltasource}, we found that the optimal CH value is independent of the factorizable condition, Eq.~\eqref{Eq:Uncorrelated}. Thus with the quantification method in Eq.~\eqref{Eq:deltasource}, one may not need to consider the factorizable condition.

For further works, it is interesting to consider joint strategies of LHVMs of the CH inequality, where the inputs of different runs are correlated. It is already shown that joint attacks to the CHSH inequality puts a very high requirement of the input randomness no matter the factorizable condition is satisfied or not \cite{Yuan15CHSH}. In addition, we can investigate the case where the input randomness is restricted by both lower and upper bounds, $P$ and $Q$. Due to the asymmetric definition of the CH inequality, the analysis may become extremely complicated with increasing number of correlated runs. Similar cases may also happen for other asymmetric Bell inequalities. 

Furthermore, it is interesting to see whether there exist Bell's inequalities such that the randomness requirement is very low. To do so, we have to solve the problem of optimizing the Bell value with all LHVM strategies. We expect that our derivation method could provide a general way to solve this problem.


\chapter{Source-independent quantum random number generation}
In general, a physical generator contains two parts¡ªa randomness source and its readout. The source is essential to the quality of the resulting random numbers; hence, it needs to be carefully calibrated and modeled to achieve information-theoretical provable randomness. In this chapter, we discuss quantum random number generation without trusting its randomness source \cite{Cao2016} hence semi-selftesting. We show in this chapter that randomness can be obtained even without characterizing its source.

\section{The prepare and measure model}
A typical QRNG can be described by the prepare and measure model that can be decomposed into two modules: a  randomness source (quantum state preparation) and its readout (measurement), as shown in Fig.~\ref{Fig:SRModel}. In general, the source emits quantum states that are superpositions of the measurement basis. The output (raw) random numbers are the measurement results. In many QRNGs, a short random seed is required to assist state preparation or measurement.

\begin{figure}[htb]
\centering \resizebox{10cm}{!}{\includegraphics{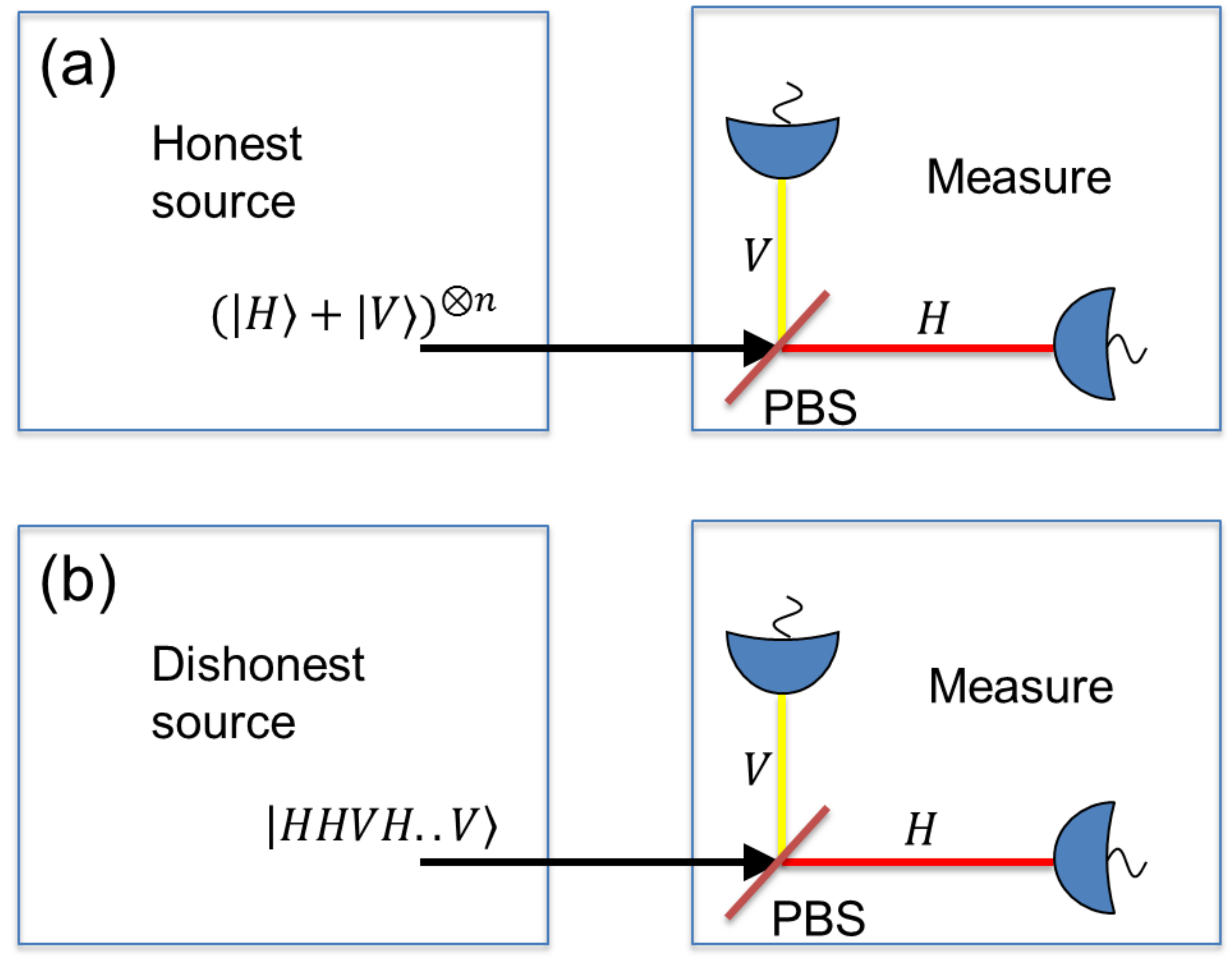}}
\caption{Illustration of a generic QRNG setup in which we take photon polarization as the example. $H$ and $V$ refer to horizontal and vertical polarizations, respectively. PBS refers to a polarizing beam splitter. (a) The source functions normally (or trusted) and sends superpositions of  $H$ and $V$ polarizations, which offers quantum randomness. (b) The source malfunctions (or untrusted) and sends $H$ and $V$ polarizations in a predetermined order, which should output no genuine randomness. From the measurement result viewpoint, one cannot distinguish these two cases.} \label{Fig:SRModel}
\end{figure}

As an example, consider a simple QRNG that projects the quantum state  $\ket{+}=(\ket{H}+\ket{V})/{\sqrt{2}}$ emitted from a single photon source on the horizontal and vertical polarization basis $\{\ket{H}, \ket{V}\}$. This QRNG can be divided into two modules, as shown in Fig.~\ref{Fig:SRModel}(a). Randomness is guaranteed by the intrinsically probabilistic nature of quantum physics. Hereafter, we denote $\ket{H}$, $\ket{V}$ ($\ket{+}$, $\ket{-}$) as the $Z$-basis ($X$-basis) eigenstates.

Existing practical QRNGs suffer from security loopholes if the devices are not perfect.
In the source readout model, the measurement device can normally be trusted due to its simple structure. For instance, in the previous example, the measurement is a simple demolition measurement on the polarization basis.
In contrast, the randomness contained in a source, such as a laser or an atomic assemble, is normally difficult to characterize completely. If the source malfunctions and emits classical signals instead of quantum ones, the outputs may not be truly random. Consider the worst-case scenario in which the devices are designed or controlled by an adversary Eve.
Eve can employ a pseudo-RNG to output a fixed (from Eve's viewpoint) string that still appears random to Alice. More concretely, in the example of the previous paragraph,
when a dishonest source emits $Z$-basis instead of $X$-basis eigenstates for the $Z$-basis measurement, the output will just be a fixed string, as shown in Fig.~\ref{Fig:SRModel}(b). From this perspective, with given measurement devices, justifying the randomness in a source is crucial to generating  randomness.

Most existing QRNGs use complicated physical models \cite{PhysRevA.75.032334,Xu:QRNG:2012} to quantify their sources. For example, the dimension of the source is sometimes assumed to be a fixed known number \cite{PhysRevA.90.052327}. The underlying models implicitly assume the existence of randomness in the first place, but this assumption cannot be verified experimentally. Therefore, to achieve truly reliable randomness, there is a strong motivation to avoid the use of such models. Note that removing the dimension assumption is the key challenge to the analysis for device-independent scenarios.

Thus, a QRNG without trusting the source (source-independent) is both theoretically and practically meaningful and greatly in need. A device-independent QRNG \cite{Vazirani12} can generate randomness without having to  trust the devices. This type of QRNG requires a short seed for device testing, which is the reason why they are also called randomness expansions \cite{colbeck2011private,Miller14,miller2014universal}. By observing the violation of a certain Bell's inequality, such as the Clauser--Horne--Shimony--Holt inequality \cite{CHSH}, one can guarantee the presence of randomness without any assumptions about the source or the measurement device. The main drawback of device-independent QRNGs is that they are not loss-tolerant, which typically imposes very severe requirements on experimental devices.
Furthermore, this type of QRNG generates random numbers at rates that are very low for practical applications. The highest speed of this type of QRNG has, so far, been reported to be 0.4~bps \cite{Christensen13}.


In this chapter, we will introduce a source-independent QRNG (SIQRNG) scheme based on the uncertainty relation that is loss-tolerant and hence highly practical. In particular, our scheme allows the source to have arbitrary and unknown dimensions. The loss-tolerance property enables potential high-loss implementations of our scheme, such as in integrated optic chips or with inefficient but cheap single photon detectors. We analyze the randomness of the scheme based on complementary uncertainty relations. Our analysis takes account of several practical issues, including finite-key-size effects, multi-photon components in the source, initial seed length, and losses. The analysis combines several ingredients from the security proof of quantum key distribution (QKD), a rich subject that has developed over the last two decades. These ingredients include phase error correction \cite{ShorPreskill_00}, random sampling \cite{Fung:Finite:2010}, and the squashing model \cite{BML_Squash_08}. Since the squashing model shows the equivalence between threshold detectors and qubit detectors \cite{BML_Squash_08}, our scheme allows the source to have an unfixed finite dimension as well as an infinite dimension. For simplicity, in the rest of this chapter, we assume a two-level (bit) output system. All our techniques can be directly applied to cases with more outputs.

\section{The Protocol}\label{protocol}
\subsection{Theoretical protocol}
A schematic of our SIQRNG protocol is shown in Fig.~\ref{Fig:EXPSetup}(a). Here, we take an optical implementation as the example, as shown in Fig.~\ref{Fig:EXPSetup}(b). All our results apply similarly to other implementation systems. Quantum signals from the source first go through a modulator that actively chooses between the $X$ and $Z$ bases. Then, a polarizing beam splitter and two threshold detectors perform a projective measurement. Since two detectors are used, there are four possible outcomes: no clicks (losses), two single clicks, and double clicks. This implementation is equivalent to the schematic setup of \emph{squashing model} as discussed in Section \ref{Sec:Squash}. The details of the protocol are presented in Fig.~\ref{Fig:Procedure}.

\begin{figure}[hbt]
\centering \resizebox{8cm}{!}{\includegraphics{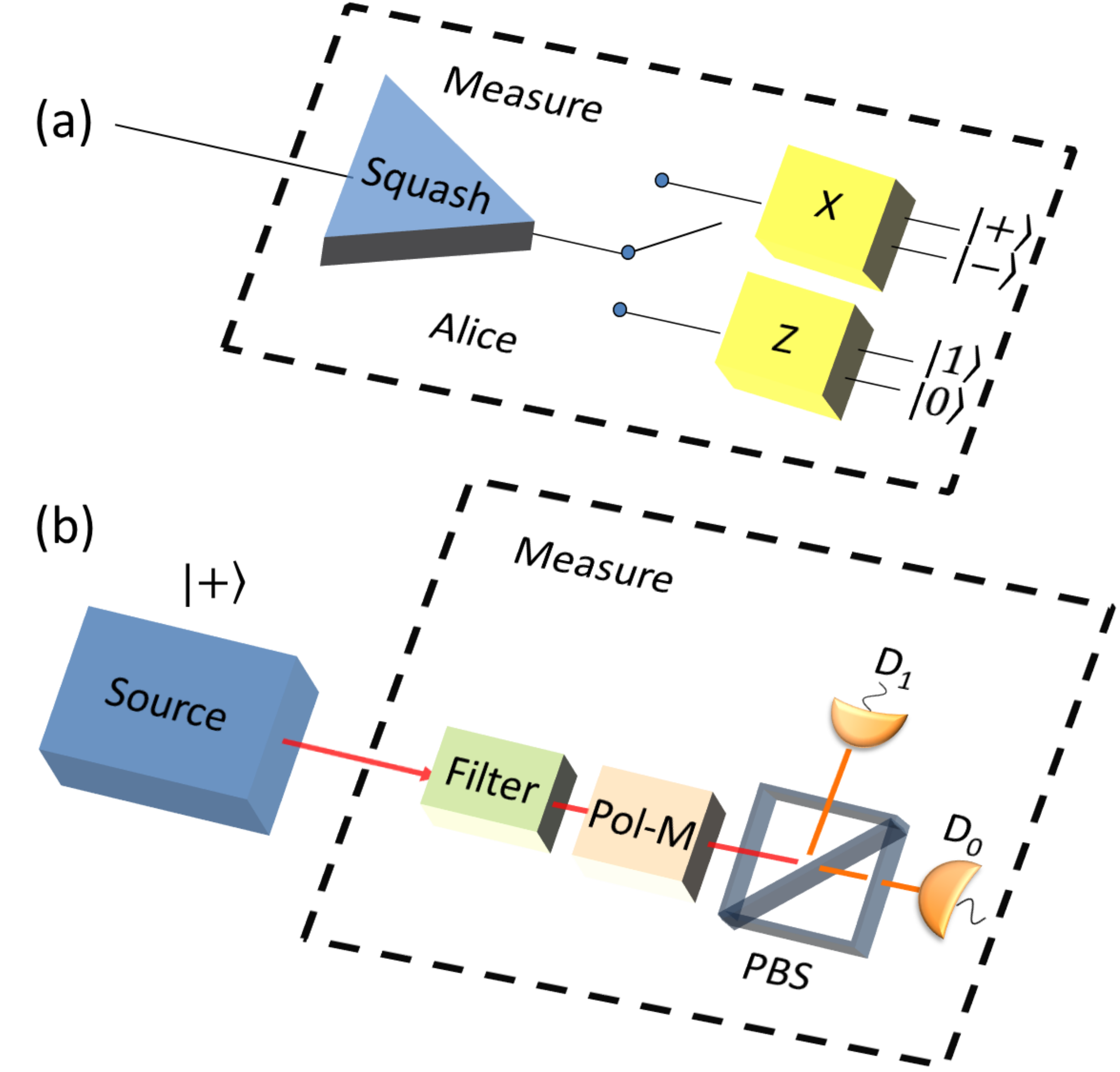}}
\caption{(a) Measurement model for  SIQRNG. The quantum state first passes through a squasher and is projected as either a qubit or a vacuum. Then, the output qubit is measured in the $X$ or $Z$ basis chosen by an active switch. There are two outcomes for each basis measurement, corresponding to the two eigenstates of the basis. (b) An optical implementation of the SIQRNG in (a), as discussed in Section~\ref{Sec:Squash}. Here Pol-M refers to a polarization modulator, PBS refers to a polarizing beam splitter, and $D_0$ and $D_1$ are the threshold detectors.} \label{Fig:EXPSetup}
\end{figure}

\begin{figure}[htbp]
\begin{framed}
\centering
\begin{enumerate}
\item
\textbf{Source:} An untrusted party, Eve, prepares many quantum states in an arbitrary and unknown dimension and feeds them into the measurement box of Alice.
\item
\textbf{Squashing:} Alice (or Eve) squashes the quantum states into qubits and vacua. Alice postselects  the vacua and obtains $n$ squashed qubits. The vacuum components take account of optical losses and quantum efficiencies.
\item
\textbf{Random sampling:} By consuming a short seed with the length given in Eq.~\eqref{Source:seed}, Alice randomly chooses $n_x$ out of the $n$ squashed qubits and measures them in the $X$ basis, each results in $\ket{+}$ or $\ket{-}$. 
\item
\textbf{Parameter estimation:} When  the  system  operates  properly, the source emits qubits $\ket{+}$ for all runs. Thus, a result of $\ket{-}$ in the $X$-basis measurement is defined as an error. A double click is considered to be  half an error. Alice evaluates the bit error rate $e_{bx}$ in the $X$ basis and its statistical deviation $\theta$ according to Eq.~\eqref{eq:Ptheta}. If $e_{bx}+\theta\ge1/2$, Alice aborts the protocol.
\item
\textbf{Randomness generation:}
For the remaining $n-n_x$ squashed qubits, Alice performs measurement in the $Z$ basis to generate $n_z = n-n_x$ random bits.
\item
\textbf{Randomness extraction:} Alice picks a parameter $t_e$ according to the desired failure probability restriction and extracts $n_z-n_zH(e_{bx}+\theta)-t_e$ bits of final randomness using Toeplitz-matrix hashing \cite{Mansour:Toeplitz:93,Ma2011Finite}\footnote{Other extraction methods, such as Trevisan's extractor \cite{trevisan2001extractors} can be applied, in which the relation between the failure probability and $t_e$ can differ.}.
\item
\textbf{Security parameter:} With the composable security definition, the security parameter (in trace-distance measure) is given by $\varepsilon=\sqrt{(\varepsilon_\theta+2^{-t_e})(2-\varepsilon_\theta-2^{-t_e})}$.
\end{enumerate}
\end{framed}
\caption{Source-independent QRNG with the finite data size effect. The results are proven in Section~\ref{analysis}.} \label{Fig:Procedure}
\end{figure}


\subsection{Analysis}\label{analysis}
In this part, we analyze the randomness output of the SIQRNG protocol. Strictly speaking, like device-independent QRNGs, our scheme is  a randomness expansion scheme, in which a random seed is used to generate extra independent randomness. The procedure of parameter estimation is an analog to the phase error rate estimation in QKD postprocessing \cite{Ma2011Finite}. Randomness extraction is mathematically equivalent to privacy amplification in QKD.  The difference between the biased measurement used here and the biased-basis choice QKD protocol \cite{Lo:EffBB84:2005} is that the number of $X$-basis measurements is a constant in our case, whereas in QKD, this number must go to infinity when the data size is infinitely large.

\subsubsection{Squashing model}  \label{Sec:Squash}
In the SIQRNG scheme, we assume that  measurement devices are trusted and well characterized. The key assumption here is that the \emph{measurement setup is compatible with the squashing model}. That is, a measurement can be treated in two steps. First, the (unknown arbitrary-dimensional) signal state emitted from the source is projected to a qubit or vacuum. The projection is called  a squasher, as shown in Fig.~\ref{Fig:EXPSetup}(a). Then, the squashed qubits are post-selected by discarding the vacua and measuring them in the $X$ or $Z$ basis.  This assumption can be satisfied when threshold detectors are used with random bit assignments for  double clicks \cite{BML_Squash_08}.
For the protocol described in Section \ref{protocol}, the $X$-basis measurement results are used for parameter estimation and  are then discarded in postprocessing. Thus, the random assignment can be replaced by adding  half of the double-click ratio to the $X$-basis error rate.

In practice, it is a challenge to verify whether a measurement setup is compatible with the squashing model. Much effort has been put into this question \cite{PhysRevA.86.042327}. The key point here is to make the two detectors respond equally to (four) different qubits, and hence make the measurement device basis-independent \cite{Fung:Mismatch:2009}. This can be done by adding a series of filters (including spectrum and temporal filters) before the threshold detectors, to ensure that the input states stay within a proper set of optical modes \cite{xu2014experimental}, in which the detectors have the same efficiencies \cite{BML_Squash_08, Fung:Mismatch:2009}. One can further assume that Alice uses a trusted source to calibrate the measurement devices beforehand; that is, Alice performs a quantum measurement tomography. A similar measurement calibration procedure should be done in most  current QKD and QRNG realizations. Here, we emphasize that the verification of the squashing model does not affect the source-independent property of our scheme. Thus, we leave detailed investigation on validating the measurement setting for future works.


Similar to the QKD case \cite{BML_Squash_08}, we can assume that the squashing operator is held by Eve in the randomness analysis. By this, we mean that Eve can choose a valid operator, so long as the output is a qubit or a vacuum. In the following discussions, we focus on the squashed qubits. We need to determine the min-entropy associated with these qubits in the $Z$-basis measurement.

\subsubsection{Complementary uncertainty relation} \label{Sec:Uncertain}

First, we show intuitively why the protocol works. According to quantum mechanics, the outcome of projecting the state $\ket{+}$ on the $Z$ basis is random. Of course, in reality, due to device imperfections, Alice would never obtain a perfect state of $\ket{+}$. Now, the key question for Alice becomes how to verify that the source faithfully emits the state $\ket{+}$. This can be done by borrowing a similar technique from the security analysis of QKD \cite{LoChauQKD_99,ShorPreskill_00,Koashi_Uncer_06} and consider an equivalent virtual protocol depicted in Fig.~\ref{Fig:Procedureeq}, where we replace steps $5$ and $6$ by $5'$ and $6'$. In steps $3$ and $4$ of the protocol, Alice occasionally performs the $X$-basis measurement and define the \emph{phase error} rate to be the ratio of detecting $\ket{-}$. In the virtual protocol, once Alice knows the phase error rate by random sampling tests, she can perform a phase error correction (step $5'$) before the final $Z$-basis measurement (step $6'$). By an smart design of the phase error correction procedure \cite{ShorPreskill_00}, Alice can make it commute with the $Z$-basis measurement. Thus, she can perform the $Z$-basis measurement (step $5$) first and then apply randomness extraction (step $6$). At this stage, all the states have already collapsed to classical results, and the phase error correction procedure becomes randomness extraction (or privacy amplification in QKD) \cite{LoChauQKD_99,ShorPreskill_00,Koashi_Uncer_06}. Besides QKD, the argument here is similar to the one used in Ref.~\cite{Yuan15Coherence}, where one can consider the error correction process $5'$ as distilling coherence or randomness extraction.

\begin{figure}[htbp]
\begin{framed}
\centering
\begin{enumerate}[1']
\setcounter{enumi}{4}

\item
\textbf{Error correction:} Based on phase error rate $e_{bx}$, Alice performs phase error correction and obtain  $n_z[1-H(e_{bx})]$ copies of perfect $\ket{+}$ with a nearly unit probability.

\item
\textbf{Randomness generation:}
After obtaining all states in $\ket{+}$, Alice performs measurement in the $Z$ basis to get $n_z[1-H(e_{bx})]$ random bits.

\end{enumerate}
\end{framed}
\caption{An equivalent protocol of source-independent QRNG.} \label{Fig:Procedureeq}
\end{figure}

It has been proved that the phase error correction (randomness extraction) can be efficiently done with Toeplitz-matrix hashing \cite{Mansour:Toeplitz:93}. Suppose the number of qubits measured in the $Z$ basis is $n_z$ and the phase error rate is $e_{pz}$, the number of bits sacrificed in the phase error correction is given by
\begin{equation} \label{eq:PAcost}
\begin{aligned}
n_zH(e_{pz})+t_e,
\end{aligned}
\end{equation}
and the probability that the phase error correction fails is $2^{-t_e}$ \cite{Ma2011Finite}. Here, $H(e)=-e\log e -(1-e)\log(1-e)$ is the binary Shannon entropy function, all the $\log$ is base 2 throughout this chapter, and $t_e$ is the parameter Alice picks up by balancing the failure probability and the final output length. Then, the number of final random bits is given by,
\begin{equation} \label{eq:R}
\begin{aligned}
K\ge n_z-n_zH(e_{pz})-t_e.
\end{aligned}
\end{equation}
In practice, Alice needs to prepare a Toeplitz matrix of size $n_z\times[n_z-n_zH(e_{pz})-t_e]$ for randomness extraction.

We note that the failure probability $2^{-t_e}$ quantifies  fidelity between the  state that results from the phase error correction and the ideal state $\ket{+}^{\otimes n_z}$. In the composable security definition \cite{BenOr:Security:05,Renner:Security:05}, a trace-distance measure security parameter $\varepsilon_t$ should be employed. Its relation to the fidelity measure $\varepsilon_f$ is given by \cite{Fung:Finite:2010}
\begin{equation} \label{eq:traceFidelity}
\begin{aligned}
\varepsilon_t=\sqrt{\varepsilon_f(2-\varepsilon_f)}
\end{aligned}
\end{equation}
In the following, we shall use the fidelity measure for the failure probability, which, in the end, can be conveniently converted to the trace-distance measure security parameter.

To construct the Toeplitz matrix of size $n_z\times[n_z-n_zH(e_{pz})-t_e]$, Alice needs to use $n_z+n_z-n_zH(e_{pz})-t_e-1$ random bits. Thanks to the Leftover Hash Lemma \cite{Impagliazzo:Leftover:1989}, the Toeplitz hashing extractor can be proven to be a strong extractor. That is, the output random bits are independent of the random bits used in the construction of the Toeplitz matrix \cite{frauchiger2013true}. Thus, the Toeplitz matrix can be reused.

Our result can also be derived via a different but  elegant approach by employing a newly developed seminal uncertainty relation \cite{tomamichel2012tight} and extending the Leftover Hash Lemma \cite{Impagliazzo:Leftover:1989} to the quantum scenario \cite{tomamichel2011leftover}. Interestingly, the result from that approach yields a security parameter (in trace distance measure) that is of the order of $2^{-t_e/2}$, which is consistent with ours. Such techniques have been successfully applied in some applications, including QRNGs \cite{PhysRevA.90.052327}.

%
%
%
%

\subsubsection{Finite key analysis} \label{Sec:Finite}
In practice, the QRNG only runs for a finite  time; consequently, the sampling tests for the $X$-basis measurements will suffer from statistical fluctuations. In the parameter estimation step, the key parameter $e_{pz}$ in Eq.~\eqref{eq:R} should be estimated (bounded) from the finite data size effect. 

In the random sampling test, Alice measures the squashed qubits in the $X$ basis and obtains the error rate, $e_{bx}$. Remember that, as required in the squashing model, this error rate includes half of the double-click ratio. Henceforth, we simply call this error rate as the $X$-basis error rate. Recall that the phase error rate $e_{pz}$ is defined as the error rate if the quantum signals measured in the $Z$ basis were measured in the $X$ basis. When the sampling size is large enough, $e_{pz}$ can be well approximated by $e_{bx}$. 

Before presenting the details of the random sampling analysis, we establish a notation. Suppose Alice receives $n$ squashed qubits and randomly chooses $n_x$ of them to be measured in the $X$ basis, leaving the remaining $n_z=n-n_x$ qubits in the $Z$ basis. Let the ratio of $X$-basis measurements be $q_x=n_x/n$, the number of errors Alice finds in the $X$ basis to be $k$, and the total number of errors to be $m$ if Alice had measured all qubits in the $X$ basis. Then, the number of errors in the qubits measured in the $Z$ basis is $m-k$, which is the key parameter we need to determine through random sampling. The quantity $m-k=n_ze_{pz}$ determines the randomness extraction rate. Define the lower bound of $e_{pz}$ by,
\begin{align} \label{eq:eptheta}
e_{pz}\le e_{bx}+\theta,
\end{align}
where $\theta$ is the deviation due to statistical fluctuations.
Following the random sampling results of Fung et al.~\cite{Fung:Finite:2010}, we can bound the probability when Eq.~\eqref{eq:eptheta} fails,
\begin{equation} \label{eq:Ptheta}
\begin{aligned}
\varepsilon_\theta &= \text{Prob}(e_{pz}>e_{bx}+\theta)  \\
&\le \frac{1}{\sqrt{q_x(1-q_x)e_{bx}(1-e_{bx})n}}2^{-n\xi(\theta)},
\end{aligned}
\end{equation}
where $\xi(\theta)= H(e_{bx}+\theta-q_x\theta)-q_x H(e_{bx})-(1-q_x)H(e_{bx}+\theta)$. Note that in the unlikely event that $e_{bx}=0$, the failure probability is unbounded, and one should rederive the failure probability or simply replace $e_{bx}$ with a small value, say, $1/n_x$.

In practice, the failure probability $\varepsilon_\theta$ is normally picked to be a small number depended on applications. In later data postprocessing, we pick up $\varepsilon_\theta=2^{-100}$. Once $\varepsilon_\theta$ is fixed, there is a trade-off between $q_x$ and $\theta$ for the ratio of the final random bit length over the raw data size. Thus, the number of samples for the $X$-basis measurement should be optimized for the randomness extraction rate.

One key property for the random sampling is that the $n_x$ locations of the $X$-basis measurements are randomly chosen from the total $n$ locations, i.e., the $\binom{n}{n_x}$ cases occur equally likely. Then, Alice needs a random seed with a length of
\begin{equation} \label{Source:seed}
n_{seed} = \log\binom{n}{n_x} \le n_x\log n.
\end{equation}
The effect of loss on the seed length will be discussed in Section \ref{Sec:Issues}.
In Appendix \ref{App:numXmeas}, we show that $n_x$ can remain a constant, given the failure probability, when $n$ is large. Then, in the large data size limit, the seed length is exponentially small compared to the length of the output random bit. Therefore, we reach an exponential randomness expansion.

\subsubsection{Practical issues} \label{Sec:Issues}
{\bf Multi photons:} In our protocol, the source is allowed to emit multi photons, since its dimension is assumed to be uncharacterized. In other words, these components would not affect the randomness of the final output. In practice,  multi photons may introduce double clicks when threshold detectors are used \cite{BML_Squash_08}; these  double  clicks will directly contribute to the error rate term $e_{bx}$. Thus, when the multi-photon ratio is very high, the double-click ratio will increase to a point that the upper bound on information leakage $e_{pz}$ increases to oen half; at that point,  when no random bits can be extracted according to Eq.~\eqref{eq:R} and Alice simply aborts the protocol.


{\bf Loss:} The loss tolerance of our protocol is guaranteed by the squashing model in which the measurement is assumed to be basis independent \cite{BML_Squash_08}. This assumption can be guaranteed by the fact that the basis is chosen after losses. Alice does not anticipate the positions of losses, so she effectively decides the (random) positions for $X$-basis measurements before losses. The effect of loss only decreases the number of effective $X$ measurements, but the positions of effective $X$ measurements are still uniformly random in squashed qubits; this fulfills the requirement of  random sampling. The detailed proof is shown in Appendix \ref{app:sampling}.

{\bf Basis-dependent detector efficiency:} Our protocol  assumes that the efficiencies of the detectors are the same. In practice, efficiency mismatches would cause the measurement to be different for the two bases (basis dependent). A viable way to deal with this imperfection is to recalculate the rate as a function of the ratio between the efficiencies of the two bases, employing the technique used in QKD \cite{Fung:Mismatch:2009}. As indicated by the result in QKD \cite{Fung:Mismatch:2009}, the random number generation rate will slightly decrease when there is a small mismatch in detector efficiencies. More precisely, denote the ratio between the minimum and maximum efficiencies of the two detectors as $r\le1$, then the key size becomes $r n_z (1-H[(e_{bx}+\theta)/r])-t_e$ bits. We leave detailed analysis of this imperfection for future work.


{\bf Double clicks:} Our analysis takes account of the effect of double clicks by adding half of the double-click ratio to the $X$-basis error rate, as required in the squashing model. This is also essentially why multi-photon states can be used on the source side without affecting final randomness. Note that double clicks should not be discarded freely in the measurement. Otherwise a security loophole will appear, namely, a strong pulse attack \cite{Lutkenhaus_99DoubleClick}. In a strong pulse attack, Eve always sends strong signals (with many photons) in the $Z$ basis. Suppose she sends a strong state in $\ket{H}$; if Alice chooses the $Z$-basis measurement, a valid raw random bit will be obtained, but if she chooses the $X$ basis, a double click is likely to happen. In our protocol, when Alice chooses the $X$-basis measurement, she should get an error (resulting in $\ket{-}$) with a probability of one half. If Alice simply discards all double clicks, Eve's attack will not be noticed. This attack cannot be explained by a qubit measurement. This is intuitively why the squashing model requires random assignments for double clicks.


{\bf Basis choice:} 
When choosing $X$- or $Z$-basis measurements, an input random string of length $N$ (as a seed) is needed to choose the basis. Suppose the number of $X$-basis measurements to be performed is $N_x$, then Alice chooses $N_x$ positions out of $N$ with equal probability, i.e., with probability $\binom{N}{N_x}^{-1}$.  Then, she needs a seed length of $\log\binom{N}{N_x}$. This is similar to Eq.~\eqref{Source:seed} with the difference  that before the measurement, Alice does not know the positions of losses. More details on how to dilute a short random seed to a longer (partially random) one are provided in Appendix \ref{app:input}.

{\bf Intensity optimization:} The intensity of the source should be optimized to maximize the randomness generation rate. With increasing intensity, the detection rate will increase along with an increases in the double-click rate (and hence $e_{pz}$ increases). There exists a trade-off between $n_z$ and $e_{pz}$, as shown in Eq.~\eqref{eq:R}.

\section{Experiment demonstration} \label{sec:exp}
In this section, we perform a proof-of-principle experimental demonstration to show the practicality of SIQRNG scheme. Our experiment setup consists of two parts, the source, owned by an untrusted party Eve, and the measurement device, owned by the user Alice. The schematic diagram is shown in Fig.~\ref{fig:expsetup}.

\begin{figure}[htb]
\centering \resizebox{10cm}{!}{\includegraphics{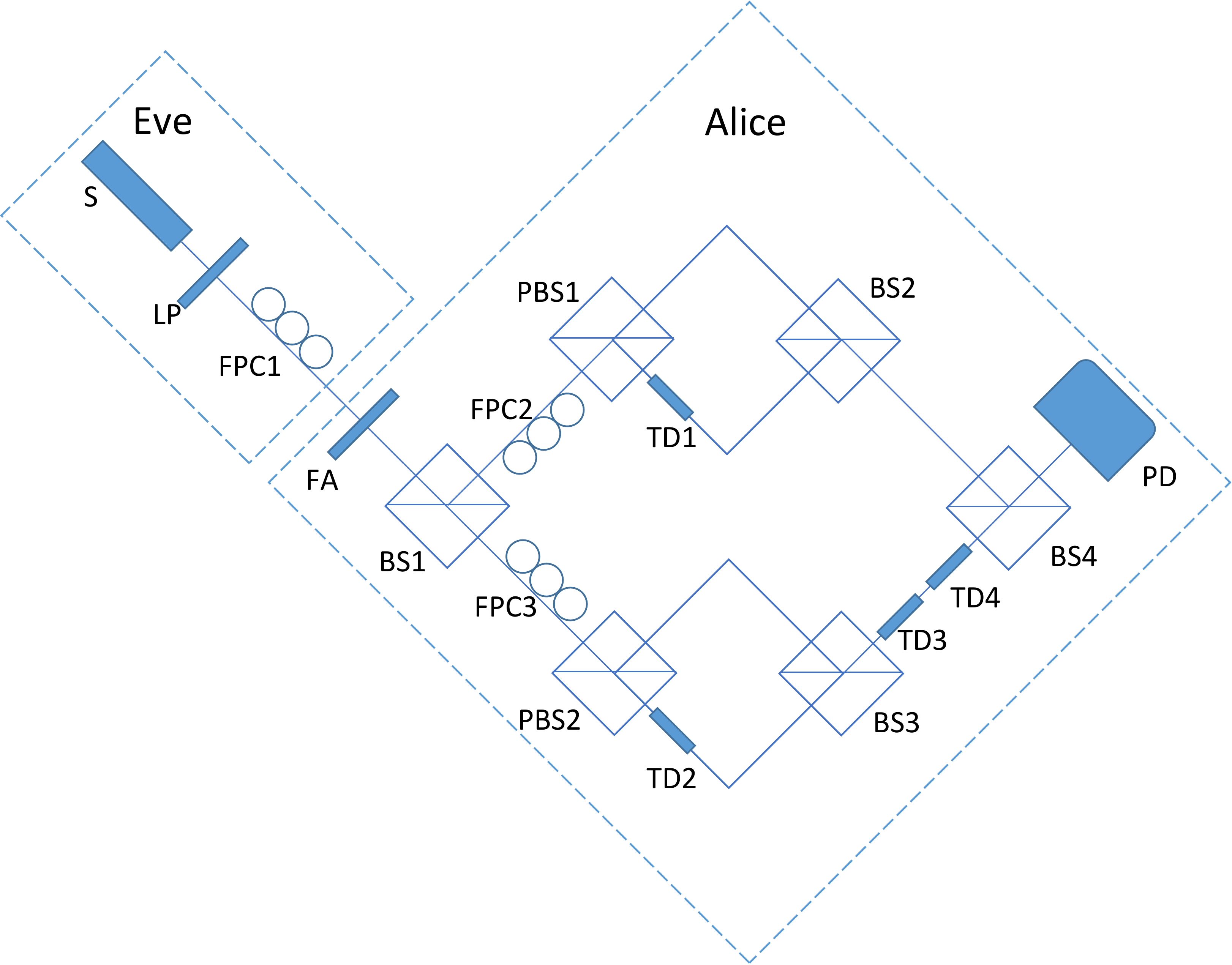}}
\caption{Experiment setup of SIQRNG. S: laser source; LP: linear polarizer; FPC: fiber polarization controller; FA: fiber attenuator; BS: beam splitter; PBS: polarizing beam splitter; TD: time delay implemented with a 12 m fiber; PD: photon detector.}
\label{fig:expsetup}
\end{figure}

On Eve's side, a laser, labeled as $S$, with a wavelength of 850 nm and a repetition rate of 1 MHz is used as a photon source. The power of the laser is adjusted to be one photon number per pulse. Instead of assuming each state to be a qubit system, each pulse that the laser sends is a coherent state of infinite dimensions.  The pulse of the laser is then modulated to $\ket{+}$ polarization by a linear polarizer (LP) and a fiber polarization controller (FPC1). Between the source and the measurement device, we put an fiber attenuator (FA) to simulate different losses in the system.

On Alice's side, first a series of filters need to be applied to ensure the measured optical mode is pure before entering the threshold detectors, as required by the squashing model. For demonstration purpose, we use a single-mode fiber to play the role of a filter. Ideally, frequency and temporal filters should be also added to further purify the optical mode in order to make the photons indistinguishable. For demonstration purpose, a biased beam splitter (BS1) with a ratio of $1:49$ is used to passively choose the $X$ or $Z$ basis. 
Finally, Alice records when the photon detector (PD) clicks. The detector is time-division-multiplexed by adding four time delays TD1 to TD4 (60~ns each) in the optical paths, so that it can simulate four detectors which detect the outcomes of both bases and each bit values. The gate width and the dead time of the detector are 10 ns and 50 ns, respectively.


The phase error rate, as calculated in Eq.~\eqref{eq:eptheta}, is plotted in
Fig.~\ref{fig:errorrate}. The related experimental parameters are listed as follows. The raw key sizes is $N=10^6$; the dark count is $10^{-5}$; the detector efficiency (without FC adaptor) is $45\%$; the misalignment error of the source is 2\%; and the failure probability is $\varepsilon_\theta=2^{-100}$.  The figure shows that the error rate increases as the loss becomes large. This is because the effect of dark counts becomes dominant when the loss is high. Due to statistical fluctuations, the phase error rate increases when the data size shrinks. Note in particular that the phase error rate can go beyond $20\%$ under high losses, which does not yield any key rates in most QKD protocols. Nevertheless, random numbers can still be generated in our SIQRNG scheme.


\begin{figure}[htb]
\centering \resizebox{10cm}{!}{\includegraphics{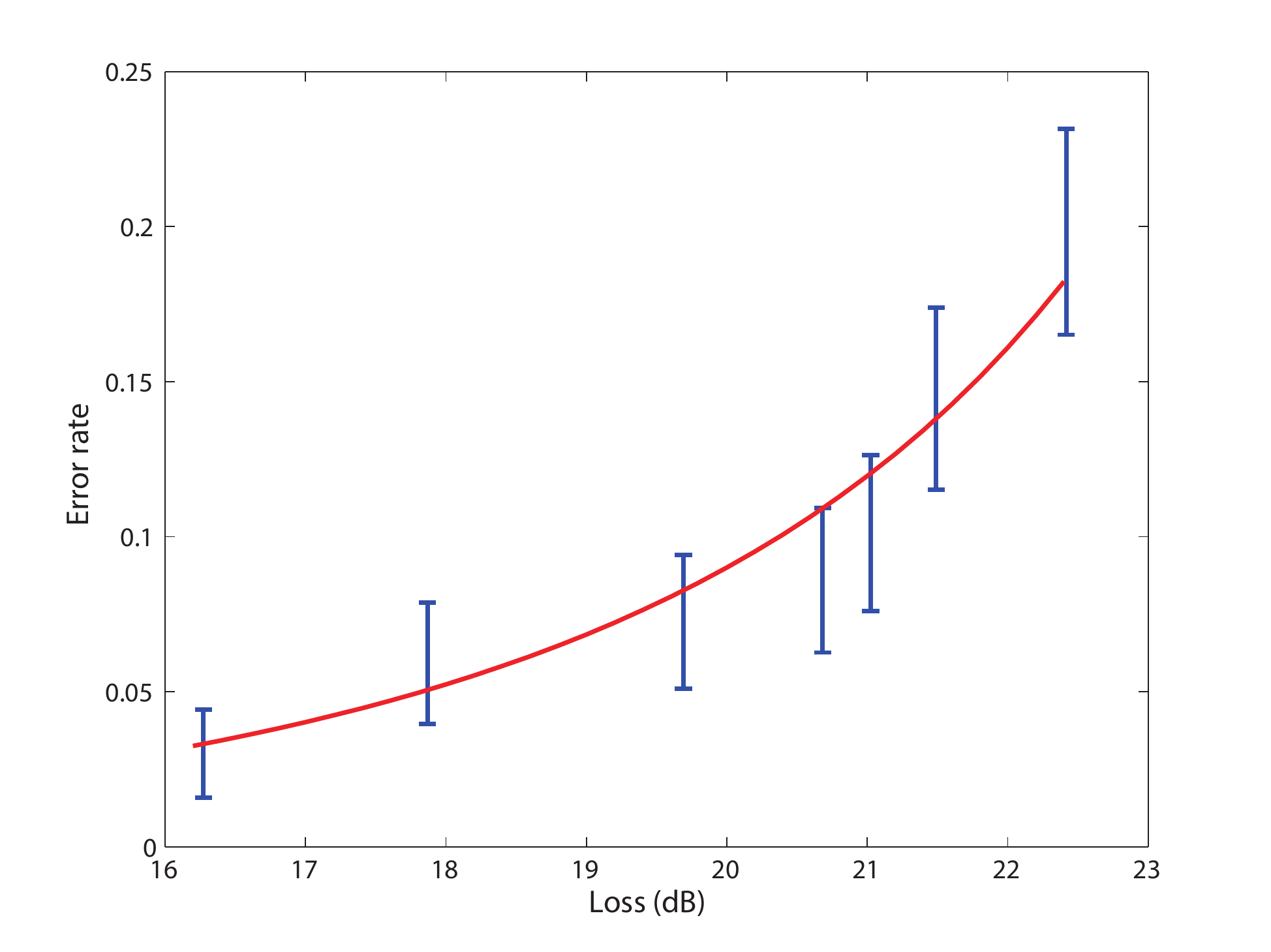}}
\caption{Relation between the phase error rate and the loss. The big error bars are caused by a very conservative estimation of statistical fluctuations and also partially by the fluctuation of experimental parameters for different losses.}
\label{fig:errorrate}
\end{figure}

The relation between the randomness generation rate and the loss is plotted in Fig.~\ref{fig:rngrate}.  It can be seen that the randomness generation rate becomes lower with a larger loss, which is consistent with Fig.~\ref{fig:errorrate}. Under practical detector efficiency, the randomness generation rate still achieves a relatively high rate of $5\times 10^{3}$~bit/s. Note that, the intensity of the source is fixed in our experimental demonstration. In practice, the intensity of the source can be increased to compensate the loss, and actually the maximum randomness generation rate in our scheme is mainly limited by the dead time of the detector. For our detector with a dead time of 50 ns, the maximum randomness generation rate is $1$ bit$/50$ ns=20 Mbps, which requires the source to be a single photon source with a repetition rate of 20 Mbps. For practical implementations with coherent-state sources, the randomness generation rate can reach the order of 2 Mbps after taking account of various errors and finite data size effects.


\begin{figure}[htb]
\centering \resizebox{10cm}{!}{\includegraphics{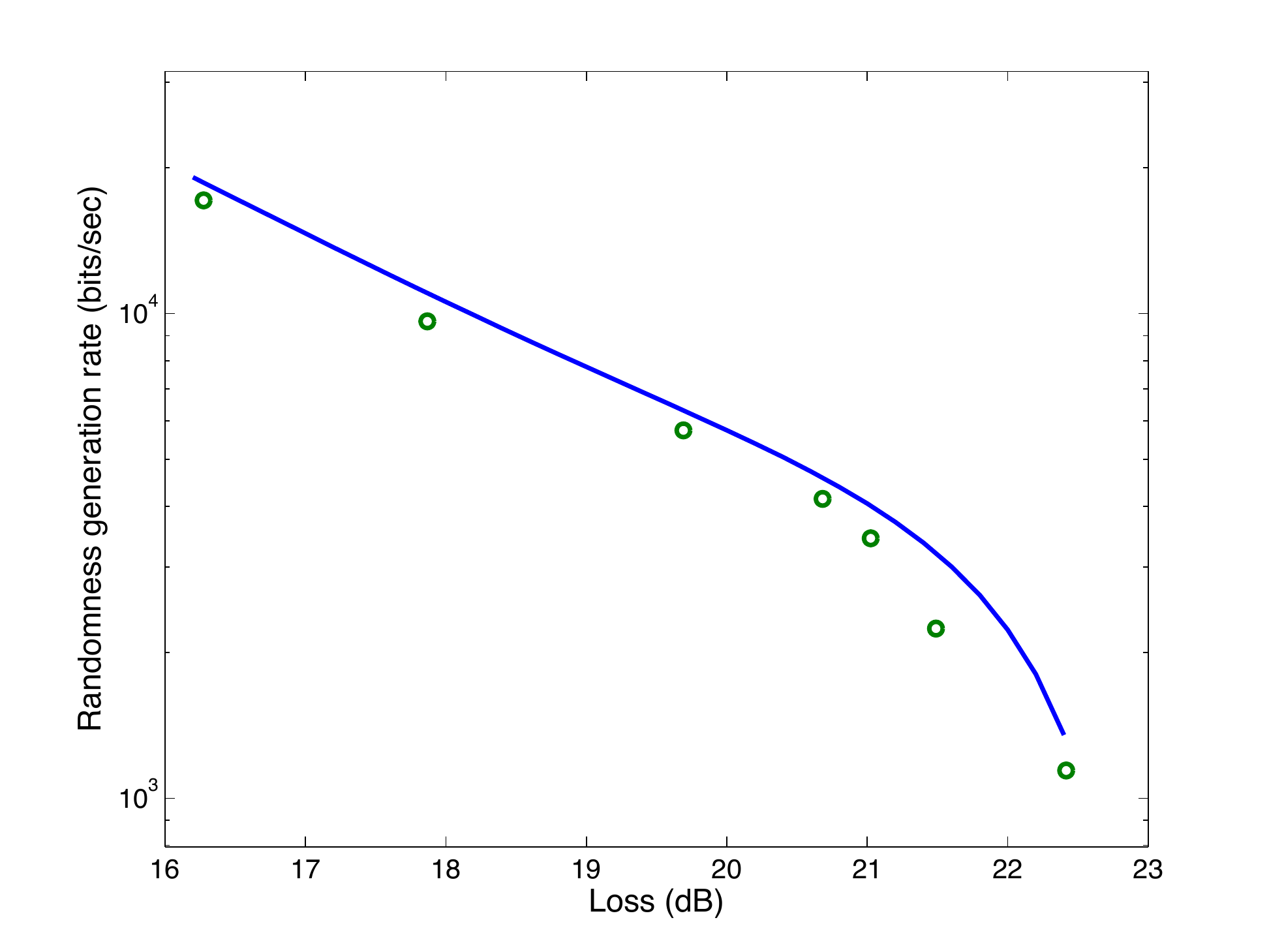}}
\caption{Dependency of randomness generation rate on the loss. The data points on the figure are taken to be the lower bound of the rate, evaluated by random sampling. The security parameter is $\varepsilon_t=2\times 2^{-50}$}
\label{fig:rngrate}
\end{figure}


After obtaining the random bits, we apply the Toeplitz-matrix hashing \cite{Mansour:Toeplitz:93} on the raw data to obtain final random numbers. To test the randomness, we further perform two statistical tests on the output of our SIQRNG, autocorrelation test and the NIST test suite \footnote{See http://csrc.nist.gov/groups/ST/toolkit/rng.}.  The autocorrelation is defined as
\begin{equation}\label{eq:auto}
 R(j) = \frac{ \mathbb{E}[ (X_i-\mu) ( X_{i+j}- \mu )]}{\sigma^2},
\end{equation}
where $j$ is the lag between the samples, $X_i$ is the $i$-th sample bit, $\mu$ and $\sigma$ are the average and the variance of the sample, and $\mathbb{E}$ stands for expectation.
The result of autocorrelation test of raw data and final data is shown in Fig.~\ref{fig:AR}. It can be seen that the autocorrelation is substantially reduced in the final data.
\begin{figure}[htb]
\centering \resizebox{12cm}{!}{\includegraphics{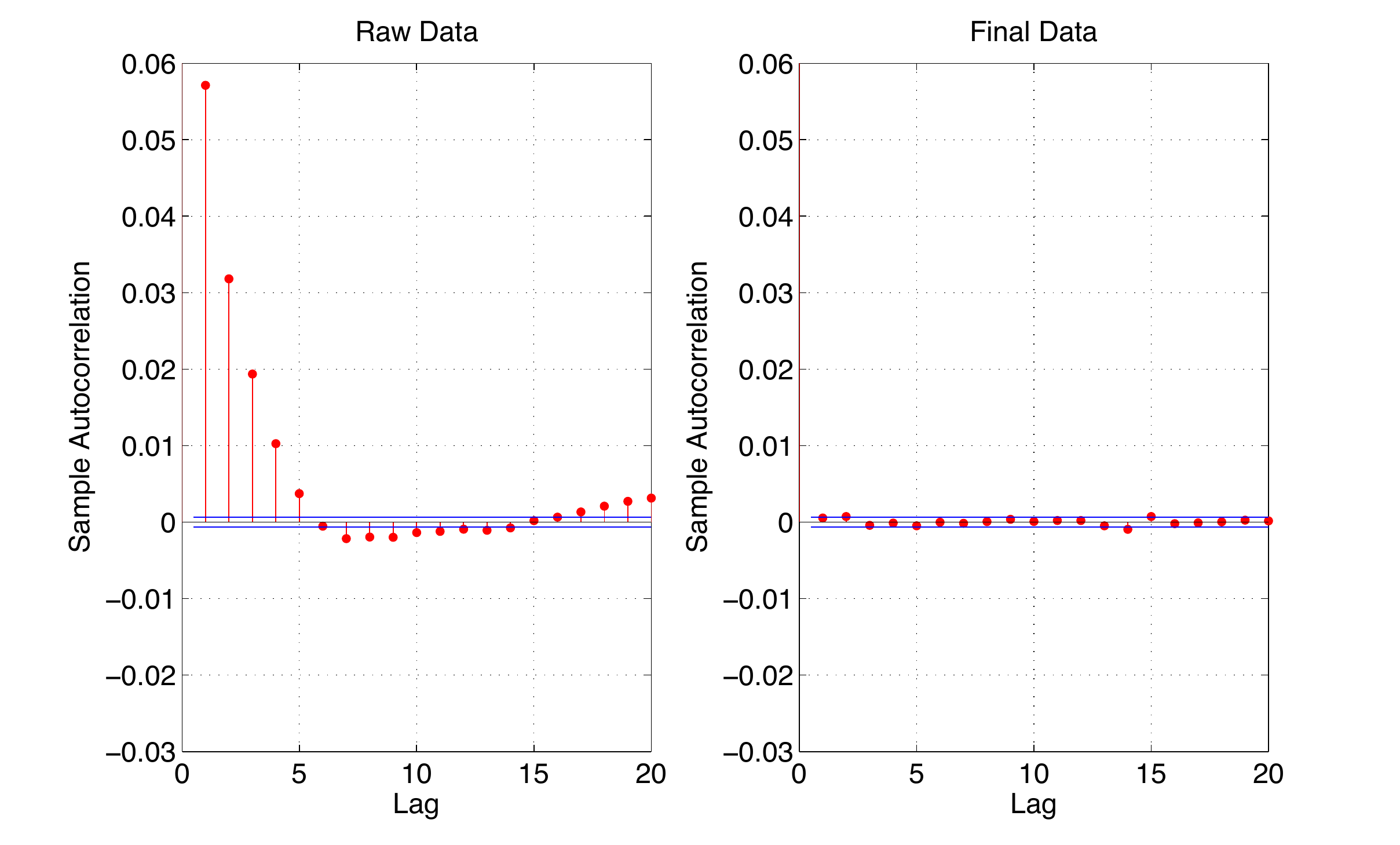}}
\caption{The autocorrelation function of the raw data and the final data. The x-axis is the lag $j$ between the sampled data $X_i$ and $X_{i+j}$, while the y-axis is the autocorrelation $R(j)$ defined in Eq.~\eqref{eq:auto}. Data sizes of both the raw data and the final data are in the order of $10^7$. The autocorrelation of the final data is significantly smaller than the raw data in absolute value.  Due to finite-key-size effect, the autocorrelation cannot be zero even for perfectly random strings.}
\label{fig:AR}
\end{figure}
The result of NIST tests on the final data is shown in Fig.~\ref{fig:p}. We can see that all tests are passed.
\begin{figure}[htb]
\centering
\includegraphics[width=10cm]{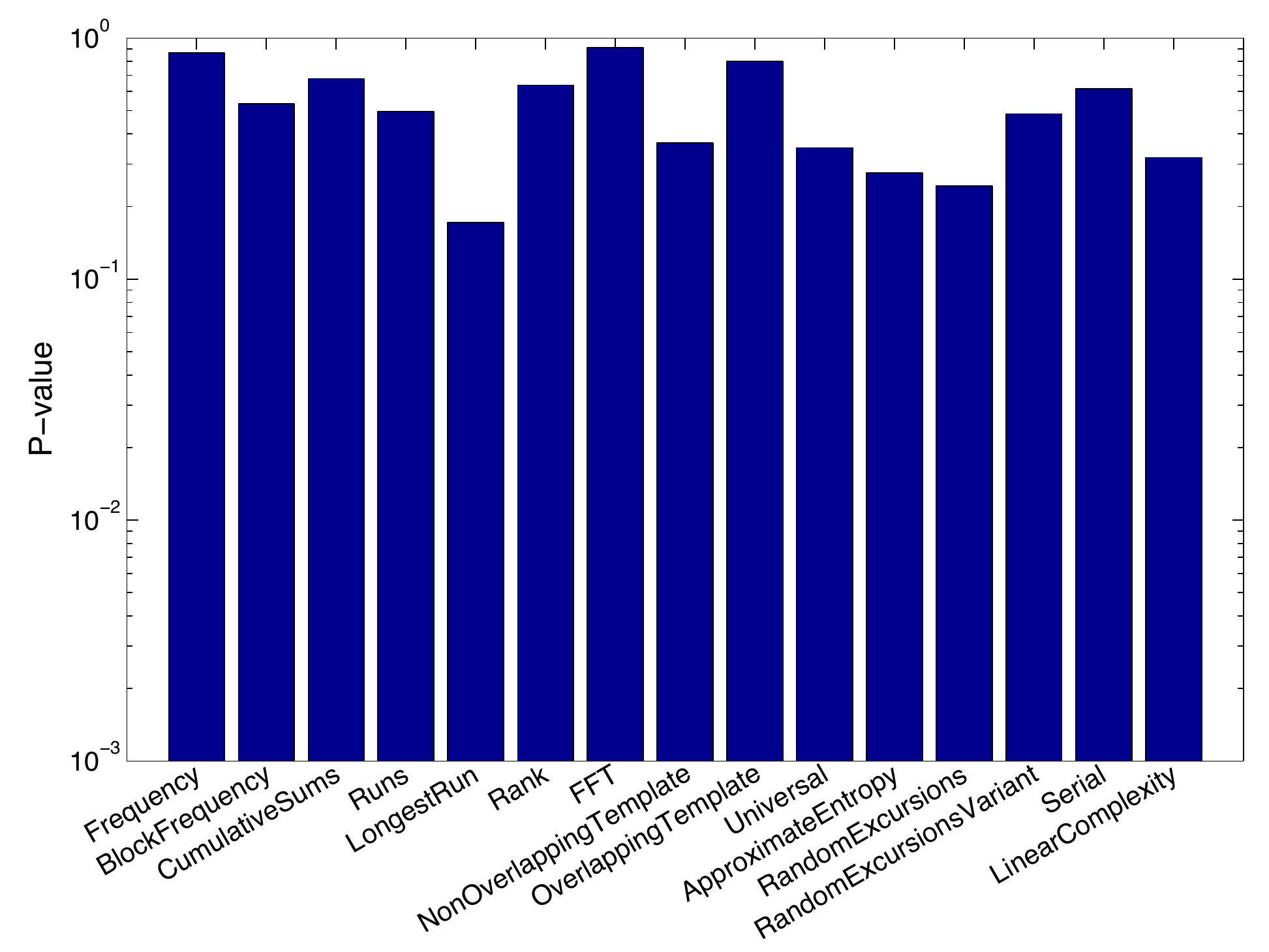}
\caption{The P-value of the statistical tests. The x-axis lists the names of statistical tests in the NIST test suite. The final data size is 91 Mbit, which is extracted from 115 Mbit raw data. To pass each test, the P-value should be at least 0.01 and the proportion of sequences that satisfy $P>0.01$ should be at least 96\%. It can be seen  in the figure that the P-values of all tests are greater than $0.01$.}
\label{fig:p}
\end{figure}

We have proposed a source-independent and loss-tolerant QRNG scheme and its experimentally demonstration in a passive basis choice realization. From an experimental point of view, the beam splitter itself, as a part of the measurement device, may also be uncharacterized. Thus, it would also be interesting to demonstrate our scheme with an active basis choice in the future. In fact, when the source operates properly, the speed of our protocol is comparable to that of a trusted polarization-based QRNG whose frequency is limited only by single photon detectors---approximately 100~Mbps \cite{comandar2014ghz}.

Some current realizations of QRNG experiments could be converted to our SIQRNG protocol. For example, an LED could be used as the source, as  regular QRNG \cite{Mobile2014}. Since the polarizations of an LED light are random, it would be convenient to add a polarizer for the $\ket{+}$ direction to make the source polarized light. Since the detector can work in a gated mode, it does not matter whether the light source is  continuous or pulsed. This shows why the repetition rate is limited only by single photon detectors. Viewed from another angle, such a setup could also be used to test  quantum features of macroscopic sources.



Recently,  a continuous-variable version of the source-independent QRNG is experimentally demonstrated\cite{2015arXiv150907390M}
and achieves a randomness generation rate over 1 Gbps. Moreover, with state-of-the-art devices, it can potentially reach the speed in the order of tens of Gbps, which is similar to the trusted-device QRNGs. Hence, semi-self-testing QRNG is approaching practical regime.

Apart from the protocol based on uncertainty relation, we can also make use of the coherence on the measurement basis to quantify the randomness output. Note that, the SI-QRNG protocols based on uncertainty relation does not maximally exploits the randomness in the source. For instance, suppose the source emits state $1/2(\ket{0}+i\ket{1})$, then we have $H(X) = 1$ and hence $H(Z|E) \ge 0$. In this case, although the measurement outcome on the $Z$ basis is genuinely random, it cannot be revealed by the information on the complementary basis $X$. Instead, the genuine randomness can be extracted if the $Y = \{\ket{+i} = (\ket{0}+i\ket{1})/\sqrt{2},\ket{-i} = (\ket{0}-i\ket{1})/\sqrt{2}\}$ basis is measured. On the other hand, if the $Y$ basis is also measured, we can directly calculate the coherence of the state and the randomness output will be maximized.

\part{Other works}
\chapter{Open timelike curves}
In general relativity, closed timelike curves (CTCs) can break causality with remarkable and unsettling consequences. At the classical level,
they induce causal paradoxes disturbing enough to motivate conjectures that explicitly prevent their existence. At the quantum
level such problems can be resolved through the Deutschian formalism, however this induces radical benefits¡ªfrom cloning
unknown quantum states to solving problems intractable to quantum computers. Instinctively, one expects these benefits to vanish
if causality is respected. This chapter  shows that in harnessing entanglement, we can efficiently solve NP-complete problems and clone
arbitrary quantum states¡ªeven when all time-travelling systems are completely isolated from the past \cite{yuan15OTC}. Thus, the many defining
benefits of Deutschian closed timelike curves can still be harnessed, even when causality is preserved.

\section{Open timelike curves}
\subsection{Causality and CTC}
Causality aligns with our natural sense of reality. We expect there to be a natural chronology to our reality -  two events should not be simultaneous causes for each other. The breaking of causality defies classical logic, resulting in causal paradoxes with no simple solution - the iconic example being the case where a man travels back in time to kill his own grandfather. Thus, physical predictions that break causality face intense scrutiny - often considered to be theoretical artifacts that are likely suppressed once we gain a more complete understanding of reality - motivating various chronology protection conjectures \cite{Hawking92}.

Nevertheless, causality breaking theories are consistent with current scientific knowledge. Closed timelike curves (CTCs) are valid solutions of Einstein's equations in general relativity \cite{Godel49,Morris88,Gott91}. Meanwhile, Deustch demonstrated that in the quantum regime, the resulting causal paradoxes always have self consistent solutions \cite{deutsch1991quantum}. This resolution, however, has radical operational consequences. Many foundational constraints of quantum theory break. Non-orthogonal quantum states can be perfectly distinguished, the uncertainty principle can be violated, and arbitrary unknown quantum states can be cloned to any fixed fidelity \cite{Brun09,brun2013quantum}. In harnessing these effects, many problems thought to be intractable to standard quantum computers now field efficient solutions \cite{brun2003computers,Bacon04,aaronson2009closed,aaronson2005guest}. Though radical, these effects seem somewhat rationalized in the context of requiring broken causality - the sentiment being that they are curiosities that will vanish once causality is imposed.

\begin{figure}[hbt]
\centering
\resizebox{10cm}{!}{\includegraphics[scale=1]{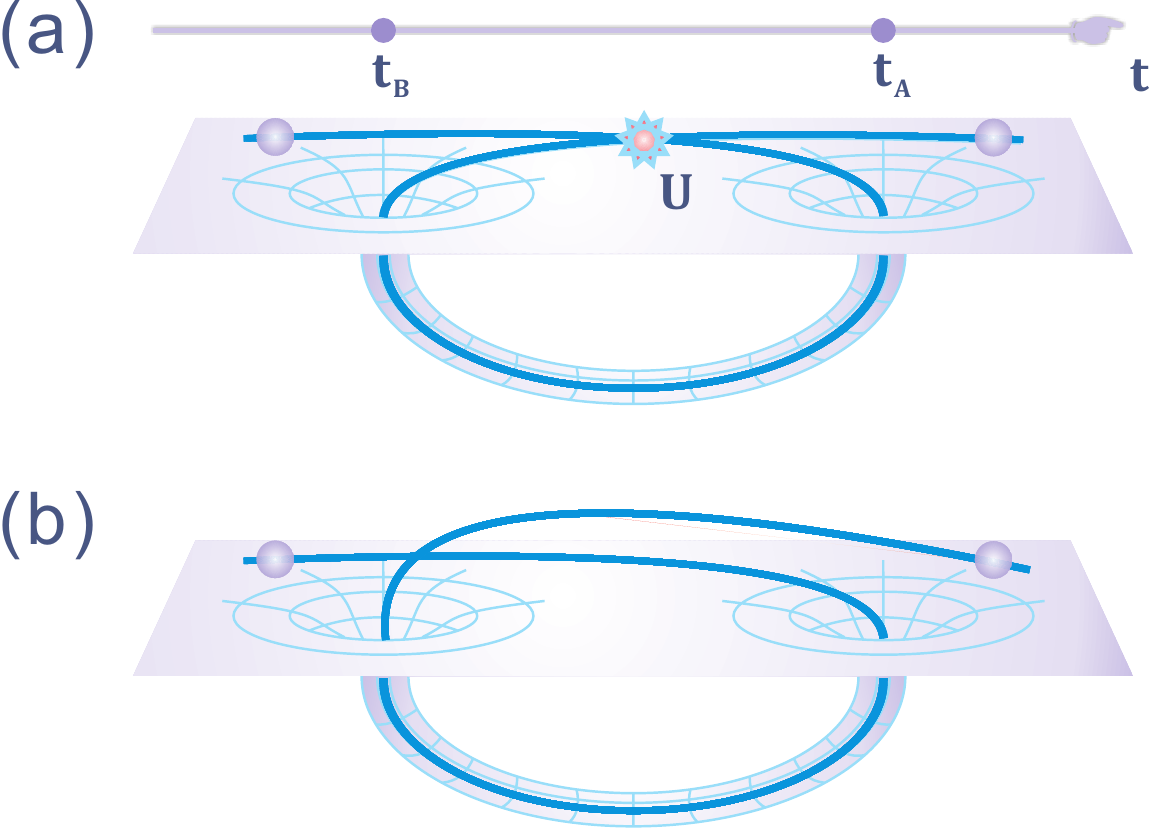}}
 \caption{\textbf{Deutschian timelike curves}. (a) depicts a physical visualization of a CTC, where an object entering one mouth of a wormhole at some point $t_A$ may jump to a prior time $t_B$ (with respect to an chronology respecting observer) and interact with its past self via some unitary $U$. (b) In the special case where no interaction occurs, we obtain an open timelike curve. This naturally occurs, for example, in instances where the wormhole mouths are spatially separated.}\label{Fig:CTC}
\end{figure}

What happens, however, if causality is not strictly broken? In this context, Pienaar et.al introduced open timelike curves \cite{Pienaar13} (OTCs). Consider a particle that travels back in time with respect to a chronology respecting observer, but is completely isolated from anything that can affect its own causal past during the time-traveling process (See Fig. \ref{Fig:CTC}). While the time-traveling particle has the potential to break causality, its complete isolation ensures that causality never breaks. Nevertheless, such OTCs can violate uncertainty principles between position and momentum. This opens a remarkable possibility - could the many other radical effects of CTCs stand independent from the breaking of causality?

Here, we demonstrate that OTCs are remarkably powerful, and can replicate many defining operational benefits of CTCs. In sending a particle back in time - even when it interacts with nothing in the past - we can clone arbitrary quantum states to any fixed accuracy, and thus violate any uncertainty principle. Meanwhile, they also grant quantum processors additional computational power, allowing efficient solution of NP-complete problems. Our results hint that the remarkable power of Deustchian CTCs may survive the censorship of chronology protection. This drastically improves the potential of harnessing such power via alternative effects - such as certain models of gravitational time dilation \cite{Pienaar13}. Thus, we open the possibility of testing the many radical protocols that harness CTCs in significantly less controversial settings.

\subsection{OTC and CTC}

In general relativity, causality can be violated due to the presence of spacetime wormholes that facilitate closed timelike curves (see Fig. \ref{Fig:CTC}). This allows a physical system $A$ to travel into its own causal past, and interact with its past self via some unitary $U$. The Deutschian model resolves potential paradoxes by enforcing temporal self-consistency \cite{deutsch1991quantum, ralph2010information}, i.e.,
\begin{equation}\label{eq:Deutsch1}
  \rho_{\mathrm{CTC}} = \mathrm{Tr}_{\neq A}\left[U(\rhoin\otimes\rho_{\mathrm{CTC}})U^\dag\right],
\end{equation}
where $\rhoin$ denotes the initial state of the system, $\rhoctc$ is the state it evolves to at the point of wormhole traversal and $\mathrm{Tr}_{\neq A}$ represents tracing over all systems other than $A$. Given a solution for $\rho_{\mathrm{CTC}}$, the final output of the process is given by
\begin{equation}\label{eq:Deutsch2}
\rhoout = \mathrm{Tr}_A\left[U(\rhoin\otimes\rho_{\mathrm{CTC}})U^\dag\right].
\end{equation}The many radical effects of CTCs rely on using specific self-interactions $U$ to break causality in different ways \cite{Brun09,brun2013quantum,brun2003computers,Bacon04,aaronson2009closed}.

\begin{figure}[hbt]
\centering
\resizebox{8cm}{!}{\includegraphics[scale=1]{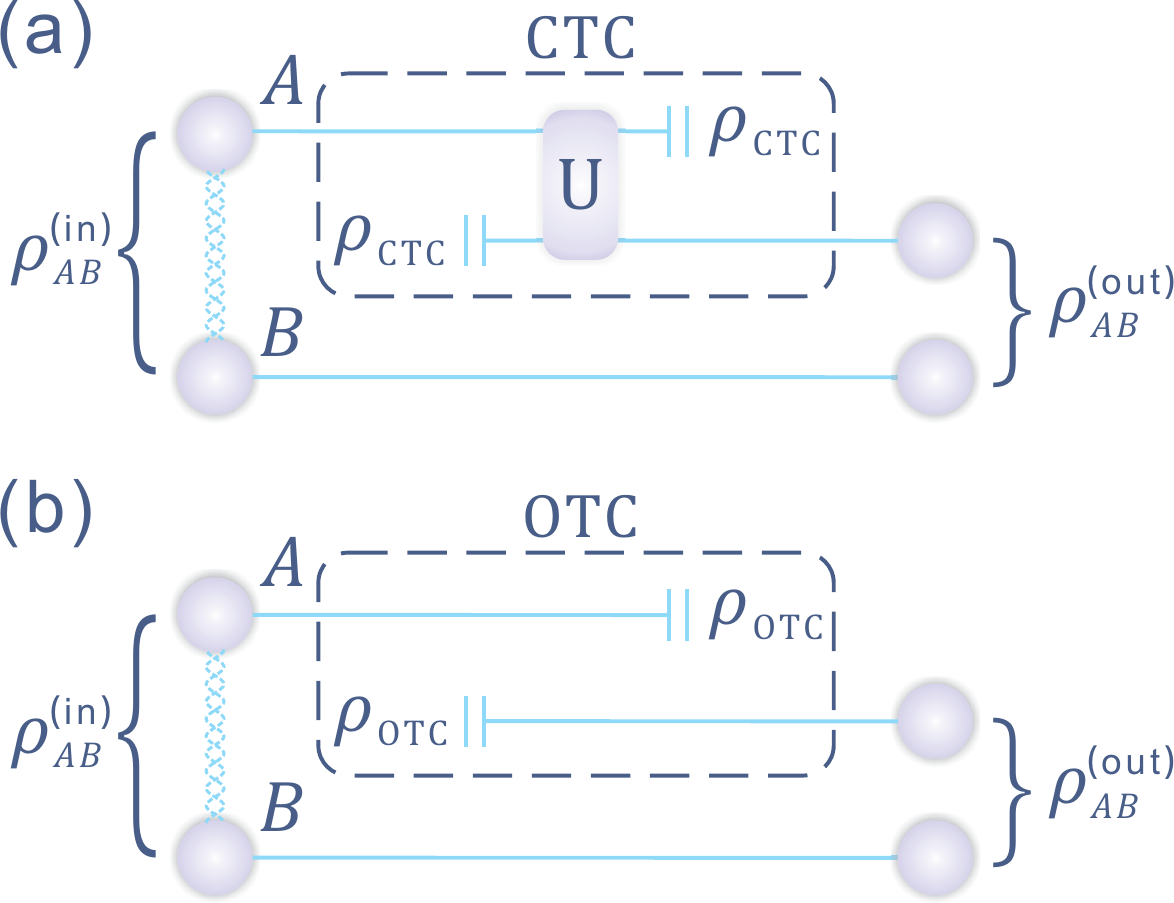}}
  \caption{\textbf{CTCs and OTCs in presence of ancilla}. $A$ represents the system to be sent through the space-time wormhole, and $B$ some chronology respecting  system initially correlated with $A$. (a) In general CTCs, temporal self-consistency demands that $\rho_{\mathrm{CTC}}$ satisfies $\rho_{\mathrm{CTC}} = \mathrm{Tr}_{\neq A}[U(\rhoin_{AB}\otimes\rho_{\mathrm{CTC}})U^\dag]$. (b) In the case of OTCs, this implies that system $A$ has state $\rho_{\mathrm{OTC}}= \mathrm{Tr}_{\neq A}\left[\rhoin_{AB}\otimes\rho_{\mathrm{OTC}}\right]=\rho_{A}$ after application of the protocol.
  }\label{Fig:OTC}
\end{figure}

Note that while the above analysis does not assume $\rhoin$ is pure, it \emph{only} applies to mixed inputs if $\rhoin$ represents one partition of a larger composite system that is pure. In the scenario where an input $\ket{\phi_k}$ is prepared with probability $p_k$, the dynamics of the CTC on each $\ket{\phi_k}$ must be analyzed separately \cite{deutsch1991quantum}. This is due to non-linearity, which implies different unravellings of the density operator yield differing outputs.

In OTCs, causality is preserved. The unitary $U$ is the identity - such that the time-travelling system does not interact with its causal past. Any observer in the frame of reference of $A$ can assign a valid chronology to all the events they witness. Meanwhile, to any outside observer, all events involving interactions with $A$ will respect causality. From an operational standpoint, there is no breaking of causality. If all information were classical, this entire procedure would only have the effect of desynchronizing $A$'s clock with that of an observer $B$.

Non-trivial effects, however, emerge when we consider quantum ancilla. Suppose we have access to a bipartite system $AB$ in state $\rho_{AB}$, where only one bipartition is sent through the OTC (see Fig.~\ref{Fig:OTC}b). The self-consistency relations imply
\begin{equation}\label{eq:decoherencor}
  \rhoout_{AB} = \mathrm{Tr}_{\neq A}[\rhoin_{AB}] \otimes \mathrm{Tr}_{A} [\rhoin_{AB}] = \rho_A \otimes \rho_B
\end{equation}
Thus, the OTC acts as a \emph{universal decorrelator} on $A$ - in sending a system $A$ though an OTC, we erase all quantum correlations between $A$ and the rest of the universe (and in particular, $B$). The resulting state, $\rho_A \otimes \rho_B$ fields identical local statistics with respect to the input $\rho_{AB}$, but none of its bipartite correlations. While this operation appears similar to trivial decoherence, it is non-linear, and shown to be impossible to synthesize with standard quantum dynamics \cite{Terno99}.

This effect is associated with the monogamy of entanglement \cite{ralph2010information} - a particle and its past self cannot be simultaneously entangled with the same external ancilla. While OTCs produce non-trivial dynamics when the input appears completely classical (e.g., when $\rhoin_{AB} = (\ket{00}\bra{00} + \ket{11}\bra{11})/2$), it applies only for mixed inputs if this mixedness is due to entanglement with some other system $C$. If we input $\ket{00}$ and $\ket{11}$ with equiprobability, then the dynamics of each input must be analyzed separately, and the OTC will have no effect.

\section{Quantum information with OTC}
\subsection{OTC enhanced measurement}
We first introduce \emph{OTC enhanced measurement}, a procedure that harnesses OTCs to measure an arbitrary observable $\hat{O}$ to any fixed precision. Specifically, given an unknown qudit ($d$ dimensional quantum system) in state $\rho$, we can determine $\langle \hat{O}\rangle = \mathrm{Tr}[\hat{O}\rho]$ to any desired accuracy $\delta > 0$ with negligible failure probability. This protocol functions as a building block for more sophisticated applications of OTCs, such as the efficient solution of NP-complete problems and cloning of unknown quantum states.

The protocol is illustrated in Fig.~\ref{Fig:measure}. Let $\ket{j}: j = 0,1,\dots,d-1$ denote a basis that diagonalizes $\hat{O}$. On this basis, we introduce the two qudit controlled addition operator, $C_+\ket{i}\ket{j} = \ket{i}\ket{j+i}$, where addition is done modulo $d$. We then
\begin{enumerate}
  \item Prepare $N$ identical ancillary states in an eigenstates of $\hat{O}$, say $\ket{0}$.
  \item Apply the $C_+$ operations $N$ times, each controlled on $\rho$ and targeting a fresh ancilla state. This correlates $\rho$ with each of the $N$ ancillaries.
  \item Pass each of the ancillaries through an OTC to destroy all correlations in this $N+1$-partite system.
\end{enumerate}

\begin{figure}[hbt]
\centering
\resizebox{8cm}{!}{\includegraphics[scale=1]{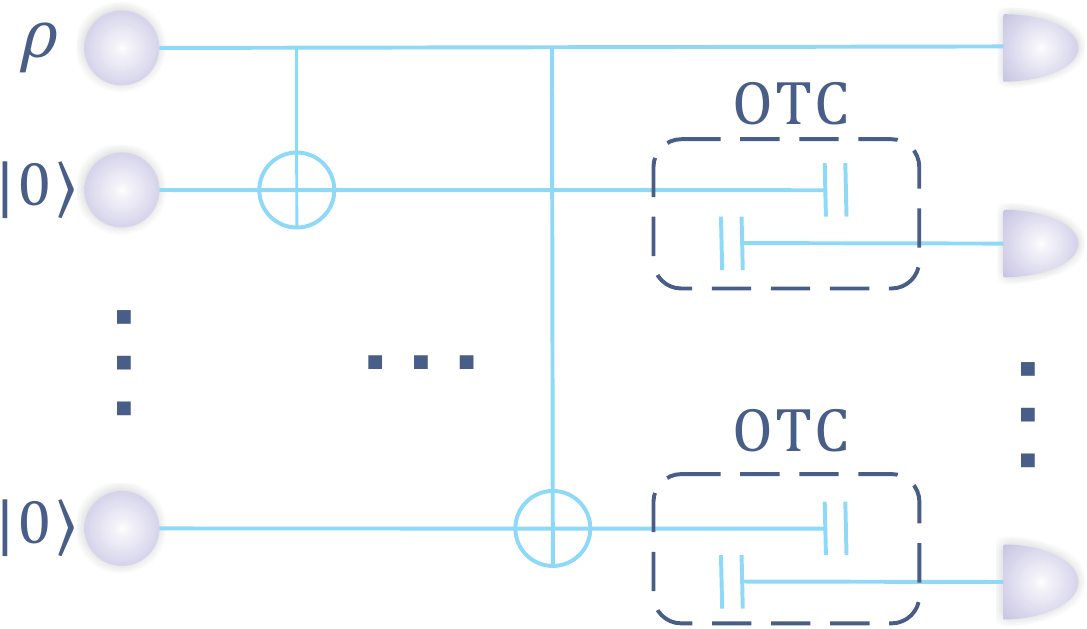}}
  \caption{\textbf{Quantum circuit of OTC enhanced measurement}. The protocol first introduces $N$ ancilla qudits, all of which are initialized in the state $\rho_E = \ket{0}\bra{0}$, where $\ket{0}$ is an eigenstate of $\hat{O}$. A sequence of $C_+$ gates then perfectly correlates each ancilla with $\rho$ with respect to $\hat{O}$ basis. The erasure of these correlations via OTCs, followed by $\hat{O}$ measurements on each individual qudit, allows determination of $\mathrm{Tr}[\hat{O}\rho]$ to a standard error that scales inversely with $N^2$.
  }\label{Fig:measure}
\end{figure}

This results in $N+1$ uncorrelated qudits, each in state $\rho_{\mathrm{diag}} = \sum_{i = 1}^d\rho_{ii}\ket{i}\bra{i}$, where $\rho_{ii}$ are the diagonal elements of $\rho$ in the $\hat{O}$ basis. Thus, each qudit exhibits identical statistics to $\rho$ when measured in the $\hat{O}$ basis. In taking the mean of these measurements, we obtain an estimate for $\langle \hat{O}\rangle$. By the central limit theorem, the error of our estimate scales linearly with $1/\sqrt{N}$. In particular, provided the eigenvalues of $\hat{O}$ are bounded, Hoeffding's bound implies we can estimate $\hat{O}$ to any desired accuracy $\delta$ and error rate $\epsilon$ using $O[1/\delta^2\log(1/\epsilon)]$ OTCs (see methods for details).

\subsubsection{Scaling Analysis}
Execution of the OTC enhanced measurement with $N$ ancillaries (and therefore $N$ uses of the OTC) to estimate $\langle \hat{O}\rangle$ gives an output $O_{\mathrm{est}} = \sum_k O_k/(N+1)$. We define the measurement as being successful if the estimate achieves a desired accuracy of $\delta$ (i.e., $|O_{\mathrm{est}} - \hat{O}| < \delta$). Application of Hoeffding's inequality  \cite{hoeffding1963probability} gives failure probability $p_{f}$ that obeys
\begin{equation}\label{}
p_{f} \leq 2\exp\left[\frac{-2(N+1)\delta^2}{(O_\mathrm{max} - O_\mathrm{min})^2}\right].
\end{equation}
Here, $O_\mathrm{max}$ and $O_\mathrm{min}$ are the respective maximum and minimum eigenvalues of $\hat{O}$. Therefore
\begin{equation}\label{}
  N > \frac{(O_\mathrm{max} - O_\mathrm{min})^2}{2\delta^2}\log\frac{2}{\epsilon},
\end{equation}
OTC applications ensures a failure probability of no more than $\epsilon$. Provided $\hat{O}$ is bounded, this scales as $O[1/\delta^2 \log(1/\epsilon)]$.

In OTC assisted cloning, we need to make $d^2$ informationally complete measurements, each to a desired accuracy $\delta > 0$ with negligible failure probability $\epsilon > 0$. Recall this is achieved via a $1 \rightarrow d^2$ universal cloner, whose imperfect copies are to be decorrelated via the use of OTCs (Fig. \ref{Fig:scheme}). For each of the $d^2$ copies, we apply an OTC enhanced measurement. To ensure this measurement is within accuracy $\delta$, an extra $O(d^2)$ overhead is required to compensate for the noise within the imperfect copies. The total number of OTCs required is then of order
\begin{equation}\label{}
  N > O\left[d^4\left(\frac{1}{2\delta^2}\log\frac{2}{\epsilon}\right)\right],
\end{equation}
where $O_\mathrm{max} - O_\mathrm{min}$ is equal to 1 for members of the information complete basis.

\subsection{Solving NP-complete problems}
We take inspiration from Bacon \cite{Bacon04}, who devised an efficient algorithm to solve the boolean satisfaction problem - a known NP-complete problem - using CTCs.

\subsubsection{Bacon's protocol}
Here we outline explicitly how a non-linear map that takes $\rho(n_z)$ to $\rho(n_z^2)$ allows the efficient solution of NP-complete problems. Specifically we study the satisfaction problem: \emph{Given a Boolean function $f:\{0,1\}^n\rightarrow\{0,1\}$, specified in conjunctive normal form, does there exist a satisfying assignment $(\exists b|f(b)=1)$?} This problem is known to be NP-complete.

Bacon \cite{Bacon04} showed that this problem can be efficiently solved if when given an input qubit $\rho = (I + \vec{n}\cdot\vec{\sigma})/2$, we can synthesize a quantum gate $S$ such that $S(\rho) = \frac{1}{2}\left(I + n_z^2\sigma_z\right)$. Here $\vec{n}$ denotes the Bloch sphere vector, and $\vec{\sigma}$ is a 3 component vector of Pauli matrices. In Fig. \ref{Fig:NP} we demonstrated how this gate can be synthesized using OTCs. With this established, the satisfaction problem is efficiently solved as follows:
\begin{enumerate}
  \item Prepare $n$ ancillary qubits in the state $1/\sqrt{2^n}\sum_{i=0}^{2^n-1}\ket{i}$ and a target qubit in state $\ket{0}$.
  \item Apply the unitary
      \begin{equation}\label{}
        U_f = \sum_{i=0}^{2^n-1}\ket{i}\bra{i}\otimes\sigma_x^{f(i)},
      \end{equation}
      on this system (with the last qubit representing the target). Tracing out the ancillary qubits leaves the target in
      \begin{equation}\label{}
        \rho = \frac{1}{2}\left[1+\left(1-\frac{s}{2^{n-1}}\right)\sigma_z\right],
      \end{equation}
      where $s$ is the number of $x$ satisfying $f(x)=1$.
  \item  Apply $S$ to the target via the use of OTCs (See Fig. \ref{Fig:NP}). Repeat this step $p$ times to get
    \begin{equation}\label{}
        \rho_p = \frac{1}{2}\left[1+\left(1-\frac{s}{2^{n-1}}\right)^{2^p}\sigma_z\right],
      \end{equation}

\end{enumerate}

Notice that, we could easily check the case of  $s=2^n$. Thus we only need to distinguish between $s=0$ and $0<s<2^n$.
With the limit of $p\rightarrow\infty$, the two output states corresponding to the cases of $s=0$ and $0<s<2^n$ are $\rho_p= \ket{0}\bra{0}$ and $\rho_p \rightarrow I/2$, respectively. By performing measurement in the $\sigma_z$ basis, one can distinguish the two types of output states $\ket{0}\bra{0}$ and $I/2$, that is, the case of  $s=0$ and $0<s<2^n$, with failure probability being $1/2$. By repeating these steps more times, say, $q$, the failing probability exponentially decays. For finite $p$ and $q$ that are polynomial in $n$, the probability of failure is given by  \cite{Bacon04}
\begin{equation}\label{}
  P_{fail}=\frac{1}{2^q}\left[1+\left(1-\frac{s}{2^{n-1}}\right)^{2^p}\right]^q.
\end{equation}

\subsubsection{Solving NP-complete problems with OTC}
We modify this algorithm to preserve causality - without losing efficiency. In the causality breaking algorithm, the key role of CTCs is to implement the non-linear map $S$ that maps an input qubit in state $\rho(n_z)$ to an output state $\rho(n_z^2)$, where $\rho(n_z) =  \frac{1}{2}\left(I + n_z \sigma_z\right)$ and $\sigma_z = \ket{0}\bra{0} -  \ket{1}\bra{1}$ denotes the Pauli $Z$ matrix (see methods for details).

\begin{figure}[hbt]
\centering
\resizebox{8cm}{!}{\includegraphics[scale=1]{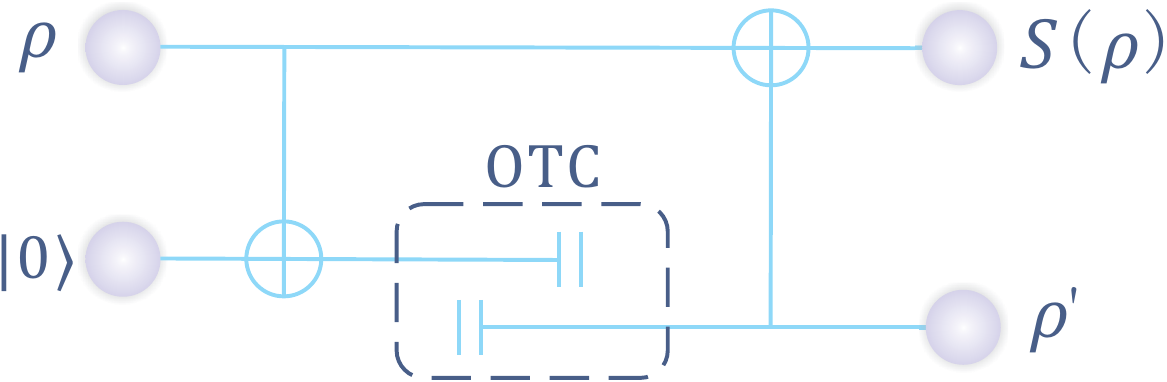}}
  \caption{\textbf{Solving NP-complete problems with OTCs.} The key non-linear gate $S$, that takes $\rho(n_z)$ to $\rho(n_z^2)$, can be implemented via open timelike curves. This is achieved by the use of a single OTC, applied between two successive $C_+$ gates.}\label{Fig:NP}
\end{figure}

This non-linear map can be replicated without breaking causality (See Fig.~\ref{Fig:NP}). Consider a special case of OTC enhanced measurement, with $\sigma_z$ as the observable of interest and a single ancilla. For the input qubit $\rho$ with matrix elements $\rho_{ij}$, application of the enhanced measurement protocol outputs two uncorrelated qubits, each in state $\rho_{\mathrm{diag}} = \rho_{00} \ket{0}\bra{0} + \rho_{11} \ket{1}\bra{1}$. Instead of measuring each in $\sigma_z$ directly, we apply a further $C_+$ gate controlled on the ancilla. After discarding the ancilla, the input qubit is now transformed to $S(\rho)$ as required.

In generating $S(\rho)$ using only OTCs, we can translate Bacon's algorithm into one that does not break causality. We note that as each call of $S(\rho)$ only takes one OTC, the translation from CTCs to OTCs incurs no overhead on the number of times a particle needs to be sent through a spacetime wormhole. Thus, for the purpose of solving NP-complete problems, an OTC, together with one bit of entanglement, is at least as a powerful as a CTC.

\subsection{Cloning with OTCs}
Given an unknown input $\rho$, OTCs allow us to generate an unlimited number of clones to arbitrary fidelity. Our approach harnesses OTC enhanced measurements as a subroutine, which allows us to accurately determine $\mathrm{Tr}[M_i \rho]$, for any observable $M_i$. First, observe that this remains possible even if we are supplied with
\begin{equation}\label{eq:cloner}
   \rho' = s\rho + \frac{1-s}{d}I;
\end{equation}
a very noisy version of $\rho$. Here $I$ is the $d$-dimensional identity matrix, and $s$ is some fixed parameter such that $0 < s < 1$.

This observation, together with imperfect quantum cloners, form the basis of OTC enhanced cloning (Fig.~\ref{Fig:scheme}). In conventional quantum theory, an unknown quantum state $\rho$ can be cloned if we are given sufficiently many copies to perform accurate tomography \cite{qubittomo_2002}.  One way to do this, is to use a set of $O(d^2)$ informationally complete measurements $\{M_i\}$, whose expectation values $\mathrm{Tr}\left[M_i \rho\right]$ has a one-to-one correspondence with the classical matrix description of $\rho$. Given only a single copy of $\rho$, this option is no longer valid. Recently, Brun et.al demonstrated that close timelike curves circumvent this restriction, and allows the estimation of each $\langle M_i\rangle$ to any desired accuracy \cite{brun2013quantum}.

OTC enhancement measurements can replicate this effect while preserving causality. We use standard methods to construct $O(d^2)$ imperfect clones in the form of Eq. \ref{eq:cloner}, where $s$ scales as $1/d$ for an optimal cloner  \cite{Werner98}. Each clone is passed through an OTC to remove all entanglement between clones. An OTC enhanced measurement is then performed on each clone with respect to a different $M_i$. The outcomes of these measurements determine the density matrix of $\rho$. In methods, we show that by using $O(d^4/\delta_c^2\log{1/\epsilon_c})$ OTCs, we can ensure that each $\langle M_i\rangle$ is obtained to an accuracy of $\delta_c$ with failure probability $\epsilon_c$.

\begin{figure}[hbt]
\centering
\resizebox{8cm}{!}{\includegraphics[scale=1]{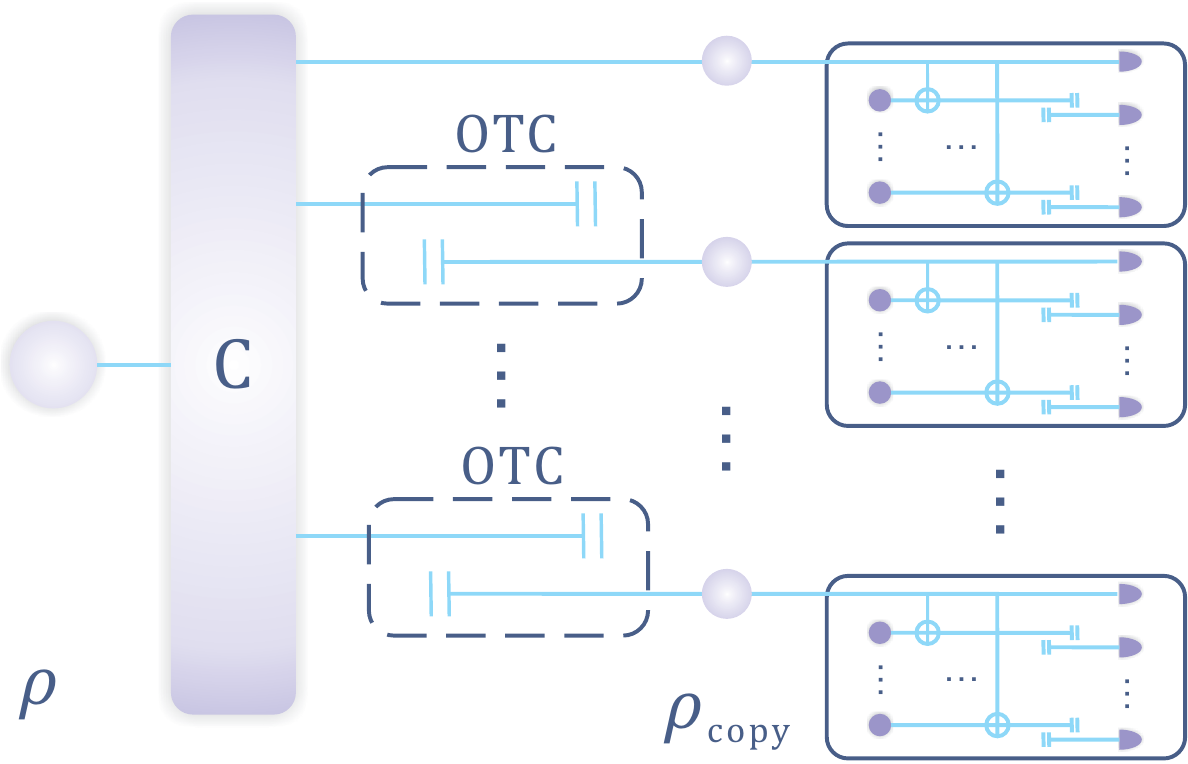}}
  \caption{\textbf{OTC Assisted Cloning.}  An arbitrary qudit $\rho$ can be cloned to any desired fidelity. The process involves (i) application of a standard quantum cloner $C$ to generate $O(d^2)$ imperfect copies, and (ii) use of OTC enhanced measurements to measure different observables $M_i$ on each imperfect copy. We can choose $M_i$ to be informationally complete, and OTCs ensure that we can determine $\mathrm{Tr}\left[{M_i \rho}\right]$ to any desired precision. Thus this protocol can yield (to any fixed precision) the classical description of $\rho$.}\label{Fig:scheme}
\end{figure}

\subsubsection{A Simple Example}
We illustrate these ideas by cloning a qubit. Here, the Pauli operators $\sigma_k, k = x,y,z$ is informationally complete - any $\rho$ is uniquely defined by the expectation values $n_k = \mathrm{Tr}[\sigma_k \rho]$. To determine each $n_k$, we first apply a universal 1-to-3 quantum cloner to obtain three imperfect clone of $\rho$, each in state $\rho' = (I + s\vec{n}\cdot\vec{\sigma})/2$ with $s = 5/9$  \cite{Scarani05}. These copies can be made independent via OTCs.

An OTC enhanced measurement of $\sigma_z$ is then performed on one such imperfect clone. We initialize $N$ ancilla qubits in state $\ket{0}$, and apply a CNOT gate between each ancilla and the clone (with the clone as the control qubit). In erasing the resulting correlations by sending each clone through an OTC, we obtain $N+1$ qubits, each in the state $(I + s n_z \sigma_z)/2$. Provided $N$ is sufficiently large, measurement of these qubits allows $n_z$ to be determined to any desired accuracy with negligible error. Repetition of this process with $\sigma_x$ and $\sigma_y$ on the two remaining imperfect clones then yields complete information about $\rho$.



Our result highlights the intricate interplay between quantum theory and general relativity. If all physical systems have a well defined local reality, open timelike curves would have no operational effect. Our protocol thus fields no classical explanation. In the quantum regime, however, entanglement exists. While the local properties of a physical system are unaffected by open timelike curves, their correlations with other chronology respecting systems are complete erased. In each of our protocols, this effect played a central role, allowing us to replicate the many defining benefits of close timelike curves. This remarkable success propels us to conjecture whether entanglement assisted open timelikes curves are operationally equivalent to their causality breaking counterparts. Could one, for example, derive a map that takes any quantum circuit with CTCs, and engineer it in a way that does not break causality?

Preserving causality has significant benefits. Breaking causality is likely to be non-trivial, and opportunities to do so are negligible in the foreseeable future. This makes it unlikely for us to directly test the predictions of Deutschian CTCs. The preservation of causality in OTCs, however, suggest that its non-linear effects may be synthesized using alternative means. For instance, from the perspective of a chronology respecting observer, a particle sent through an OTC exhibits nothing more than time delay. Thus, in order to reconcile quantum field theory with non-hyperbolic space-times, gravitational time-dilation has been conjectured to share similar operational effects as OTCs  \cite{Ralph09}. If true, our protocols suggest the exotic benefits of quantum processing in the general relativistic regime can be tested much sooner than previously expected.



\chapter{Quantum theory from axioms}
Quantum Information provided a new angle on the foundations of Quantum Mechanics, where the
emphasis is placed on operational tasks pertaining information-processing and computation. In this
spirit, several authors have proposed that the mathematical structure of Quantum Theory could
(and should) be rebuilt from purely information-theoretic principles.
This chapter  reviews the particular
route proposed by D'Ariano, Perinotti, and Chiribella \cite{chiribella10,chiribella11,chiribella2013quantum} and probes its application in nonlocality and contextuality \cite{chiribella2014measurement}.

\section{Axiomization of quantum theory}
Quantum mechanics has been one of the greatest scientific breakthroughs of the 20th century, and at the beginning of the 21st it still provides the most accurate predictions about the microscopic world and a key to new exciting discoveries. A young branch of quantum mechanics is quantum information science.        Over the past three decades, this new field brought to light a wealth of operational consequences of the mathematical structure of quantum theory: no-cloning \cite{wootters1982single}, quantum teleportation \cite{tele}, dense coding \cite{wiesner},  quantum key distribution \cite{bb84,Ekert1991}, and  quantum algorithms \cite{Grover96,Shor97} are just a few examples showing that quantum theory entails a powerful and totally new model of information processing.

Which are the conceptual \emph{ingredients} of this power? What makes the quantum model  special with respect to the classical one or to other more exotic models that we can conceive? Stimulating new questions are  the best route to answering older, long-standing questions.  Inspired by the surprising features of quantum information, many researchers---notably Fuchs \cite{fuchs} and Brassard \cite{brassard}---suggested that the whole structure of quantum theory, with its dowry of Hilbert spaces and operator algebras,  could be reconstructed from a few simple principles of information-theoretic nature.  Reconstructing quantum theory means rigorously proving that, once we define a  sufficiently general class of possible theories, quantum theory is the only one compatible  with the principles.
 Many proposals of reconstructions of quantum theory have been made, in different frameworks and with emphasis on different features \cite{chiribella11, hardy01, goyal10, masanes, hardy11, masanes12, dakic11, dariano}, bringing a breeze of fresh air on the long standing problem of a conceptual axiomatization of quantum mechanics.   Here we will focus on the reconstruction of Ref. \cite{chiribella11}, which presented a new set of information-theoretic principles directly inspired by quantum information.

The structure of Ref. \cite{chiribella11} is devised to highlight one  particular feature as the ingredient  that generates the surprising features of quantum theory:  the ingredient is the \emph{purification principle} \cite{chiribella11, chiribella10}, stating that every physical process can be simulated using only pure states and reversible interactions with an environment.     The  principle is directly inspired by quantum information, where the applications of purification are countless. In addition to  purification, Ref. \cite{chiribella11}  contains five principles that are influence the outcome  probabilities of present experiments.  Combining these five, seemingly innocuous requirements with the purification principle, Ref. \cite{chiribella11}  singled out quantum theory uniquely.  This  result  has been recently presented in a non-technical terms in Ref. \cite{chiribella12}.   In the following we will  introduce the reader to the principles of Ref.  \cite{chiribella11}, providing their  translation in the ordinary mathematical language of quantum theory. The principles will be first stated in the non-technical terms of Ref. \cite{chiribella12}, and then illustrated in the special example of quantum theory.

\subsection{Operational vs Hilbert space framework}

Most works on the reconstruction of quantum theory adopt an operational framework describing the experiments that a physicist can perform in a laboratory. Each device has an input system and an output system, a set of possible outcomes $X$ that the experimenter can read out, and a set of random processes $\{\mathcal{C}_i\}_{i\in X}$ occurring in conjunction with the outcomes.
As a special case, a device can have trivial input (i.e. no input at all), in which case its action  consists in preparing a system in a particular set of states.     Likewise,  the device can have trivial output, in which case its action consists in  a demolition measurement that absorbs the system, producing an outcome with some probability (think for example of the absorption of a photon on a photograhic plate).
Quantum theory can be cast in the operational framework as a special example.  In the mathematical language of quantum theory, systems are described by Hilbert spaces, and the outcomes of a device correspond to outcomes of an (indirect) measurement.  A device that prepares states of a system with Hilbert space $\mathcal H$ corresponds to an ensemble $\{  \rho_i, p_i\}_{i\in X}$, where
  each $\rho_i$ is a density matrix (non-negative operator with unit trace) and $\{p_i\}_{i=1}^X$ are probabilities. Here $p_i$ is the probability that the device outputs the system in the state $\rho_i$.   A device that performs a demolition measurement  will be described by a  positive operator valued measure (POVM), namely a collection $\{  P_i  \}_{i\in  X}$ of non-negative operators satisfying the condition  $\sum_{i\in  X}  P_i  =  I$,  where $I$ is the identity operator on the Hilbert space of the system.
  When a POVM measurement $\{  P_i  \}_{i\in  X}$   is performed on a system prepared in the quantum state $\rho$, one obtains the outcome $i$ with probability $p_i  =  \trg [  P_i  \rho]$.  The textbook example  of POVM measurement is the measurement on an orthonormal basis $\{|i\rangle  \}_{i=1}^d$ ($d$ being the dimension of the Hilbert space), corresponding to the POVM $\{   P_i\}_{i=1}^d$ with  $P_i  = |i\rangle\langle  i|$.

  General quantum processes (with both non-trivial input and non-trivial output) are described by the theory of open quantum systems  \cite{krausbook,lathi}.
   When a system $S$, initially prepared in state $\rho$, interacts  with an environment $E$, initially prepared in state $\sigma$, through a joint unitary evolution $U_{SE}$, the joint state changes to $U_{SE}  (\rho  \otimes \sigma)  U^\dag_{SE}$.     If we perform a measurement on the environment, with POVM  $\{  P_i  \}_{i\in  X}$, the outcome $i$ will occur with probability $p_i  = \trg  [  U_{SE}  (\rho  \otimes \sigma)  U^\dag_{SE}  (I\otimes P_i)  ] $ and, for $p_i  \not =  0$,  the state of the system conditionally to outcome $i$ will be given by  $\rho_i'  =  \map C_i (\rho)/ \trg[\map C_i (\rho)]$. Here, $\map C_i$ is the linear map
   \begin{align}\label{map}
  \map C_i  (\rho) = \trg_{\spc H_E}   [   U_{SE}  (\rho  \otimes \sigma)  U^\dag_{SE}  (I\otimes P_i) ] ,
 \end{align}
   $\trg_{\spc H_E}$ denoting the partial trace over the environment.     Any such map is \emph{completely positive}  (it transforms positive operators into positive operators even if we apply it only on one part of a composite system) and \emph{trace-decreasing}   ($\trg [\map C_i (\rho)]   \le \trg [\rho]$ for every  density matrix  $\rho$).    Completely positive trace-decreasing maps are known as \emph{quantum operations} \cite{krauspaper}.  A collection of quantum operations $\{\map C_i\}_{i\in  X}$ as in Eq. (\ref{map}) is known as \emph{quantum instrument} \cite{davieslewis}. A quantum instrument describes the action of a device that couples the system with an environment and  performs a measurement on the latter.
    A quantum operation can be always written in the Kraus form \cite{krauspaper}
    \begin{align}\label{kraus}
    \map C_i  (\rho)  =  \sum_{k=1}^K  C_{i,k}  \rho  C_{i,k}^\dag,
    \end{align}
    where $\{  C_{i,k}\}_{k=1}^K$  is a set of operators   on the system's Hilbert space.  Here we can think of each operator $C_{i,k}$ as representing a ``quantum jump" that randomly changes the state of the system to $\rho'_{i,k}  =  C_{i,k}  \rho  C_{i,k}^\dag/ \trg[ C_{i,k}  \rho  C_{i,k}^\dag]$.

 When there is a \emph{single} operator in Eq. (\ref{kraus}), we know exactly \emph{which} quantum jump has occurred for outcome $i$.     In this case, we say that the process $\map C_i$ is \emph{fine-grained}, because knowing the outcome $i$  gives full information about the jump undergone by the system.   On the contrary, a quantum operation $\map C_i$ with more than one Kraus operator represents a \emph{coarse-grained process}: in principle, one can devise a finer measurement on the environment with outcomes $(i,k)$, inducing  the quantum operations $\map C_{i,k}  (\rho)  =  C_{i,k} \rho C_{i,k}^\dag$.  For example,  if the environment is a particle with spin,  measuring only the position of the particle  would result in a coarse-grained process, because the information coming from the spin degree of freedom has been ignored.


Coarse-graining over all possible outcomes $i$, the evolution of the system's state under the random-process $\{\map C_i\}_{i\in X}$ is given by $\rho  \mapsto \rho'  =  \sum_{i\in X}  \map C_i (\rho) $.     The linear map $\map C:   \rho  \mapsto \map C(\rho)=\sum_{i\in  X}  \map C_i(\rho)$  is called \emph{quantum channel}.  Due to the normalization of the probabilites,   the map is \emph{trace-preserving}, namely
\begin{align}\label{tracepres}
\sum_{i\in X}    \trg[\map C_i (\rho  )]  =  \trg[\rho]   \qquad \forall \rho
\end{align}
as it can be easily checked by taking the sum in Eq. (\ref{map}) and using the normalization $  \sum_{i\in X}  P_i = I$.

\subsection{Informational principles and their translation in the Hilbert space language}

At the general operational level, devices can  be connected with one another, if the output system of one device coincides with the input system of the next. Note that the notion of input and output of a device singles out in a privileged direction when we compose systems.   The first principle in Ref.  \cite{chiribella11,chiribella12}  states that information can only flow from input to output.

\begin{flushleft}
\textbf{Causality:} \emph{ The probability of an outcome at a certain step does not depend on the choice of experiments performed at later steps.}
\end{flushleft}

Let us illustrate the meaning of the principle in the Hilbert space framework.   Suppose that  a quantum system, initially prepared in the state $\rho$, is sent first to the quantum instrument  $\{  \map C_i\}_{i \in  X}$ and then to another quantum instrument  $\{  \map D_j\}_{j\in Y}$  (for example, the two instruments could be to   successive Stern-Gerlach devices).    The joint probability of observing the sequence of outcomes $(i,j)$ will be  $p(i,j)  =  \trg [\map D_j  \map C_i  (\rho) ] $.    In the general operational framework, causality states that the probability of observing outcome $i$, given by $p(i)  =  \sum_{j\in  Y}   p(i,j)$ should not depend on the choice of the particular quantum instrument  $\{\map D_j\}_{j\in Y}$  (in the example, it should not depend on the orientation of the magnetic field in the second Stern-Gerlach device).    This condition is indeed satisfied, thanks to the normalization condition  of Eq. (\ref{tracepres}), which gives $p_i  =  \sum_{j\in Y}   \trg [\map D_j \map C_i  (\rho) ]  =  \trg [\map C_i (\rho)]$, independently of the choice of   $\{\map D_j\}_{j\in  Y}$.

 In we place our operational theory in  a relativistic spacetime, than we have that causality forbids information to travel faster than the speed of light, in the sense that the outcome probabilities for an experiment performed at a  spacetime point  $P= (x,y,z,t)$ can not  depend on the  settings of experiments performed at spacetime points that do not contain $P$ in their light cone.
   As a consequence,  causality implies the \emph{no signalling principle}:  when two experiments take place in space-like separated regions, the outcome probabilities for an experiment  should not depend on the settings of the other experiment.


The second principle in Ref. \cite{chiribella12} pertains the extraction of information from composite systems.
\begin{flushleft}
\textbf{Local tomography:} \emph{The state of a composite system is determined by the statistics of local measurements on its components.}
\end{flushleft}

Let us illustrate the principle in the quantum case:  suppose that   two density matrices $\rho, \sigma$ on the Hilbert space $\spc H_A \otimes \spc H_B$ give the same statistics for every pair of local POVM measurements $\{P_i\}_{i\in X}$ and  $\{Q_j\}_{j\in Y}$ on $\spc H_A$ and $\spc H_B$, namely  $\trg[  (P_i \otimes Q_j)\rho]   =   \trg[  (P_i\otimes Q_j)  \sigma]  $ for all  $ i,j$.
Choosing  $P_i$ and $Q_j$ to be arbitrary rank-one projectors $P_i  =|\alpha_i\rangle\langle  \alpha_i |$   and $Q_j  =|\beta_j\rangle\langle  \beta_j |$ and using the polarization identity, we then conclude that $\rho$ must be equal to $\sigma$. This means that the probability distributions for local measurements are sufficient to identify the density matrix of the composite system $\spc H_A \otimes \spc H_B$.   Note that the property of local tomography, which holds both for classical and quantum theory, fails to hold for the variant of quantum theory on \emph{real Hilbert spaces} \cite{wootters10}.


The following principle states  an information-theoretic property of the composition  of physical processes.

\begin{flushleft}
\textbf{Fine-Grained Composition}: \emph{The sequence of two fine-grained processes is a fine-grained process.}
\end{flushleft}

As we already noted,  a fined-grained process in quantum theory is represented by a quantum operation  $\map C_i$ of the form $\map C_i (\rho)  =  C_i \rho C_i^\dag$, so that the information about the outcome $i$ is enough to identify the quantum jump  $\rho    \mapsto  \rho_i'  =  C_i\rho C_i^\dag/\trg[C_i \rho C_i^\dag]$.    Suppose now that we have two devices in a sequence, and that the two devices produce two  outcomes $i$ and $j$ corresponding to two fine-grained processes  $\map C_i (\rho)  =  C_i \rho C_i^\dag$ and  $\map D_j (\rho)  =  D_j \rho D_j^\dag$, respectively.    The composite process will be given by the quantum operation  $\mathcal{D}_j  \mathcal{C}_i(\rho) = (D_jC_i) \rho(D_jC_i)^\dag$, which still represents a fine-grained process.





The next principle guarantees that  we can encode classical bits   using the physical systems available in our operational theory:

\begin{flushleft}
\textbf{Perfect Distinguishability:} \emph{If a state is not compatible with some preparation, then it is perfectly distinguishable from some other state.}
\end{flushleft}
What does it mean that a state is \emph{not compatible} with some preparation?   Consider a mixed quantum state, with density matrix $ \rho  = \sum_{i}  p_i  |  \psi_i    \rangle\langle \psi_i|$.   The state $\rho$ can be interpreted as the average state of the ensemble $\{    |  \psi_i    \rangle\langle \psi_i|, p_i\}$, where  the system is prepared in the pure state $|\psi_i\rangle$  with probability $p_i$. According to this interpretation, the state $\rho$ is \emph{compatible} with the system being prepared in every  pure state $|\psi_i\rangle$.    Saying that a density matrix $\rho$ is \emph{not compatible} with some preparation means  that there exists  some pure state $|\psi \rangle$ that cannot appear in any ensemble decomposition of $\rho$.   In formula, this means that the only solution to the equation
    \begin{equation}\label{}
    \rho  =  p |  \psi\rangle\langle \psi| + (1-p)  \sigma
    \end{equation}
 with $\sigma$ a density matrix, is $p=0$  and $\sigma  =  \rho$. Technically, if a density matrix $\rho$ is not compatible with some pure state, then $\rho$ cannot be invertible.   This means that $\rho$ will have a non-trivial kernel $\mathsf{Ker} (\rho)  :=  \{  |\varphi\rangle \in \spc H~|~  \rho  |\varphi \rangle   = 0\}$.   Hence, every unit vector $|\varphi\rangle  \in \mathsf {Ker} (\rho)$ will be orthogonal to $\rho$, and, therefore, it will represent a pure state of the system that is  perfectly distinguishable from $\rho$.    For example, the state  $\rho  =  1/2   (  |  0\rangle  \langle0  |  +  |  1 \rangle  \langle 1  | )$  of a three-level system   is not compatible with the system being prepared in the  pure state $|\psi\rangle   =    1/\sqrt 2     (  |1  \rangle  +  | 2 \rangle ) $. The perfect distinguishability principle then imposes that there must exist a state  $\rho'$   that is perfectly distinguishable from $\rho$: in this particular example,  $\rho' =|2 \rangle \langle 2 |$. Perfect distinguishability enables us to encode classical bits into quantum states.
 For instance, we can encode a classical bit into the angular momentum of a nucleous,  by encoding the logical 0 in the  the spin up state $|  \uparrow \rangle$ and the logical 1 in the spin down state $|\downarrow\rangle  $.

The next principle  ensures the possibility of encoding the state of a system $A$ in the state of another system $B$ of  potentially smaller ``size".
\begin{flushleft}
\textbf{Ideal Compression:} \emph{Information can be compressed in a lossless and maximally efficient fashion.}
\end{flushleft}

Let us make mathematically precise the meaning of this principle through an example in quantum theory.   Suppose that  a three-level system $A$ is prepared in the mixed state    $\rho  =  1/2   (  |  0\rangle  \langle0  |  +  |  1 \rangle  \langle 1  | )$. This preparation is compatible with the system being in every pure  state $|\psi\rangle =  \alpha |0\rangle  +  \beta  |1\rangle$, with $|\alpha|^2  +  |\beta|^2  =1$.   It is clear that we can encode these states in a two-level system using the encoding channel $\map C (\rho) =  C \rho C^\dag  +    \langle 2 |  \rho  |2\rangle   ~  |\!\uparrow\rangle \langle \uparrow\! |$ with $C  |0\rangle  =  |\!\uparrow \rangle$ and  $C  |1\rangle  =  |\!\downarrow \rangle$.  With this encoding,  the state $|\psi'\rangle =   \alpha |\!\uparrow\rangle  +  \beta  |\!\downarrow\rangle  \in  \spc H_B$ is the ``codeword" for the state $|\psi\rangle =  \alpha |0\rangle  +  \beta  |1\rangle\in\spc H_A$.     The encoding is \emph{lossless} for the information compatible with $\rho$, because we can always restore the initial state $|\psi\rangle$ from $|\psi'\rangle$.   The encoding is also \emph{maximally efficient}, because we cannot encode without losses the pure states  $\{  |\psi\rangle =  \alpha |0\rangle  +  \beta  |1\rangle , |\alpha|^2  +  |\beta|^2  =1\}$ in a system of Hilbert space dimension smaller than 2.    In general, the encoding of a mixed state $\rho$  into another physical system $B$ is maximally efficient if every pure state of $B$  is the codeword for some pure state compatible with $\rho$.    For a density matrix of rank $r$, the ideal compression is obtained by encoding the information in  a Hilbert space of dimension  $r$.

All the principles discussed so far are satisfied both by classical and quantum theory. The principle that singles out uniquely quantum theory is the following
\begin{flushleft}
\textbf{Purification Principle:} \emph{Every random process can be simulated in an essentially unique way as a reversible interaction of the system with a pure environment.}
\end{flushleft}

Let us illustrate the meaning of the principle in the Hilbert space framework, where random  processes are described by  quantum channels.   One way to simulate a  random process   $\map C (\rho)  =  \sum_{i\in  X}  C_i  \rho  C_i^\dag$  is to introduce an environment with Hilbert space $\spc H_E  =  \mathsf {Span}  \{  |i\rangle\}_{i\in  X}$.     For example,  every pure state $|\eta\rangle \in  \spc H_E$ and every unitary operator $U_{SE}$ satisfying $U_{SE} |\psi\rangle |\eta\rangle  =  \sum_{i\in  X}   C_i |\psi\rangle |i\rangle$ will give
\begin{align}
\map C (\rho)  =  \trg_{\spc H_E}   [   U_{SE}  (\rho  \otimes |\eta\rangle\langle \eta|) U_{SE}^\dag].
\end{align}
This means that the random process $\map C$ can be simulated by letting the system interact unitarily with an environment prepared in the pure state  $|\eta\rangle$.   Since unitary interactions are reversible, this is an example of the pure and reversible simulation required by the purification principle.  The pure and reversible simulation is not unique, because we are free to choose $|\eta\rangle$ to be \emph{any pure state} of the environment and to choose the basis $\{|i\rangle\}_{i\in X}$ used in the definition of $U_{SE}$ to be \emph{any orthonormal basis} for $\spc H_E$.  However, the pure and reversible simulation is \emph{essentially unique}:  once we fix the environment there is no remaining freedom except for the choice of bases for Hilbert space $\spc H_E$.

The ability  to purify every random process is a unique feature of quantum theory. The main message of Refs.\cite{chiribella11,chiribella12}  is that, among all physical theories satisfying the first five reasonable requirements,  quantum theory is the only one that enables a pure and reversible simulation of every random process.  Such a key feature places quantum theory at the core of the theory of reversible computation.  From this angle, the usual picture is turned upside down: instead of regarding quantum theory as ``incomplete" \cite{EPR35} because it fails to give deterministic predictions about the outcomes of arbitrary measurements, we are brought to regard classical theory as ``incomplete"  because it fails to provide a pure and reversible simulation of arbitrary random processes.   This type of simulation is essential if we want to reconcile information theory and physics, the former being based on the notions of random variable and noisy channel, and the latter trying to model phenomena in terms of pure states and fundamentally reversible interactions.


\section{Measurement sharpness trims nonlocality and contextualize in every physical theory}
Nonlocality \cite{EPR35,bell} and contextualize \cite{kochen,mermin93} are among the most striking features of quantum mechanics,
 in radical conflict with the worldview of classical physics.  
 Still, quantum mechanics is neither the most nonlocal theory one can imagine, nor the most contextual.   For nonlocality,  this observation dates back to the seminal work of Popescu and Rohrlich \cite{pr95}, who showed that relativistic  no-signalling  is compatible with correlations that are much stronger than those allowed by  quantum theory.
Their work stimulated the question whether other fundamental principles, yet to be discovered, characterize the peculiar set of  correlations observed in the quantum world.  Up to now, several candidates that partly retrieve the set of quantum correlations have been proposed, including Non-Trivial Communication Complexity \cite{Dam05, brassard2006}, No-Advantage in Nonlocal Computation \cite{linden2007}, Information Causality \cite{Paw09}, Macroscopic Locality \cite{Navascues09},  and, most recently,  Local Orthogonality (LO) \cite{Fritz12}.        The observation that quantum theory is not maximally contextual is more recent \cite{cabello2010,cabello2014} and so is the search for principles that characterize the quantum set of contextual probability distributions.  On this front, the only principle put forward so far is   Consistent Exclusivity  (CE) \cite{cabello2013,henson2012,acin2012}.

Despite many  successes, a complete characterization of the  quantum set is still challenging. What makes the problem hard is the fact that---intendedly---the principles considered so far dealt only with   input-output probability distributions, without making any hypothesis on how  these distributions are generated.  
  On the other hand, a physical theory does not provide only probability distributions, but also specifies  rules on how to combine physical systems together, how to measure them, and how to evolve their state in time \cite{coecke2010}.   Considering that fundamental quantum features like no-cloning  and the possibility of universal   computation cannot be expressed just in terms of input-output distributions,  it is natural to wonder whether also quantum nonlocality and contextualize could better understood in a broader framework of general probabilistic theories  (GPTs)  \cite{hardy01,barrett2007,chiribella10,barnum11}.    Further motivation to extend the framework  comes from the latest principles in the nonlocality and contextualize camps, LO and CE.  Both principles  refer to a notion of orthogonal events and impose that the sum of the  probabilities of a set of mutually orthogonal events shall no not exceed one.    This  is a powerful requirement, which in the case of LO is even capable to rule out non-quantum correlations that are compatible with every bipartite principle \cite{gallego2011}.   But why  should Nature obey such a requirement?  And what does this requirement tell us about the fundamental laws that govern   physical processes?  

Here we tackle the problem of understanding quantum nonlocality and contextualize from a new angle, which focuses  on the fundamental structure of  measurements.     In an arbitrary physical theory, we  introduce a class of ideal measurements, called sharp, that are repeatable and  cause the minimal amount of disturbance on future observations.
We postulate that all measurements are  sharp at the fundamental level and we explain  the apparent unsharpness of real life experiments as due to the interaction with the environment.
Assuming that sharp measurements remain sharp under elementary operations, such as joining  two outcomes together  and applying two measurements in parallel,  we show that  the fundamental sharpness of measurements  implies the validity of CE and LO, thus providing a strong constraint on the set of probability distributions.  Our result demonstrates that  principles formulated in the broader framework of  GPTs   can offer an extra power in the characterization of the quantum set and  identifies the fundamental sharpness  of measurements as a candidate principle for  future axiomatizations of quantum theory.

\subsection{Framework}
In a general theory,  a measurement is described by a collection of events,   each event  labelled by an outcome.   We first consider demolition measurements, which adsorb the measured system.    In this case,
the measurement events are called \emph{effects} and the measurement is a collection of effects $\{m_x\}_{x\in\set X}$.
For a system prepared in the state $\rho$, the probability of the outcome $x$ is denoted by $p_x  =  (  m_x|\rho)$.    In quantum theory this is a notation for the Born rule $p_x  =  \trg  [  m_x  \rho]$, where $\rho$ is a density matrix and $m_x$ is a measurement operator.     In general theories, however, $(m_x|\rho)$ does not denote a trace of matrices and in fact the actual recipe for computing the probability  $(m_x|\rho)$ is irrelevant  here.
We will often use the notation   $(m_x|$   and $ |\rho  ) $ for effects and states,  respectively.
It is understood that two different states  give different probabilities for at least one effect, and two different effects take place with different probabilities on at least one state.

 When two measurements $ \{m_x\}$ and $\{n_y\}$ are performed in parallel on two systems $A$ and $B$, we denote by $m_x \otimes n_y$ the measurement event labelled by the pair of outcomes $x,y$.  Similarly, when  two states of systems $A$ and $B$, say $\alpha$ and $\beta$, are prepared independently,  we denote by $\alpha  \otimes \beta$ the corresponding state of the composite system $AB$.   In quantum theory, this is the ordinary tensor product of operators, but this may not be the case in a general theory and, again,  the actual recipe for computing $\alpha\otimes \beta$  is irrelevant here.
What is relevant, instead, is that the notation is consistent with the operational notion of performing independent operations on different systems:        If two  systems   are independently prepared in states $\alpha$ and $\beta$ and undergo to independent measurements $ \{m_x\}$ and $\{n_y\}$, we impose that the  probability has the product form $ p_{xy}  =    (  m_x|\alpha) \, (n_y  |  \beta)$.  

The most basic operation one can perform on a measurement is to join  some outcomes together, thus obtaining a new, less informative measurement. This operation, known as \emph{coarse-graining},  is achieved by dividing the outcomes  of the original measurement $\{m_x\}_{x\in\set X}$ into disjoint groups $   \{   \set X_z\}_{z\in\set Z}$, and by identifying outcomes that belong to the same group.  The result of this procedure is a new  measurement $ \{m'_z\}_{  z\in\set Z}$ satisfying  the relation
$(m'_z  |  \rho)  =    \sum_{x\in\set X_z}  (m_x|\rho)$
for every every $z$ and for every possible state $\rho$.   For brevity,  we write $ m'_z  =  \sum_{x\in\set X_z}   m_x$.

Coarse-graining allows one  to  express the  principle of  \emph{causality},
 which states that the settings of future measurements do not influence the outcome probabilities of present experiments  \cite{chiribella10}. Causality is   equivalent to the requirement that for every system $A$  there exists an effect   $u_A$, called the \emph{unit},    such that
\begin{equation}\label{causal} \sum_{x\in\set X}   m_x  =  u_A
\end{equation}
for every measurement   $\{m_x\}_{x\in\set X}$   on $A$.
  In quantum theory,   $u_A$ is the identity operator on the Hilbert space of the system and Eq. (\ref{causal})  expresses the fact that quantum measurements are resolutions of the identity.   When there is no ambiguity, we drop the subscript from $u_A$.

Causality has major consequences.    First of all, it implies that the probability distributions generated by local measurements satisfy the no-signalling principle \cite{chiribella10}.
Moreover,  it allows  to perform adaptive operations:  for example, if $\{m_x\}_{x\in \set X}$ is a measurement on system $A$ and  $\{n^{(x)}_y\}_{y\in\set Y}$ is a measurement on system $B$ for every value of $x$, then causality guarantees that it is possible to choose the measurement on $B$ depending on the outcome on system $A$, \ie that $\{  m_x\otimes n^{(x)}_y\}_{x\in\set X,  y\in\set Y}$ is a legitimate measurement.   Finally, causality allows one to describe non-demolition measurements.  
For a non-demolition measurement  $\{\map M_x\}_{x\in\set X}$, the measurement events  are \emph{transformations}, which turn the initial state of the system, say $\rho$, into a new  unnormalized state $  \map M_x   |\rho)$.
For a system prepared in the state $\rho$,   the probability of the outcome $x$  is   $  p_x = (u|  \map M_x  |\rho)$ and, conditionally on outcome $x$,   the post-measurement state is $ \map M_x   | \rho ) /   (u|\map M_x |  \rho)$.   We  will often refer to  the non-demolition measurements as \emph{instruments}, in analogy with the usage in quantum theory  \cite{Davies76,Ozawa84}.
Note that, thanks to causality, every instrument   $\{\map M_x\}$ is associated to a unique demolition measurement $\{m_x\}$ via  the relation
\begin{equation}\label{prob}
(m_x|     =   (u|  \map M_x \qquad \forall  x\in\set X\, .
\end{equation}
 By definition, $\{m_x\}$ describes the statistics of the instrument:  for every state $\rho$ and for every outcome $x$, one has  $p_x   =   (  u  | \map M_x  |\rho)  \equiv   (m_x|\rho)$.

\medskip
\subsubsection{Sharp measurements in arbitrary theories}   In textbook  quantum mechanics,  physical quantities are associated to self-adjoint operators, called observables  CITA.  The values of a quantity are the eigenvalues of the corresponding operator and the probability that a measurement outputs the value  $x$ is given by the Born rule $p_x  =  \trg [  P_x\rho ]$, where $P_x $ is the projector on the eigenspace for the eigenvalue $x$ and $\rho$ is the density matrix of the system before the measurement.  If the measurement gives the outcome $x$,   then the state after the measurement is  $\rho_x'  = P_x\rho P_x/  \trg[P_x\rho]$, according to the projection postulate.
These canonical  measurements, where all the measurement operators are orthogonal projectors,  are called  sharp  \cite{Busch96}.
 While it is clear that sharp measurements play a  key  role  in quantum theory 
 it is by far less clear how to define them in an arbitrary GPT.     Here we propose a  simple definition based on the notions of repeatability and minimal disturbance.

 Let us start from repeatability.
An instrument   $\{ \map M_x\}$ is \emph{repeatable} if it gives the same outcome when performed two consecutive times,  namely
\begin{align}\label{rep}    (     m_x|  \map M_x    =  (   m_x|    \qquad \forall x\in\set X  \, ,
\end{align}
where $\{m_x\}$ is the  measurement of Eq. (\ref{prob}).
Repeatability poses a fairly weak requirement on  $\{m_x\}$:
every measurement  that discriminates perfectly among a set of states  $\{\rho_x\}$ can be realized by a repeatable instrument, which consists in measuring $\{m_x\}$ and, if the outcome is $x$, re-preparing the system in state $\rho_x$.

The second ingredient entering in  our definition of  sharp measurements is  minimal disturbance.     We say that the instrument $ \{\map M_x\}_{x\in\set X}$    does not disturb the  measurement   $\st n  = \{n_y\}_{y\in\set Y}$  if   the former does not  affect the statistics of the latter, namely
\begin{align}\label{nodist}
(  n_y|  \map M  =  (  n_y|  \qquad \forall   y\in\set Y   \, ,
\end{align}
where   $(n_y| \map M:  =  \sum_{x\in\set X}   (n_y| \map M_x$.
 Then, we ask which instruments disturb the smallest possible set of measurements.
 Clearly, if   $\{\map M_x\}$  does not disturb   $\st n$, then  $\st n$ must be compatible with  the measurement $\st m  =  \{  (u|  \map M_x\}$, in the sense that $\st m$ and $\st n$  can be measured jointly.    Indeed, by measuring $\st n$ after   $ \{    \map M_x  \}$ one obtains   the  probability distribution  $ p_{xy}   =   (   n_y  |\map M_x  |  \rho)$, whose  marginals on $x$ and $y$ are equal to the probability distributions of $\st m$ and $\st n$, respectively.    Read in the contrapositive, this means that if $\st m$ and $\st n$ are incompatible,  the instrument $\{\map M_x\}$ must disturb $\st n$.
This leads us to the following definition:   an instrument  $\{\map M_x\}$ has \emph{minimal disturbance} if it  disturbs only the measurements that are incompatible with   $\st m  = \{  (u|  \map M_x\}$.

We define an instrument to be sharp if it  is both repeatable and with minimal disturbance. We say that a measurement is  sharp if  it describes the statistics of  a sharp instrument and we call an effect sharp if it belongs to a sharp measurement.
 In quantum theory, our definition coincides with the usual one:  one can prove that the only sharp instruments are the L\"uders instruments  \cite{Luders50}, of the form $\map M_x  (\rho)   =   P_x  \rho P_x$ where $\{P_x\}$ is a collection of orthogonal projectors.  Hence, the sharp measurements are projective measurements.
In addition, we can prove that  when a sharp measurement extracts a coarse-grained information,  the experimenter can still retrieve the finer details at a later time. In fact, this is a necessary and sufficient condition for a measurement to be sharp, as proven in the Methods section.

\medskip

\subsubsection{ Fundamental sharpness of measurements}       Sharp measurements are an ideal standard---they are the measurements that generate outcomes in a repeatable way, while at the same time causing the least disturbance on future observations.  Unfortunately though, most measurements  in real life  appear to be noisy and not repeatable.  Hence the natural question:  Is noise is fundamental? Or rather it is contingent to the fact that the experimenter has  incomplete control  on the conditions of the experiment?
Here we state that noise is not fundamental and only arises from the fact that the realistic measurements do not extract information only from the system, but also  from the surrounding environment:
\begin{axiom}[Fundamental Sharpness of Measurements]\label{ax:sharpness} Every measurement arises from a sharp measurement performed jointly on the system and on the  environment.
\end{axiom}
Precisely, we require that for every measurement   $\st m  =  \{m_x\}_{x\in\set X} $   there exists an environment $E$, a state  $\sigma$  of $E$, and a sharp measurement $\st M  =  \{M_x\}_{x\in\set X} $  on the composite system $SE$  such that, for every state $\rho$ of system $S$, one has  $(  m_x | \rho)      =    \left (  M_x \right|      \rho  \otimes   \sigma)$ for every outcome  $ x\in\set X $.
In quantum theory, this is the content of the celebrated Naimark's theorem  \cite{Helstrom76,Holevo11}.
This is a deep property,  hinting at the idea there exists a  fundamental level where all measurements  are ideal.

Let us push  the idea  further.  If  measurements are sharp at the fundamental level, it is natural to assume that the set of sharp measurements is closed under the basic operation of coarse-graining, which transforms an initial measurement $\st m$ into a new, less informative measurement $\st m'$.  Indeed, since $\st m'$ provides less information than $\st m$, one expects that $\st m'$ should not be less repeatable, nor create more disturbance, than $\st m$.
 This intuition leads to the following requirement:
\begin{axiom}[Less Information, More Sharpness]\label{ax:coarse}
If a measurement is less informative than a sharp measurement, then it is sharp.
\end{axiom}
Suppose now that two experimenters,  Alice and Bob, perform two sharp measurements on two systems $A$ and $B$ in their laboratories.    Again, if measurements are sharp at the fundamental level, one expects the result of Alice's and Bob's measurements to be a sharp measurement  on the composite system $A B$.  If this were not the case, it would mean that at the fundamental level some measurements require  nonlocal interactions, even though at the operational level the they appear to be implemented locally by Alice and Bob.
 We then postulate the following
\begin{axiom}[Locality of Sharp Measurements]\label{ax:loc}  If  two sharp measurements are applied in parallel on  systems $A$ and $B$, then the result is a sharp measurement  on the composite system $AB$.
\end{axiom}
Axioms 1-3 lay down the fundamental structure of sharp measurements, summarized in Fig. \ref{fig:sharp}.
\begin{figure}
\centering
\includegraphics[totalheight=45mm]{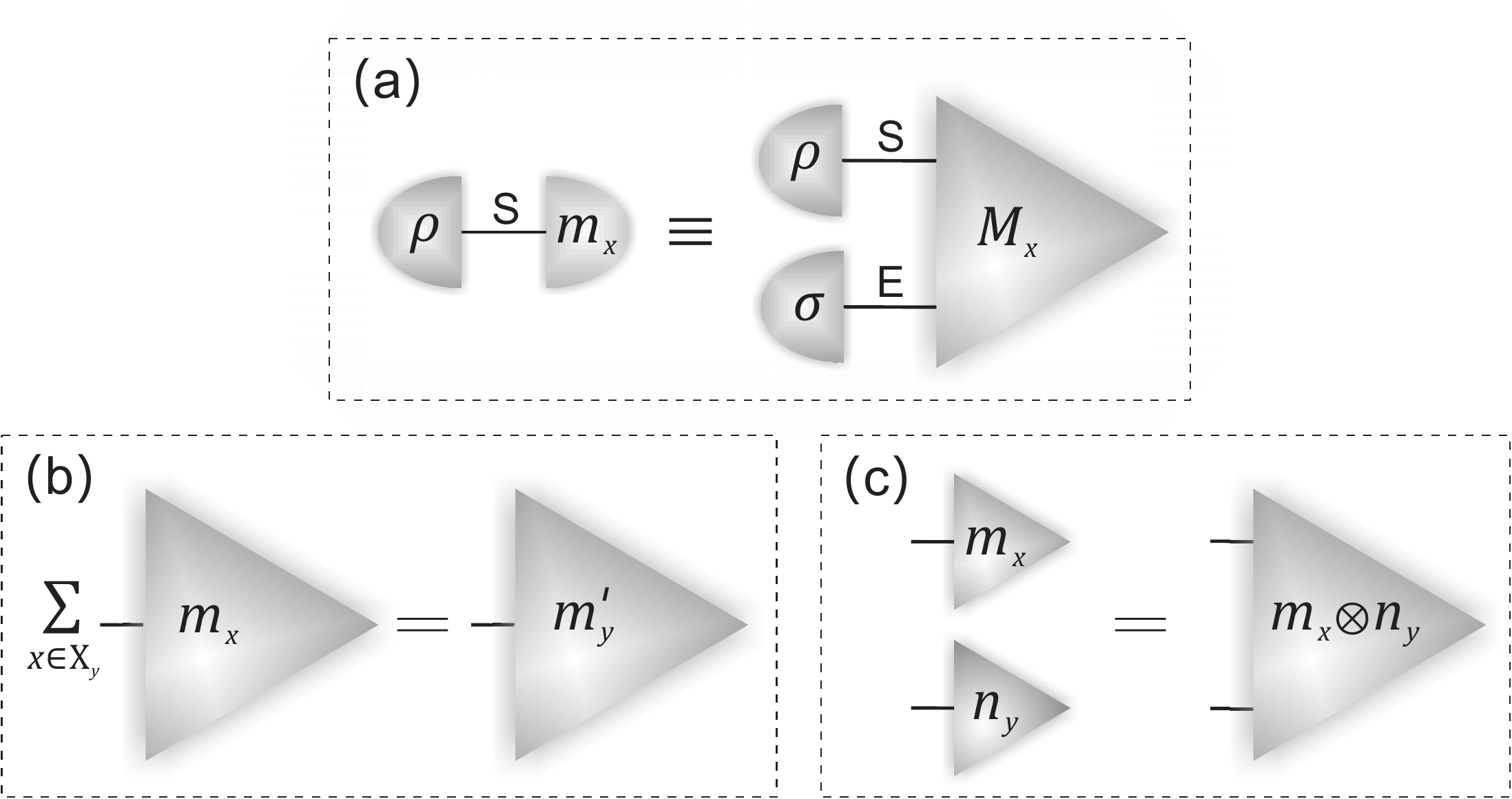}
\caption{{\bf The structure of sharp  measurements.}    ({\bf a})  Every non-sharp measurement  $\{m_x\}_{x\in\set X}$ (round diagram on the l.h.s.)  is equivalent to a sharp measurement   $\{M_x\}_{x\in\set X}$  (triangular diagram on the r.h.s.) performed on the system along with an environment.   ({\bf b})    Coarse-graining a  sharp measurement $\{ m_x\}_{x\in \set X}$ yields a new  sharp measurement $\{ m'_y \}_{y\in\set Y}$.  ({\bf c})  When two sharp measurements $\{m_x\}$ and $\{m_y\}$  are performed in parallel, they yield a new sharp measurement  $\{m_x \otimes n_y\}$.
}
\label{fig:sharp}
\end{figure}
They  are  satisfied by classical theory and by quantum theory, both on complex and real Hilbert spaces.    In the following we will show that the fundamental structure of sharp measurements has an enormous impact on the amount of nonlocality and contextualize that can be found in a physical theory.


\medskip

\subsection{ Derivation of CE}
At present, CE  is the only principle known to constrain the amount of contextualize of a generic theory.
Operationally, the principle can be formulated as follows:   Consider  a collection of sharp measurements $\{\st m^{(x)} , \,    x\in\set X\}$,  each measurement having outcomes in a set $\set Y_x$.
Suppose  that the possible events have been labelled so that two effects corresponding to the same outcome coincide, \ie   $m_y^{(x)}   \equiv   m_y$, independently  of $x$.   Letting $\set Y  =  \cup_x \set Y_x$ be the set of all  outcomes, one  calls  two distinct outcomes  $y,y'  \in\set Y$ \emph{exclusive} if    there exists a measurement setting $x$   such that  both $y$ and $y'$ belong to  $ \set Y_x$.    We say that   a theory satisfies CE if for every set of mutually exclusive outcomes   $\set E$ and for  every state $\rho$ the probabilities $p_y  =   ( m_y|\rho )$  obey the bound
$\sum_{y\in \set E}   p_y \le 1 \, .$

Our first key result is the derivation of  CE.  In fact, we prove a stronger result: We define  two sharp effects  $m$ and $m'$   to be \emph{orthogonal} if they belong to the same measurement and   we prove that mutually orthogonal effects can be combined   into a single sharp measurement  (see Methods). Clearly, since mutually exclusive outcomes correspond to mutually orthogonal effects, the existence of a joint measurement  containing the effects  $\{  m_y\}_{y\in\set E}$ implies the bound $\sum_{y\in \set E}   (m_y|  \rho)  \le 1$.    Our result implies that in a theory where measurements are fundamentally sharp  the violation of Kochen-Specker inequalities is upper bounded by the value set by CE  \cite{cabello2014}.  What is remarkable here is that a single requirement on measurements influences directly the strength of contextualize in an arbitrary physical theory.   This situation contrasts with that of the known  axiomatizations of quantum theory  \cite{hardy01,chiribella11,hardy11,masanes,Brukner,masanes12}, where the quantum bounds on contextualize    are retrieved  only indirectly through the derivation of the  Hilbert space framework.

Our axioms do not imply only CE, but also the whole hierarchy of extensions of this principle defined in Ref.   \cite{acin2012}.  The $L$-th level of the hierarchy can be defined by considering independent measurements on $L$ copies of the state $\rho$.  Denoting by  $ \st y  =  (y_1,\dots, y_L)   $  the string of all outcomes,  one says that  two strings $\st y$ and $\st y'$ are exclusive if there exists some $i$ such that $y_i$ and $y_i'$ are exclusive.      A physical theory satisfies the  $L$-th level of the hierarchy if
 the probabilities  $p_L  (\st y)  =   \prod_{i=1}^L (  m_{y_i}  |  \rho)$ obey the bound
 $\sum_{\st y  \in  \set E}    p_L (\st y)   \le 1  $
for every set $\set E$ of mutually exclusive strings.  In the Methods section we show that our axioms on sharp measurements imply that this bound is satisfied for every possible $L$.

\medskip

\subsection{ Derivation of LO}   In the nonlocality camp, LO  occupies a special position, being up to now the  only known principle   that rules out  non-quantum correlations that are not detected  by any bipartite principle \cite{gallego2011}.
LO refers to a scenario where $N$  parties perform local measurements on $N$  systems, initially prepared in some  joint state.  The $i$-th party  can choose among different measurement settings in a set $\set X_i$ and her measurements give outcomes in another set $\set Y_i$.   Let $ \st x  =  (x_1,\dots, x_N)  $ be the string of all settings,  $ \st y  =  (y_1,\dots, y_N)   $ be the string of all outcomes, and $\st e$ be the pair $\st e = (\st  x,\st  y)$.  In this context, the pair $\st e = (\st  x,\st  y)$ is called an event and   two events are called \emph{locally orthogonal} iff  there exists a party $i$ such that $x_i =  x_i'$ and $y_i\not=y_i'$.        Setting  $p(\st e)$ to be the conditional probability distribution $   p(\st y| \st x)$,    one says that theory satisfies local orthogonality if all the probability distributions generated by local measurements   obey the bound
$\sum_{\st e  \in  \set O}  p  ( \st e)   \le 1$
for every  set  $\set O$ of pairwise locally orthogonal events.

To derive LO, we specify how the probability distribution $p(\st y|\st x)$ is generated:  In the most general scenario,  the $N$ parties share a state $\rho$ and that, for setting $x_i$,  party $i$ performs a measurement  $\st m^{(i,x_i)}$.  Denoting the product effects $P^{(\st x)}_{\st y}  : =  \bigotimes_{i=1}^N  m^{(i,x_i)}_{y_i} $,  the probability distribution of the outcomes is given by  $  p(\st e)  : =  \left(    P^{(\st x)}_{\st y} |  \rho  \right)$.
The proof that LO follows from the axioms, provided in Methods,  consists of three steps:  First, thanks to the Fundamental Sharpness of Measurements, the problem is reduced  to proving that LO holds for probability distributions generated in a scenario  where all parties perform sharp measurements.     Then, we observe that, in the case of  sharp measurements, locally orthogonal events correspond to orthogonal effects.     Finally, we use the fact that mutually orthogonal effects  can coexist in a single measurement.  As a corollary, we obtain the bound $\sum_{\st e  \in  \set O}  p  ( \st e)   \le 1$, establishing the validity of LO for all the probability distributions generated by measurements in our theory.

Like in the case of CE, our axioms imply the whole hierarchy of extensions of LO introduced in \cite{acin2012}. The hierarchy is defined as follows:   the probabilities $p(\st y|\st x)$ satisfy the $L$-th level of the hierarchy if their product $  p(\st y_1|  \st x_1)  \cdots  p(  \st y_L |\st x_L)$ satisfies LO.
  Now, we can think of the product as being generated by measurements on $N$ copies of the state $\rho$.  In this way, we reduce the problem of proving the $L$-th level of the hierarchy to the problem of proving LO for measurements performed on the state $\rho^{\otimes  L}$.  But we already proved the validity of LO for arbitrary measurements and arbitrary states.    In conclusion, the structure of sharp measurements  implies that LO is satisfied at every possible level.  A striking consequence of this argument is that the fundamental sharpness of measurements rules out PR box correlations, as the latter violate the LO hierarchy \cite{Fritz12}.  In other words, in the world of PR boxes some measurements must be fundamentally noisy.
Finally, since our axioms imply LO, in particular they imply all the limitations  that LO sets on Bell inequalities and nonlocal games.  We will elaborate on this point in the following.

\medskip

\subsection{ Sharp Bell inequalities}   The request that measurements are ideal at the fundamental level exerts a censorship on the amount of nonlocality that can be detected by experiments.
To illustrate this fact, we show a number of Bell inequalities where the sharpness of measurements prevents every violation.  We call such inequalities sharp.  

Consider a game played by $N$ non-communicating parties and a referee, who sends to party $i$ an input $x_i$ and receives back an output $y_i$.   The referee chooses the input string   $\st x$ at random with probability $q(\st x)$ and assigns a payoff $\omega  (\st x,\st y)$ to the players, assumed without loss of generality to be nonnegative for every $(\st x,\st y)$.  The expected payoff obtained by the players is given by $  \omega  =   \sum_{\st x,\st y}    q(\st x)  \omega  (\st x,\st y)   p(\st y|\st
 x)$,      where $p(\st y|\st x)$ is the probability distribution describing their strategy.
For a given game, the maximum payoff that can be achieved by classical strategies---call it $\omega_c$---defines a Bell inequality,  $\omega  \le \omega_c$. The  game can be associated with a graph  $\set G$, here called the \emph{winning graph}, by choosing as vertices the events  $(\st x, \st y)$ such that  $q(\st x)  \omega  (\st x,\st y)  \not = 0$ and placing an edge between two events  $\st e$ and $\st e'$  if they are not locally orthogonal, as illustrated in Figure  \ref{fig:graphs}.
\begin{figure}[hbt]
\centering
\includegraphics[totalheight=100mm]{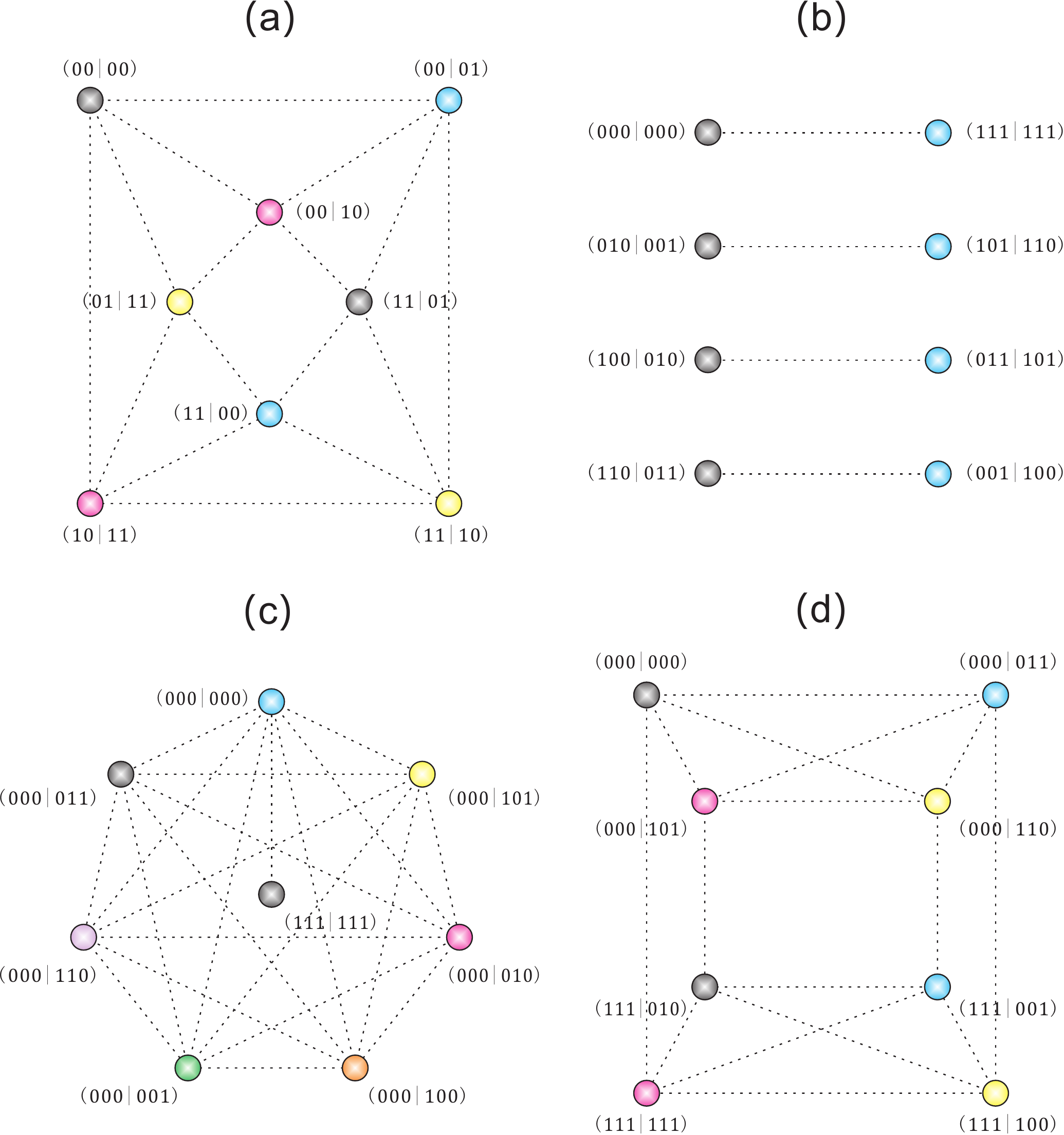}
\caption{{\bf Winning graphs for  examples of nonlocal games.}  The vertices are coloured  so that two connected vertices have distinct colours, using the minimum number of colours.    ({\bf a})  Winning graph for the CHSH game. The player   win if $y_1  \oplus  y_2  = x_1 x_2$  and 0 otherwise.    The graph is not perfect, because the largest clique in the graph has 3 vertices while the number of colours in the graph is 4.  Here classical strategies are not optimal  among the strategies that satisfy LO.       ({\bf b}) Winning graph for Guess Your Neighbor's Input  \cite{Almeida10} in the case of $N=4$ parties. The players win $+1$ if $y_i  =  x_{i+1}$ for every $i$.  The graph is a disjoint union of disconnected cliques and  therefore classical strategies are optimal among all strategies satisfying LO.      ({\bf c})  Winning graph for the game ``Guess the Product" in the case of $N=3$ parties.  The players win  $+1$ if   $y_i  =  x_1x_2 x_3$ for every $i$ and 0 otherwise.   The graph is perfect and therefore the classical strategy is optimal  \cite{acin2012}.         ({\bf d})  Winning graph for the game  ``Guess the Parity" in the case of $N=3$ parties.  The players win  + 1 if $y_i  =  x_1 \oplus  x_2 \oplus x_3$ for every $i$ and 0 otherwise. The graph is not perfect, because it contains odd  cycles with more than 3 vertices.  Still, classical strategies are optimal for this game, as shown in Supplementary Note 2 for arbitrary number of players.
}
\label{fig:graphs}
\end{figure}
In this picture, the maximum payoff  achieved by classical strategies is
\begin{equation}
\omega_c =  \max_{\set C  \subseteq \set G}   \sum_{  (\st x,\st y)\in\set C }    q(\st x) ~ \omega  (\st x,\st y)  \, ,
\end{equation}
where $\set C$ is a clique, \ie  a subset of $\set G$ with the property that every two vertices in $C$ are connected \cite{Fritz12}.

A first class of games leading to sharp Bell inequalities  is the class of games with a graph $\set G$ that is the disjoint union of mutually disconnected cliques  $\set C_k,    \,  k\in  \{1,\dots, K\}$.     This class contains the game  Guess Your Neighbor's Input \cite{Almeida10}  and the maximally difficult Distributed Guessing Problems of Ref. \cite{Fritz12}.
In addition, it contains other games such as, for even $N$,  the ``Guess the Parity" game where each player is asked to guess the  parity of the input string  $\st x$.
For all these games, LO implies that the classical payoff is an upper bound.   Indeed,   picking    for every  $k$  the event  $\st e_k  \in\set C_k$  that has maximum probability, the payoff can be bounded as
\begin{eqnarray*}
\omega   &  = & \sum_{k=1}^K        \sum_{(\st x ,\st y)\in \set C_k }    q(\st x) ~   \omega  (\st x , \st y)    ~ p\left (\st y|   \st x \right)  \\
 & \le&  \sum_{k  = 1}^K       p\left ( \st e_k \right)          \left[ \sum_{(\st x,\st y)  \in \set   C_k}       q(\st x)   ~ \omega (\st x, \st y  )     \right]
\end{eqnarray*}
and since the events $\{  \st e_k\}_{k=1}^K$ are locally orthogonal by construction,  one has  $\sum_k  p(\st e_k)  \le 1$ and, therefore,  $\omega   \le   \max_k    q  \left(  \set   C_k \right)  \equiv  \omega_c $.  In conclusion, every sharp game defines a  Bell inequality,  $\omega  \le \omega_c$  that cannot be violated by any theory satisfying LO, and, in particular, by any theory where measurements are fundamentally repeatable and minimally disturbing.
Using a result of Ref. \cite{acin2012},  the proof that LO cuts the payoff down to its classical value can be extended  to a larger class of games, defined by the property that the winning graph $\set G$ is a perfect  graph  \cite{BergeGraphs}.  For example, one such game is the ``Guess the Product" game where the players are win if they guess the product of their inputs.  Finally, there are examples of games where the Bell inequality $\omega \le \omega_c$ is sharp even if the winning graph is not perfect, such as  Guess the Parity when the number of players is odd (see Supplementary Note 1 for the proof that the payoff is upper bounded by the classical value).



Our results are derived in a minimal framework,  which avoids some assumptions  commonly made   in  GPTs. In particular,  our arguments  do not invoke local tomography \cite{hardy01,mauro2006,barrett2007}, but only the requirement that sharp measurements are local.   This requirement is strictly weaker: for example,  it is satisfied by quantum theory on real  Hilbert spaces, where local tomography fails.   Also, the validity of our results does not require   that the states of a given physical system  form a convex set.  Thanks to this feature, the results apply  also to  non-convex theories, like Spekkens' toy theory \cite{SpekkensToy}. Interestingly, probabilities themselves do not play a crucial role in our arguments and it is quite straightforward to extend the   results to theories that  only specify which outcomes are possible, impossible or certain, without specifying their probabilities,  such as Schumacher's and Westmoreland's theory \cite{ModalQT}.

Since sharp measurements play  a central role in quantum mechanics,  it is not surprising that they have been the object of extensive investigation since the early days \cite{vonNeumann,Luders50,Dirac}.   
Later, Holevo proposed a purely statistical definition of  sharp measurement, which does not refer to post-measurement states  \cite{HolevoObservable}. Although in the quantum case Holevo's definition  reduces to that of projective measurement, in general theories it is inequivalent to ours, and it is not clear how one could use it to derive features like LO and CE.    Different notions of ideal measurements were put forward by Piron \cite{PironBook,PironIdeal} and Beltrametti-Cassinelli \cite{Beltrametti81} in the framework of quantum logic.  In general, they differ form our definition in the way the condition of minimal disturbance is defined.  Most recently,  measurement disturbance  came back  to play  an important role in the search for basic principles principles, as shown \eg by the No Disturbance  Without Information principle of Ref. \cite{Pfister}.

Our work joined the insights from two different approaches to the foundations of quantum mechanics: the characterization of quantum correlations \cite{pr95,Dam05, brassard2006,linden2007,Paw09,Navascues09,Fritz12}  and the study of general probabilistic theories  \cite{hardy01,chiribella11,hardy11,masanes,Brukner,masanes12}.
  Although these two approaches have developed on separate tracks so far,   they share the same fundamental goal: understanding which picture of Nature lies behind the mathematical laws of quantum mechanics and guiding our intuition towards the formulation of new protocols and new  physical  theories.      Our results demonstrate that the interaction between the two approaches can be beneficial for both.    Here  LO and CE stimulated the search for new  principles in the GPT  framework, leading to  a compelling picture of nature where measurements are repeatable and cause minimal disturbance  at the fundamental level.   The idea that a noisy physical process can be reduced to an ideal process at the fundamental level reminds immediately of another quantum feature: Purification \cite{chiribella10}. Operationally, Purification is the property that every mixed state can be generated  from a pure state of a composite system by discarding one component.    This principle implies directly entanglement    and is at the core of the reconstruction of quantum theory of Ref.   \cite{chiribella11}.
  Our result suggests the possibility that purification and sharpness could be sufficient to derive quantum theory.
 In terms of quantum correlations, this would lead to the tantalizingly simple picture ``Purification brings nonlocality in, sharpness cuts it down".    Going even further, it is intriguing  wonder whether purification and sharpness can be viewed as two sides of the same medal by imposing that physical theories must satisfy  a suitable requirement of time symmetry, similarly to what was done in quantum theory by Aharonov, Bergmann and Lebowitz \cite{ABL,APT}.

\subsection{Methods}

{\bf Characterization of sharp instruments.} The starting point of our results is the observation that   an  instrument $\{\map M_x\}$ is sharp if and only if
\begin{equation}\label{alte}
(  r_{xy}  |    \map M_x  =   (r_{xy}|    \qquad  \forall x\in \set  X  \, , \forall y \in \set Y
\end{equation}
for every measurement $\st r  =  \{r_{xy}\}_{(x,y)\in\set X\times \set Y}$ that refines $\{m_x\}$, \ie  $ \sum_{y\in\set Y}  r_{xy}  =  m_x$ for every $x$.
 Let us see why Eq. (\ref{alte}) is equivalent to sharpness.   First, suppose that Eq. (\ref{alte}) holds. Clearly, this implies that $\st m$ is repeatable, as one can see by summing over $y$.  Moreover,  Eq. (\ref{alte})  implies that $\st m$ is a minimal disturbance measurement. Indeed, take a generic measurement  $\st n$ that is compatible with $\st m$.   By definition, this means that  there exists a joint measurement $\st r  = \{  r_{xy}\}$ such that $\sum_x r_{xy}  = n_y$ for every $y$ and $\sum_y  r_{xy}  = m_x$ for every $x$.
 We then obtain
\begin{align}
\nonumber (n_y|   \map M  &  =
 \sum_{x}  (r_{xy}|    \map M_{x}   +    (s_y|   \qquad  (s_y| : =    \sum_{x,x':  x\not = x'}  (j_{xy}|    \map M_{x'} \\
 \nonumber &  =   \sum_{x}  (r_{xy}|     +   (s_y|   \\
\label{proof} &  =    (n_{y}|     +      (s_y|  \, ,
\end{align}
having used Eq. (\ref{alte}) in the second equality.
 Summing over $y$ and using the normalization of the measurement $\st n$ we obtain the condition $ \sum_y  (s_y|  = 0$, or, equivalently, $\sum_y  (s_y|  \rho)=0$ for every $\rho$.  Since probabilities are non-negative, this implies that each term in the sum vanishes, leading to the relation $s_y=  0$ for every $y$.  Inserting this relation back in Eq. (\ref{proof}) we conclude that $(n_y|  \map M  =  (n_y|$, that is, the instrument does not disturb $\st n$.  Hence,   Eq. (\ref{alte})  implies that $\{\map M_x\}$ is a sharp instrument.         Conversely,  if $\{\map M_x\}$  is a sharp instrument then Eq. (\ref{alte}) must be satisfied.  
  By definition, one has
  \begin{equation*}
  (  m_x|  =   (e|  \map M_x =  \sum_{x'}  (m_{x'}|  \map M_x
  \end{equation*}
  and using the repeatability condition  $(m_x|    = (m_x|  \map M_x  $ one obtains  $\sum_{x'\not = x}   (m_{x'}|\map M_x   =  0$.
  Again, the fact that probabilities are nonnegative implies  that each term in the sum must vanish, namely
  $(m_{x'}|\map M_x   =  0 $ for every $x'  \not = x$.
   Now, let $\st r$ be a measurement such that $ \sum_y r_{xy}   =  m_x$.   Since the measurement $\st r$ is compatible with $\st m$, Eq. (\ref{nodist}) implies $ (r_{xy} |  \map M  =  (r_{xy}|$ with $\map M  =  \sum_{x'} \map M_{x'}$.    On the other hand, for $x  \not =  x'$ the condition $( m_x|  \map M_{x'}   =  0$  implies $  (r_{xy}|  \map M_{x'} = 0 $.  Hence, we conclude that $(r_{xy} |  \map M_x  =  (r_{xy} |  \map M   =  (r_{xy}|$ for every $x$ and $y$. \qed

\medskip

{\bf  Joint measurability of orthogonal effects.}     
 The characterization of sharp instruments, combined with the Less Information-More Sharpness principle, leads directly  to the  first key result of our work:  a construction showing that  mutually orthogonal effects can be measured jointly in a single sharp measurement.  Precisely,   if   $m_k$ is orthogonal to $m_l$ for every $k,l  \in\{1,\dots, K\}$, we show that there exists a joint sharp measurement $\st j$ such that $\{m_k\}_{k=1}^K  \subseteq \st j$.

 Let us see how to  construct the  joint measurement.    Let $\st m^{(k)}$ be  the sharp measurement that contains the effect $m_k$.  By coarse graining of $\st m^{(k)}$ one obtains  the binary measurement $\st m^{(k)}  =  \{m_0^{(k)}  ,  m^{(k)}_1\}$, with $m_0^{(k)} : =  m_i$ and $m_1^{ (k) } : =  u-m_k$. By  the Less Information, More Sharpness postulate, $\st m^{(k)}$ is  sharp.    Let  $\{ \map M^{(k)}_0,  \map M_1^{(k)}  \}$ be the corresponding instrument.  Now, since $m_k$ and $m_l$ are orthogonal, $\st m^{(kl)}  = \{m_k, m_l,  e-m_k  -  m_l\}$ must be a valid measurement.     Since $\st m^{(k)}$ is a coarse-graining of  $\st m^{(kl)}$, Eq. (\ref{alte}) gives
\begin{equation}\label{aa}  (  m_l  |    \map M_1^{(k)}    =   (m_l|\,.
\end{equation}
 Now, consider the following measurement procedure:  \emph{i)} perform the first instrument   \emph{ii)}  if the outcome is $1$, then perform the second instrument,  \emph{iii)}  for every $k< K $,  if the outcome of the $k$-th instrument is $1$, perform the  $(k+1)$-th instrument.   The resulting instrument, denoted by $\{\map J_k\}_{i=1}^{K+1}$ consists of the transformations
\begin{align*}
\nonumber \map J_{1}   & :=  \map M^{(1)}_{0} \\
\nonumber \map J_{2}   &:=  \map M_0^{(2)}  \map M_1^{(1)}   \\
\nonumber \map J_{3}       &: =  \map M_0^{(3)}   \map M_1^{(2)}    \map M_{1}^{(1)}  \\
 \nonumber &\, ~~\vdots  \\
\map J_{K}   &:=   \map M_{0}^{  (K) } \map M_1^{(K-1)}\cdots    \map M_1^{(1)}  \\
\map J_{K+1}   &:=   \map M_{1}^{  (K) } \map M_1^{(K-1)}\cdots    \map M_1^{(1)}  \, .
\end{align*}
The measurement $\st j  = \{j_k\}_{k=1}^{K+1}$ associated to  the instrument $\{\map J_k\}_{k=1}^{K+1}$ is the desired joint measurement: indeed,  and using Eqs.   (\ref{prob}) and (\ref{aa})  we obtain  $(j_k|  =  (e|  \map  M_0^{(k)}   \map M_1^{(k-1)}  \cdots  \map M_1^{(1)}   =    ( m_k|   \map M_1^{(k-1)}  \cdots  \map M_1^{(1)}     =   (m_k|$ for every $k  \in   \{1,\dots,  K\}$.   In addition, the measurement $\st j$ is sharp.  Indeed, if a measurement $\{r_{kl}\}$  is a refinement of $\st j$, \ie  $  \sum_l  r_{kl}  =  j_k$ for all $k$,  then $\st r$ is also a refinement of the sharp measurement  $\st m^{(k')}$ for every fixed $k'$.   Hence, one has
\begin{align*}
(  r_{kl}|  \map M_0^{(k)}    &=  (r_{kl}  |   \qquad  \forall k \in  \{1,\dots K\}  \\
(  r_{kl} | \map M_1^{(k')}    & =  (r_{kl}|  \qquad
\forall k, k' \in  \{1,\dots K\}   ,  \,  k\not= k' \, .
\end{align*}
Using this fact and the definition of $\map J_k$ it is immediate to obtain the relation $(  r_{kl}|  \map J_k  =  (r_{kl}|$ for every $k,l$.   Thanks to Eq. (\ref{alte}), this proves that the instrument $\{  \map J_k\}$ is sharp, and so is the corresponding measurement $\st j$.   \qed

\medskip

The ability to combine orthogonal effects into a single measurement is a powerful  asset.   As we already observed,  it implies CE at its basic level.  In the following we show that it can be used also to obtain the whole CE  hierarchy.

\medskip

{\bf Orthogonality of product effects.}
 In a causal theory  the information available at a given moment of time can be used to make decisions about the settings of future experiments, thus allowing for adaptive measurements where the choice of setting for a system   $B$ depends on the outcome of a measurement on system  $A$.   In particular, if $\{m_x\}$ is a sharp measurement on $A$ and $\{  n^{(x)}_y\}_{y\in\set Y}$ is a sharp measurement for every $x$, then $\{m_x\otimes n^{(x)}_y\}$ is a legitimate measurement.  Now, the Locality of Sharp Measurements implies that      $\{m_x\otimes n^{(x_0)}_y\}$ is sharp for every fixed $x_0$.  Since $x_0$ if arbitrary, this means that each effect $m_x\otimes  n^{(x)}_y$ is sharp and that two effects  $m_x\otimes  n^{(x)}_y$  and $m_{x'}\otimes  n^{(x')}_{y'}$ are orthogonal unless $x=x'$ and $y=y'$.


Thanks to this observation,  it is easy to see that  every level of the CE hierarchy is satisfied.  The key is to note that if two strings of outcomes $\st y$ and $\st y'$ are exclusive,  then the corresponding effects $P_{\st y}$  and $P_{\st y'} $ are orthogonal.  This is clear because, by definition,  the effects corresponding to two exclusive strings are of the form $   P_{\st y}   =  m_{y_i}  \otimes  n$  and $P_{\st y'}  =   m_{y_i'} \otimes n'$ where  the effects   $m_{y_i}$ and $m_{y_i'}$ are orthogonal and the effects $n  =  \otimes_{j\not  = i}  m_{y_j}$ and $  n'  =  \otimes_{j\not  = i}  m_{y_j} $ are sharp  thanks to the Locality of Sharp Measurements.    Using our result about product effects, we then have that   $P_{\st y}$ and $P_{\st y'}$ are orthogonal.
 Now,  a set of mutually exclusive strings  $\set E$ corresponds to a set of mutually orthogonal effects  $\{  P_{\st y}\}_{\st y\in\set E}$.   Since mutually orthogonal effects can be combined into a joint measurement, the probabilities $p_L (\st y)  =   (  P_{\st y}|   \rho^{\otimes L}) $ obey  the bound $\sum_{\st y\in\set E}  p_L(\st y)\le 1$, meaning that the theory satisfies the $L$-th level of the CE hierarchy for arbitrary $L$.

Note that the same argument can be used to prove the validity of LO for the probability distributions generated in a scenario where all parties perform sharp measurements.   In this scenario,  two locally orthogonal events $(\st x,\st y)$ and $(\st x',\st y')$ correspond to two orthogonal effects   $P^{(\st x)}_{\st y}$  and $P^{(\st x')}_{\st y'}$, for exactly the same reason mentioned above.  Hence,  the joint measurability of orthogonal effect implies the bound $  \sum_{(\st x,\st y)  \in\set O} p(\st x|\st y)\le 1$ for every set $\set O$ of locally orthogonal events.    In other words, all the  probability distributions  generated by sharp measurements obey LO.

\medskip

{\bf Reduction to sharp measurements.}  While CE applies only to sharp measurements, LO applies to arbitrary measurements.     This is because  every probability distribution that we can encounter in our theory is a probability distribution generated by sharp measurements.   This fact can be seen as follows:  combining the Fundamental Sharpness with the Locality of Sharp Measurements,   one can show that for  every party $i$ and every measurement $\st m^{(i,x_i)}$, there exists an ancilla   $A_i$, a state of $A_i$, call it $\sigma_i$, and a  sharp measurement $\st M^{(i,x_i)}$  such that $     \left(  m^{(i,x_i)}_{y_i}  |    \rho_i   \right)   =  \left(    M^{(i,x_i)}_{y_i} |   \rho_i\otimes \sigma_i \right)  $ for every $x_i$ and for every $y_i$
(cf. Supplementary Note 2 for the proof).     Now, since all  measurements that party $i$ can perform can be replaced by sharp measurements by adding an ancilla in a fixed state $\sigma_i$, the input-output distribution $p(\st y|\st x)$ generated by arbitrary measurements on the state $\rho$ coincides with the input-output distribution generated by sharp measurements on the state $\rho'  =  \rho  \otimes  \sigma_1\otimes \dots\otimes \sigma_N$.
In other words, at the level of correlations there is no difference between sharp and non-sharp measurements.  Thanks to this fact, deriving LO for sharp measurements is equivalent to deriving LO   for  arbitrary  measurements.

\appendix
\chapter{Coherence Distillation Procedure}\label{App:coherence}
A coherence distillation procedure refers to a series of incoherent operations by which a large number of identical partly coherent states can be transformed into a smaller number of maximally coherent states. This chapter  introduces a coherence distillation procedure for pure qubit states. With $N$ copies of states $\ket{\psi} = (\alpha \ket{0} + \beta \ket{1})$, we show that we can asymptotically obtain $l$ copies of $\ket{\Psi_2} = (\ket{0} + \ket{1})/\sqrt{2}$, where $l$ and $N$ satisfy $l/N \approx R_I(\ket{\psi})$. The derivation method can be generalized  to an arbitrary dimension.

\section{Coherence distillation: qubit}
First we prepare $MN$ copies of a partially coherent qubit state which will be uniformly divided into $M$ groups. The initial state of each group can be expressed according to
\begin{equation}
\ket{\psi}^{\otimes N}=\left(\alpha \ket{0}+\beta \ket{1}\right)^{\otimes N}.
\end{equation}
A binomial expansion on the computational basis contains $N+1$ distinct coefficients $\beta^N, \alpha^{1}\beta^{N-1}, \dots, \alpha^N$. Thus we can divide the original $2^N$-dimensional Hilbert space into $N+1$ subspaces according to the coefficients. For the $k$th coefficient $\alpha^{N-k}\beta^k$, the corresponding $k$th subspace is a  $D_k = C_N^k$ dimensional Hilbert space, whose basis are denoted by
\begin{equation}
\alpha^{N-k}\beta^k:\left\{\ket{e_1^k}, \ket{e_2^k}, \cdots , \ket{e_{D_k}^k}\right\}.
\end{equation}
When considering the computational basis, $\ket{e_i^k}$ $(i=, 1, 2, \cdots, D_k)$ is an N-qubit basis with $(N-k)$ $\ket{0}$s and $k$ $\ket{1}$s.

Next, we perform a projection measurement on $\ket{\psi}^{\otimes N}$ to the subspaces. In our case, the projection operator that maps onto the $k$th subspace is given by
\begin{equation}
P_k=\ket{e_1^k}\bra{e_1^k}+\ket{e_2^k}\bra{e_2^k}+\cdots+\ket{e_{D_k}^k}\bra{e_{D_k}^k}.
\end{equation}
The probability of obtaining the $k$th outcome is
\begin{equation}
p_k=C_N^k\left|\alpha\right|^{2N-2k}\left|\beta\right|^{2k}.
\end{equation}
Note that as the coefficients for the expansion are the same, the post-selection of the $k$th outcome corresponds to a maximally coherent state $\ket{\Psi_{D_k}}$  of dimension $D_k$.

If $D_k = 2^r$, we can directly convert to $r$ copies of $\ket{\Psi_2}$ as desired.
Or, we can repeat this process $M$ times, and take the tensor product of the post selected state to obtain a maximally coherent state of dimension $D$,
 \begin{equation}\label{}
      \ket{\Psi_{D}} = \ket{\Psi_{D_{k_1}}} \ket{\Psi_{D_{k_2}}} \dots\ket{\Psi_{D_{k_M}}},
\end{equation}
where $k_j$ is the outcome of the $j$th measurement, and the total dimension is $D = D_{k_1}D_{k_2}\cdots D_{k_M}$.
The total dimension $D$ will lie between $2^r$ and $2^r(1+\epsilon)$ $(0<\epsilon<1)$ for some power $r$. It can be proved \cite{Bennett96} that as $M$ increases, $\epsilon$ will asymptotically approach $0$.

Therefore, we can perform a second projection measurement to the $2^r$-dimensional Hilbert subspace and directly get obtain a final state
\begin{equation}
\ket{\Psi_2}^{\otimes r}=\left(\frac{1}{\sqrt{2}}\left(\ket{0}+\ket{1}\right)\right)^{\otimes r}
\end{equation}
Using the above procedure, $NM$ copies of a partly coherent qubit state $\alpha \ket{0}+ \beta \ket{1}$ have been distilled into $r$ copies of maximally coherent state $\ket{\Psi_2} = \ket{0}+\ket{1}$.

In the following, we will show that all the operations of the distillation protocol are incoherent operations. In addition, we will show that the number of distilled maximally coherent state $r$ and the number of initial qubit $MN$ satisfy the relation $NMR_I(\ket{\psi}) \approx r$.
\subsection{Incoherent operations}
As the only operations are the two projective measurements, we only need to prove the following lemma.
\begin{lemma}
Suppose an $n$-dimensional Hilbert space has a complete basis $I_n = \{\ket{1}, \ket{2},\cdots, \ket{n}\}$. A projection measurement that divides $I_n$ into its complementary subsets are incoherent operations on the basis of $I_n$.
\end{lemma}
\begin{proof}
Suppose that the basis $I_n$ is divided into $m$ complementary subsets $I_{n_1}, I_{n_2}, \cdots, I_{n_m}$, such that $I_{n_\alpha}\cap I_{n_\beta}= \emptyset$, for all $\alpha \neq \beta \in \{1,2,\cdots,m\}$, and $I_n = I_{n_1}\cup I_{n_2}\cup \cdots\cup I_{n_m}$. Denote the projector that projects onto the $I_{n_\alpha}$ subspace by $P_\alpha$. Thus, we can show that the projection measurement is a set of Kraus operators $\{\hat{P}_\alpha\}$ that satisfy $\hat{P}_\alpha^\dag\hat{P}_\beta = \delta_{\alpha,\beta}\hat{P}_\alpha$ and $\sum_\alpha P_\alpha = I_n$.
To prove the projection measurement to be an incoherent operation, we additionally need to show that $\hat{P}_\alpha\mathcal{I}_n \hat{P}_\alpha^\dag \subset \mathcal{I}_n$, where $\mathcal{I}_n$ is the set of all incoherent states that can be represented by $\delta=\sum_{i=1}^n \delta_i \ket{i}\bra{i}$. As the definition of $P_\alpha$, we have
\begin{equation}\label{}
\begin{aligned}
  P_\alpha \ket{i} &= \delta(\ket{i}\in I_{n_\alpha})\ket{i},
\end{aligned}
\end{equation}
where $\delta(\ket{i}\in I_{n_\alpha})=1$ if $\ket{i} \in I_{n_\alpha}$ and $\delta(\ket{i}\in I_{n_\alpha})=0$ otherwise.
Thus, we can show that for an arbitrary state $\delta=\sum_{i=1}^d \delta_i \ket{i}\bra{i}\in \mathcal{I}_n$, we have
\begin{equation}
\begin{aligned}
\hat{P}_\alpha \delta \hat{P}_\alpha^\dag &=\hat{P}_\alpha\sum_{i=1}^n \delta_i \ket{a_i}\bra{a_i} \hat{P}_\alpha^\dag \\
&= \sum_{i=1}^n \delta_i\delta(\ket{i}\in I_{n_\alpha}) \ket{a_i}\bra{a_i}\in \mathcal{I}_n.
\end{aligned}
\end{equation}
\end{proof}
Therefore, we have proven that the operations in the distillation protocol are incoherent.
\subsection{Coherence loss}
To explain why we have $NMR_I(\ket{\psi}) \approx r$, we only need to consider the coherence loss during the distillation process.
The initial state in each group can be rewrite as
\begin{equation}
\ket{\psi}^{\otimes N}=\sum_{k=0}^{N} \sqrt{C_N^k} \alpha^{N-k}\beta^k \ket{\Psi_{D_k}},
\end{equation}
where $\ket{\Psi_{D_k}}$ is a maximally coherent state of dimension $D_k$.
Thus the density matrix of the initial state is
\begin{equation}
\rho=\sum_{k,k'} \sqrt{C_N^k} \sqrt{C_N^{k'}}\alpha^{N-k}\beta^k (\alpha^*)^{N-k'}(\beta^*)^{k'} \ket{\Psi_k} \bra{\Psi_{k'}}
\end{equation}

As the coherence of $\rho$ is defined by its von Neumann entropy of its diagonal terms, we first look at $\rho^{diag}$. That is,
\begin{equation}\label{1}
\begin{aligned}
\rho^{diag}=&\sum_{i=1}^{2^N}\bra{e_i}\rho\ket{e_i}\ket{e_i}\bra{e_i}\\
=&\sum_{i=1}^{2^N}\sum_{k,k'}\sqrt{C_N^k} \sqrt{C_N^{k'}}\alpha^{N-k}\beta^k (\alpha^*)^{N-k'}(\beta^*)^{k'} \\
&\bra{e_i}{\Psi_k}\rangle \bra{\Psi_{k'}}{e_i}\rangle\ket{e_i}\bra{e_i}
\end{aligned}
\end{equation}
Here, we can see that when $k\neq k'$, $\bra{e_i}{\Psi_k}\rangle \bra{\Psi_{k'}}{e_i}\rangle=0$. Therefore Eq.~\eqref{1} can be simplified as
\begin{equation}
\rho^{diag}=\sum_{k=0}^{N} C_N^k\left|\alpha\right|^{2N-2k}\left|\beta\right|^{2k}\left(\ket{e_i}\bra{e_i}\right)^{diag}= \sum_{k=0}^{N} p_k \rho_k^{diag}
\end{equation}
Here, $\rho^{\mathrm{diag}}$ has the decomposition $\{p_k, \rho_k^{diag}\}$. Thus, we have
\begin{equation}\label{2}
S(\rho^{diag})= H(p_k)+ \sum_{k=0}^{N}p_k S(\rho_k^{diag})
\end{equation}
where $S(\rho^{diag})$ is the von Neumann entropy of $\rho^{diag}$ and $H(p_k)$ is the Shannon entropy. Considering our coherence (intrinsic randomness) definition, Eq.~\eqref{2} is equivalent to
\begin{equation}
C(\rho)=H(p_k)+\sum_{k=0}^{N} p_kC(\rho_k),
\end{equation}
where $C(\rho)$ is the average initial coherence and $\sum_{k=0}^{N} p_kC(\rho_k)$ is the average coherence left after the first projection measurement.
Therefore, the coherence loss in the first operation is
\begin{equation}
H(p_k)=-\sum_{k=0}^{N}p_k \log_2{(p_k)}\leq \log_2{N}.
\end{equation}
The coherence loss for the second projection measurement can be easily estimated by $\log_2(1+\epsilon)\approx \epsilon$. Thus the total coherence loss has an upper bound given by
\begin{equation}
M\log_2N+\log_2(1+\epsilon)
\end{equation}
which is negligible relative to the initial coherence $MNC(\ket{\psi})$ when $M$ and $N$ are large.

\section{General definition}
Generally, when considering the intrinsic randomness of multiple copies of $\rho$, we can define the average intrinsic randomness in a manner similar  to the definition of entanglement cost \cite{hayden2001asymptotic, Plenio2007Measures, Horodecki09} by
\begin{equation}\label{Eq:Nrho}
  R_I^C(\rho) = \inf\left\{r:\lim_{N\rightarrow\infty}\left[\inf_{\Phi_{\mathrm{ICPTP}}}D\left(\rho^{\otimes N}, \Phi_{\mathrm{ICPTP}}\left(\ket{\Psi_{2^{rN}}}\right)\right)\right]=0\right\},
\end{equation}
where $D(\rho_1,\rho_2)$ is a suitable measure of distance, which, for instance, could be the trace norm. In this case, the intrinsic randomness is understood as the average coherence cost in preparing $\rho$. Compared to the definition of the regulated entanglement of formation \cite{hayden2001asymptotic}, we conjecture that  $R_I^C(\rho)$ equals the regulated intrinsic randomness measure,
\begin{equation}\label{Eq:RNrho}
  R_I^\infty(\rho) = \lim_{N\rightarrow\infty}\frac{ R_I\left(\rho^{\otimes N}\right)}{N}.
\end{equation}

In the other direction, we can apply intrinsic operations to transform $N$ non-maximally coherent copies of $\rho$ to $l$ maximally coherent state $\ket{\Psi_2}$.
Similarly, we can define the distillable coherence by the supremum of $l/N$ over all possible distillation protocols \cite{Rains98, Plenio2007Measures, Horodecki09},
\begin{equation}\label{Eq:Nrho}
  R_I^D(\rho) = \sup\left\{l:\lim_{N\rightarrow\infty}\left[\inf_{\Phi_{\mathrm{ICPTP}}}D\left(\Phi_{\mathrm{ICPTP}}\left(\rho^{\otimes N}\right)-\ket{\Psi_{2^{lN}}}\right)\right]=0\right\}.
\end{equation}
This distillable coherence $R_I^D(\rho)$ can thus be considered as the amount of intrinsic randomness when a quantum extractor is performed before measurement, as shown in the main context. For a general reasonable regularized coherence measure $C_I^\infty(\rho)$ similar to Eq.~\eqref{Eq:RNrho}, we conjecture that the two measures $R_I^D$ and $R_I^C$ are equivalent for all possible distance measures. They serves as two extremal measures, such that,
$R_I^D \leq C_I^\infty \leq R_I^C$ for all regularized $C_I^\infty$. In addition, similar to entanglement measures \cite{vidal2000entanglement, Horodecki00, donald2002uniqueness}, we show in the next section that the coherence measure for pure states is unique under regularization

\section{A unique measure for pure states}
In this section, we show that the measure of randomness is unique for pure quantum states.

First note that $R_I^D(\rho) \leq R_I^C(\rho)$. Otherwise, we could first distill $N$ copies of $\rho$ into $NR_I^D(\rho)$ copies of $\ket{\Psi_2}$, and then convert into $NR_I^D(\rho)/R_I^C>N$ copies of $\rho$.
For a pure state $\rho$, we have already give a distillation protocol such that $R_I^D(\rho) = R_I^C(\rho)$.
Now, suppose that $C_I$ is a coherence measure for a single quantum state. By regularization, the coherence measure is given by
\begin{equation}\label{Eq:RNrho2}
  C_I^\infty(\rho) = \lim_{N\rightarrow\infty}\frac{ C_I\left(\rho^{\otimes N}\right)}{N}.
\end{equation}
Suppose that $C_I(\ket{\Psi_d}) = \log_2d$, we now prove that $C_I^\infty(\rho) = R_I^D(\rho) = R_I^C(\rho)$ for pure state $\rho$.
\begin{proof}
For a given pure state $\rho$, suppose that $C_I^\infty(\rho) < R_I^D(\rho)$, that is, $C_I^\infty(\rho) = R_I^D(\rho) - \Delta$, where $\Delta$ is a finite positive number. After the distillation process, we can convert $N$ copies of $\rho$ into approximately $NR_I^D(\rho)$ copies of $\ket{\Psi_2}$. From the main context, we know that the remaining coherence is $NR_I^D(\rho)$. For the $C_I^\infty(\rho)$ measure, the coherence after distillation is also $NR_I^D(\rho)$, while the initial coherence is given by $NC_I^\infty(\rho) = NR_I^D(\rho) - N\Delta$. As $\Delta$ is a finite positive number, the distillation process increases coherence, which leads to a contradiction.

If $C_I^\infty(\rho) > R_I^D(\rho)$, we can follow a similar method by considering the transformation $NR_I^C(\rho)$ copies of $\ket{\Psi_2}$ into $N$ copies of $\rho$. The contradiction originates from checking the coherence increase with the $C_I^\infty(\rho)$ measure during the incoherent transformation process.
\end{proof}

\chapter{SI-QRNG}
This chapter discusses finite size effect of the source independent quantum random number generator.
\section{Calculation of the number of effective $X$-basis measurements} \label{App:numXmeas}
In this appendix, we show that in the asymptotic limit, the number of effective $X$-basis measurements is independent of $n$. Our starting point is Eq.~\eqref{eq:Ptheta} and $\varepsilon_\theta < 2^{-100}$. Notice that normally $n$ is smaller than $10^{12} < 2^{40}$ to ease fast post-processing; thus, the term $1/\sqrt{n}$ and the other polynomial terms in Eq.~\eqref{eq:Ptheta} play a relatively small role in making $\varepsilon_\theta < 2^{-100}$. In the following, we consider only the exponent in Eq.~\eqref{eq:Ptheta}.

For ease of notation, let $x=e_{bx}$, $y=e_{bx}+\theta$ and $q=q_x$. Then the exponent of Eq.~\eqref{eq:Ptheta} becomes
\begin{equation*}
n[H((1-q)y+q x)-q H(x)-(1-q) H(y)]
\end{equation*}
and the inequality $\varepsilon_\theta < 2^{-100}$ is approximately equivalent to
 \begin{equation}
 \begin{aligned}
& n [q(H((1-q)y+q x)-H(x))+ \\
 & (1-q) (H((1-q)y+q x )- H(y))] \ge 100.
  \end{aligned}
\label{eq:A1}
 \end{equation}
Since $q$ is very small, one can make three approximations:
\begin{equation} \label{firstapprox}
H ((1-q)y+qx )-H(y)\approx -H'(y)q(y-x),
 \end{equation}
 \begin{equation} \label{secondapprox}
q[H((1-q)y+qx)-H(x)]\approx q(H(y)-H(x))
 \end{equation}
and
\begin{equation} \label{thirdapprox}
q^2 \approx 0.
 \end{equation}
Then, by applying Eqs.~\eqref{firstapprox} and \eqref{secondapprox}, the inequality \eqref{eq:A1} becomes
 \begin{equation}
n[q(H(y)-H(x))-(1-q) (H'(y)q(y-x))] \gtrsim 100.
 \end{equation}
 Applying Eq.~\eqref{thirdapprox} yields
 \begin{equation}
n[q(H(y)-H(x))- H'(y)q(y-x)] \gtrsim 100,
 \end{equation}
 and rearranging terms, we have
 \begin{equation}
q \gtrsim \frac{100}{n[H(y)-H(x)-H'(y)(y-x)]},\\
 \end{equation}
 Substituting the definitions of $x$ and $y$, we obtain
 \begin{equation}
 q \gtrsim  \frac{100}{n[H(e_{bx}+\theta)-H(e_{bx})-H'(e_{bx}+\theta)\theta]}.
 \end{equation}
Finally, we substitute $q=n_x/n$ and get
\begin{equation} \label{eq:c}
n_x \approx \frac{100}{H(e_{bx}+\theta)-H(e_{bx})-H'(e_{bx}+\theta)\theta},
\end{equation}
which is independent of $n$.

\section{Proof of the random sampling property for a type of QRNG input after loss}\label{app:sampling}
In this appendix, we first restate the setting. In the idealistic protocol, the measurement device chooses its measurement basis after confirming that the state received from the source is not a vacuum (or equivalently, not lost). In practice, confirming whether a state is
a vacuum is usually done by observing whether detectors in the measurement device click or not. Thus, it is desirable for the measurement device to choose its basis before confirming whether loss happens.

We prove that for a specific input that defines the measurement basis choices before the potential loss, the positions of $n_x$ valid $X$-basis measurements (after excluding loss events) are randomly drawn from the positions of the total of $n$ valid measurements. This proves that the random sampling technique from  Fung {\it et al.} can still be applied  when the measurement basis is chosen before the loss.

For ease of presentation, we state the input that specifies the measurement choices before the loss as follows. The input is a string of length $N=N_x+N_z$ that contains $N_x$ 0s and $N_z$ 1s. The $\binom{N}{N_z}$ possibilities for choosing the positions of $N_z$ 1s from the total $N_x+N_z$ positions are equally likely. Here, 0 stands for an $X$-basis measurement and 1 stands for a $Z$-basis measurement.  After loss, the numbers of valid $X$-basis measurements and $Z$-basis measurements are denoted by $n_x$ and $n_z$, respectively, with a total string length of
\begin{equation}
n=n_x+n_z.
\end{equation}
We need to show that the output is uniform for the $\binom{n_x+n_z}{n_z}$ possibilities of choosing the positions of $n_z$ 1s from the total $n$ positions.

The proof proceeds through a symmetry argument. The input is symmetric, i.e., if we exchange the indices of two positions, the distribution will not change. Suppose that the initial positions are $1,2,\dots, n$ and the probability of choosing specific positions for $N_z$ 1s from the total $N$ positions is
\begin{equation}
p=\frac{1}{\binom{N_x+N_z}{N_z}}.
\end{equation}
 For ease of presentation, denote the left positions after loss as $i_1<i_2< \dots<i_n$. Then each possibility with $n_x$ 0s in the left $n$ positions has the same probability
 \begin{equation}
 p_1=p \times \binom{N-n}{N_x-n_x},
 \end{equation}
which proves our claim.

As a side remark, we could see that the proof does not depend on whether the loss is basis dependent or independent. Thus, the same property also holds for a more general class of losses that could be useful in other settings. Another remark is that independent and identically distributed input also satisfies the property, as in the work of Fung {\it et al}.

\section{Random seed dilution}\label{app:input}
The input is either given directly or expanded from a uniformly random seed. Here, we provide a method for performing the expansion. The expansion is straightforward since the input is also uniformly random within its support. We can simply map a uniform seed of length $\log \binom{N}{c_1}$ bijectively to the input support, which is the $\binom{N}{c_1}$ possibilities of choosing the positions of $c_1$ 0s from the string of of length $N$. Then, we have obtained the desired input. Furthermore, note that this construction is deterministic; thus, input randomness is only needed for the uniformly random seed of length $n$.

For the input of our protocol, the ratio of the initial random seed length to the number of runs $N$ becomes negligible as $N$ goes to infinity because the number of $X$-basis measurements $c_1$ is a constant, as derived in Appendix \ref{App:numXmeas}. More precisely, the min-entropy of the input as well as the length of the uniformly random seed has an upper bound given by
\begin{equation}
\log \binom{N}{c_1}\le c_1 \log N.
\end{equation}
Note that since the detector completely controls this random seed length, calculating the exact input min-entropy is possible. This is very different from estimating the error rate in the finite-key analysis section, in which we can only estimate the range of the error rate with a high  probability of success. Apart from the input specified in the main text, independent and identically distributed bit strings are also a possible choice  for the input. Finally, we remark that the reason to include this input seed length analysis is to make our QRNG composable.

\chapter{Proof for randomness requirement for the CH inequality}
This chapter is the proof of the randomness requirement for the CH inequality.
\section{Proof for finite strategies of choosing input settings}\label{App:proof1}
As we discussed in Sec. 9.2, there are two levels of strategies. One is the strategy of choosing the input settings and the other is about the outputs conditioned on inputs of Alice and Bob. As there are finite deterministic strategies of Alice and Bob, here, we prove that the strategies of choosing input settings is finite and can be characterized by all the possible optimal strategies of Alice and Bob.

Essentially, even the strategies of Alice and Bob are finite, the strategies of choosing input settings can always be infinite. Here, what want to prove is that any optimal strategy (including both levels) can be realized with finite strategies of choosing input settings.

Suppose there exist an optimal strategy that gives maximal CH value with LHVMs. For this strategy, we suppose there are finite strategies of choosing the input settings (the proof for infinite case follows similarly). Then, it is easy to check that for a given $\lambda$ and hence $(p_0(\lambda), p_1(\lambda), p_2(\lambda), p_3(\lambda))$  in the optimal strategy, the optimal strategy for the output of Alice and Bob should be from the set  Eq.~\eqref{eq:chq}. This also proves why we only take account of the possibly optimal deterministic strategies of Alice and Bob.

Now, suppose that there exist $m$ strategies of $\lambda$ of choosing input settings for the first strategy of Alice and Bob, $(p_2-p_0)/2$, that is,
\begin{equation}\label{}
  \begin{array}{ccc}
    \lambda_{1}^1 :& q(\lambda_{1}^1),&(p_0(\lambda_{1}^1), p_1(\lambda_{1}^1), p_2(\lambda_{1}^1), p_3(\lambda_{1}^1)), \\
    \lambda_{1}^2 :& q(\lambda_{1}^2),&  (p_0(\lambda_{1}^2), p_1(\lambda_{1}^2), p_2(\lambda_{1}^2), p_3(\lambda_{1}^2)), \\
     & \dots \\
    \lambda_{1}^m :& q(\lambda_{1}^m),&  (p_0(\lambda_{1}^m), p_1(\lambda_{1}^m), p_2(\lambda_{1}^m), p_3(\lambda_{1}^m)).
  \end{array}
\end{equation}
Here the superscript denotes the $m$ strategies of $\lambda$ and the subscript denotes the strategy for Alice and Bob.
It is easy to see that we can always take an average of all the $m$ strategies without decreasing the Bell value and violate the constraints. In this case, we can define one $\lambda_1$ to the denote all the $\lambda_1^1$, $\lambda_1^2$, $\dots\lambda_1^m$. That is,
\begin{equation}\label{}
  \lambda_1: q(\lambda_1) = \sum_{t=1}^mq(\lambda_1^t), (p_0(\lambda_{1}), p_1(\lambda_{1}), p_2(\lambda_{1}), p_3(\lambda_{1})),
\end{equation}
where
\begin{equation}\label{}
  p_i(\lambda_1) = \frac{1}{q(\lambda_1)}\sum_{t=1}^mq(\lambda_1^t)p_i(\lambda_1^t), \forall i \in\{0,1,2,3\}.
\end{equation}
Thus, we show that the $m$ strategies for choosing input settings can be combined into one for any strategy of Alice and Bob. In the following, we prove this argument in more detail.

\begin{proof}
We use label $t$ to denote the $t$th strategy of choosing input settings for a given strategy of Alice and Bob, $j$ to denote the strategies of Alice and Bob, and $i$  to denote the number of inputs meaning the subscript of $(p_0(\lambda), p_1(\lambda), p_2(\lambda), p_3(\lambda))$.

We denote $\lambda_j^t$ to be the $t$th strategy of choosing input settings when the optimal strategy for Alice and Bob is $j$.
The prior probability for $\lambda$ and input settings of each strategy are denoted as $q(\lambda_j^t)$ and $p_i(\lambda_j^t)$, where $j\in\{1,2,3,4,5\}$ and $i\in \{0,1,2,3\}$, respectively. Denote the Bell value for the $j$th strategy to be $J_j$, which is linear function of $\{p_i(\lambda_j^t)\}$. Thus, the total Bell value is given by
\begin{equation}\label{}
  J = \sum_j \sum_t q(\lambda_j^t)J_j(p_i(\lambda_j^t)),
\end{equation}
and the constraints of $q(\lambda_j^t)$ and $p_i(\lambda_j^t)$ are given by,
\begin{equation}\label{infinite}
\begin{aligned}
  &\sum_{j,t} q(\lambda_j^t) p_i(\lambda_j^t) = 1/4, \forall i \\
  &\sum_i p_i(\lambda_j^t) = 1, \forall j,t \\
  &\sum_{j,t} q(\lambda_j^t) = 1.\\
\end{aligned}
\end{equation}

Just as mentioned above, we can add up $t$ by defining $q(\lambda_{j})$ and $p_i(\lambda_{j})$ by
\begin{equation}\label{normalize}
\begin{aligned}
  &q(\lambda_{j})=\sum_{t} q(\lambda_j^t),\forall j\\
  &p_i(\lambda_{j}) = \frac{\sum_t q(\lambda_j^t)p_i(\lambda_j^t)}{q(\lambda_{j})},\forall i,j.\\
\end{aligned}
\end{equation}
Take Eq.~\ref{normalize} into the Eq.~\ref{infinite}, consequently we find the constraints of $q(\lambda_{j})$ and $p_i(\lambda_{j})$ are given by
\begin{equation}\label{}
\begin{aligned}
 &\sum_{j} q(\lambda_{j}) p_i(\lambda_{j}) = 1/4, \forall i \\
  &\sum_i p_i(\lambda_{j}) = 1, \forall j \\
  &\sum_{j} q(\lambda_{j}) = 1.\\
\end{aligned}
\end{equation}

We should also note that the substitution in Eq.~\eqref{normalize} will not affect the Bell value,
\begin{equation}\label{}
\begin{aligned}
  J &= \sum_j q(\lambda_j)J_j(p_i(\lambda_j)),\\
  &=\sum_j  q(\lambda_j)J_j\left(\frac{\sum_t q(\lambda_j^t)p_i(\lambda_j^t)}{q(\lambda_{j})}\right),\\
  & = \sum_j \sum_t q(\lambda_j^t)J_j(p_i(\lambda_j^t)),
\end{aligned}
\end{equation}
where the last equality is because $J_j$ is a  linear function.

\end{proof}

\section{Optimal strategy of the CH test}
\subsection{General condition}\label{app:CH1}
In this section, we present the optimal strategy in order to maximizing $J_{\mathrm{CH}}^{\mathrm{LHVM}}$ defined in Eq.~\eqref{Eq:mathJlambda} under constraints defined in Eq.~\eqref{eq:constraints}.
\subsubsection{$Q=0$}\label{app:P}
For simplicity, we first consider the randomness requirement $P$ and set $Q$ to be 0. That is, the input randomness is upper bounded by $P$,
\begin{equation}\label{}
  0\leq p(x,y|\lambda) \leq P, \forall x,y,\lambda
\end{equation}

The Bell value $J_{\mathrm{CH}}^{\mathrm{LHVM}}$ with LHVMs is given by
\begin{equation}\label{Eq:appJ}
\begin{aligned}
  &J^{\mathrm{LHVM}}_{\mathrm{CH}}\\
  &= 4\{q(\lambda_1)(p_2(\lambda_1)-p_0(\lambda_1))/2 + q(\lambda_2)(p_1(\lambda_2)-p_0(\lambda_2))/2\\
  &+q(\lambda_3)(p_1(\lambda_3)-p_2(\lambda_3))/2+q(\lambda_4)(p_2(\lambda_4)-p_1(\lambda_4))/2\\
  &+q(\lambda_5)[(p_2(\lambda_5)+p_1(\lambda_5))/2-p_3(\lambda_5)]\}.
\end{aligned}
\end{equation}
Hereafter, we denote  $J_{\mathrm{CH}}^{\mathrm{LHVM}}$ by $J$ for simple notation. Group $J$
by the index of the strategies of Alice and Bob $p_i$, $i\in\{0,1,2,3\}$, instead of $\lambda_i$, then we have
\begin{equation}\label{eq:Proof2}
\begin{aligned}
  J &= 4(J_0 + J_1 + J_2 + J_3),
\end{aligned}
\end{equation}
where
\begin{equation}\label{eq:Proof3}
\begin{aligned}
  J_0 =& \frac{1}{2}[-q(\lambda_1)p_0(\lambda_1) -q(\lambda_2)p_0(\lambda_2)],\\
  J_1 =& \frac{1}{2}[q(\lambda_2)p_1(\lambda_2) +q(\lambda_3)p_1(\lambda_3) \\
  &- q(\lambda_4)p_1(\lambda_4) + q(\lambda_5)p_1(\lambda_5)],\\
  J_2 =& \frac{1}{2}[q(\lambda_1)p_2(\lambda_1) -q(\lambda_3)p_2(\lambda_3) \\
  &+ q(\lambda_4)p_2(\lambda_4) + q(\lambda_5)p_2(\lambda_5)],\\
  J_3 =& -q(\lambda_5)p_3(\lambda_5).\\
\end{aligned}
\end{equation}
And the constraints are given by
\begin{equation}\label{eq:appconstraints}
\begin{aligned}
  &\sum_j q(\lambda_j) p_i(\lambda_j) = 1/4, \forall i \\
  &\sum_i p_i(\lambda_j) = 1, \forall j, \\
  &\sum_j q(\lambda_j) = 1.\\
\end{aligned}
\end{equation}

In the following, we investigate the optimal strategy based on value of $P$.

(1) when $\frac{1}{4}\le P\le \frac{1}{3}$.

With the normalization condition of $p_i(\lambda_j)$, we can rewrite $J$ as
\begin{equation}\label{eq:Proof4}
\begin{aligned}
 J_0 =& -\frac{1}{8} + \frac{1}{2}[q(\lambda_3)p_0(\lambda_3) + q(\lambda_4)p_0(\lambda_4) + q(\lambda_5)p_0(|\lambda_5)],\\
  J_1 =& -\frac{1}{8} +\frac{1}{2}[q(\lambda_1)p_1(\lambda_1) +2q(\lambda_2)p_1(\lambda_2) +2q(\lambda_3)p_1(\lambda_3)\\
  &+2q(\lambda_5)p_2(\lambda_5)],\\
  J_2= & -\frac{1}{8} +\frac{1}{2}[2q(\lambda_1)p_2(\lambda_1) +q(\lambda_2)p_2(\lambda_2) +2q(\lambda_4)p_2(\lambda_4)\\
  &+ 2q(\lambda_5)p_2(\lambda_5)],\\
  J_3 =& -\frac{1}{4} +\frac{1}{2}[2q(\lambda_1)p_3(\lambda_1) +2q(\lambda_2)p_3(\lambda_2) +2q(\lambda_3)p_3(\lambda_3)\\
  &+2q(\lambda_4)p_3(\lambda_4) ].\\
\end{aligned}
\end{equation}

In this case, we can write $J$ by
\begin{equation}\label{Eq:appendixbeta}
  J = 4\left(-\frac{5}{8} + \sum_{ij} \beta_{ij} q(\lambda_j)p_i(\lambda_j)\right),
\end{equation}
where the coefficient is given in Table~\ref{Table:coe2}.
\begin{table}[hbt]
\centering
\caption{The coefficient of $p_i(\lambda_j)$ in the expression of J.}
\begin{tabular}{cccccc}
  \hline
   &$q(\lambda_1)$&$q(\lambda_2)$&$q(\lambda_3)$&$q(\lambda_4)$&$q(\lambda_5)$\\
   \hline
   $p_0$&$0$&$0$&$\frac{1}{2}$&$\frac{1}{2}$&$\frac{1}{2}$\\
   $p_1$&$\frac{1}{2}$&$1$&$1$&$0$&$1$\\
   $p_2$&$1$&$\frac{1}{2}$&$0$&$1$&$1$\\
   $p_3$&$1$&$1$&$1$&$1$&$0$\\
  \hline
\end{tabular}\label{Table:coe2}
\end{table}

Note that $p_i$ is upper bounded by $P$, then we have
\begin{equation}\label{eq:Proof5}
\begin{aligned}
  J_0 &\leq -\frac{1}{8} + \frac{P}{2}\left[q(\lambda_3)+ q(\lambda_4) + q(\lambda_5)\right],\\
  J_1 &\leq -\frac{1}{8} +\frac{P}{2}[q(\lambda_1) +2q(\lambda_2) +2q(\lambda_3)+ 2q(\lambda_5)],\\
  J_2 &\leq -\frac{1}{8} +\frac{P}{2}[2q(\lambda_1)+q(\lambda_2) +2q(\lambda_4)+ 2q(\lambda_5)],\\
  J_3 &\leq -\frac{1}{4} +\frac{P}{2}[2q(\lambda_1)+2q(\lambda_2) +2q(\lambda_3)+2q(\lambda_4)].\\
\end{aligned}
\end{equation}
Therefore, we have
\begin{equation}\label{eq:Proof6}
\begin{aligned}
  J &\leq -4\{\frac{5}{8} + \frac{5P}{2}[q(\lambda_1)+q(\lambda_2) +q(\lambda_3)+q(\lambda_4)+q(\lambda_5)]\}\\
  &= \frac{5}{2}(4P-1).
\end{aligned}
\end{equation}

In addition, we can see that the equality holds by simply letting $p_i(\lambda_j)$ to be $P$ for $\beta_{i,j}\neq0$ and $p_i(\lambda_j)$ to be $1-3P$ for $\beta_{i,j}=0$. This special strategy is valid when $P\leq 1/3$, we have to consider differently for the other cases.


(2) When $\frac{1}{3}\le P\le \frac{3}{8}$.

With the constraints defined in Eq.~\eqref{eq:Proof4}, we can also write $J$ as follows,
\begin{equation}\label{eq:Proof q1}
\begin{aligned}
J=&4\{-\frac{5}{8}+q(\lambda_1)[1-\frac{1}{2} p_1(\lambda_1)-p_0(\lambda_1)]\\
&+q(\lambda_2)[1-\frac{1}{2} p_2(\lambda_2)-p_0(\lambda_2)]\\
&+q(\lambda_3)[1-\frac{1}{2} p_0(\lambda_3)-p_2(\lambda_3)]\\
&+q(\lambda_4)[1-\frac{1}{2} p_0(\lambda_4)-p_1(\lambda_4)]\\
&+q(\lambda_5)[1-\frac{1}{2} p_0(\lambda_5)-p_3(\lambda_3)]\}.\\
\end{aligned}
\end{equation}
Then $J$ can be similarly expressed by
\begin{equation}\label{Eq:appendixbeta2}
  J = 4\left(\frac{3}{8} + \sum_{ij} \beta_{ij} q(\lambda_j)p_i(\lambda_j)\right),
\end{equation}
with coefficient defined in Table~\label{table:2}.

\begin{table}[hbt]
\centering
\caption{The coefficient of $p_i(\lambda_j)$ in the expression of J.}
\begin{tabular}{cccccc}
  \hline
   &$q(\lambda_1)$&$q(\lambda_2)$&$q(\lambda_3)$&$q(\lambda_4)$&$q(\lambda_5)$\\
   \hline
   $p_0$&$-1$&$-1$&$-\frac{1}{2}$&$-\frac{1}{2}$&$-\frac{1}{2}$\\
   $p_1$&$-\frac{1}{2}$&$0$&$0$&$-1$&$0$\\
   $p_2$&$0$&$-\frac{1}{2}$&$-1$&$0$&$0$\\
   $p_3$&$0$&$0$&$0$&$0$&$-1$\\
  \hline
\end{tabular}
\label{table:2}
\end{table}

The intuition to maximize Eq.~\eqref{Eq:appendixbeta2} is to assign smaller values to $p_i(\lambda_j)$ for smaller corresponding coefficients. Because $\frac{1}{3}\le P\le \frac{3}{8}$, we can see that
\begin{equation}\label{}
  \frac{1}{2} p_i(\lambda_j)+p_{i'}(\lambda_j)\ge \frac{1}{2}(1-2P), \forall i,j,j'.
\end{equation}
Therefore, the Bell value defined in Eq.~\eqref{Eq:appendixbeta2} can be upper bounded by
\begin{equation}\label{eq:Proof qi4}
\begin{aligned}
J&\le4\left[ -\frac{5}{8}+ \sum_i q(\lambda_i)\left(1-\frac{1-2P}{2}\right)\right],\\
&=4P-\frac{1}{2}.
\end{aligned}
\end{equation}
This equal sign can be achieved by following parameter:
\begin{equation}\label{eq:Proof qi5}
\begin{aligned}
&q(\lambda_1)=q(\lambda_2)=\frac{1}{2}-\frac{1}{8(1-2P)};\\
&q(\lambda_3)=q(\lambda_4)=\frac{1}{8(1-2P)}+\frac{1}{8P}-\frac{1}{2};\\
&q(\lambda_5)=1-\frac{1}{4P};\\
&p_0(\lambda_1)=p_0(\lambda_2)=0,p_0(\lambda_3)=p_0(\lambda_4)=p_0(\lambda_5)=1-2P;\\
&p_1(\lambda_1)=1-2P,p_1(\lambda_2)=p_1(\lambda_3)=p_1(\lambda_5)=P,p_1(\lambda_4)=0;\\
&p_2(\lambda_2)=1-2P,p_2(\lambda_1)=p_1(\lambda_4)=p_1(\lambda_5)=P,p_2(\lambda_3)=0;\\
&p_3(\lambda_1)=p_3(\lambda_2)=p_3(\lambda_4)=p_3(\lambda_4)=p,p_3(\lambda_5)=0;\\
\end{aligned}
\end{equation}

(3) When $P\ge \frac{3}{8}$.

For this case, we can easily see that maximal Bell value can be achieved to be $1$, which is the algebra maximum of $J$. We show in the following that the Bell value cannot exceed $1$.

From Eq.~\eqref{Eq:appendixbeta2}, we know that $J$ can be expressed by
\begin{equation}\label{eq:Proof qi1}
\begin{aligned}
J=&\frac{3}{2}-4N\\
\end{aligned}
\end{equation}
where $N$ denotes the part contribute negatively,
\begin{equation}\label{eq:Proof qi2}
\begin{aligned}
N
&=\frac{1}{2}\sum_i q(\lambda_i) p_0(\lambda_i)+\frac{1}{2}[q(\lambda_1)p_0(\lambda_1)+q(\lambda_1) p_1(\lambda_1)\\
&+q(\lambda_2)p_0(\lambda_2)+q(\lambda_2) p_2(\lambda_2)+2q(\lambda_3)p_2(\lambda_3)\\
&+2q(\lambda_4)p_1(\lambda_4)+q(\lambda_5)p_3(\lambda_5)]\\
&=\frac{1}{8}+\frac{1}{2}[q(\lambda_1)(p_0(\lambda_1)+p_1(\lambda_1))+q(\lambda_2)(p_0(\lambda_2)+p_2(\lambda_2))\\
&+2q(\lambda_3)p_2(\lambda_3)+2q(\lambda_4)p_1(\lambda_4)+q(\lambda_5)p_3(\lambda_5)]\\
&\ge \frac{1}{8}.
\end{aligned}
\end{equation}
Therefore, we show that $J\leq1$. The equal sign is satisfied with the following strategy
\begin{equation}\label{eq:Proof qi3}
\begin{aligned}
&q(\lambda_1)=q(\lambda_2)=0;\\
&p_0(\lambda_3)=p_0(\lambda_4)=p_0(\lambda_5)=\frac{1}{4};\\
&p_1(\lambda_3)=p_1(\lambda_5)=\frac{3}{8},p_1(\lambda_4)=0;\\
&p_2(\lambda_4)=p_2(\lambda_5)=\frac{3}{8},p_2(\lambda_3)=0;\\
&p_3(\lambda_3)=p_1(\lambda_4)=\frac{3}{8},p_3(\lambda_5)=0.\\
\end{aligned}
\end{equation}

\subsubsection{$Q\neq0$}\label{P,Q}
In this part, we consider the input randomness quantification of $p_i(\lambda_j)$ with both $P$ and $Q$, which are defined in Eq.~\eqref{eq:randomness}. In this case, we have
\begin{equation}\label{}
  Q\leq p_i(\lambda_j) \leq P, \forall x,y,\lambda
\end{equation}

Note that, if we we substitute $p_i(\lambda_j)$ by
\begin{equation}\label{Eq:apppp}
  p'_i(\lambda_j)=\frac{p_i(\lambda_j)-Q}{1-4Q},
\end{equation}
we can show that the constraints on $p'_i(\lambda_j)$ are given by
 \begin{equation}\label{eq:constraints qi}
\begin{aligned}
  &0\le p'_i(\lambda_j)\le P' = \frac{P-Q}{1-4Q},\forall i,j \\
  &\sum_i q(\lambda_i) p'_j(\lambda_i) = 1/4, \forall j\\
  &\sum_i p'_i(\lambda_j) = 1, \forall j.
\end{aligned}
\end{equation}
Compared to Eq.\ref{Eq:appJ}, if we replace $p_i(\lambda_j)$ by $p'_i(\lambda_j)$, we obtain a new Bell value $J'$,
\begin{equation}\label{eq:def}
\begin{aligned}
  &J'/4\\
  &= q(\lambda_1)(p'_2(\lambda_1)-p'_0(\lambda_1))/2 + q(\lambda_2)(p'_1(\lambda_2)-p'_0(\lambda_2))/2\\
  &+q(\lambda_3)(p'_1(\lambda_3)-p'_2(\lambda_3))/2\\
  &+q(\lambda_4)(p'_2(\lambda_4)-p'_1(\lambda_4))/2+q(\lambda_5)[(p'_2(\lambda_5)+p'_1(\lambda_5))/2\\
  &-p'_3(\lambda_5)]
\end{aligned}
\end{equation}
Because $p_i(\lambda_j)$ and $p'_i(\lambda_j)$ are related by Eq.~\eqref{Eq:apppp}, we can prove that
\begin{equation}\label{Eq:appJJ}
  J(p_i(\lambda_j)) = (1-4Q)J'(p'_i(\lambda_j)).
\end{equation}

Therefore, instead of considering both upper and lower bound of $p_i(\lambda_j)$ in the original Bell's inequality, we can equivalently consider the same Bell inequality with $p'_i(\lambda_j)$, which has upper bound $P$ and lower bound $0$. We have our result as follows,

(1) When $\frac{P-Q}{1-4Q}\le \frac{1}{3}$, that is $3P+Q\le 1$
\begin{equation}
\begin{aligned}
J(P,Q)&=(1-4Q)J\left(\frac{P-Q}{1-4Q},0\right)\\
 &=(1-4Q)\frac{5}{2}\left(\frac{4P-4Q}{1-4Q}-1\right)\\
 &= \frac{5}{2}(4P-1)
 \end{aligned}
\end{equation}

(2) When $\frac{1}{3}\le \frac{P-Q}{1-4Q}\le \frac{3}{8} $, that is $3P+Q\ge 1$ and $2P+Q\le \frac{3}{4}$
\begin{equation}
\begin{aligned}
J(P,Q)&=(1-4Q)J\left(\frac{P-Q}{1-4Q}, 0\right)\\
 &=(1-4Q)\left(4\frac{P-Q}{1-4Q}-\frac{1}{2}\right)\\
 &= 4P-2Q-\frac{1}{2}\\
 \end{aligned}
\end{equation}

(3) When $ \frac{P-Q}{1-4Q}\ge \frac{3}{8} $, that is $2P+Q\ge \frac{3}{4}$
\begin{equation}
\begin{aligned}
J(P,Q)&=(1-4Q)J\left(\frac{P-Q}{1-4Q}, 0\right)\\
 &= 1-4Q\\
 \end{aligned}
\end{equation}

Therefore, the optimal CH value $J^{\mathrm{LHVM}}_{\mathrm{CH}}$ with LHVMs,
\begin{equation}\label{}
  J^{\mathrm{LHVM}}_{\mathrm{CH}}(P,Q) =
  \left\{
  \begin{array}{cc}
    \frac{5}{2}(4P-1) & 3P+Q\le 1\\
    1-4Q & 2P+Q\ge \frac{3}{4}\\
     4P-2Q-\frac{1}{2} & \mathrm{else}
  \end{array}
  \right.
\end{equation}

\subsection{Factorizable condition}\label{app:fac}
Now, we consider the optimal strategy of the CH test with LHVMs under factorizable condition,
\begin{equation}\label{}
  p(i,j) = p_A(i)p_B(j).
\end{equation}
As we denote $p(i, j)$ by $p_{2*i+j}$, we have
 \begin{equation}\label{eq:Proof2}
\begin{aligned}
  p_0 & = p_A(0)p_B(0)\\
  p_1 & = p_A(0)p_B(1)\\
  p_2 & = p_A(1)p_B(0)\\
  p_3 & = p_A(1)p_B(1)\\
\end{aligned}
\end{equation}

\subsubsection{$Q=0$}
Similarly, we consider first the case with $Q=0$. In the following, we show that all the five possible strategies are upper bounded by $P-1/4$.

(1) When $P\le \frac{1}{2}$.

The result is based on the order of $p_1$, $p_2$, $p_3$, and $p_4$.

(a) $p_3 \geq p_2 \geq p_1 \geq p_0$ and $p_3 \geq p_1 \geq p_2 \geq p_0$.

This case is equivalent to $p_A(1) \geq p_A(0)$ and $p_B(1) \geq p_B(0)$. Thus we have $p_A(1)p_B(1) \leq P$. Amongst the five strategies, the biggest one is $ (p_2 - p_0)/2$, which can be upper bounded by
 \begin{equation}\label{eq:}
\begin{aligned}
  (p_2 - p_0)/2 & = (2p_A(1) - 1)(1 - p_B(1))/2,\\
   & \leq \frac{1}{2}[2p_A(1) + p_B(1) - 2p_A(1)p_B(1) - 1],\\
   & \leq P -\frac{1}{4}.
\end{aligned}
\end{equation}

(b) $p_1 \geq p_0 \geq p_3 \geq p_2$ and $p_3 \geq p_1 \geq p_2 \geq p_0$.

This case is equivalent to $p_A(0) \geq p_A(1)$ and $p_B(1) \geq p_B(0)$. Thus we have $p_A(0)p_B(1) \leq P$. Amongst the five strategies, the biggest one is  $(p_1 - p_2)/2$, which can be upper bounded by
 \begin{equation}\label{eq:}
\begin{aligned}
  (p_1 - p_2)/2 & = (p_A(0)p_B(1) -(1 - p_A(0))(1 - p_B(1)))/2,\\
   & \leq \frac{1}{2}[p_A(0) + p_B(1) - 1 ],\\
   & \leq P -\frac{1}{4}.
\end{aligned}
\end{equation}

(c) $p_2 \geq p_3 \geq p_0 \geq p_1$ and $p_2 \geq p_0 \geq p_3 \geq p_1$.

This case is equivalent to $p_A(1) \geq p_A(0)$ and $p_B(0) \geq p_B(1)$. Thus we have $p_A(1)p_B(0) \leq P$. Amongst the five strategies, the biggest one is  $(p_2 - p_1)/2$, which can be upper bounded by
 \begin{equation}\label{eq:}
\begin{aligned}
  (p_2 - p_1)/2 & =  (p_A(1)p_B(0) -(1 - p_A(1))(1 - p_B(0)))/2,\\
   & \leq \frac{1}{2}[p_A(1) + p_B(0) - 1 ],\\
   & \leq P -\frac{1}{4}.
\end{aligned}
\end{equation}

(d) $p_0 \geq p_1\geq p_2 \geq p_3$ and $p_0 \geq p_2 \geq p_1 \geq p_3$.

This case is equivalent to $p_A(0) \geq p_A(1)$ and $p_B(0) \geq p_B(1)$. Thus we have $p_A(0)p_B(0) \leq P$. Amongst the five strategies, the biggest one is  $(p_1 + p_2)/2 - p_3$, which can be upper bounded by
 \begin{equation}\label{eq:}
\begin{aligned}
  &(p_1 + p_2)/2 - p_3 \\
  & =p_A(0)(1 - p_B(0))/2 + (1 - p_A(0))p_B(0)/2 \\
  &- (1 - p_A(0))(1 -  p_B(0)),\\
   & \leq \frac{3}{2}[p_A(0) + p_B(0)] - 2p_A(0)p_B(0) - 1,\\
   & \leq P -\frac{1}{4}.
\end{aligned}
\end{equation}

Therefore, we show that all the strategies are upper bounded by $P-1/4$. Then the total Bell value
\begin{equation}\label{}
  J\leq 4(P-1/4) = 4P - 1,
\end{equation}
and the equal sign holds.

(2) When $P\ge \frac{1}{2}$.

It is easy to see that the maximal Bell value $J$ reaches 1 when $P\ge \frac{1}{2}$.

Consequently, we show the optimal Bell value  $J$ with LHVMs,
\begin{equation}\label{}
  J(P) =
  \left\{
  \begin{array}{cc}
    (4P-1) & P\le \frac{1}{2}\\
    1 & P> \frac{1}{2}\\
  \end{array}
  \right.
\end{equation}

\subsubsection{$Q\neq 0$}
We can follow a similar way in Appendix~\ref{P,Q} to take account of nonzero $Q$.

(1) When $\frac{P-Q}{1-4Q}\le \frac{1}{2}$, that is $P+Q\le \frac{1}{2}$
\begin{equation}
\begin{aligned}
J(P,Q)&=(1-4Q)J\left(\frac{P-Q}{1-4Q}, 0\right)\\
 &=(1-4Q)\left(4\frac{P-Q}{1-4Q}-1\right)\\
 &= 4P-1.
 \end{aligned}
\end{equation}

(2) When $\frac{P-Q}{1-4Q}> \frac{1}{2}$, that is $P+Q> \frac{1}{2}$
\begin{equation}
\begin{aligned}
J(P,Q)&=(1-4Q)J\left(\frac{P-Q}{1-4Q}, 0\right)\\
 &= 1-4Q\\
 \end{aligned}
\end{equation}

Thus, the Bell value $J^{\mathrm{LHVM, Fac}}_{\mathrm{CH}}$ with LHVMs under factorizable condition is,
\begin{equation}\label{}
  J^{\mathrm{LHVM, Fac}}_{\mathrm{CH}}(P,Q) =
  \left\{
  \begin{array}{cc}
    4P-1 & P+Q\le \frac{1}{2}\\
    1-4Q & P+Q> \frac{1}{2}\\
  \end{array}
  \right.
\end{equation}

\section{Optimal strategy of the CHSH inequality}
\subsection{CH and CHSH inequalities under NS}\label{app:equiva}
In this section, we prove that the CH and CHSH inequality are equivalent when NS is assumed. We refer to \cite{Rosset14} for detail discussion about the connection between CH and CHSH.

\begin{proof}
According to the inputs, we can divide the CHSH inequality into four parts. When inputs are $ij$, define :
\begin{equation}\label{}
  J_{ij} = \sum_{a,b\in\{0,1\}}(-1)^{ a + b + ij}p(a, b|i, j).
\end{equation}
Owing to the NS condition, $J_{ij}$ can be rewritten by probabilities with output $0$,
\begin{equation}\label{}
\begin{aligned}
  J_{ij} &= \sum_{a,b\in\{0,1\}}(-1)^{ a + b + ij}p(a, b|i, j),\\
          &=  (-1)^{ij}[p(0,0|i,j)-p(0,1|i,j)-p(1,0|i,j)+p(1,1|0i,j)],  \\
          &=   (-1)^{ij}[p(0,0|i,j)-(p_A(0|i)- p(0,0|i,j))-(p_B(0|j)\\
          &- p(0,0|i,j))+(p_A(1|i)-p_B(0|j)+p(0,0|i,j))],\\
          &=(-1)^{ij}[1+4p(0,0|i,j)-2p_A(0|i)-2p_B(0|j)].\\
\end{aligned}
\end{equation}

Therefore, we have
\begin{equation}\label{}
\begin{aligned}
J_{\mathrm{CHSH}}&= \sum_{ij} J_{ij},\\
&=\sum_{ij} (-1)^{ij}(1+4p(0,0|i,j)-2p_A(0|i)-2p_B(0|j)),\\
 &= 2+4[p(0,0|0,0)+p(0,0|0,1)+p(0,0|1,0)\\
 &-p(0,0|1,1)-p_A(0|0)-p_B(0|0)],\\
 &= 2+4J_{\mathrm{CH}}.
 \end{aligned}
\end{equation}

\end{proof}
Hence, under the NS assumption, the value of the CH and the CHSH inequality are linearly related.  To analyze the best LHVMs strategy for the CH test, we can therefore consider the CHSH Bell test instead.

\subsection{General condition}\label{app:CHSH}
Follwing the similar method described above, we first consider deterministic strategies, i.e., $p_A(0|x),p_B(0|y)\in\{0,1\}$ for the reason that any probabilistic LHVM could be realized with convex combination of deterministic ones.
Denote $p(i,j)$ as $p_{2*i + j}$, it is easy to show that the possible optimal deterministic strategies for $J_\lambda$ are
\begin{equation}
\begin{aligned}
   \{p_0+p_1+p_2-p_3&, p_0+p_1+p_3-p_2, \\
   p_0+p_2+p_3-p_1&, p_1+p_2+p_3-p_0\},
\end{aligned}
\end{equation}

\subsubsection{$Q=0$}
Here, we also first consider that $Q=0$.

(1) When $P\le \frac{1}{3}$.

We can show that, all the four strategies are upper bounded by $6P-1$. Take  the strategy of $p_0+p_1+p_2-p_3$ as an example,

\begin{equation}
\begin{aligned}
   p_0+p_1+p_2-p_3&\le P+P+P-(1-3P)\\
   &=6P-1.
\end{aligned}
\end{equation}

In this case, we can see that the CHSH value $J$ is upper bounded by $4(6P-1)$.

(2) When $P> \frac{1}{3}$.

In this case, LHVMs reaches the maximum Bell value, that is  $J$ can be 4.

Thus, the Bell value $J$ LHVMs is
\begin{equation}\label{}
  J(P) =
  \left\{
  \begin{array}{cc}
    24P-4 & P\le \frac{1}{3},\\
    4 & P> \frac{1}{3}.
  \end{array}
  \right.
\end{equation}

\subsubsection{$Q\neq0$}
For the case that $Q$ is nonzero, we apply the same transformation as Appendix~\ref{P,Q}. After the transformation defined in Eq.~\eqref{Eq:apppp}, the relation between $J(p_i(\lambda_j))$ and $J'(p'_i(\lambda_j))$ is given by
\begin{equation}\label{Eq:appJJ2}
  J(p_i(\lambda_j)) = (1-4Q)J'(p'_i(\lambda_j)) + 8Q.
\end{equation}
In this case, the optimal Bell value $J^{\mathrm{LHVM}}_{\mathrm{CHSH}}$ for the CHSH inequality with LHVMs is
\begin{equation}\label{}
  J^{\mathrm{LHVM}}_{\mathrm{CHSH}}(P,Q) =
  \left\{
  \begin{array}{cc}
    24P-4 & 3P+Q\le 1,\\
    4-8Q & 3P+Q\ge 1.
  \end{array}
  \right.
\end{equation}
And the optimal CH value $J^{\mathrm{LHVM, NS}}_{\mathrm{CH}}$ with LHVMs under NS is
\begin{equation}\label{}
  J^{\mathrm{LHVM, \mathrm{NS}}}_{\mathrm{CH}}(P,Q) =
  \left\{
  \begin{array}{cc}
    6P-3/2 & 3P+Q\le 1,\\
    1/2-2Q & 3P+Q\ge 1.
  \end{array}
  \right.
\end{equation}

\subsection{Factorizable condition}\label{app:NSFac}
In addition, we consider the factorizable condition,
\begin{equation}\label{}
  p(i,j) = p_A(i)p_B(j).
\end{equation}
In this case, we have
 \begin{equation}\label{eq:}
\begin{aligned}
  p_0 & = p_A(0)p_B(0)\\
  p_1 & = p_A(0)p_B(1)\\
  p_2 & = p_A(1)p_B(0)\\
  p_3 & = p_A(1)p_B(1)\\
\end{aligned}
\end{equation}

\subsubsection{$Q=0$}

(1) When $P\leq\frac{1}{2}$.

For the case that $Q=0$, $p_i$ are upper bounded by $P$ only. As the four strategies are symmetric, suppose that $p_3$ is the smallest one, which  is equivalent to $p_A(0) \geq p_A(1)$ and $p_B(0) \geq p_B(1)$. Thus we can see that $p_0+p_1+p_2-p_3$ is the largest strategy and is also upper bounded by
 \begin{equation}\label{eq:Proof2}
\begin{aligned}
  p_0+p_1+p_2-p_3 & = p_A(0) +(1 - p_A(0))(2p_B(0)-1)\\
   &=1-2(1-p_A(0))(1 - p_B(0)),\\
   & \leq 1-(1-2P),\\
   &=2P
\end{aligned}
\end{equation}

Thus, we see that all the strategies are upper bounded by $2P$. Then the Bell value is upper bounded by $8P$.

(2) When$P> \frac{1}{2}$.

We can easily see that LHVMs reaches the maximum Bell value, that is  $J$ can be 4.

Consequently, the CHSH Bell value $J$ with LHVMs with factorizable condition is given by,
\begin{equation}\label{}
  J^{\mathrm{LHVM}}_{\mathrm{CHSH}}(P) =
  \left\{
  \begin{array}{cc}
    8P & P\le \frac{1}{2}\\
    4 & P> \frac{1}{2}\\
  \end{array}
  \right.
\end{equation}

\subsubsection{$Q\neq0$}
For the case where $Q\neq0$, we can similarly derive our result. The  CHSH Bell value $J^{\mathrm{LHVM,Fac}}_{\mathrm{CHSH}}(P,Q)$ is given by
\begin{equation}\label{}
  J^{\mathrm{LHVM,Fac}}_{\mathrm{CHSH}}(P,Q) =
  \left\{
  \begin{array}{cc}
    8P & P+Q \le \frac{1}{2},\\
    4-8Q & P+Q> \frac{1}{2}.\\
  \end{array}
  \right.
\end{equation}
And the optimal CH value $J^{\mathrm{LHVM, NS, Fac}}_{\mathrm{CH}}$ with LHVMs under NS and factorizable condition is
\begin{equation}\label{}
  J^{\mathrm{LHVM, NS, Fac}}_{\mathrm{CH}}(P,Q) =
  \left\{
  \begin{array}{cc}
    2P-1/2 & P+Q \le \frac{1}{2},\\
    1/2-2Q & P+Q> \frac{1}{2}.\\
  \end{array}
  \right.
\end{equation}

\begin{singlespace}
\bibliographystyle{unsrt}
\bibliography{Thesisbib}
\bibliographystyle{plain}
\addcontentsline{toc}{chapter}{Bibliography}
\end{singlespace}

\begin{center}
{\Huge{\textbf{Publication}}}
\end{center}
\begin{enumerate}
\item	L. Chen, Z. Li, X. Yao, M. Huang, W. Li, H. Lu, X. Yuan, Y. Zhang, X. Jiang, C. Peng, L. Li, N. Liu, X.
Ma, C. Lu, Y. Chen, J. Pan, New source of ten-photon entanglement from thin BiB3O6 crystals.
Optica, 4(1): 77-83, 2017
\item X. Yuan, K. Liu, Y. Xu, W. Wang, Y. Ma, F. Zhang, Z. Yan, R Vijay, L. Sun, X. Ma, Experimental
quantum randomness processing using superconducting qubits, Phys. Rev. Lett. 117 1 10502,
2016;
\item X. Yuan, Z. Zhang, N. Lutkenhaus, X. Ma, Simulating single photons with realistic photon
sources; Phys. Rev. A 94, 062305, 2016;
\item X. Ma, X. Yuan, Z. Cao, B. Qi, Z, Zhang, Quantum random number generation, NPJ Quantum
Information 21 16021, 2016;
\item Z. Cao, H, Hongyi, X. Yuan, X. Ma, Source-independent quantum random number generation,
Phys. Rev. X 6 1 11020, 2016;
\item X. Yuan, Q. Mei, S. Zhou, X. Ma, Reliable and robust entanglement witness, Phys. Rev. A 93 (4),
042317, 2016;
\item G. Chiribella, X. Yuan, Bridging the gap between general probabilistic theories and the device-
independent framework for nonlocality and contextuality; Information and Computation, 250
15, 2016;
\item Q. Zhao, X. Yuan, X. Ma, Efficient measurement-device-independent detection of multipartite
entanglement structure, Phys. Rev. A 94 012343, 2016; (contributed equally);
\item X. Yuan, S. Assad, J. Thompson, J. Haw, V. Vedral, T. Ralph, P. Lam, C. Weedbrook, M. Gu,
Replicating the benefits of closed timelike curves without breaking causality, NPJ Quantum Information 1 15007, 2015;
\item X. Yuan, H. Zhou, Z. Cao, X. Ma, Xiongfeng; Intrinsic randomness as a measure of quantum coherence; Physical Review A 92 2 22124, 2015;
\item X. Yuan, Q. Zhao, X. Ma, Clauser-Horne Bell test with imperfect random inputs, Physical Review A 92 2 22107, 2015;
\item X. Yuan, Z. Cao, X. Ma, Randomness requirement on the Clauser-Horne-Shimony-Holt Bell test in the multiple-run scenario, Physical Review A 91 3 32111, 2015;
\item H. Zhou, X. Yuan, X. Ma, Randomness generation based on spontaneous emissions of lasers, Physical Review A 91 6 62316, 2015;
\item P. Xu, X. Yuan, L. Chen, H. Lu, X. Yao, X. Ma, Y. Chen, J. Pan, Implementation of a measurement- device-independent entanglement witness; Physical Review Letters 112 14 140506, 2014, (contributed equally)
\item Z. Zhang, X. Yuan, Z. Cao, X. Ma, Round-robin differential-phase-shift quantum key distribution, arXiv:1505.02481, 2015; (Accepted by New J. Phys.)
\item G. Chiribella, X. Yuan, Quantum theory from quantum information: the purification route, Canadian Journal of Physics 91 6 475-478, 2013;
\item X. Yuan, G. Bai, T. Peng, X. Ma, Quantum uncertainty relation of coherence, arXiv:1612.02573,
2016;
\item X. Yuan, Q. Zhao, D. Girolami, X. Ma, Interplay between local quantum randomness and non-
local information access, arXiv:1605.07818, 2016;
\item G. Chiribella, X. Yuan, Measurement sharpness cuts nonlocality and contextuality in every
physical theory, arXiv:1404.3348, 2014.
\end{enumerate}
\end{document}